\documentclass[11pt,letter]{article}
\usepackage{amsmath}
\usepackage{xspace}
\usepackage{todonotes}
\usepackage{amsthm}
\usepackage{fullpage}

\usepackage[utf8]{inputenc}
\usepackage{indentfirst,verbatim} 
\usepackage{listings} 
\usepackage[normalem]{ulem}
\usepackage{soul}
\usepackage{amsfonts}
\usepackage{amsthm}
\usepackage{enumitem}
\usepackage{hyperref}

\usepackage{amsmath,amsfonts,amsthm}
\usepackage{mathbbol}
\usepackage{amsthm}
\usepackage{graphicx,color}
\usepackage{boxedminipage}
\usepackage[linesnumbered, boxed]{algorithm2e}
\usepackage{framed}
\usepackage{thmtools}
\usepackage{thm-restate}
\usepackage{xspace}
\usepackage{todonotes}
\usepackage{pgfplots}
\usetikzlibrary{calc}
\usetikzlibrary{shapes}
\usetikzlibrary{arrows.meta}
\usetikzlibrary{intersections}
 \usepackage{colortbl}
 \usepackage{tcolorbox} 

\usepackage{xifthen}
\usepackage{tabularx}
\usepackage{capt-of}
\usepackage{thmtools}
\usepackage{thm-restate}
\usepackage{subcaption}
\usepackage{booktabs}
\usepackage{calc}
\usepackage{cleveref}

\usepackage[margin=1in]{geometry} 

\theoremstyle{plain}
\newtheorem{theorem}{Theorem}

\newtheorem{corollary}{Corollary}
\newtheorem{lemma}{Lemma}
\newtheorem{claim}{Claim}
\newtheorem{proposition}{Proposition}
\theoremstyle{definition}
\newtheorem{definition}{Definition}
\newtheorem{problem}{Open Question}

\newcommand{\pname}{\textsc}
\newcommand{\polyn}{n^{\Oh(1)}}

\newenvironment{claimproof}{\medskip\noindent \emph{Proof of Claim~\theclaim.}  }{\hfill\claimqed\medskip}

\newlength{\RoundedBoxWidth}
\newsavebox{\GrayRoundedBox}
\newenvironment{GrayBox}[1]%
   {\setlength{\RoundedBoxWidth}{.93\textwidth}
    \def\boxheading{#1}
    \begin{lrbox}{\GrayRoundedBox}
       \begin{minipage}{\RoundedBoxWidth}}%
   {   \end{minipage}
    \end{lrbox}
    \begin{center}
    \begin{tikzpicture}%
       \node(Text)[draw=black!20,fill=white,rounded corners,%
             inner sep=2ex,text width=\RoundedBoxWidth]%
             {\usebox{\GrayRoundedBox}};
        \coordinate(x) at (current bounding box.north west);
        \node [draw=white,rectangle,inner sep=3pt,anchor=north west,fill=white] 
        at ($(x)+(6pt,.75em)$) {\boxheading};
    \end{tikzpicture}
    \end{center}}     

\newenvironment{defproblemx}[2][]{\noindent\ignorespaces%
                                \FrameSep=6pt%
                                \parindent=0pt%
                \vspace*{-1.5em}
                \ifthenelse{\isempty{#1}}{%
                  \begin{GrayBox}{#2}%
                }{%
                  \begin{GrayBox}{#2 parameterized by~{#1}}%
                }
                \begin{tabular*}{\textwidth}{@{\hspace{.1em}} >{\itshape} p{1.8cm} p{0.8\textwidth} @{}}%
            }{
                \end{tabular*}%
                \end{GrayBox}%
                \ignorespacesafterend
            }  

\newcommand{\defparproblema}[4]{
  \begin{defproblemx}[#3]{#1}
    Input:  & #2 \\
    Task: & #4
  \end{defproblemx}
}%

\newcommand{\defparproblemanl}{\defparproblema}

\newcommand{\Oh}{\mathcal{O}}

\newcommand{\cl}{{\sf cl}}

\newcommand{\probDC}{\pname{Long Dirac  Cycle}\xspace}
\newcommand{\probDP}{\pname{Long Dirac  Path}\xspace}
 \newcommand{\probstP}{\pname{Long Erd{\H{o}}s-Gallai  $(s,t)$-Path}\xspace}
\newcommand{\probKPath}{\pname{Longest Path}\xspace}
\newcommand{\probKCycle}{\pname{Longest Cycle}\xspace}

\newcommand{\probTLDP}{\pname{Long $(s,t)$-Cycle}\xspace}
	\newcommand{\gbref}[2]{R_{#1}(#2)}

	\newcommand{\bananadec}{Erd{\H {o}}s-Gallai decomposition\xspace}
\newcommand{\banana}{Erd{\H {o}}s-Gallai component\xspace}
\newcommand{\bananas}{Erd{\H {o}}s-Gallai components\xspace}
	\newcommand{\Bananadec}{Erd{\H {o}}s-Gallai decomposition\xspace}

\newcommand{\cyclebananadec}{Dirac decomposition\xspace}
\newcommand{\cyclebanana}{Dirac component\xspace}
\newcommand{\Cyclebananadec}{Dirac decomposition\xspace}

\DeclareMathOperator{\operatorClassP}{{\sf P}}
\newcommand{\classP}{\ensuremath{\operatorClassP}}
\DeclareMathOperator{\operatorClassNP}{{\sf NP}}
\newcommand{\classNP}{\ensuremath{\operatorClassNP}}

\lstnewenvironment{code}{\lstset{language=Haskell,basicstyle=\small}}{}

 \title{Algorithmic Extensions of Dirac's Theorem\thanks{An extended abstract of this paper is published in the proceedings of SODA 2022 \cite{Fomin2022}. This research was supported by the Research Council of Norway via the project  MULTIVAL (grant no. 263317) and BWCA (grant no. 314528). Kirill Simonov acknowledges support by the Austrian Science Fund (FWF) via project Y1329 (Parameterized Analysis in Artificial Intelligence).}}
\author{
Fedor V. Fomin\thanks{
Department of Informatics, University of Bergen, Norway.}\\fomin@ii.uib.no
\and
Petr A. Golovach\addtocounter{footnote}{-1}\footnotemark{}\\petr.golovach@ii.uib.no
\and
Danil Sagunov\thanks{
    St.\ Petersburg Department of V.A.\ Steklov Institute of Mathematics, Russia.
}\\danilka.pro@gmail.com
\and 
Kirill Simonov\thanks{Algorithms and Complexity Group, TU Wien, Austria}\\kirillsimonov@gmail.com
}

\date{}

\usetikzlibrary{arrows}

\begin{document}

\maketitle	

	\begin{abstract} In 1952, Dirac proved the following theorem about long cycles in graphs with large minimum vertex degrees: 
Every $n$-vertex $2$-connected graph $G$ with minimum vertex degree $\delta\geq 2$ contains a cycle with at least $\min\{2\delta,n\}$ vertices. In particular, if $\delta\geq n/2$, then $G$ is Hamiltonian. The proof of Dirac's theorem is constructive, and it yields an algorithm computing the corresponding cycle in polynomial time. The combinatorial bound of Dirac's theorem is tight in the following sense. There are 2-connected graphs that do not contain cycles of length more than $2\delta+1$. Also, there are non-Hamiltonian graphs with all vertices but one of degree at least $n/2$.  This prompts naturally to the following algorithmic questions. For $k\geq 1$, 
\begin{itemize}
\item[(A)] How difficult is to decide whether a 2-connected graph contains a cycle of length at least $\min\{2\delta+k,n\}$? 
\item[(B)] How difficult is to decide whether a  graph $G$ is Hamiltonian, when 
 at least $n - k$ vertices of $G$ are of degrees at least $n/2-k$?
 \end{itemize}
 The first question  was asked by Fomin,   Golovach,  Lokshtanov,  Panolan,  Saurabh, and
  Zehavi. The second question is due to Jansen,  Kozma, and  Nederlof.
Even for a very special case of $k=1$, the existence of a polynomial-time algorithm deciding whether $G$ contains a cycle of length at least $\min\{2\delta+1,n\}$ was open.
We resolve both questions by proving 
 the following algorithmic generalization of Dirac's theorem: If all but $k$ vertices of a $2$-connected graph $G$ are of degree at least $\delta$, then deciding whether $G$ has a cycle of length at least $\min\{2\delta +k, n\}$ can be done in time $2^{\Oh(k)}\cdot n^{\Oh(1)}$.  
 
The proof of the algorithmic generalization of Dirac's theorem builds on new graph-theoretical results that are interesting on their own. 

\end{abstract}

	\newpage 
\tableofcontents
\newpage


\section{Introduction}\label{sec:intro}



%

The fundamental theorem of Dirac from 1952 guarantees the existence of a Hamiltonian cycle in a graph with a large minimum vertex degree. 

 \begin{theorem}[Dirac~{\cite[Theorem~3]{Dirac52}}] \label{thm:diracs}
 If every vertex of an $n$-vertex graph $G$ is of degree at least $n/2$, then  $G$  is Hamiltonian, that is, contains a Hamiltonian cycle. 
 \end{theorem}
 
\Cref{thm:diracs} follows from a more general statement of Dirac about long cycles in a graph. 

\begin{theorem}[Dirac~{\cite[Theorem~4]{Dirac52}}]\label{thm:circum} 
Every $n$-vertex $2$-connected graph $G$ with  minimum vertex degree $\delta(G)\geq 2$, contains  a cycle with at least $\min\{2\delta(G),n\}$ vertices.
\end{theorem}



Both Dirac's theorems were the first instances of results that developed into one of the core areas in Extremal Graph Theory. One of the  main questions in this research domain is to establish vertex degree characterization of Hamiltonian graphs and conditions enforcing long paths or cycles in graphs. The (very) incomplete list of results in this area includes the classical theorems of Erd{\H{o}}s  and Gallai \cite{ErdosG59},
Ore~\cite{ore60}, Bondy and Chv\'{a}tal~\cite{bondyC76}, P\'{o}sa \cite{posa62},
 Meyniel \cite{Meyniel73}, and Bollob\'{a}s and Brightwell \cite{bollobasB93}, see also the  Wikipedia entry on the Hamiltonian path.\footnote{\url{https://en.wikipedia.org/wiki/Hamiltonian_path}} The chapters of Bondy \cite{MR1373656} and Bollob\'{a}s \cite{MR1373679} in the Handbook of Combinatorics, as well as Chapter~3 in the Extremal Graph Theory book \cite{MR506522} provide excellent introduction to this important part of graph theory. The survey of Li  \cite{Li13survey} is a comprehensive (but a bit outdated) overview of the area.   After almost 70 years, the field remains active, see for example the very recent proof of the Woodall's conjecture by Li and Nung \cite{MR4140611}.
 
Computing long cycles and paths  is also an important topic  in parameterized complexity. 
  It served as a test-bed for developing several fundamental algorithmic techniques including the color coding  of  Alon, Yuster and Zwick \cite{AlonYZ95}, the algebraic approaches of Koutis and Williams \cite{Koutis08,Williams09}, matroids-based methods 
\cite{FominLS14},  and the determinants-sum technique of Bj{\"{o}}rklund from his FOCS 2010 Test of Time Award paper
\cite{DBLP:journals/siamcomp/Bjorklund14}. We refer to   \cite{FominK13},  \cite{KoutisW16}, and 
\cite[Chapter~10]{cygan2015parameterized} for an overview of algorithmic  ideas and techniques developed for  computing long paths and cycles in   graphs.

Despite the tremendous progress in graph-theoretical and algorithmic studies of longest cycles, all the developed tools do not answer the following natural and ``innocent'' question. 
By \Cref{thm:circum}, deciding whether a $2$-connected graph $G$ contains a cycle of length at least $\min\{2\delta(G),n\}$ can be trivially done in polynomial time by checking degrees of all vertices in $G$. 
  \begin{tcolorbox}[colback=green!5!white,colframe=blue!40!black]
 \textbf{Question~1:} Is there a polynomial time algorithm to decide whether a $2$-connected graph $G$ contains a cycle of length at least 
 $\min\{2\delta(G) +1,n\}$?
 \end{tcolorbox}
The methods developed in the extremal Hamiltonian graph theory do not answer this question.  
The combinatorial bound in \Cref{thm:circum} is known to be sharp; that is, there exist graphs that have no cycles of length at least $\min\{2\delta(G)+1,n\}$. Since the extremal graph theory studies the existence of a cycle under certain conditions, such type of questions are beyond its applicability. The   techniques of  parameterized algorithms do not seem to be much of use here either. Such algorithms compute a cycle of length at least $k$ in time $2^{\Oh(k)} \cdot n^{\Oh(1)}$, which in our case is $2^{\Oh(\delta(G))}\cdot n^{\Oh(1)}$. Hence  when  $\delta(G)$ is, for example, at least $n^{1/100}$, these algorithms  do not run in polynomial time.

Similarly,  the existing methods do not answer the question about another ``tiny algorithmic step'' from Dirac's theorem, what happens when all vertices of $G$ but  one are of large degree?

  \begin{tcolorbox}[colback=green!5!white,colframe=blue!40!black]  \textbf{Question~2:} Let $v$ be a vertex of the minimum degree of  a $2$-connected graph $G$.
 Is there a polynomial time algorithm to decide whether $G$ contains a cycle of length at least 
 $\min\{2\delta(G-v),n\}$?
 \end{tcolorbox}
(We denote by $G-v$  the induced subgraph of $G$ obtained by removing vertex  $v$.)  Note that  graph  $G-v$ is not necessarily   2-connected and we cannot apply \Cref{thm:circum}   to it. 

The incapability of existing techniques to answer Questions 1 and 2 was the primary motivation for our work.
We answer both questions affirmatively and in a much more general way. Our result implies that in polynomial time one can decide whether $G$ contains a cycle of length at least $2\delta(G-B)+k$ for $B\subseteq V(G)$ and $k\geq 0$ as long as $k+ |B| \in \Oh(\log{n})$. (We denote by $G-B$ the induced subgraph of $G$ obtained by removing vertices of $B$.) To state our result more precisely, we define the following problem.

%
%
%

\defparproblema{\probDC}%
{A graph $G$ with vertex set $B\subseteq V(G)$ and an integer $k\ge 0$.}%
{$k+|B|$}
{Decide whether $G$ contains a cycle of length at least $\min\{2\delta(G-B), |V(G)|-|B|\}+k$.
}

In the definition of \probDC we use the minimum of two values for the following reason. The question whether an $n$-vertex graph 
 $G$ contains a cycle of length at least $ 2\delta(G-B)+k$ is meaningful only for $\delta(G-B)\leq n/2$. Indeed,  for   $\delta(G-B)>n/2$,    $G$ does not  contain a cycle of length at least  $2\delta(G-B)+k>n$. However, even when  $\delta(G-B)>n/2$, deciding whether $G$ is Hamiltonian, is still very intriguing.  By taking the minimum of the two values, we  capture both  interesting situations. 
%

The   main result of the paper is the following theorem providing an algorithmic generalization of Dirac's theorem.

\begin{theorem}[\textbf{Main Theorem}]\label{theorem:main}On an $n$-vertex 2-connected graph $G$,  
\probDC is solvable in time $2^{\Oh (k+|B|)} \cdot n^{\Oh (1)}$.
\end{theorem}


In other words, \probDC is fixed-parameter tractable parameterized by $k+|B|$ and the dependence on the parameters is single-exponential. 
This dependence is asymptotically optimal up to the Exponential Time Hypothesis (ETH) of Impagliazzo, Paturi,  and Zane~\cite{ImpagliazzoPZ01}. Solving \probDC in time $2^{o (k)} \cdot n^{\Oh (1)}$ even with $B = \emptyset$ yields  recognizing  in time $2^{o(n)}$ whether a graph is Hamiltonian. A subexponential algorithm deciding Hamiltonicity would fail ETH. In \Cref{thm:hard_on_B} we show that solving \probDC in time $2^{o (|B|)} \cdot n^{\Oh (1)}$ even for $k = 1$ would contradict ETH as well.
It is also \classNP-complete to decide whether a $2$-connected graph $G$ has a cycle of length at least $(2+\varepsilon)\delta(G)$ for any $\varepsilon> 0$ (\Cref{thm:tightness}).

The 2-connectivity requirement in the statement of the theorem is  important---without it \probDC is already \classNP-complete for $k=|B|=0$. Indeed, for an $n$-vertex graph $G$ construct a graph $H$ by attaching to each vertex of $G$  a clique of size $n/2$. Then $H$ has a cycle of length at least $2\delta (H)\geq n$ if and  only if 
$G$ is Hamiltonian.
However, when instead of a cycle we are looking for a long path, the 2-connectivity requirement could be omitted. More precisely, consider the following problem. 

\defparproblema{\probDP}%
{A graph $G$ with vertex set $B\subseteq V(G)$ and an integer $k\ge 0$.}%
{$k+|B|$}
{Decide whether $G$ contains a path of length at least $\min\{2\delta(G-B), |V(G)|-|B|-1\}+k$.
}
 
 Theorem~\ref{theorem:main} yields the following. 
 \begin{corollary}\label{theorem:main_path}
 On a connected $n$-vertex graph $G$,  
\probDP is solvable in time $2^{\Oh (k+|B|)} \cdot n^{\Oh (1)}$.
\end{corollary}

Indeed, when $G$ is connected, the graph $G+v$, obtained by adding a vertex $v$ and making it adjacent to all vertices of the graph, is $2$-connected. The minimum vertex degree of $G+v$ is equal to $\delta(G)+1$, and $G$ has a path of length at least $t$ if and only if $G+v$ has a cycle of length at least $t+2$.

\Cref{theorem:main}  answers  several open questions from the literature. Fomin, Golovach, Lokshtanov, Panolan, Saurabh and Zehavi in~\cite{fomin_et_al:LIPIcs:2019:11168} asked about 
 the parameterized complexity of   problems (with parameter $k$) where for a given   (2-connected) graph $G$ and $k\geq 1$, the task is to check whether $G$ has a path (cycle) with at least $2\delta(G)+k$ vertices. By \Cref{theorem:main} and Corollary~\ref{theorem:main_path}  (the case $B=\emptyset$), both problems are fixed-parameter tractable.

 Jansen, Kozma, and Nederlof  in \cite{DBLP:conf/wg/Jansen0N19} 
  conjectured that if 
 at least $n - k$ vertices of graph $G$ are of  degree at least $n/2-k$, then deciding whether $G$ contains a Hamiltonian cycle 
 can be done in time  $2^{\Oh (k)} \cdot n^{\Oh (1)}$.  \Cref{theorem:main} resolves this conjecture. Indeed, if $G$ is Hamiltonian, it is $2$-connected. Then let $B$, $|B|\leq k$,  be the set of vertices such that every vertex from $V(G)\setminus B$ is of degree (in $G$) at least $n/2-k$. Then $\delta(G-B)\geq n/2 -k - |B|\geq n/2 -2k$.
$n-|B|\leq n-2\delta(G-B)$, 
we put $k'=|B|$,  otherwise  we put $k'=n-2\delta(G-B)$. Note that because  $2\delta(G-B)\geq n -4k$, in both cases we have that $k'\leq 4k$.  Also by the choice of $k'$, 
 $\min\{2\delta(G-B), n-|B|\}+k'=n$ and hence 
 $G$ has a cycle of length at least  $\min\{2\delta(G-B), n-|B|\}+k'$
 if and only if $G$ is Hamiltonian. 
 By \Cref{theorem:main}, deciding whether $G$ has a cycle of length at least 
 $\min\{2\delta(G-B),n-|B|\}+k'$  can be done in time  $2^{\Oh (k'+|B|)} \cdot n^{\Oh (1)}=2^{\Oh (k)} \cdot n^{\Oh (1)}$.
Interestingly, while the conjecture of Jansen, Kozma and Nederlof follows from the statement of \Cref{theorem:main}, to prove the theorem, we need to resolve this conjecture directly.

 
We state  \Cref{theorem:main}  for the decision variant  of the problem. However, the proof  is constructive and the corresponding cycle can be found within the same running time.  Note that standard self-reduction arguments are not applicable here because deleting or contracting edges could change the minimum vertex degree.
 
 
 \medskip\noindent\textbf{Related work.}
Until very recently,  graph-theoretical and algorithmic studies  of  the longest paths and cycles  coexisted in parallel universes without almost any visible interaction. 
In 1992,  H\"{a}ggkvist \cite{Hagvist92}, as a corollary of his structural theorem, provided an algorithm that decides in time $n^{\Oh(k)}$ whether a graph with the minimum vertex degree at least $n/2-k$ is Hamiltonian. 
 In 2019,  Jansen, Kozma, and Nederlof  in \cite{DBLP:conf/wg/Jansen0N19}    gave two algorithms of running times  $2^{\Oh (k)} \cdot n^{\Oh (1)}$ that decide whether the input graph $G$ is Hamiltonian when either the minimum vertex degree of $G$ is at least $n/2-k$ or  at least $n-k$ vertices of $G$ are of degree at least $n/2$.  The first result of Jansen, Kozma, and Nederlof strongly improves the algorithm of H\"{a}ggkvist.  
  However, the methods they use, like the structural theorem of  H\"{a}ggkvist \cite{Hagvist92},  are specific for Hamiltonicity and are not applicable for the more general 
 problem of computing the longest cycle. Second, their parameterized algorithms   work only in one of the scenarios: either when all vertices are of degree at least   $n/2-k$ or  when at least $n-k$ vertices are of degree at least $n/2$. Whether  both scenarios could be combined, that is, the existence of a parameterized algorithm deciding Hamiltonicity when $n-k$ vertices are of degree at least $n/2-k$, was left open. 

   Fomin, Golovach, Lokshtanov, Panolan, Saurabh and Zehavi in~\cite{fomin_et_al:LIPIcs:2019:11168}  gave an algorithm that in time $2^{\Oh (k)} \cdot n^{\Oh (1)}$ decides whether a $2$-connected graph $G$ contains a cycle of length at least $d+k$, where $d$ is the degeneracy of $G$.  Since the minimum vertex degree $\delta(G)$ does not exceed the degeneracy of $G$, this result also implies an algorithm for finding a cycle of length at least  $\delta(G)+k$ in $2$-connected graphs. 

None of the works  \cite{DBLP:conf/wg/Jansen0N19} and  \cite{fomin_et_al:LIPIcs:2019:11168}  could be used to address Questions~1 and~2, the very special cases of \Cref{theorem:main}. 

%
%
 More generally,   \Cref{theorem:main}  fits into a popular trend in parameterized complexity called 
 ``above guarantee'' parameterization. The general idea of this paradigm is that the natural parameterization of, say,  a maximization problem by the solution size is not satisfactory if there is a lower bound for the solution size that is sufficiently large. For example, there always exists a satisfying assignment that satisfies half of the clauses or there is always a max-cut  containing at least half the edges. Thus nontrivial solutions occur only for the values of the parameter that are above the lower bound. This indicates that for such cases, it is more natural to parameterize the problem by the difference of the solution size and the bound. Since the work of 
Mahajan and Raman~\cite{MahajanR99} on  \textsc{Max Sat} and \textsc{Max Cut}, the above guarantee approach was 
 successfully applied to various problems, see e.g.~\cite{AlonGKSY10,CrowstonJMPRS13,GargP16,DBLP:journals/mst/GutinKLM11,GutinIMY12,GutinP16,GutinRSY07,LokshtanovNRRS14,MahajanRS09}. 
In particular, \cite{BezakovaCDF17}  and \cite{fomin_et_al:LIPIcs:2020:11724} study the longest path above the shortest $s,t$-path and the girth of a graph.

%

\section{Overview of the proof}\label{section:overview}
The original proof of Dirac is not constructive because it does not provide any procedure for constructing a cycle of length at least $2\delta(G)$. There are algorithmic proofs of Dirac's theorem; see, e.g., the thesis of Locke \cite{locke1983extremal}. The idea of Locke's proof that also provides a polynomial-time algorithm for constructing a cycle of length at least $2\delta(G)$ is to start from some cycle and to grow by inserting new vertices and short paths. Thanks to the conditions on the graph's degrees, such a procedure always constructs a cycle of the required length. 
On a very general level, our proof of Theorem~\ref{theorem:main} uses the same strategy. For an instance $(G,B,k)$ of \probDC, we try to grow a cycle iteratively. However, enlarging the cycle by ``elementary'' improvements could get stuck with a cycle of length significantly smaller than
$\min\{2\delta(G-B), |V(G)|-|B|\}+k$. It appears that the cycles that cannot be improved by ``elementary'' operations induce a very particular structure in a graph.
These structural theorems play a crucial role in our algorithm. 

The main technical contribution is the new graph decomposition that we call \cyclebananadec. The formal definition is given in Section~\ref{sec:bananas}.   \cyclebananadec
  is defined for a cycle $C$ in $G$. Let $C$ be a cycle of length less than $2\delta(G-B) +k$. 
Informally, the components of Dirac decomposition are connected components in $G-V(C)$. (For an intuitive description of the decomposition, we will assume that $B=\emptyset$. Handling  vertices of $B$ requires more technicalities---we have to refine the graph and work with its refinement.)
Since $G$ is $2$-connected, we can reach $C$ by a path starting in such a component in $G$.
One of the essential properties of \cyclebananadec is a limited number of vertices in $V(C)$ that have neighbors outside of $C$.
In fact, we can choose two short paths $P_1$ and $P_2$ in $C$ (and short means that their total length is of order $k$) such that all connections between 
connected components of $G-V(C)$ and $C$ go through $V(P_1)\cup V(P_2)$. The second important property is that each connected component of $G-(V(P_1)\cup V(P_2))$ is connected with $P_i$ in $G$ in a very restricted way: The maximum matching size between its vertex set and the vertex set of $P_i$ is at most one.

\cyclebananadec appears to be very useful for algorithmic purposes. For a cycle $C$ and a vertex set $B$, given a \cyclebananadec for $C$ and $B$, in time $2^{\Oh (k+|B|)} \cdot n^{\Oh (1)}$ we either solve the problem or succeed in enlarging $C$
 (Theorem~\ref{thm:cyclebanana}). We also provide an algorithm that either constructs a \cyclebananadec in polynomial time or obtains additional structural information that again can be used either to solve the problem or to enlarge the cycle. More precisely,  
first, we need to eliminate the ``extremal'' cases. When $\delta(G-B)\in \Oh(k)$, the classical result of Alon, Yuster, and Zwick \cite{AlonYZ95} solves the problem in time
$2^{\Oh (k+|B|)} \cdot n^{\Oh (1)}$. Another extremal case is when $|B|\le k$ and $\delta(G - B) \ge \frac{n}{2}-k$. In that case, for solving \probDC, we have to decide in time $2^{\mathcal{\Oh}(k)}\cdot n^{\Oh(1)}$ whether $G$ is almost Hamiltonian, i.e., a cycle in $G$ that   cover all but $\Oh(k)$ vertices. The existence of such an algorithm for Hamiltonian cycles was conjectured in 
  \cite{DBLP:conf/wg/Jansen0N19}  and Theorem~\ref{theorem:hamiltonian} settles this conjecture. We give an overview of the proof of  Theorem~\ref{theorem:hamiltonian} later in this section. If we are in none of the extremal cases, then  (Lemma~\ref{lemma:main_cycle_lemma}) in polynomial time we can either (a) enlarge the cycle $C$, or (b) compute a vertex cover of $G-B$ of size at most $\delta(G-B)+2k$, or (c) compute a \cyclebananadec. In cases (a) and (c), we can proceed iteratively. For the case (b) we give an algorithm that solves the problem in time $2^{\Oh(k+|B|)}\cdot\polyn$ (Theorem~\ref{thmVCad}). 

 The most critical and challenging component of the proof is Theorem~\ref{thm:cyclebanana} about algorithmic properties of \cyclebananadec.  
 We use the properties of \cyclebananadec to show that an enlargement of a cycle $C$ of length at most $2\delta(G-B)+k-1$ can be done in a very particular way. By an extension of Dirac's existential theorem, Theorem~\ref{thm:relaxed_long_cycle}, we can assume that $C$ is of length at least $2\delta(G-B)$.  The most interesting and not-trivial situation that could occur is that for some vertices $x\in V(P_1)$ and $y\in V(P_2)$, we replace the shortest $(x,y)$-path in $C$ by a detour with a particular property. This detour leaves $x$, moves to a vertex $s$ of some $2$-connected component of $G-V(C)$, visits some vertices in this component, leaves it from a vertex $t$, and goes to vertex $y$. Since the length of the longest $(x,y)$-path in $C$ is at least $\delta(G-B)$, to decide whether such a detour exists, it is sufficient to solve the following problem.  
 For vertices $s,t$ of a $2$-connected graph $G$, decide whether $G$ contains an $(s,t)$-path of length at least $\delta(G-B)+k$. We give an algorithm solving this problem in time $2^{\mathcal{O}(k+|B|)}\cdot n^{\mathcal{O}(1)}$ (Theorem~\ref{thmEG}). The combinatorial bound that an $(s,t)$-path of length $\delta(G)$ always exists if $G$ is $2$-connected, is the classical theorem of Erd{\H{o}}s and Gallai \cite[Theorem~1.16]{ErdosG59}. Because of that, we name the problem of computing an $(s,t)$-path of length at least $\delta(G-B)+k$ by the  
\probstP problem. \probstP is an interesting problem on its own, and to prove Theorem~\ref{thmEG}, we use another structural result which we call \bananadec. Similar to \cyclebananadec, this decomposition is very useful from the algorithmic perspective. In Section~\ref{sec:erdos-gallaiPath}, we define this decomposition, provide efficient algorithms for constructing it, and use it to solve \probstP. Another interesting component of the solution to  \probstP is the algorithm for computing the longest cycle passing through two specified vertices (Theorem~\ref{thmTLDP}). 
We are not aware of the previous work in parameterized algorithms on this natural problem. 
  
Figure~\ref{fig:proofsketch} displays the most important steps of the proof and the dependencies between them. In the remaining part of this section, we highlight the  ideas behind each of the auxiliary steps  (Theorems~\ref{thmTLDP}, \ref{thmEG}, \ref{thmVCad}, and~\ref{theorem:hamiltonian})  in the proofs of  Theorem~\ref{theorem:main} and Theorem~\ref{thm:cyclebanana}. 

\medskip
		
%
%
%
%
%

%
%


\begin{figure}
	\centering
	\begin{tikzpicture}[every text node part/.style={align=center}]
		\tikzstyle{problem}=[rectangle, draw, rounded corners]
		\tikzstyle{problemM}=[rectangle, fill=black!10,draw=red, rounded corners]
		
		\tikzstyle{implies}=[-latex']
		\node [problem] (two-paths) at (-3, 0) {\probTLDP\\ (Theorem~\ref{thmTLDP})};
		
		\ifdefined\STOC
		\coordinate (st-path-C) at (-2.1,-1.5);
		\else
		\coordinate (st-path-C) at (-3, -1.5);
		\fi		
		
		\node [problem] (st-path) at (st-path-C) {\probstP \\ (Theorem~\ref{thmEG})};
		
		\ifdefined\STOC
		\coordinate (diracdec-C) at (-2.7,-3);
		\else
		\coordinate (diracdec-C) at (-3, -3);
		\fi		

		\node [problem] (diracdec) at (diracdec-C) {\cyclebananadec \\ (Theorem~\ref{thm:cyclebanana})};

		\ifdefined\STOC
		\coordinate (small-vc-C) at (1,0);
		\else
		\coordinate (small-vc-C) at (6.2,0);
		\fi
		
		\node [problem] (small-vc) at (small-vc-C)	{\textsc{\probDC\ / Vertex Cover} \\\textsc {Above Degree}\\ (Theorem~\ref{thmVCad})};
		

		\ifdefined\STOC
		\coordinate (hamiltonian-C) at (1.7,-3);
		\else
		\coordinate (hamiltonian-C) at (6.2,-3);
		\fi

		\node [problem] (hamiltonian) at (hamiltonian-C) {\textsc{Almost Hamiltonian} \\ \textsc{Dirac Cycle} \\(Theorem~\ref{theorem:hamiltonian})};
		
		\ifdefined\STOC
		\coordinate (dirac-cycle-C) at (-1,-4.5);
		\else
		\coordinate (dirac-cycle-C) at (2,-4.5);
		\fi
		
		\node [problemM] (dirac-cycle) at (dirac-cycle-C) {\probDC \\(Theorem~\ref{theorem:main})};

		\draw [implies] (two-paths) -- (st-path);
		\draw [implies] (small-vc) -- (hamiltonian);
		\draw [implies] (hamiltonian) -- (dirac-cycle);
		\draw [implies] (st-path) -- (diracdec);
		\draw [implies] (small-vc) -- (dirac-cycle);
	 	\draw [implies] (diracdec) -- (dirac-cycle);
	\end{tikzpicture}
	\caption{The main steps  and connections  in  the proof of Theorem~\ref{theorem:main}.}\label{fig:proofsketch}
	\end{figure}
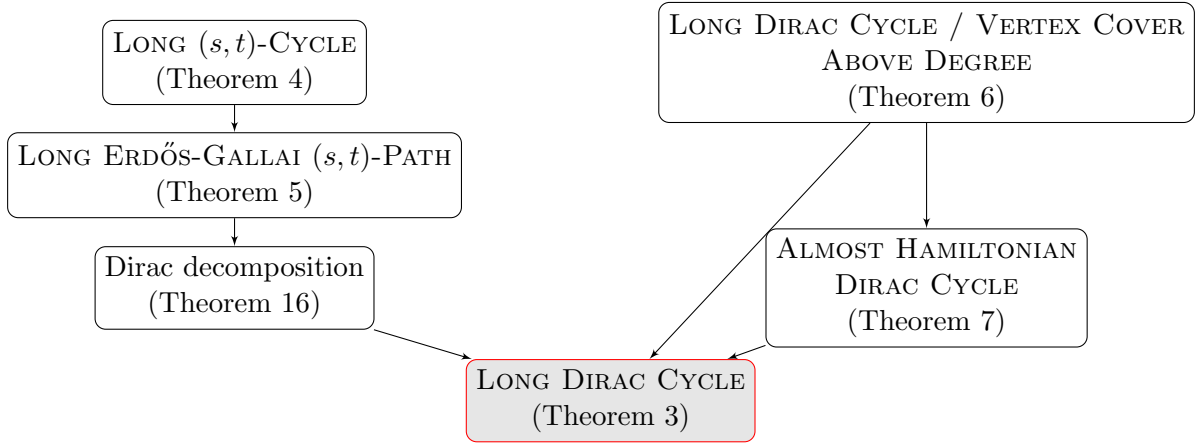

The first auxiliary problem whose solution we use in  the proof of Theorem~\ref{thmEG} is the following. 



 \defparproblema{\probTLDP}%
{A graph $G$ with two vertices $s,t\in V(G)$ and an integer $k \ge 0$.}%
{$k$}
{Decide whether there is a cycle in $G$ of length at least $k$ that passes through $s$ and $t$.
}

When $s\neq t$, an equivalent formulation is to decide whether $G$ contains two internally disjoint $(s,t)$-paths of total length at least $k$.
%
%
%
 In Section~\ref{sec:tldp} we prove the following theorem.  
\begin{theorem}\label{thmTLDP}   \probTLDP   is solvable in time 
	 $2^{\mathcal{O}(k)}\cdot n^{\mathcal{O}(1)}$.
\end{theorem}
While the first idea to design an algorithm claimed in Theorem~\ref{thmTLDP} would be to use the color coding technique of Alon, Yuster and Zwick \cite{AlonYZ95}, this idea does not work directly. The reason is that color coding can be used only to  find in the claimed running time the cycle whose length is of order of $k$. However, it is quite possible that  the lengths of all solutions are much larger than $k$; in such situation color coding cannot be applied directly. Our approach in proving  Theorem~\ref{thmTLDP} 
builds on  ideas from~\cite{DBLP:journals/ipl/FominLPSZ18,Zehavi16}, where a parameterized algorithms for finding a directed  $(s,t)$-path and a directed cycle of length at least $k$ were developed.   
The main idea of the proof is the following. First, we use color coding to verify whether the considered instance has a solution composed by two  $(s,t)$-paths of total length at most $3k$. If the instance has a solution, we return it and stop. Otherwise, we conclude that the total length of the paths of every solution is at least $3k+1$. This allows to use structural properties of paths. Let $P_1$ and $P_2$ be the $(s,t)$-paths of a solution of minimum total length. Then there are vertices $x_1$ and $x_2$ on $P_1$ and $P_2$, respectively, such that (i) the total length of the $(s,x_1)$-subpath $P_1'$ of $P_1$ and the $(s,x_2)$-subpath $P_2'$ of $P_2$ is exactly $k$, (ii) either $x_1=s$ or the length of the $(x_1,t)$-subpath $P_1''$ of $P_1$ is at least $k$, and, symmetrically, (iii) either $x_2=s$ or  the length of the $(x_2,t)$-subpath $P_2''$ of $P_2$ is at least $k$. Then $P_1''$  and  $P_2''$  are internally disjoint paths that are shortest disjoint paths avoiding $V(P_1')\cup V(P_2')\setminus\{x_1,x_2\}$. We use the method of random separation  to distinguish the following three sets: $V(P_1')\cup V(P_2)\setminus\{x_1,x_2\}$, the last $\min\{k,|V(P_1)|-2\}$ internal vertices of $P_1''$, and the last $\min\{k,|V(P_2)|-2\}$ internal vertices of $P_2''$. This allows to highlight the crucial parts of the shortest solution and then find a solution.


The second problem whose solution we use in the proof of  Theorem ~\ref{theorem:main},  comes from another classical theorem due to Erd{\H{o}}s and Gallai from \cite[Theorem~1.16]{ErdosG59}, see also \cite{locke1985generalization}.
%
 	For every pair of vertices $s,t$ of a 2-connected graph $G$,  there is a path of length at least $\min_{v\in V(G)\setminus \{s,t\}}\deg v$. The proof of this result is constructive, and it implies a polynomial time algorithm that finds such a path.     We define \probstP as follows.
	
	\medskip
 \defparproblema{\probstP}%
 {A graph $G$ with a vertex set $B\subseteq V(G)$, two vertices $s,t\in V(G)$ and an integer $k \ge 0$.}%
 {$k+|B|$}
 {Decide whether $G$ contains an $(s,t)$-path of length at least $\delta(G-B)+k$.
 }
	
In Section~\ref{sec:erdos-gallaiPath}, we prove the following theorem.  This theorem plays an important role in the proof of Theorem~\ref{thm:cyclebanana}. 
\begin{theorem}\label{thmEG} 
 \probstP     is solvable in time 
  $2^{\mathcal{O}(k+|B|)}\cdot n^{\mathcal{O}(1)}$
on $2$-connected graphs.
\end{theorem}
Similar to \probDP, the requirement that the input graph is $2$-connected is important. It is easy to prove that  \probstP is NP-complete for $k=|B|=0$ when the input graph is not $2$-connected.

To prove Theorem~\ref{thmEG}, we apply the following strategy. We take an $(s,t)$-path $P$ and try to extend it as much as possible. The principal tool in enlarging the path 
 $P$ is Corollary~\ref{thm:relaxed_st_path}, which is an extension of the theorem of Erd{\H{o}}s and Gallai that takes into account the vertices of $B$. In the extremal case, when we cannot extend the path anymore, we obtain a graph decomposition whose properties become useful from the algorithmic perspective. We call this decomposition by the name of \bananadec and prove that, in that case, the graph can be decomposed in a very particular way. Roughly speaking, after a certain refinement of the graph, the $(s,t)$-path $P$ consists of a prefix $P_1$ and a suffix $P_2$ with the following properties. These parts of the path are sufficiently far from each other in $P$. Moreover, all components of the graph $G-V(P_1\cup P_2)$, we call them \banana, are connected to $P_1$ and $P_2$ in a very restricted way. Such a graph-theoretical insight helps us to characterize how a long $(s,t)$-path traverses through an \banana. This property allows us to design the recursive algorithm that proves the theorem. 



\medskip
The next auxiliary result required for the proof of Theorem~\ref{theorem:main}, concerns \probDC parameterized by the vertex cover of a graph. It is well-known, see e.g., \cite{CyganFKLMPPS15}, that a longest path in a graph $G$ could be found in time $2^{\Oh(t)}n^{\Oh(1)}$, where $t$ is the size of the minimum vertex cover of $G$. However, we need a much more refined result for  the proof of the main theorem, where the parameter is not just the size of the vertex cover, but the difference between that size and $\delta(G - B)$. We define the following parameterized problem.

\medskip
\defparproblemanl{\textsc{\probDC\ / Vertex Cover Above Degree}}%
{A graph $G$ with a vertex set $B\subseteq V(G)$, a vertex cover $S$ of $G$ of size $\delta(G-B)+p$ and an integer $k\ge 0$.}%
{$p+|B|$}
{Decide whether $G$ contains a cycle of length at least $2\delta(G-B)+k$.
}

Section~\ref{sec:vcalgo} is devoted to the proof of the following theorem, which we need for both  Theorem~\ref{theorem:hamiltonian} and Theorem~\ref{theorem:main}.

\begin{theorem}\label{thmVCad}
	\textsc{\probDC\ / Vertex Cover Above Degree} is solvable in $2^{\Oh(p+|B|)}\cdot n^{\Oh(1)}$ running time.
 \end{theorem}
 To prove Theorem~\ref{thmVCad}, we establish the new structural result, Lemma~\ref{thm:path_cover}. The lemma reduces the crucial case of the problem about the long cycle to a particular path cover problem. This equivalence becomes very handy because  we can use color-coding to compute the particular path cover, and thus by the lemma, to compute a long path. In spirit,  Lemma~\ref{thm:path_cover} is close to  the classical theorem of Nash-Williams~\cite{NashWil71}, stating that  a 2-connected graph $G$ with $\delta(G) \ge (n + 2)/3$ is either Hamiltonian or contains an independent set of size $\delta(G) + 1$. An extension of this theorem is due to H\"{a}ggkvist \cite{Hagvist92}, which was used by Jansen, Kozma and Nederlof \cite{DBLP:conf/wg/Jansen0N19} in their algorithm for Hamiltonicity below Dirac's condition. In our case, we cannot use the structural theorem of H\"{a}ggkvist as a black box, and build on the new graph-theoretic lemma instead.   
  
\medskip
The last ingredient we need to prove Theorem~\ref{theorem:main}, is its special case when the minimum degree of $\delta(G - B)$ is nearly $\frac{n}{2}$.
Specifically, the problem is defined as follows.

\defparproblema{\textsc{Almost Hamiltonian Dirac Cycle}}%
{A graph $G$, integer $k\ge 0$ and a vertex set $B \subset V(G)$, such that $|B| \le k$ and $\delta(G - B) \ge \frac{n}{2}-k$.}%
{$k$}
{Find the longest cycle in $G$.
}
Observe that for a 2-connected graph, Theorem~\ref{thm:relaxed_long_cycle} always gives a cycle of length $2\delta(G - B) \ge n - 2k$. Thus it is more natural to state the problem in the form above, as the length of the longest cycle is necessarily between $n - 2k$ and $n$, which is at most $2 \delta(G - B) + 2k$. In other words, we look for an almost Hamiltonian cycle, in a sense that it does not cover only $\Oh(k)$ vertices. Now we state our result for \textsc{Almost Hamiltonian Cycle} that we prove in Section~\ref{sec:HamCycles}.

\begin{restatable}{theorem}{theoremhamiltonian}
		\label{theorem:hamiltonian}
		Let $G$ be a given 2-connected graph on $n$ vertices and let $k$ be a given integer.
		Let $B\subseteq V(G)$ be such that $|B|\le k$ and $\delta(G - B) \ge \frac{n}{2}-k$.
		There is a $2^{\mathcal{O}(k)}\cdot n^{\mathcal{O}(1)}$ running time algorithm that finds the longest cycle in $G$.
	\end{restatable}


    The key obstacle for proving the theorem is the low-degree set $B$, since for empty $B$, we could simply apply the Nash-Williams theorem~\cite{NashWil71} and obtain either a Hamiltonian cycle or an independent set of size $\delta(G) + 1$, and in the latter case use our result for \textsc{\probDC\ / Vertex Cover Above Degree}. Assume there exists a Hamiltonian cycle in $G$ (for almost Hamiltonian cycles the algorithm is similar), it induces a certain path cover of the vertices of $B$, where the endpoints of paths belong to $V(G)\setminus B$, and their total length is $\Oh(k)$. Such a path cover can be found by color-coding and dynamic programming in time $2^{\Oh(k)} n ^{\Oh(1)}$. Now either the rest of the graph is not 2-connected, and we have a $\Oh(k)$-sized separator, or we can apply the Nash-Williams theorem and obtain a cycle covering everything except the path cover, or a large independent set. The latter case is dealt with by Theorem~\ref{thmVCad}, and for the case of the small separator we design a special algorithm that leverages the fact that the resulting components are very dense. So the main case is when the graph splits into a long cycle and the path cover. Now we crucially use that the paths in the path cover start and end outside of $B$, thus the endpoints of a path have high degree, each of them sees roughly half of the vertices of the long cycle. This makes it ``hard'' to not be able to insert the path somewhere in the cycle and make it longer. However, this last intuitive idea is achieved by a very intricate case analysis that constitutes the most of technical difficulty of the proof. 
    Also, in some of the cases, we cannot make the cycle longer nor conclude that it is impossible, but instead we are able to find either a small separator or a large independent set. Again, we settle these cases by using the respective specialized algorithms.

\section{Preliminaries and classical theorems}\label{section:prelim}

\medskip\noindent\textbf{Graph notation.} Most of the  graph notation that we use  here are standard and are compatible  with the notation used in the textbook of Diestel~\cite{Diestel}.
Graphs in this paper are finite and undirected.
The vertex set of a graph $G$ is denoted by $V(G)$ and the edge set of $G$ is denoted by $E(G)$.  We use shorthands $n=|V(G)|$ and $m=|E(G)|$. An edge of an undirected graph with endpoints $u$ and $v$ is denoted by $uv$.

Graph $H$ is a {\em{subgraph}} of graph $G$ if $V(H)\subseteq V(G)$ and $E(H)\subseteq E(G)$.  
For a subset $S\subseteq V(G)$, the subgraph of $G$ {\em{induced}} by $S$ is denoted by $G[S]$; its vertex set is $S$ and its edge set consists of all the edges of $E(G)$ that have both endpoints in $S$. 
For $B\subseteq V(G)$, we use $G-B$ to denote the graph $G[V\setminus B]$, and for $F\subseteq E(G)$ by $G-F$ we denote the graph $(V(G),E(G)\setminus F)$. We also write $G- v$ instead of $G- \{v\}$.

For graph $G$ and edge $uv\in E(G)$, by {\em{contracting}}  edge $uv$ we mean the following operation. We remove $u$ and $v$ from the graph, introduce a new vertex $w_{uv}$, and connect it to all the vertices $u$ or $v$ were adjacent to.  
The \emph{neighborhood}  of a vertex $v$ in $G$ is
$N_G(v)=\{u\in V~|~  uv\in E(G)\}$ and the \emph{closed neighborhood} of $v$ is $N_G[v]=N_G(v)\cup \{v\}$.
For a vertex set $S\subseteq V$,  we define $N_G[S]=\bigcup_{v \in S} N[v]$ and $N_G(S)=N_G[S]\setminus S$. We denote by $\deg_G(v)$ the \emph{degree}  of a vertex $v$ in graph $G$, which is just the number of edges incident with $v$. We may omit indices if the graph under consideration is clear from the context.  
We use $\delta(G)$ for minimum vertex degree of graph $G$.

A {\em{path}} $P$ in a graph  is a nonempty sequence of vertices $v_0,\ldots,v_{k}$ such that for every $i=0,\ldots,k-1$ we have $v_iv_{i+1}\in E(G)$ and $v_i\neq v_j$ for all $i\neq j$. Vertices $v_0$ and $v_k$ are the \emph{endpoints} of path $P$ and $v_1, \dots, v_{k-1}$ are \emph{internal}.
If $P=v_0v_1\dots v_k$ is a path, then the graph obtained from $P$ by adding edge $x_k x_0$ is a cycle. The {\em{length}} of a path or cycle is equal to the cardinality of its edge set. The \emph{distance}  between vertices $u$ and $v$ in a graph $G$  is the shortest length of a path between $u$ and $v$. For vertices $s,t\in V(G)$, an $(s,t)$-path is a path with the first vertex $s$ and the last vertex $t$. 
 A \emph{Hamiltonian path (cycle)} in a graph $G$ is a path (cycle) passing through all the vertices of $G$. Two paths $P$ and $Q$ are internally disjoint if 
 every internal vertex of one path is not a vertex of the other path, that is, $P$ and $Q$ may only share their endpoints. 
 The \emph{concatenation} of internally vertex-disjoint paths $P= v_0,\ldots,v_{k}$  and $Q=v_{k}, v_{k+1}, \ldots, v_\ell$ is $PQ= v_0,\ldots,v_{k}, v_{k+1}, \ldots, v_\ell$. Note that $PQ$ is a path if $v_0\neq v_\ell$ and is a cycle if $v_0= v_\ell$. 
 An \emph{arc}  in a cycle $C$ is a path whose all edges belong to $C$. A \emph{chord} of a cycle $C$ is a path connecting two non-adjacent vertices of $C$ that is internally vertex-disjoint with $C$.

 An undirected graph $G$ is connected if for every pair $u,v$ of its vertices there is a path between $u$ and $v$. 
A vertex set $X\subseteq V(G)$ is \emph{connected}  if the subgraph $G[X]$ is connected. A connected component of $G$ is the subgraph induced by a maximal connected vertex subset of $G$.  A connected graph $G$ with at least three vertices is \emph{2-connected} if for every $v\in V(G)$, $G-v$ is connected. Similarly, a  vertex set $X\subseteq V(G)$ is \emph{2-connected}  if the subgraph $G[X]$ is 2-connected.  A \emph{block} of $G$ is the subgraph induced by a maximal  \emph{2-connected} subset. A vertex $v$ is a \emph{cut-vertex} if it belongs to at least two blocks.  
All other vertices of a block are \emph{inner} vertices. Blocks in a graph form a forest  structure (viewing each block as a vertex of the forest and two blocks are adjacent if they share a cut-vertex). The blocks corresponding to the leaves of the block-forest, are referred as \emph{leaf-blocks}. A connected component is \emph{separable} if it contains a cut-vertex, or equivalently, if it is not 2-connected.

A \emph{vertex cover}  $X$ of a graph $G$ is a subset of the vertex set $V(G)$
such that $X$ covers the edge set $E(G)$, i.e., every edge of $G$ has at least one endpoint
in $X$. An \emph{independent set} $I$ in a graph $G$ is a subset of the vertex set $V(G)$
such that the vertices of $I$ are pairwise nonadjacent.  
A \emph{path cover} of a graph $G$ is a family of disjoint paths in $G$ such that every vertex of $G$ belongs to some of these paths. 

\medskip\noindent\textbf{Classical results.}
Besides  Dirac's theorem from \cite{Dirac52}, already stated  as \Cref{thm:diracs}, we use the result that guarantees a long path between two fixed vertices of a $2$-connected graph.
Its different versions can be found throughout the  works of Locke \cite{locke1983extremal,locke1985generalization}. The version below is from the paper of Egawa, Glas,  and Locke \cite[Lemma 5]{Egawa1991}.

\begin{lemma}[Egawa, Glas,  and Locke \cite{Egawa1991}]\label{lemma:path_by_large_degree}
	Let $G$ be a $2$-connected graph with at least $4$ vertices, and let $s,t \in V(G)$ be a pair of vertices in $G$, and let $d$ be an integer.
	If all vertices in $G$, except $s$, $t$ and one other vertex,  have degree at least $d$, then there exists an $(s,t)$-path of length at least $d$ in $G$.
\end{lemma}

In addition, we rely on several other classical theorems. 
In some parts of the proof we  use one more result from Dirac's work~\cite{Dirac52}.

\begin{theorem}[Dirac \cite{Dirac52}]\label{prop:cycle_delta}
	 Every graph $G$ contains a cycle of length at least $\delta(G)+1$.
\end{theorem}

We remark that \Cref{lemma:path_by_large_degree} and  \Cref{prop:cycle_delta}  are constructive in the following sense. Their proofs can be turned into polynomial time algorithms producing an $(s,t)$-path of length at least $d$ and a cycle of length at least $\delta(G)+1$.

In 1976, Bondy  and Chv\'{a}tal~\cite{bondyC76} proved the following generalization of  \Cref{thm:diracs} that we use in Section~\ref{sec:vcalgo}. Let $G$ be an $n$-vertex graph. The \emph{closure} $\cl(G)$ of $G$ is the graph obtained from $G$ by iteratively making two distinct vertices $u$ and $v$ adjacent whenever the sum of their degrees is at least $n$. Note that if $\deg_G(u)+\deg_G(v)\geq n$ for all pairs of vertices $u$ and $v$, then $\cl(G)=K_n$. In particular,  the closure of a $n$-vertex graph satisfying the conditions of \Cref{thm:diracs} is $K_n$. 

\begin{theorem}[Bondy  and Chv\'{a}tal \cite{bondyC76}]\label{thm:bh}
A graph $G$ has a Hamiltonian cycle if and only if $\cl(G)$ has a Hamiltonian cycle. 
\end{theorem}

We remark that the proof of   \Cref{thm:bh}  
yields a polynomial time algorithm constructing from a cycle with an added edge a cycle in a graph without the new edge. Thus by repeating this argument for every added edge, we obtain a polynomial time algorithm for constructing a Hamiltonian cycle in $G$ from   a Hamiltonian cycle in $\cl(G)$.

Another classical result that we require is the well-known  Menger's theorem.

\begin{theorem}[Menger's theorem, \cite{menger1927allgemeinen,goring2000short}]
	Let $G$ be a graph and $A, B\subseteq V(G)$ be two subsets of its vertices.
	Let $s$ be the minimum number of vertices separating $A$ and $B$ in $G$.
	There are $s$ vertex-disjoint paths going from $A$ to $B$.
\end{theorem}

Throughout the paper we are mostly working with $2$-connected graphs, so we just need the following corollary of the Menger's theorem.

\begin{corollary}
	Let $G$ be a $2$-connected graphs and $A,B\subseteq V(G)$ be two subsets of its vertices such that $|A|,|B|\ge 2$.
	There exist two vertex-disjoint paths going from $A$ to $B$ in $G$.
\end{corollary}

Finally, we will make use of a strengthening of Dirac's theorem due to Nash-Williams \cite{NashWil71}. We state it in the form following Jansen, Kozma and Nederlof \cite{DBLP:conf/wg/Jansen0N19}, where the following algorithmic statement is proven.
\begin{theorem}[Nash-Williams \cite{NashWil71, DBLP:conf/wg/Jansen0N19}]
    \label{proposition:cycle_or_is}
    Let $G$ be a $2$-connected graph with $n$ vertices, with $\delta(G) \ge (n + 2)/3$. Then, we can find in $G$, in time $\Oh(n^3)$, either a Hamiltonian cycle, or an independent set of size $\delta(G) + 1$.
\end{theorem}

%

\medskip\noindent\textbf{Parameterized algorithms.}
We will use several results from parameterized complexity as black boxes. 
Let us recall that \probKCycle (\probKPath) are the problems where for given graph $G$ and integer $k$, the task is to decide whether $G$ has a cycle (a path) of length at least $k$. In \textsc{Long $(s,t)$-Path}, for $s,t\in V(G)$,  the task is to decide whether an $(s,t)$-path of length at leas $k$ exists. The first algorithms for \probKCycle  and \probKPath of running time  $2^{\Oh(k)}\cdot\polyn$ are due to 
Alon, Yuster and Zwick \cite{AlonYZ95}.
The fastest known randomized
algorithm for \probKPath\ on undirected graph is due to Bj{\"{o}}rklund, Husfeldt, Kaski and Koivisto~\cite{BjHuKK10} and runs
in time $1.657^k \cdot n^{\Oh(1)} $.   Tsur  gave the fastest known deterministic algorithm for the problem  running in time $2.554^k \cdot n^{\Oh(1)}$~\cite{Tsur19b}. For \probKCycle, the current fastest randomized algorithm  running in time $4^k\cdot n^{\Oh(1)}$ is due to  Zehavi~\cite{Zehavi16} and the best deterministic algorithm 
runs in time $4.884^k\cdot n^{\Oh(1)}$~\cite{DBLP:journals/ipl/FominLPSZ18}. For  
\textsc{Long $(s,t)$-Path} the best known running time is $4.884^k\cdot n^{\Oh(1)}$~\cite{DBLP:journals/ipl/FominLPSZ18}.

\begin{theorem}[\cite{AlonYZ95},\cite{DBLP:journals/ipl/FominLPSZ18}]\label{prop:longest_cycle}
	\probKPath, \probKCycle,  and \textsc{Long $(s,t)$-Path} admit algorithms with running time $2^{\Oh(k)}\cdot\polyn$.
\end{theorem}

We also use the following algorithms of Jansen, Kozma, and Nederlof   \cite{DBLP:conf/wg/Jansen0N19}. 
\begin{theorem}[\cite{DBLP:conf/wg/Jansen0N19}]\label{theorem:JansenKN}
	If a graph $G$ has at least $n-k$ vertices of degree at least $\frac{n}{2}$ or if a graph $G$ has $\delta(G)\ge\frac{n}{2}-k$, a Hamiltonian cycle in $G$ can be found in time $2^{\Oh(k)}\cdot\polyn$.
\end{theorem}
\section{Generalized theorems}\label{section:generaltheorems}
The classical theorems of Dirac  and Erd{\H{o}}s-Gallai provide bounds on the length of cycles and paths in terms of vertex degrees in graph $G$. 
In our algorithmic extension of Dirac's theorem, we deal with a more general problem when the cycle's length is bounded by vertex degrees of graph $G-B$. In our algorithm we use  generalizations of 
 these classical results stated for vertex degrees in graph $G-B$. These generalizations are simple and most likely they are known as a folklore. However we could not find them in the literature and  prove them here for completeness. 


	
The first theorem is the generalization of Dirac's theorem (Theorem~\ref{thm:circum}):   Theorem~\ref{thm:circum}	 is its special case with $B=\emptyset$. 
	\begin{theorem}\label{thm:relaxed_long_cycle}
		Let $G$ be a $2$-connected $n$-vertex graph.
		For any $B \subseteq V(G)$ there exists a simple cycle in $G$ of length at least $\min\{n-|B|,2\delta(G- B)\}$. Moreover, there is a polynomial time algorithm constructing a cycle of such length. 
	\end{theorem}
	\begin{proof}
		We assume that $0 < |B| \le n-1$ and $\delta(G-B)>1$, other cases are trivial.
		
		Consider   graph $G-B$.
		It consists of one or more connected components.
		If $G$ has at least one connected component, say $H$, that is $2$-connected and is of size at least $2\delta(G-B)$, then a cycle of length at least $\min(|V(H)|, 2\delta(H))\ge 2\delta(G-B)$ can be found inside $H$ by Theorem~\ref{thm:circum}.
		
		Now assume that each connected component in $G-B$ either contains a cut-vertex or consists of less than $2\delta(G-B)$ vertices.
		Assume that there are at least two connected components in $G-B$, say $H_1$ and $H_2$.
		Since $\delta(G-B)>1$, both $H_1$ and $H_2$ consist of at least three vertices.
		
		If $H_1$ contains a cut vertex, take one of its leaf-blocks, say $L_1$, and put $S_1=L_1$.
		Note that all vertices but one are of degree at least $\delta(G- B)$ in $L_1$.
		As $\delta(G-B)>1$, $L_1$ consists of at least three vertices and is $2$-connected.
		If $H_1$ is $2$-connected, put $S_1=H_1$.
		Find $S_2$ in the same way for $H_2$.
		By Menger's theorem, there are two vertex-disjoint paths from $V(S_1)$ to $V(S_2)$.
		Thus, there are two distinct vertices $u_1, v_1 \in V(S_1)$ that are connected correspondingly with $u_2, v_2 \in V(S_2)$ with two vertex-disjoint paths.
		Note that the total length of these paths is at least four, since the are no edges $u_1 u_2$ and $v_1 v_2$ in $G$.
		By Lemma~\ref{lemma:path_by_large_degree}, there is a path of length at least $\delta(G- B)$ between $u_1$ and $v_1$ in $S_1$ if $|V(S_1)|\ge 4$.
		If $|V(S_1)|<4$, then $S_1$ is a cycle on three vertices, so there also exists a path of length at least $2\ge |V(S_1)|-1\ge \delta(G-B)$ between $u_1$ and $v_1$ in $S_1$.
		Analogously, there is a path of length at least $\delta(G-B)$ between $u_2$ and $v_2$ in $S_2$.
		Combine these two paths and the two paths outside and obtain a cycle of length at least $2\delta(G- B)+4$ in $G$.
		
		Now assume that there is exactly one connected component in $G-B$ of size $n-|B|$.
		Note that it consists of at least three vertices as $\delta(G-B)>1$.
		If it is $2$-connected, then its size is less than $2\delta(G-B)$.
		Hence, the desired cycle is obtained automatically by Theorem~\ref{thm:diracs}.
		If it is not $2$-connected, take any of its cut vertices and add it to $B$ to obtain $B'$.
		Note that $\delta(G- B')\ge \delta(G- B)-1$.
		Now 
		$G- B'$ consists of at least two connected components, so apply the discussion above for this case and obtain a cycle of length at least $2\delta(G- B')+4\ge 2\delta(G- B)+2$.
		
The proof is constructive and all its steps (computing 2-connected components, finding a cut-vertex, computing two vertex-disjoint paths, etc.) are implementable in polynomial time.		
	\end{proof}

	The similar theorem for paths can be now derived.

	\begin{theorem}\label{thm:relaxed_long_path}
	Let $G$ be a connected $n$-vertex graph.
	For any $B\subseteq V(G)$ there exists a simple path in $G$ of length at least $\min\{n-|B|-1, 2\delta(G-B)\}$. Moreover, there is a polynomial time algorithm constructing a path  of such length. 
	\end{theorem}
	\begin{proof}
	Construct graph $G'$ from $G$ by adding to it  a universal vertex, that is the vertex adjacent to all vertices of $G$.
		Note that $\delta(G'-B)=\delta(G-B)+1$ and $G'$ consists of $n+1$ vertices.
		Also, $G'$ is $2$-connected, since $G$ is connected.
		Thus, by Theorem~\ref{thm:relaxed_long_cycle},  $G'$ contains a simple cycle of length at least  $\min\{n+1-|B|, 2\delta(G-B)+2\}$.
		
		If this cycle does not contain the universal vertex, this cycle is contained in $G$ as well, and we automatically obtain a path of length at least $\min\{n+1-|B|, 2\delta(G-B)+2\}-1$ in $G$.
		
		If the cycle contains the universal vertex, remove this vertex from the cycle.
		Since this vertex is incident with two edges of the cycle, we obtain a path of length at least $\min\{n+1-|B|, 2\delta(G-B)+2\}-2$ in $G$. Again, the construction can be easily turned into a polynomial time algorithm.
	\end{proof}
The following Corollary generalizes the theorem of  Erd{\H{o}}s and Gallai from \cite[Theorem~1.16]{ErdosG59}.
	\begin{corollary}\label{thm:relaxed_st_path}
		Let $G$ be a $2$-connected graph and let $s, t$ be a pair of distinct vertices in $G$.
		For any $B \subseteq V(G)$ there exists a path of length at least $\delta(G- B)$ between $s$ and $t$ in $G$. Moreover, there is a polynomial time algorithm constructing a cycle of such length. 
	\end{corollary}
	\begin{proof}
		Suppose that $n-|B|\ge2\delta(G- B)$.
		Use Theorem~\ref{thm:relaxed_long_cycle} to find a cycle of length at least $2\delta(G-B)$.
		By Menger's theorem, there are two vertex-disjoint paths from $\{s,t\}$ to this cycle in $G$.
		Take these paths and the longer arc of the cycle and obtain a path of length at least $\delta(G- B)$ between $s$ and $t$.
		
		Consider the case when $n-|B|< 2\delta(G- B)$, so $\delta(G- B)\ge (|V(G- B)|+1)/2$.
		If $n-|B|\le 3$, then $\delta(G- B)\le 2$, so it is enough to find any path of length two between $s$ and $t$.
		If $n-|B|\ge 4$, then $G- B$ is $2$-connected, as it contains a Hamiltonian cycle by classical Dirac's theorem.
		Apply Menger's theorem to $G$ and find two vertex-disjoint paths from $\{s,t\}$ to $V(G)\setminus B$.
		Let these paths be a path going from $s$ to $s'\in V(G-B)$ and from $t$ to $t'\in V(G-B)$.
		By Lemma~\ref{lemma:path_by_large_degree}, there is a path of length at least $\delta(G- B)$ between $s'$ and $t'$ in $G- B$. 
		Combine this path with the paths from $s$ to $s'$ and from $t$ to $t'$.
		This yields a path of length at least $\delta(G- B)$ between $s$ and $t$ in $G$.
	\end{proof}

\newcommand\claimqed{\hfill$\lrcorner$}
\section{Long $(s,t)$-Cycle}\label{sec:tldp}
In this section we give an  FPT algorithm that finds a cycle of length at least $k$ passing through designated terminal vertices $s$ and $t$. When the length of such cycle is of order  $\Oh(k)$, then the classical methods like color-coding solve the problem. The difficulty is that the length of the cycle can be arbitrarily bigger than $k$. For that  case we  build on the approach from ~\cite{DBLP:journals/ipl/FominLPSZ18} that was used to design an algorithm for a longest $(s,t)$-path. 


Now we are ready to prove Theorem~\ref{thmTLDP}. We restate it here.

\medskip
\noindent
{\bf Theorem~\ref{thmTLDP}.}
{\it The \probTLDP problem  is solvable in $\Oh((2e)^{3k}\cdot mn)$ time by a randomized Monte Carlo algorithm and in $(2e)^{3k}k^{\Oh(\log k)}\cdot mn\log n$ deterministic time.
}


\begin{proof}
Let $(G,s,t,k)$ be an instance of \probTLDP. Clearly, we can assume that $G$ is connected, because if $s$ and $t$ are in distinct connected components, then we have a trivial no-instance, and if $s$ and $t$ are in the same connected component of a disconnected graph, then we can consider the problem on the component containing $s$ and $t$ instead of $G$.
To avoid additional case analysis, we assume that $s\neq t$. Otherwise, if $s=t$, we can do the following. If $k\leq 3$, then 
to solve the problem, it is sufficient to check whether $G$ has a cycle containing $s$ and this easily can be done in linear time.
If $k\geq 4$, then we apply the algorithm from ~\cite{DBLP:journals/ipl/FominLPSZ18}. To be able to do it formally, we create a new vertex $t'$ that is a false tween of $s$ and then check whether the obtained graph has an $(s,t')$-path of length at least $k$. Fomin, Lokshtanov, Panolan, Saurabh and Zehavi~\cite{DBLP:journals/ipl/FominLPSZ18} do not state explicitly the dependency of their algorithm on the graph size. However, it can be seen that the running times of the randomized and deterministic variants of their algorithm are dominated by $\Oh((2e)^{3k}\cdot mn)$ and $(2e)^{3k}k^{\Oh(\log k)}\cdot mn\log n$, respectively. 
We also assume that $k\geq 4$. If $k\leq 3$, then to solve the problem, it is sufficient to find any two internally disjoint  $(s,t)$-paths, and this can be done by the standard  flow techniques (see, e.g., the recent textbook~\cite{Williamson19}) in time $\Oh(n+m)$, because we are looking for a flow of volume 2.

The algorithm works in two stages. First we  try to find two internally vertex-disjoint $(s,t)$-paths   of total length $\ell$ for $\ell\in\{k,\ldots,3k\}$. If such paths are found, they form the required cycle, so we stop. Otherwise, we proceed to Stage 2, where we assume that the long $(s,t)$-cycle,  if it exists, is longer that $3k$. 

\medskip
\noindent\textbf{Stage~1.}
First, we check whether there are two internally disjoint $(s,t)$-paths of total length $\ell$ for some $\ell\in\{k,\ldots,3k\}$. For this, we apply the classical color-coding technique of  Alon, Yuster, and Zwick~\cite{AlonYZ95}.   Here the arguments are standard and 
  we only sketch how to solve the decision version of the problem. The algorithm may be easily modified to construct the paths. We describe a randomized Monte Carlo algorithm and explain how to derandomize it in the concluding part of the theorem proof.

We color the vertices of $G$ uniformly at random by $3k$ colors $\{1,\ldots,3k\}$.  We say that two $(s,t)$-paths $P_1$ and $P_2$ form a \emph{colorful solution} if the vertices of each of the paths have distinct colors and the colors of the internal vertices of $P_1$ are distinct from the colors of the internal vertices of $P_2$.  (Clearly, the colors of $s$ and $t$ are the same in both paths.) In other words,  in the $(s,t)$-cycle formed by $P_1$ and $P_2$ all vertices are colored in different colors. 

We find a colorful solution by dynamic programming. Denote by $c(x)$ the color of a vertex $x\in V(G)$,  and let $p=c(s)$ and $q=c(t)$.
If $p=q$, then there is no colorful solution. Suppose that $p\neq q$.
For a vertex $x\in V(G)$ and a nonempty set of colors $X\subseteq\{1,\ldots,3k\}$, define 
$\alpha(x,X)={\tt true}$ if there is an $(s,x)$-path $P$ with $|X|$ vertices that are colored by distinct colors from $X$, and we set $\alpha(x,X)={\tt false}$ otherwise. 
The values of $\alpha(x,X)$ are computed for all $x\in V(G)$ and all $X\subseteq \{1,\ldots,3k\}$ starting from   sets of size one.

For every $x\in V(G)$ and every $i\in\{1,\ldots,3k\}$, we define
\begin{equation}\label{eq:dp-cc-one}
\alpha(x,\{i\})=
\begin{cases}
{\tt true}&\mbox{if }x=s\text{ and }i=p\\
{\tt false}&\mbox{otherwise}.
\end{cases}
\end{equation}
Assume that $|X|\geq 2$ and the table of values of $\alpha(x,Y)$ is constructed for all $x\in V(G)$ and all $Y\subseteq\{1,\ldots,3k\}$ such that $|Y|<|X|$. 
Then for $x\in V(V)$, we set
\begin{equation}\label{eq:dp-cc-two}
\alpha(x,X)=
\begin{cases}
\bigvee\limits_{y\in N_G(x)}\alpha(y,X\setminus\{c(x)\})&\mbox{if }c(x)\in X\\
{\tt false}&\mbox{if }c(x)\notin X. 
\end{cases}
\end{equation}
By exactly the same arguments as for the color-coding algorithm for \probKPath (see, e.g.,~\cite[Chapter~5]{CyganFKLMPPS15}), we obtain that (\ref{eq:dp-cc-one}) and (\ref{eq:dp-cc-two}) allow to compute the table of values of $\alpha(x,X)$ for all $x\in V(G)$ and all nonempty $X\subseteq\{1,\ldots,3k\}$ in time $\Oh(2^{3k}\cdot mn)$, because we compute $\alpha(x,X)$ for $n$ vertices $x$ and $2^{3k}-1$ sets $X$, and to compute $\alpha(x,X)$ using  (\ref{eq:dp-cc-two}), we consider the neighbors of $x$.
 
To complete the description of the algorithm that verifies the existence of a colorful solution, we observe that such a solution exists if and only if there are $X,Y\subseteq \{1,\ldots,3k\}$ such that $X\cap Y=\{p,q\}$, $|X|+|Y|\geq k+2$, and 
$\alpha(t,X)=\alpha(t,Y)={\tt true}$.  

Hence, it takes time $\Oh(2^{3k}\cdot mn)$  to decide whether there is a colorful solution. If there is a colorful solution, $(G,s,t,k)$ is a yes-instance of \probTLDP. However, the absence of a colorful solution does not imply that we have a no-instance.  

Assume that there are two internally disjoint $(s,t)$-paths $P_1$ and $P_2$ in $G$ whose total length is between  $k$ and $3k$.  That is,  $k\leq |V(P_1)\cup V(P_2)|\leq 3k$. Then the probability that all  vertices of $V(P_1)\cup V(P_2)$ are colored by distinct colors is at least $\frac{(3k)!}{(3k)^{3k}}\geq e^{-3k}$. The probability that there is no colorful solution is at most $1-e^{3k}$. Therefore, by trying to find a colorful solution for $N=\lceil e^{3k}\rceil$ random colorings, we either conclude that we have a yes-instance, or return  no-answer with the mistake probability at most $(1-e^{3k})^N\leq e^{-1}$. This gives us a Monte Carlo algorithm with running time
$\Oh((2e)^{3k}\cdot mn)$.

\medskip
\noindent\textbf{Stage~2.}
From now, we assume that we failed to solve the problem at Stage~1. This means that   each solution is an $(s,t)$-cycle of length $3k+1$. As in Stage~1, we find
 two disjoint $(s,t)$-paths of total length  at least $3k+1$. This is done by generalizing the technique of  Fomin, Lokshtanov, Panolan, Saurabh and Zehavi from~\cite{DBLP:journals/ipl/FominLPSZ18} for finding an $(s,t)$-path of length at least $k$.  
Now instead of color-coding, we use the technique of random separation~\cite{CaiCC06}. 
 The main step of our procedure for Stage~2 is given in Algorithm~\ref{alg:step}.
 \medskip
 
%

\begin{algorithm}[h]
Color the vertices of $V(G)\setminus\{s,t\}$ uniformly at random by three colors $1, 2,$ and $3$, and denote by $X_1,X_2,X_3$ the vertices colored by the corresponding colors.\\
\For{$i=1,2,3$}
{
\ForEach{$v\in X_i$ at distance $k$ from $t$ in $G_i=G[X_i\cup\{t\}]$}
{
Find a shortest $(v,t)$-path $P$ in $G_i$\;
Find  an $(s,v)$-path $P_1$ and an $(s,t)$-path $P_2$, such that $P_1$ and $P_2$ are
 internally disjoint  and both these paths  avoid internal vertices of $P$\;
\tcc*[h]{that is,  $P_1$ and  $P_2$ are 
paths in $G-(V(P)\setminus\{v,t\})$}\\
\If{ such paths $P_1$ and $P_2$ exist}
{\Return{the paths  $P_1'=P_1P$ and $P_2$}\; 
{\bf quit}.
}
}
}
\caption{Main step of Stage~2.}\label{alg:step}
\end{algorithm}

\medskip
Due to the conditions that 
  $P_1$ does not contain internal vertices of $P$, avoids $t$, and  is internally disjoint with $P_2$, 
  we have that 
  the concatenation $P_1'$ of $P_1$ and $P$ is a path. Moreover, $P_1'$ and $P_2$ are internally vertex disjoint $(s,t)$-paths. 
Since the length of $P$ is $k$, the length of $P_1'$ is at least $k$.  We conclude that if the algorithm returns $P_1'$ and $P_2$, then these paths form a required $(s,t)$-cycle of length at least $k$. The algorithm runs in $\Oh(n+m)$ time, as $P_1$ and $P_2$ can be found (if they exist) by the standard flow algorithm, see e.g., ~\cite{Williamson19}. 

However, the proof that the algorithm finds a solution in a yes-instance with a reasonable probability is non-trivial.  It follows from the following lemma.

\begin{lemma}\label{cl:probability}
If $(G,s,t,k)$ is a yes-instance of \probTLDP,  then the described algorithm finds a solution with probability at least $\frac{2}{3^{3k-1}}$.
\end{lemma}
 
\begin{proof}[Proof of Lemma~\ref{cl:probability}]
Suppose that $(G,s,t,k)$ is a yes-instance. Then there are two internally disjoint $(s,t)$-paths $P_1$ and $P_2$ with total length at least $k$. We assume that the total length of paths  $P_1$ and $P_2$ is minimum. Recall that the total length of $P_1$ and $P_2$ is at least $3k+1$. We assume that the vertices of $P_1$ and $P_2$ are ordered in the path order starting with $s$. Thus whenever we refer to the first or the last vertices of the paths, these vertices respect this ordering.  We follow the same convention for all $(x,y)$-paths, that is, we order the vertices starting from $x$.  We consider two cases.

\medskip
\noindent
{\bf Case~1.}  \emph{The shortest path among  $P_1$ and $P_2$ is of length at most $k$.  }
Without loss of generality, 
  we assume that the length of $P_2$ is at most $k$. Then the length of $P_1$ is at least $2k+1$. Denote by $A$ the set of the first $k-1$ internal vertices of $P_1$,   by $B$ the set of the last $k$ internal vertices of $P_1$, and by  $C$  the set of internal vertices of $P_2$. Because the vertices of $V(G)\setminus \{s,t\}$ are colored uniformly at random, with probability  at least $\frac{3!}{3^{|A|}\cdot 3^{|B|}\cdot 3^{|C|}}\geq \frac{2}{3^{3k-1}}$ 
\begin{itemize}
\item[(i)]  vertices of each of the sets $A$, $B$, and $C$ receive the same colors, and 
\item[(ii)]   vertices of distinct sets are of distinct colors.
\end{itemize}
We show that if (i) and (ii) holds, then the algorithm finds a solution. For further analysis, we assume that $A\subseteq X_1$, $B\subseteq X_2$, and $C\subseteq X_3$.

Let $v$ be the internal vertex of $P_1$ at distance $k$ from $t$ in the path. Then  $v\in B\subseteq X_2$. Denote by $P_1'$ the $(s,v)$-subpath of $P_1$. Let $P$ be an arbitrary shortest $(v,t)$-path in $G_2=G[X_2\cup\{t\}]$. 
Notice that $P$ is internally disjoint with $P_2$, because $V(P_2)\setminus \{s,t\}\subseteq X_3$,  and $X_2\cap X_3=\emptyset$. 
We claim that $P_1'$ and $P$ are internally disjoint. 

Targeting towards a contradiction, assume that $V(P_1')\cap V(P)\setminus\{v\}\neq\emptyset$. Let $u$ be the first vertex of $P_1'$ that is in $V(P)$. Let $Q_1$ be the $(s,u)$-subpath of $P_1'$ and let $Q$ be the $(u,t)$-subpath of $P$.  Recall that the first $k-1$ internal vertices of $P_1$ are in $A\subseteq X_1$. This implies that $A\subseteq V(Q_1)$. Threfore,
the length of $Q_1$ is at least $k$. Let $\hat{P}_1=Q_1Q$. We obtain that $\hat{P}_1$ is an $(s,t)$-path of length at least $k$ that is internally disjoint with $P_2$. Hence, $\hat{P}_1$ and $P_2$ form a solution. However, the length of $\hat{P}_1$ is less than the length of $P_1$, contradicting the condition that $P_1$ and $P_2$ form a solution of minimum total length. Hence  $P_1'$ and $P$ are internally disjoint. 

We have that $\hat{P}_1=P_1'P$ is a path internally disjoint with $P_2$. We also have that $A\subseteq V(P_1)$ and, therefore, the length of $\hat{P}_1$ is at least $k$. Thus $\hat{P}_1$ and $P_2$ is a solution. By the construction of $\hat{P}_1$, the length of $\hat{P}_1$ is at most the length of $P_1$. Since $P_1$ and $P_2$ compose a solution of minimum total length, the length of $P_1$ is the same as the length of $\hat{P}_1$. Hence, $v$ is at distance $k$ from $v$ in $G_3$.  

Summarizing, there is a vertex $v$ at distance $k$ from $t$ in $G_2$ such that for any shortest $(v,t)$-path $P$ in graph $G_2$, in graph $G-(V(P)\setminus\{v,t\})$
there exist  an $(s,v)$-path $P_1$ and an $(s,t)$-path $P_2$   that are internally disjoint. In this case  our algorithm  finds a solution.

\medskip
\noindent
{\bf Case~2.} \emph{The length  of each of the paths $P_1$ and $P_2$ is at least $k+1$.} 
Let $B$ be the last $k$ internal vertices of $P_1$ and let $C$ be  the last $k$ internal vertices of $P_2$.  Since each of the paths $P_1$ and $P_2$ is of  length at least $k+1$ because the total length of both paths is at least $3k+1$, we conclude the following. 
For some positive integers $k_1$ and $k_2$ such that $k_1+k_2=k-1$, 
   the first $k_1$ internal vertices of $P_1$ are not in $B$, and the first $k_2$ internal vertices of $P_2$ are not in $C$. Denote by $A_1$ the   first $k_1$ internal vertices of $P_1$ and by $A_2$    the first $k_2$ internal vertices of $P_2$. We set $A=A_1\cup A_2$. Thus $|A|=k-1$ and the sets $A$, $B$ and $C$ are pairwise disjoint.   Since the vertices of $V(G)\setminus \{s,t\}$ are colored uniformly at random, we have that with probability  $\frac{3!}{3^{|A|}\cdot 3^{|B|}\cdot 3^{|C|}}\geq \frac{2}{3^{3k-1}}$
\begin{itemize}
\item[(i)]  vertices of each of the  sets $A$, $B$, and $C$ are of the same color, and
\item[(ii)]   vertices of distinct sets are of different colors.
\end{itemize}
As in Case~1, we show that if a coloring satisfies (i) and (ii),  then the algorithm finds a solution. Without loss of generality, we assume that $A\subseteq X_1$, $B\subseteq X_2$, and $C\subseteq X_3$.

Let $v_1$ be the  internal vertex of $P_1$ at distance $k$ from $t$ in $P_1$, 
and let  $v_2$ be the internal vertex of $P_2$ at distance $k$ from $t$ in $P_2$. 
Note that $v_1\in B\subseteq X_2$ and $v_2\in C\subseteq X_3$.
Denote by $P_1'$ the $(s,v_1)$-subpath of $P_1$ and by $P_1''$ the $(v_1,t)$-subpath of $P_1$.
Similarly, we define $P_2'$
as the $(s,v_2)$-subpath of $P_2$ and $P_2''$ as the $(v_2,t)$-subpath of $P_2$.


We prove the following claim.
\begin{claim}\label{cl:oneoftwo}
At least one the following options holds.
\begin{itemize}
\item 
Either for every shortest $(v_1,t)$-path $Q_1$ in $G_2=G[X_2\cup\{t\}]$, paths  $Q_1$ and $P_2$ are internally disjoint,
\item or for every shortest $(v_2,t)$-path $Q_2$ in  $G_3=G[X_3\cup\{t\}]$, paths $Q_2$ and $P_1$ are internally disjoint.
\end{itemize}
\end{claim}
\medskip
\noindent \emph{Proof of Claim~\ref{cl:oneoftwo}.}  
The proof is by contradiction. Assume that there is a   shortest  $(v_1,t)$-path $Q_1$ in $G_2$ and a shortest $(v_2,t)$-path $Q_2$ in  $G_3$ such that neither $Q_1$ and $P_2$ are internally disjoint, nor $Q_2$ and $P_1$ are internally disjoint. See Figure~\ref{fig:paths}.

\begin{figure}[ht]
\centering
\scalebox{0.7}{
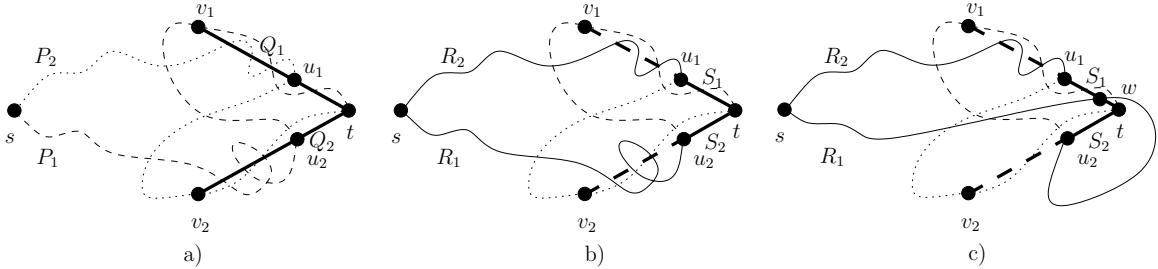
}
\caption{Structure of the paths $P_1$, $P_2$, $Q_1$,  and $Q_2$. a) A dashed line shows   path $P_1$, a dotted line indicates $P_2$, and solid lines are used for $Q_1$ and $Q_2$. 
b) Solid lines indicate paths $R_1$, $R_2$, $S_1$, and $S_2$.  The thin lines are used for $R_1$ and $R_2$, while the thick lines for $S_1$ and $S_2$. The choice of $w$ is demonstrated in c).}
\label{fig:paths}
\end{figure}

Notice that $Q_1$ and $Q_2$ are internally disjoint since they are paths in $G_2$ and $G_3$ respectively, and $t$ is the unique common vertex of these graphs.  
Let $u_1$ be the vertex of $V(Q_1)\cap V(P_2)$ distinct from $t$ that is at the minimum distance from $t$ in $Q_1$. 
Similarly, let  $u_2$ be the vertex of $V(Q_2)\cap V(P_1)$ distinct from $t$ that is at the minimum distance from $t$ in $Q_2$. The choice of $u_1$ and $u_2$ is shown in Figure~\ref{fig:paths} (a).
Because  $V(P_2'')\subseteq X_3$ and $V(P_1'')\subseteq X_2$, we have that  $u_1\in V(P_2')$  and $u_2\in V(P_1')$.  
Denote by $R_1$ the $(s,u_2)$-subpath of $P_1$ and by $R_2$ the $(s,u_1)$-subpath of $P_2$.
Let also $S_1$ be the $(u_1,t)$-subpath of $Q_1$ and $S_2$ be the $(u_2,t)$-subpath of $Q_2$. The construction of these paths is shown in Figure~\ref{fig:paths} (b). 


We claim that paths $S_1$ and $R_1$ have no common vertices. For the sake of contradiction, let $V(S_1)\cap V(R_1)\neq \emptyset$ and assume that $w$ is the first vertex of $R_1$ in $V(S_1)$ (see Figure~\ref{fig:paths} (c)). Since $R_1$ and $R_2$ are internally vertex disjoint, $w$ is an internal vertex of $S_1$.
By the choice of $u_1$, there is no internal vertex of $S_1$ that belongs to $P_2$.
Hence, the concatenation $\hat{P}_1$ of the $(s,w)$-subpath of $R_1$ and the $(w,t)$-subpath of $S_1$, gives a path that is internally vertex disjoint with $P_2$. Observe also that $A_1\subseteq V(\hat{P}_1)$, because $w\in B$ and the first $k_1$ internal vertices of $P_1$ are in $A_1$. 
Therefore, the total length of paths $\hat{P}_1$ and $P_2$ is at least $k$. However, the length of the $(s,w)$-subpath of $R_1$ is less than the length of $P_1'$ and the length of the $(w,t)$-subpath of $S_1$ is less than $k$. Therefore, the length of 
 $\hat{P}_1$ is less than the length of $P_1$, contradicting the choice of $P_1$ and $P_2$. 

This proves that 
 $S_1$ and $R_1$ have no common vertices.  By the same arguments,  $S_2$ and $R_2$ also have no common vertices. 

Consider  $(s,t)$-paths $\hat{P}_1=R_2S_1$ and $\hat{P}_2=R_1S_2$.  Since  paths $S_1$ and $R_1$ do not intersect and paths  $S_2$ and $R_2$ also do not intersect, we have that paths $\hat{P}_1$ and $\hat{P}_2$ are internally disjoint. 
 Because $A_1\subseteq V(\hat{P}_2)$ and  $A_2\subseteq V(\hat{P}_1)$, the total length of paths  $\hat{P}_1$ and $\hat{P}_2$ is at least $k$. 
 However, because the length of $P_1''$ is less than the length of $S_1$ and because  the length of $P_2''$ is less than the length of $S_2$, the total length of   $\hat{P}_1$ and $\hat{P}_2$ is less than the total length of $P_1$ and $P_2$. This contradict the choice of $P_1$ and $P_2$ and proves the claim.
  \claimqed
 
By Claim~\ref{cl:oneoftwo}, without loss of generality, we assume  that  for every shortest $(v_1,t)$-path $Q_1$ in $G_2$, paths $Q_1$ and $P_2$ are internally disjoint.
 Now we repeat the arguments from Case~1. 
We observe that every shortest $(v_1,t)$-path $Q_1$ in $G_2$ is internally disjoint with $P_1'$. Indeed, if this is not the case, we can select the first vertex $u$ of $P_1'$ that is in $Q_1$. Then by replacing $P_1$ by the concatenation of the $(s,u)$-subpath of $P_1'$ and the $(u,t)$-subpath of $Q_1$, we obtain a solution with a shorter  total length. But this  contradicts the choice of $P_1$ and $P_2$. 
Since $Q_1$ and $P_1'$ are internally vertex disjoint, we have that the cycle formed by paths  $\hat{P}_1=P_1'Q_1$ and $P_2$ is a solution. This implies that $Q_1$ and $P_1''$ have the same length. Therefore $v_1$ is at distance $k$ from $t$.

We conclude that there is $v_1$ at distance $k$ from $t$ in $G_2$ such that for every shortest $(v,t)$-path $Q_1$ in $G_2$, there are  an $(s,v_1)$-path $P_1'$ and an $(s,t)$-path $P_2$ in $G-(V(P)\setminus\{v,t\})$ that are internally disjoint. 
Then the algorithm   finds a solution. This concludes Case~2 and the proof of Lemma~\ref{cl:probability}.
 \end{proof} 
 
 By Lemma~\ref{cl:probability}, if we iterate Algorithm~\ref{alg:step} $3^{3k-1}/2$ times, then we either find a solution, or return the no-answer with the error probability at most $(1-\frac{2}{3^{3k-1}})^{3^{3k-1}/2}\leq e^{-1}$.  
 Thus we have  a Monte Carlo algorithm with false negatives that  runs in time $\Oh(3^{3k}\cdot (n+m))$.
 
 \medskip\noindent\textbf{Derandomization.} 
  For the  Monte Carlo algorithm that we use in the first stage (finding a short cycle), derandomization uses the standard technique. 
We replace random colorings by functions from the \emph{$(n,3k)$-perfect hash family} of functions of size $e^{3k}k^{\Oh(\log k)}\cdot \log n$ that can be constructed in time 
$e^{3k}k^{\Oh(\log k)}\cdot n\log n$ by the results of Naor, Schulman, and Srinivasan~\cite{NaorSS95} (we refer to~\cite[Chapter~5]{CyganFKLMPPS15} for the detailed introduction to the technique). This allows us to check in 
$(2e)^{3k}k^{\Oh(\log k)}\cdot mn\log n$ deterministic time whether there are two internally vertex disjoint $(s,t)$-paths in $G$ whose total length is at least $k$ but at most $3k$.

 To derandomize  the algorithm from the second stage that uses random separation, we have to do an extra work.
  This is because   commonly random separation is used to distinguish two sets~ \cite[Chapter~5]{CyganFKLMPPS15}. In our algorithm we  distinguish three sets; derandomization here  is slightly different and is based on Lemma~\ref{lem:derand}. Lemma~\ref{lem:derand} could be a folklore, but we did not find it in the literature and prove it here  for completeness. 
 
Let $k$ and $n$ be positive integers.   An \emph{$(n,k)$-universal} set is a family $\mathcal{U}$ of subsets of $\{1,\ldots,n\}$ such that for any $S\subseteq \{1,\ldots,n\}$ of size $k$, the family 
$\{A\cap S\mid A\in \mathcal{U}\}$ contains all $2^k$ subsets of $S$.
We use the following result of  Naor, Schulman, and Srinivasan~\cite{NaorSS95}.  

\begin{proposition}[\cite{NaorSS95}]\label{prop:derand}
For any $n, k\geq 1$, one can construct an $(n,k)$-universal set of size $2^kk^{\Oh(\log k)}\cdot\log n$ in time $2^kk^{\Oh(\log k)}\cdot n\log n$.
\end{proposition}

Using Proposition~\ref{prop:derand}, we prove the following lemma.

\begin{lemma}\label{lem:derand}
For an $n$-element set $\Omega$ and a positive $k$, there is a family of functions $\mathcal{F}_{n,k}$ mapping $\Omega$ to $\{1,2,3\}$ of size $2^{5k}k^{\Oh(\log k)}\cdot (\log n)^2$
such that for every triple of disjoint nonempty sets $A_1,A_2,A_3\subseteq \Omega$, each  of size at most $k$, there is $f\in\mathcal{F}_{n,k}$ with the property that 
\begin{itemize}
\item $f(x)=f(y)$ if $x,y\in A_i$ for some $i\in\{1,2,3\}$,
\item $f(x)\neq f(y)$ if $x\in A_i$ and $y\in A_j$ for distinct $i,j\in\{1,2,3\}$.
\end{itemize}
Moreover, $\mathcal{F}_{n,k}$ can be constructed in $2^{5k}k^{\Oh(\log k)}\cdot n^2\log n$ time.
 \end{lemma}

\begin{proof}
If $n\leq 3k$, then we define $\mathcal{F}_{n,k}$ to be the family of all at most $3^{3k}$ mappings $f\colon \Omega\rightarrow \{1,2,3\}$.  Hence, from now we assume that $n\geq 3k$.
Let $\Omega=\{\omega_1,\ldots,\omega_n\}$.

We apply Proposition~\ref{prop:derand} to construct the following family of universal sets.
We construct an $(n,3k)$ universal set $\mathcal{U}^{(1)}$. Then for every positive $p\leq n$, we construct an $(p,2k)$-universal set $\mathcal{U}^{(2)}_p$. Then $\mathcal{F}_{n,k}$ is constructed as follows.
For every $U=\{i_1,\ldots,i_p\}\in \mathcal{U}^{(1)}$ and every $W=\{j_1,\ldots,j_q\}\in \mathcal{U}^{(2)}_p$, we construct $f\colon \Omega\rightarrow \{1,2,3\}$ such that for every $h\in\{1,\ldots,n\}$,
$$
f(\omega_h)=
\begin{cases}
1&\mbox{if }h\notin \{i_1,\ldots,i_p\}\\
2&\mbox{if }h\in\{i_1,\ldots,i_p\}\setminus\{i_{j_1},\ldots,i_{j_q}\}\\
3&\mbox{if }h\in\{i_{j_1},\ldots,i_{j_q}\}.  
\end{cases}
$$ 

To see that $\mathcal{F}_{n,k}$ satisfies the required property, consider arbitrary disjoint sets $A_1,A_2,A_3\subseteq \Omega$ of size at most $k$. We  assume without loss of generality that each $A_i$ is of  size exactly $k$ (otherwise, we can complement the sets by adding elements of $\Omega$ that are outside these sets). Let $A_i=\{\omega_{i_1^i},\ldots,\omega_{i_k^i}\}$ for $i\in\{1,2,3\}$.  
Let $S=\{i_1^1,\ldots,i_k^1\}\cup \{i_1^2,\ldots,i_k^2\}\cup\{i_1^3,\ldots,i_k^3\}$.
By definition, the $(n,3k)$-universal set $\mathcal{U}^{(1)}$, contains a set $X$ such that 
$S\cap X=\{i_1^2,\ldots,i_k^2\}\cup\{i_1^3,\ldots,i_k^3\}$. Let $p=|X|$ and assume that $X=\{j_1,\ldots,j_p\}$.
Again by definition, the $(p,2k)$-universal set $\mathcal{U}_p^{(2)}$ contains a set $Z$ such that for every $s\in Z$, $j_s\neq i_1^2,\ldots,i_k^2$, and for every $t\in\{1,\ldots,k\}$, there is $s\in Z$ such that $j_s=i_t^2$.
This implies that for $f\in \mathcal{F}_{n,k}$ constructed for $X\in \mathcal{U}^{(1)}$ and $Y\in\mathcal{U}_p^{(2)}$, $f(x)=i$ if $x\in A_i$ for $i\in\{1,2,3\}$. Therefore, $f$ distinguishes the sets $A_1,A_2,A_3$.

By Proposition~\ref{prop:derand}, $|\mathcal{U}^{(1)}|=2^{3k}k^{\Oh(\log k)}\cdot\log n$ and $|\mathcal{U}^{(2)}_p|=2^{2k}k^{\Oh(\log k)}\cdot\log n$. Therefore,
$|\mathcal{F}_{n,k}|\leq 2^{5k}k^{\Oh(\log k)}(\log n)^2$.
By Proposition~\ref{prop:derand}, the universal sets can be constructed in time $2^{3k}k^{\Oh(\log k)}\cdot n^2\log n$. Then we construct  $\mathcal{F}_{n,k}$ in time $2^{5k}k^{\Oh(\log k)}\cdot n(\log n)^2$.
\end{proof}

 To derandomize our algorithm, we apply Lemma~\ref{lem:derand}. Notice that the only property of random colorings that we use in the algorithm is that with sufficiently high probability the sets $A$, $B$,  and $C$ defined in the proof of Lemma~\ref{cl:probability} are colored by distinct colors. The sets $A$, $B$,  and $C$ have sizes at most $k$, and they are subsets of $V(G)\setminus\{s,t\}$. This implies that the random colorings can be replaced by functions of the family $\mathcal{F}_{n-2,k}$ for $\Omega=V(G)\setminus\{s,t\}$. 
Since Algorithm~\ref{alg:step} runs in $\Oh(n+m)$ time, the running time is $2^{5k}k^{\Oh(\log k)}\cdot (n+m)(\log n)^2$.
Taking into account the time for constructing $\mathcal{F}_{n-2,k}$, we conclude that 
 the problem can be solved in $2^{5k}k^{\log k}\cdot mn\log n$ deterministic time.
 
 Recall that in the first stage of our algorithm, we  try to find  two internally disjoint $(s,t)$-paths of total length $\ell$ for some $\ell\in\{k,\ldots,3k\}$, and this can be done in 
 $(2e)^{3k}\cdot mn$ randomized and $(2e)^{3k}k^{\Oh(\log k)}\cdot mn\log n$ deterministic time. Since $(2e)^{3}\geq 2^{5}\geq 3^{3}$ and $nm\geq n(n-1)$ as $G$ is assumed to be connected, we obtain that the running time of the first stage dominates the running time of the second. 
 We conclude that the problem can be solved in
 $(2e)^{3k}\cdot mn$ randomized and $(2e)^{3k}k^{\Oh(\log k)}\cdot mn\log n$ deterministic time. It is  plausible that the running time for the first stage can be improved by making use of more sophisticated techniques   for   \probKPath and \probKCycle (see, e.g.,~\cite{FominLS14,Zehavi16}) but such an improvement goes beyond the scope of our paper.
\end{proof}
\section{\probstP}\label{sec:erdos-gallaiPath}

In this section we prove Theorem~\ref{thmEG}:  
 \emph{\probstP   is solvable in time 
	 $2^{\mathcal{O}(k+|B|)}\cdot n^{\mathcal{O}(1)}$.}
 The proof of the theorem relies on the structural properties of graphs with a long path. The notions of \emph{\bananadec} and \emph{\banana} are crucial here. We prove several combinatorial and algorithmic properties of \bananadec, and then apply the obtained properties in the proof of   \Cref{thmEG}.
 
 \subsection{Erd{\H {o}}s-Gallai decompositions and structures}

We need to introduce the operation of \emph{$B$-refinements}.
The intuition behind this operation is the following. In our proof, we will be using the following rerouting strategy. Suppose we have an $(s,t)$-path $P$, and we want to construct a longer path by rerouting some parts of $P$ through a connected component $H$ of $G-V(P)$. If $H$ is 2-connected, we can try to apply  \Cref{thm:circum}
  to argue that such an enlargement of $P$ is possible. However, when $H$ is not $2$-connected, we want to eliminate some ``insignificant'' parts of $H$. While in the refinement we contract some of the edges  inside $H$, all edges between $H$ and the remaining part of the graph remain.

%

	\begin{definition}[\textbf{$B$-refinement of   $H$}]
		Let $H$ be a connected subgraph of a graph $G$  and $B\subset V(G)$. 
		The \emph{$B$-refinement} of $H$, denoted by $\gbref{B}{H}$,  is the graph obtained by the following process.
		Start with $\gbref{B}{H}:=G$. 
		While  ${H}$ is not $2$-connected  and contains a leaf-block 
	with all inner vertices from  $B$,  contract all edges in $H$ from this leaf-block to its cut-vertex. 
	\end{definition}	
    In other words, 	$\gbref{B}{H}$ is obtained from $G$ by repeatedly contracting edges of $H$ from the leaf-blocks of $H$ whose inner vertices are from $B$.  Note that in  $B$-refinement only edges with both endpoints in $B$ can be contracted. We also say that $\gbref{B}{H}$ is obtained from $G$ by applying $B$-refinement to $H$.

\begin{figure}[ht]
	\begin{center}
		\ifdefined\STOC
\begin{tikzpicture}[scale=0.4]
\else
\begin{tikzpicture}[scale=0.8]
\fi
\tikzstyle{vertex}=[draw, fill, circle, black, minimum size=2,inner sep=0pt]

\draw[fill=lightgray]  plot[smooth cycle, tension=.7] coordinates {(2.4,-1.4) (2.4,-1.7) (2.5,-1.8) (2.7,-2.1) (2.9,-2.1) (2.8,-1.9) (2.8,-1.5)};
\draw[fill=lightgray]  plot[smooth cycle, tension=.7] coordinates {(9.55,-3.35) (9.2,-3.2) (9.25,-2.95) (9.6,-2.9) (9.8,-2.75) (10.2,-2.85) (10.15,-3.15) (9.9,-3.1) (9.65,-3.2)};
\draw[fill=lightgray]  plot[smooth cycle, tension=.7] coordinates {(4.25,-1.25) (4.5,-1.6) (4.7,-1.4) (4.65,-1.25) (4.35,-1) (4.15,-1.05)};
\draw[fill=lightgray]  plot[smooth cycle, tension=.7] coordinates {(-4.3,-0.6) (-4.3,-0.5) (-4.3,-0.4) (-4.2,-0.4) (-4,-0.4) (-3.7,-0.5) (-3.5,-0.4) (-3.1,-0.4) (-2.8,-0.3) (-2.5,-0.3) (-2.4,-0.2) (-2.2,-0.3) (-2.1,-0.4) (-2.2,-0.6) (-2.4,-0.6) (-2.8,-0.8) (-3.3,-1) (-3.8,-0.9) (-3.9,-0.8) (-4.1,-0.8)};
\draw[fill=lightgray]  plot[smooth cycle, tension=.7] coordinates {(1.7,0.3) (1.4,0.3) (1.4,0.7) (1.9,0.7) (2,0.9) (2.1,1.2) (2.4,1.2) (2.8,1.2) (3.1,1.1) (3.4,1) (3.8,0.8) (4.3,0.9) (4.6,0.8) (4.7,0.6) (4.7,0.3) (4.4,0.4) (4.2,0.7) (3.9,0.6) (3.2,0.7) (2.7,0.9) (2.4,0.8) (2.7,0.4) (2,0.3)};
\draw[fill=lightgray]  plot[smooth cycle, tension=.7] coordinates {(4.1,2.8) (3.2,2.7) (3.2,2.5) (3.5,2.4) (3.9,2.3) (4.6,2.2) (5,2.1) (5.1,2.4) (5,2.6) (4.6,2.7) (4.4,2.8)};

\draw [very thick, red] (-4.5,0.5) node [vertex] (v30) {} -- (-3.5,0.5) node [vertex] (v9) {} -- (-2.5,0.5) node [vertex] (v21) {} -- (-1.5,0.5) node [vertex] (v2) {} -- (-0.5,0.5) node [vertex] (v11) {};
\draw [very thick,red] (9.5,0.5) node [vertex] (v6) {} -- (10.5,0.5) node [vertex] (v7) {} -- (11.5,0.5) node [vertex] {} -- (12.5,0.5) node [vertex] (v8) {} -- (13.5,0.5) node [vertex] (v17) {};

\draw[fill=lightgray]  plot[smooth cycle, tension=.7] coordinates {(0.5,2.7) (0.5,2.9) (0.4,2.9) (0.3,2.8) (0.3,2.7) (0.2,2.5) (0.2,2.3) (0.4,2.1) (0.7,2.3) (0.6,2.5)};
\draw[fill=lightgray]  plot[smooth cycle, tension=.7] coordinates {(7.9,3.2) (7.9,3.4) (8.1,3.4) (8.2,3.1) (8.5,2.9) (8.6,2.7) (8.5,2.5) (8.3,2.6) (8.2,2.7) (8.1,2.9) (8,3) (8,3.1)};
\draw[fill=lightgray]  plot[smooth cycle, tension=.7] coordinates {(-0.6,3.8) (-0.4,3.5) (-0.2,3.3) (-0.3,3.2) (-0.5,3.3) (-0.6,3.4) (-0.7,3.5) (-0.9,3.6)};

\draw[very thick, blue]  plot[smooth cycle, tension=.7] coordinates {(-0.4,2.6) (0.1,4.1) (2.6,4.6) (4.1,4.6) (6.1,4.6) (7.6,4.6) (8.6,3.6) (7.6,3.1) (6.6,3.6) (4.6,3.6) (2.6,3.6) (1.2,3.1) (0.1,2.6)};
\node [vertex] (v3) at (-0.7,3.6) {};
\node [vertex] (v1) at (-0.2,2.7) {};
\node [vertex] (v4) at (-0.3,3.3) {};

\node [vertex] (v5) at (8.3,2.9) {};
\draw  (v5) edge (v6);
\draw  (v5) edge (v8);
\draw  (v1) edge (v9);
\draw  (v3) edge (v9);
\draw  (v4) edge (v9);
\draw (8.5,0.5) node [vertex] (v10) {} -- (7.5,0.5) node [vertex] {} -- (6.5,0.5) node [vertex] {} -- (5.5,0.5) node [vertex] {} -- (4.5,0.5) node [vertex] {} -- (3.5,0.5) node [vertex] {} -- (2.5,0.5) node [vertex] {} -- (1.5,0.5) node [vertex] {} -- (0.5,0.5) node [vertex] (v12) {};
\draw [very thick, blue] plot[smooth cycle, tension=.7] coordinates {(1.85,-0.05) (1.4,-0.05) (0.55,0.1) (0,0.5) (0.4,1.1) (1.4,1.3) (3,1.3) (4,1) (5.1,1.3) (6.1,1) (7.2,1.1) (7.7,0.9) (8.6,0.7) (9,0.2) (8,0) (7.1,-0.2) (6,-0.15) (5,0) (4,0) (3.3,-0.2) (2.5,0)};
\draw  (v10) edge (v6);
\draw  (v11) edge (v12);

\node [vertex] (v15) at (8,3.3) {};

\node [vertex] (v13) at (0.3,2.3) {};
\node [vertex] (v14) at (0.6,2.3) {};
\draw  (v9) edge (v13);
\draw  (v9) edge (v14);

\node [vertex] at (0.4,2.8) {};

\draw  (v15) edge (v6);
\node [vertex] (v16) at (8.5,2.6) {};
\draw  (v16) edge (v7);

\draw[very thick, blue]  plot[smooth cycle, tension=.7] coordinates {(0,-3.5) (-0.4,-3) (-0.7,-2.6) (-0.3,-2.3) (-0.3,-2) (0.8,-1.9) (2.3,-1.9) (2.9,-2.1) (2.5,-2.7) (2,-3.1) (1.4,-3.5) (0.7,-3.7)};
\draw  plot[smooth cycle, tension=.7] coordinates {(5,-3.4) (3.3,-2.3) (2.7,-2.1) (3.2,-1.7) (3.9,-1.5) (5,-1.4) (5.3,-1.7) (5.8,-2.4) (6.3,-2.6) (6,-3.1)};
\draw[very thick, blue]  plot[smooth cycle, tension=.7] coordinates {(8.3,-3.6) (7.4,-3.5) (6.7,-3) (6.1,-2.7) (6.8,-2.4) (7.8,-2.4) (8.3,-2.2) (9.1,-2.1) (9.4,-2.7) (9.4,-3.1) (9.2,-3.7)};
\node [vertex] (v18) at (4,-2) {};
\node [vertex] (v19) at (10,-2.9) {};
\draw  (v17) edge (v18);
\draw  (v17) edge (v19);

\node [vertex] at (6.2,-2.7) {};
\node [vertex] (v22) at (2.8,-2.05) {};
\node [vertex] (v20) at (2.55,-1.6) {};
\draw  (v11) edge (v20);
\draw  (v21) edge (v22);
\node [vertex] (v23) at (1.2,-3.15) {};
\draw  (v17) edge (v23);
\node [vertex] at (9.3,-3.1) {};
\node [vertex] (v24) at (8.75,-2.65) {};
\draw  (v17) edge (v24);

\node [vertex] at (4.55,-1.5) {};
\node [vertex] (v25) at (4.4,-1.2) {};
\draw  (v2) edge (v25);

\node [vertex] at (4.55,-1.25) {};
\node [vertex] (v28) at (2.3,1.1) {};
\node [vertex] at (3.2,0.9) {};

\node [vertex] (v26) at (3.8,2.6) {};
\node [vertex] (v27) at (4.8,2.4) {};
\draw  (v2) edge (v26);
\draw  (v6) edge (v27);
\draw  (v21) edge (v26);
\draw  (v11) edge (v28);
\node at (-4,0.9) {$P_1$};
\node at (13,0.9) {$P_2$};
\draw[bend right]  (v10) edge (v7);
\node at (3.8,4.1) {R1};
\node at (5,-2.7) {R3};
\node at (6.6,0.2) {R1};
\node at (4.2,2.5) {R0};
\node at (-3.3,-0.7) {R0};
\node [vertex] (v29) at (-4,-0.6) {};
\node [vertex] (v31) at (-2.3,-0.4) {};
\draw  (v29) edge (v30);
\draw  (v31) edge (v21);
\draw[bend right]  (v30) edge (v21);
\end{tikzpicture}
	\end{center}
	\caption{A schematic example of an \bananadec  for a path. The components are denoted by their respective types in the decomposition, R0 denotes components consisting entirely of vertices from $B$ (marked by light gray). The four \bananas are marked by thick blue borders.}
	\label{fig:bananapath}
\end{figure}
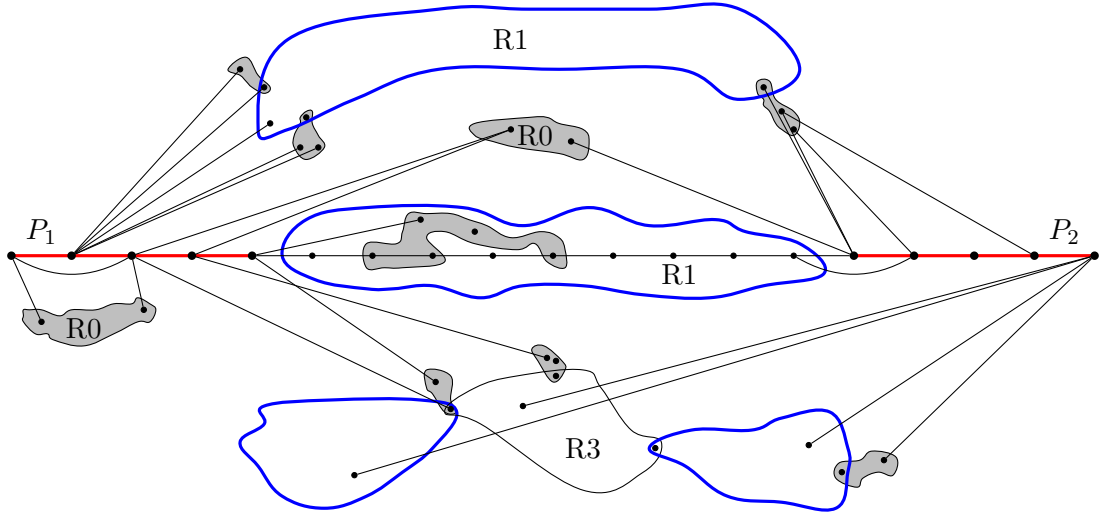

We are ready to introduce the primary tool for solving \probstP. This structure arises in the extremal cases when we cannot enlarge an $(s,t)$-path by local replacement used in the proof of the Erd{\H {o}}s-Gallai's theorem. This is where the name we use for the decomposition comes from.

\begin{definition}[\textbf{\Bananadec and \banana}]
Let $P$ be a path in  a $2$-connected graph $G$ and let $B\subseteq V(G)$.
We say that  two disjoint paths $P_1$ and $P_2$ in $G$ induce \emph{an \bananadec   for   $P$ and $B$}  in $G$ if
	\begin{itemize}
		\item 
		Path $P$ is of the form $P=P_1 {P'}P_2$, where the inner path ${P'}$ has at least $\delta(G- B)$ edges.
%
		\item 
            Let $G'$ be the graph obtained from $G$ by applying $B$-refinement to every  connected component $H$ of $G- V(P_1  \cup P_2)$,  except those components $H$ with  $V(H)\subseteq B$. Note that no edges of the paths $P_1$ and $P_2$ are contracted. 
            There are at least two connected components $H'$ in $G'-V(P_1\cup P_2)$ with $V(H')\not\subseteq B$.
For  every such connected component $H'$ holds $|V(H')|\ge 3$ and one of the following. 
		\begin{enumerate}[label=(R\arabic*)]
			\item\label{enum:tunnel_path_bic} $H'$ is $2$-connected and the maximum size of a matching in  $G'$ between $V(H')$ and $V(P_1)$  is one,  and between $V(H')$ and $V(P_2)$ is also  one;
						\item\label{enum:tunnel_path_cut_left} $H'$ is not 2-connected,   
exactly one vertex of $P_1$ has neighbors in $H'$, that is 			
			$|N_{G'}(V(H'))\cap V(P_1)|=1$, and no inner vertex from a  leaf-block of $H'$ has a neighbor in $P_2$;
			\item\label{enum:tunnel_path_cut_right} The same as  \ref{enum:tunnel_path_cut_left}, but with $P_1$ and $P_2$ interchanged. That is, 
$H'$ is not 2-connected,  			
			$|N_{G'}(V(H'))\cap V(P_2)|=1$, and no inner vertex from  a leaf-block of $H'$ has a neighbor in $P_1$.			
		\end{enumerate}
	
	\end{itemize}
The set of \emph{\banana}s for an \bananadec 
	 is defined as follows.
	First,  for each component $H'$ of type \ref{enum:tunnel_path_bic}, $H'$ is an \banana of the \bananadec.
	Second, for each  $H'$ of type \ref{enum:tunnel_path_cut_left}, or of type \ref{enum:tunnel_path_cut_right}, all its leaf-blocks are also \banana{s} of the \bananadec. The example of an \bananadec is given in Figure~\ref{fig:bananapath}.
\end{definition}

The following lemma provides a polynomial time algorithm that either finds a long path in the given graph or constructs an \bananadec. 

 \begin{lemma}\label{lemma:st_path_or_tunnel}
	Let $G$ be a   $2$-connected graph with two distinct vertices $s$ and $t$, $B\subseteq V(G)$ be a subset of vertices such that $s,t\in B$, and $k>0$ be an integer such that $4k+2|B|+4\le\delta(G-B)$. 
	There is a polynomial time algorithm that 
	\begin{itemize}
		\item either outputs an  $(s,t)$-path $P$  of length at least $\delta(G-B)+k$, 
		\item or outputs an  $(s,t)$-path $P$  with $V(P)\cup B =V(G)$,
		\item or outputs an  $(s,t)$-path $P$ with paths  $P_1, P_2$ that induce an \bananadec   for $P$ and $B$  in $G$.
	\end{itemize}
\end{lemma}
\begin{proof}
	By Corollary~\ref{thm:relaxed_st_path}, an $(s,t)$-path $P$  of length at least $\delta(G-B)$ can be  found in polynomial time.
	If the length of  $P$ is at least $\delta(G-B)+k$, we output it and stop. Otherwise, we try to make $P$ longer by replacing some of its parts with paths in $G- V(P)$.

	We first contract some edges of $G$ in a way similar to the definition of \bananadec{s}. For each connected component $H$ of $G-V(P)$ such that $V(H)$ is not in $B$, we perform $B$-refinement of $H$. That is, while $H$ is not 2-connected and has a leaf-block with all inner vertices from $B$, we contract all edges of this leaf-block. We denote the resulting graph by $G'$. Note that $G'$ still contains $P$ as a subgraph and that $\delta(G'-B)\ge\delta(G-B)$.
	If we find an $(s,t)$-path that is longer than $P$ in $G'$, this path can be easily transformed into a path of the same or greater length in $G$.
	Moreover, if we find paths $P_1$ and  $P_2$ that induce an \bananadec for $P$ in $G'$, then $P_1, P_2$ induce an \bananadec for $P$ in graph $G$ as well.
	Thus, from now on, we proceed  with the graph $G'$.
	
	We start with the trivial case. If  $V(G')\setminus P\subseteq B$,   then $V(P)\cup B=V(G)$.
	Hence, the algorithm just outputs $P$ and stops.
	
From now on we assume that 	
	  $(V(G')\setminus B) \setminus V(P)\neq\emptyset$. 
 Let $H'$ be a connected component  in $G'-V(P)$ that contains at least one vertex in $V(G')\setminus B$.
	We  consider several cases. The first case is a trivial case when $H'-B$ contains at most two vertices.
	The second case corresponds to {\banana}s of type
	\ref{enum:tunnel_path_bic}, while the third case to {\banana}s of types  \ref{enum:tunnel_path_cut_left} and \ref{enum:tunnel_path_cut_right}.

	If we find out that $P$ can be enlarged, we replace $P$ with the longer path in $G$ and start trying to make it longer again.
	Throughout the proof and all its claims, we consider that $P$ cannot be made longer with the replacement operation.

	\medskip\noindent	
	\textbf{Case 1:} \emph{$H'-B$ contains at most two vertices.}
	In this case, each vertex in $V(H'- B)$ has at least $\delta(G'-B)-2$ neighbors in $P$.
	If the length of $P$ is less than $2\delta(G-B)-4\ge \delta(G-B)+4k$, then each vertex in $V(H'-B)$ has two consecutive vertices in $P$ as neighbors.
	Hence, any such vertex can be inserted in $P$ between such two neighbors, so the length of $P$ increases by one.
	
	\textbf{Conclusion of Case 1.}  Either $H'-B$ consists of at least three vertices, or the length of $P$ can be increased (in polynomial time).
	
	\medskip\noindent
	\textbf{Case 2:} \emph{$H'$ is $2$-connected.} We start with the following claim. 
	\begin{claim}\label{claim:case2-connected}
	If there is a matching of size at least three between $V(H')$ and $V(P)$ in $G'$, then the length of $P$ can be enlarged  in polynomial time.
	\end{claim} 
	\noindent \emph{Proof of Claim~\ref{claim:case2-connected}.}  
	As $\delta(G'-B)\ge \delta(G-B)$,
	$2\delta(G'-B)-2>\delta(G-B)+4k+2|B|>\delta(G-B)+k$.
	So we assume that the length of $P$ is at most $2\delta(G'-B)-1$. 
	Let $u_1v_1, u_2v_2, u_3v_3$ be a matching in $G'$ such that $u_1, u_2, u_3 \in V(H')$ and $v_1, v_2, v_3 \in V(P)$.
	
	If no vertex in $V(H'-B)$ has a neighbor in $P$, then  $\delta(H'-B)\ge \delta(G'-B)$.
	By Corollary~\ref{thm:relaxed_st_path}, there is a path of length at least $\delta(G'-B)$ between any pair of vertices in $H'$.
	Because the length of $P$ is at most $2\delta(G'-B)-1<2\delta(G'-B)+4$, at least for one pair $\{v_i,v_j\}$, $i\neq j$, 
	the distance  
	  between  $v_i $ and $v_j$  in $P$ is less than $\delta(G'-B)+2$. Then    we   replace the $(v_i,v_j)$-subpath in $P$ with the path $v_i u_i \leadsto u_j v_j$, where $u_i \leadsto u_j$ is a path between $u_i$ and $u_j$ in $H'$ of length at least
	$\delta(G'-B)$. The length of  $v_i u_i \leadsto u_j v_j$ is  at least $\delta(G'-B)+2$ and hence we can enlarge $P$. 
	
	Now we assume that there is at least one vertex $w\in V(H'-B)$ with a neighbor in $P$. We can assume that in the matching $u_1v_1, u_2v_2, u_3v_3$,    one of the vertices $u_i=w$. (If all $u_i\in B$, we just replace $u_1$ with $w$.) 
Vertex $w$ has at least $\max\{1,\delta(G'-B)-\delta(H'-B)\}$ neighbors in $P$. Let $S$ be the set of neighbors of $u_1, u_2, u_3$ in $P$, that is,  $S:=(N_{G'}(u_1)\cup N_{G'}(u_2) \cup N_{G'}(u_3))\cap V(P)$.
%
%
%
%
%
 Then the size of $S$ is at least $\max\{\delta(G'-B)-\delta(H'-B),3\}$. Let $s_1,s_2, \dots, s_{|S|}$, be the order of vertices from $S$ in the path $P$. If the length of one of the subpaths $s_i, s_{i+1}$, $i\in \{1, \dots, |S|-1\}$,  of $P$  is   $1$, we can enlarge $P$ by 
replacing  $s_i, s_{i+1}$ with an $(s_i, s_{i+1})$-path of length at least 2  going through $V(H')$.
	Moreover, at least two of these subpaths go between neighbors of $u_i$ and $u_j$ for distinct $i$ and $j$.
If one of these paths, say  between $s_\ell$ and   $s_{\ell+1}$,  is of  length  less than  $\delta(H'-B)+2$, we 
can increase $P$ by  replacing it with  path $s_\ell  u_i \leadsto u_j s_{\ell+1}  $, where $u_i \leadsto u_j$ is a path between $u_i$ and $u_j$ in $H'$ of length at least $\delta(H'-B)$. This means that if we cannot enlarge $P$, then the length of $P$ is at least $2(|S|-3)+  2(\delta(H'-B)+2)\geq 2(\delta(G'-B)-\delta(H'-B)-3)+ 2(\delta(H'-B)+2)=2\delta(G'-B)-2>\delta(G-B)+k$.
%
%
\claimqed

By the claim and the fact that $G$ is $2$-connected, we can assume that  the maximum size of a matching between $V(H')$ and $V(P)$ in $G'$ is exactly two.

	\begin{claim}\label{claim:eg_path_bic_degree}
		There is a path of length at least $\delta(G'-B)-2$ between any pair of vertices in $H'$.
	\end{claim}
	\begin{claimproof}
		Let $h_1v_1, h_2v_2$ be the edges of the maximum matching between $V(H')$ and $V(P)$ in $G'$, where $h_1, h_2\in V(H')$, $v_1, v_2 \in V(P)$.
Note that no vertex in $V(H')\setminus\{h_1,h_2\}$ has neighbours in $V(P)\setminus \{v_1,v_2\}$.

If $h_1$ and $h_2$ have no neighbours other than $v_1$ and $v_2$ in $V(P)$, then, trivially, $N_G(V(H'))\cap V(P)=\{v_1,v_2\}$, so $\delta(H'-B)\ge \delta(G'-B)-2$.

Without loss of generality, we now assume that $h_1$ has a neighbour $v_3 \in V(P)\setminus \{v_1,v_2\}$.
Then no vertex in $V(H')\setminus \{h_1,h_2\}$ can have $v_1$ as a neighbour.
Analagously, if $h_2$ has a neighbour other than $v_1$ or $v_2$, no vertex in $V(H')\setminus\{h_1,h_2\}$ can have $v_2$ as a neighbour.
Hence, if $N_G(h_i)\not\subseteq \{v_1,v_2\}$ for both $i=1$ and $i=2$, then $\delta(H'-(B\cup \{h_1,h_2\}))\ge \delta(G'-B)-2$.

We now assume that $h_2$ has no neighbours other than $v_1$ and $v_2$ in $V(P)$.
If $h_2$ is adjacent to $v_1$, then no vertex in $V(H')\setminus \{h_1, h_2\}$ can be adjacent to $v_2$, as we would obtain a matching $h_1v_3$, $h_2v_1$, $h_3v_2$ of size at least three.
Hence, if $h_2v_1 \in E(G')$, then $\delta(H'-(B\cup\{h_1,h_2\}))\ge \delta(G'-B)-2$.
If $h_2$ is not adjacent to $v_1$, then all vertices in $V(H')\setminus\{h_1\}$ only can have $v_2$ as a neighbour, so $\delta(H'-(B\cup\{h_1\}))\ge \delta(G'-(B\cup\{h_1\}))-1\ge \delta(G'-B)-2$.

It is left to apply \Cref{thm:relaxed_st_path} to all of the cases.		
	\end{claimproof}

	Hence, if there is a matching between $V(H')$ and two vertices on $P$ that are closer than $\delta(G'- B)$ to each other and $P$ can be made longer.
	Suppose that we have a matching between $V(H')$ and $V(P)$ with endpoints $h_1,h_2 \in V(H')$ and $v_1, v_2 \in V(P)$, where $v_1$ is closer to $s$ on $P$ than $v_2$.
	Let $a_1$ denote the distance from $s$ to $v_1$ on $P$ and $a_2$ denote the distance from $v_2$ to $s$ on $P$.
	Then, if $|V(P)|+1-(a_1+a_2)<\delta(G'-B)$, $P$ can be enlarged using the long $(h_1,h_2)$-path of length at least $\delta(G'-B)-2$ in $H'$.
	Otherwise, $a_1+a_2\le |V(P)|+1-\delta(G'-B)$.
	In particular, $a_1,a_2 \le |V(P)|+1-\delta(G'-B)$.
	Thus, $v_1$ is within the first $|V(P)|+2-\delta(G'-B)$ vertices of $P$ and $v_2$ is within the last $|V(P)|+2-\delta(G'-B)$ vertices of $P$.

	\textbf{Conclusion of Case 2.} If $H'$ is $2$-connected, then either $P$ can be made longer or the following holds. 
	The size of the maximum matching between $H'$ and $P$ is exactly 2. Moreover, for any maximum matching between $H'$ and $P$, the endpoint of one edge of the matching is one of the first  $|V(P)|-\delta(G'- B)+2$ vertices of $P$ and one is within the last  $|V(P)|-\delta(G'- B)+2$  vertices of $P$. 
	
	\medskip\noindent
	\textbf{Case 3:} \emph{$H'$ is not $2$-connected.} Let $L$ be  a leaf-block $L$ of $H'$ and let $c$ be 
 the cut-vertex of the leaf-block  $L$.
	Note that $V(L)\setminus B\setminus \{c\}$ is not empty and $\delta(L- (B\cup \{c\}))\ge \delta(H'-B)-1$.
	
	Assume first that there is a matching of size three between $V(L)$ and $V(P)$ in $G'$.
	Similar to Case 2, then there is a vertex in $V(L-(B\cup \{c\}))$ with at least $\delta(G'-B)-\delta(L-(B\cup \{c\}))$ neighbors in $V(P)$. In this case, since the length of $P$ is at most  $\delta(G-B)+k< 2(\delta(G'-B)-\delta(L-(B\cup \{c\}))-1)+2\delta(L-(B\cup\{c\}))$, we can reroute a part of $P$ through $H'$ and thus make it longer. 
%

Now we may assume that the maximum matching size between $V(L)$ and $V(P)$ in $G'$ is at most two.
Again, similar to  Case 2 and \Cref{claim:eg_path_bic_degree} we derive  that $\delta(L- (B'\cup\{c\}))\ge \delta(G'-(B\cup\{c\}))-2\ge \delta(G'-B)-3$ for some $B'\supseteq B$.
	Hence, there is a path of length at least $\delta(G'-B)-3$ between any pair of vertices in $L$ by \Cref{thm:relaxed_st_path}.
	It follows that there is a path of length at least $\delta(G'-B)-3$ between an inner vertex of a leaf-block of $H'$ and any other vertex in $H'$.
	
	For each leaf-block in $H'$, there is at least one inner vertex that has at least one neighbor in $P$, otherwise $G'$ is not $2$-connected.
	
Suppose first that there are two inner vertices of two distinct leaf-blocks in $H'$ that have two distinct neighbors in $V(P)$.
There is always path between these two inner vertices of length at least $2(\delta(G'-B)-3)$: we can find two paths in each leaf-block starting in the cut vertex and ending in an inner vertex of length at least $\delta(G'-B)-3$.
	Since the length of $P$ is at most $\delta(G-B)+k\leq 2(\delta(G'- B)-3)+2$, the subpath of $P$ between their neighbours is shorter than if than the path between them through $H'$.
	So we can enlarge $P$ by using this path.
	
	Note that if there are at least two vertices $V(P)$ having at least one inner leaf-block vertex of $H'$ as a neighbour, then we can always pick two inner vertices as described in the previous paragraph.
	Hence, if $P$ cannot be made longer, there is exactly one vertex $v \in V(P)$ that is connected to inner vertices of the leaf-blocks of $H'$.
	Then, in fact, $\delta(L-(B\cup\{c\}))\ge \delta(G'-(B\cup\{c\}))-1\ge \delta(G'-B)-2$ for each leaf-block $L$ of $H$ with cut vertex $c$.
	The following claim holds.
	
	\begin{claim}\label{claim:eg_path_sep_degree}
		There is a path of length at least $\delta(G'-B)-2$ between any inner vertex of a leaf-block and any other vertex of $H'$.	
	\end{claim}

	Since $G$ is $2$-connected, there is at least one other vertex $u \in V(P)$ that has neighbors in $V(H')$.
	If the distance between $u$ and $v$ on $P$ is less than $(\delta(G'- B)-2)+2$, then $P$ can be made longer.
	As there is a path of length at least $\delta(G'- B)-2$ between their neighbours in $H$.
	Hence, $H'$ can only have neighbors among the first and among the last $|V(P)|+2-\delta(G'- B)$ vertices of $P$ analogously to Case 2.

	\textbf{Conclusion of Case 3.}
	If $H'$ contains at least three vertices and is not 2-connected, then 
	 either $P$ can be made longer, or the following properties hold. 
	All inner vertices of its leaf-blocks that have neighbors in $V(P)$ have exactly one neighbor on $P$, and this neighbour is the same for all inner vertices.
	This neighbour vertex is within the first (or the last) $|V(P)|+3-\delta(G'- B)$ vertices of $P$.
	All other neighbours of $V(H')$ on $P$ are, oppositely, within the last (or the first) $|V(P)|+2-\delta(G'- B)$ vertices of $P$.
	
	\medskip\noindent
	\textbf{Constructing \bananadec.}
We use the structural properties of $G'$ to construct an \bananadec in graph $G'$, and hence in $G$.
	We start from an $(s,t)$-path $P$ in $G'$ and try to increase its length by applying one of the algorithms from Cases 1--3. Assume that we cannot increase the length of $P$ anymore. Then we have a path $P$ and every connected component 
$H'$ of $G'-V(P)$ should satisfy the properties summarized in the conclusion of 	
 Case~2 or Case~3. We show that in this case we either could construct in polynomial time a new path $P$ of length at  least $\delta(G- B)+k$, or to construct an \bananadec.

	Then each $H'$ has neighbors within the first $k+2$ vertices of $P$ and within the last $k+2$ vertices of $P$.
	Denote by $P_1$ the shortest subpath of $P$ starting in $s$ that contains all starting neighbors (that is, neighbours that are closer to $s$ than to $t$ in $P$) among all possible components $H'$.
	Analogously, denote by $P_2$ the shortest subpath of  $P$ ending in $t$ that contains all ending neighbors (that is, neighbours that are closer to $t$ than to $s$ in $P$) among all possible $H'$. Thus $P=P_1 P' P_2$.
			
	\begin{claim}\label{claim:p'_long}
		The length of $P'$ is at least $\delta(G-B)-k$.
	\end{claim}
	\begin{claimproof}
	We know that $|V(P_1)|,|V(P_2)|\le |V(P)|-\delta(G'- B)+2$.	
	The length of each of $P_1$ and $P_2$ is at most $|V(P)|-\delta(G'-B)+1$, so the length of $P'$ is at least \[(|V(P)|+1)-2(|V(P)|-\delta(G'-B)+1)=2\delta(G'-B)-|V(P)|>\delta(G-B)-k.\]	
		\end{claimproof}
		
Denote by $s'$ and $t'$ the endpoints of $P'$, so $P_1$ and $P_2$ are the $(s,s')$-subpath and the $(t',t)$-subpath of $P$ respectively.
			The following claim is rather useful.
	
	\begin{claim}\label{claim:p'_no_neighbours}
		There is a connected component $H'$ in $G'-(V(P_1)\cup V(P_2))$ with $V(H')\setminus B=V(P'-\{s',t'\})\setminus B$.
	\end{claim}
	\begin{claimproof}
		We actually need to show that each vertex $v \in V(P'-\{s',t'\})$ can only have neighbours in $V(P)$ or $B$.
		Suppose that there is $v\in V(P'-\{s',t'\})$ with a neighbour $u \in V(G')\setminus V(P)\setminus B$.
		Then $u$ is in some connected component $H''$ of $G'-V(P)$ with $|V(H''-B)|\ge |\{u\}|>0$.
		Note that then $H''$ has a vertex with a neighbour in $P$ that is not in $V(P_1)\cup V(P_2)$.
		This contradicts the choice of $P_1$ and $P_2$.
	\end{claimproof}

We now show that the length of $P'$ can be actually assumed to be at least $\delta(G-B)$, as agrees with the definition of \bananadec.
This strengthens \Cref{claim:p'_long}.

\begin{claim}\label{claim:distP1}
If the distance  between $P_1$ and $P_2$ in $P$ is less than $\delta(G- B)$, then $G'$ contains an $(s,t)$-path of length at least $\delta(G-B)+k$. Moreover, this path can be computed in polynomial time.  
\end{claim}	
\begin{claimproof}
	Suppose that the distance between $P_1$ and $P_2$ in $P$ is less than $\delta(G- B)$.
	Equivalently, $|V(P'-\{s',t'\})|<\delta(G- B)-1$.
	Vertices in $P'-\{s',t'\}$ are adjacent in $G'$ only to vertices in $B$ and vertices from $V(P)$ by \Cref{claim:p'_no_neighbours}.
	Hence, each vertex in $V((P'-\{s',t'\})- B)$ has at least three neighbors in $V(P_1)\cup V(P_2)$, as it has at most $|V(P'-\{s',t'\})|-1\le\delta(G- B)-3$ neighbors in $V(P'-\{s',t'\})$.
	
	Consider the first vertex in $P'$ that is not in $B$ and is at distance at least $\delta(G- B)/2$ from the start of $P'$.	Denote this vertex by $v$.
	Note that the length of the $(s',v)$-subpath of $P'$ is at most $\delta(G-B)/2+|B|$.
	By \Cref{claim:p'_long}, the distance from $v$ to the last vertex of $P'$, i.e.\ the length of the $(v,t')$-subpath of $P'$, is 
	at least $(|V(P')|-1)-(\delta(G- B)/2+|B|)\ge \delta(G-B)-k-\delta(G-B)/2-|B|=\delta(G-B)/2-k-|B|$ .
	Hence, the distance from $v$ to each of the endpoints of $P'$ is  at least $\delta(G-B)/2-k-|B|$.
Vertex $v$  has at least two neighbors  in $V(P_1)$ or in $V(P_2)$, as it has at least three neighbours in $V(P_1)\cup V(P_2)$.
	Without loss of generality, assume that it has two neighbors in $P_1$.
	
	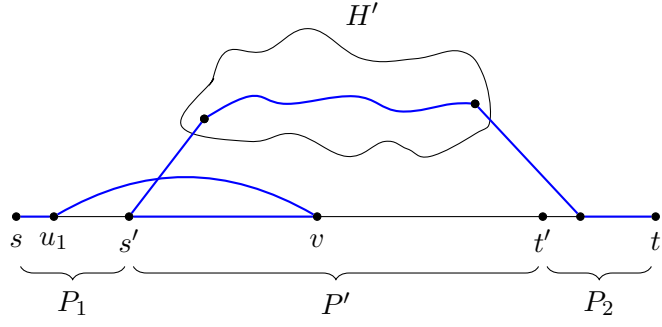
\begin{figure}
		\centering
		\usetikzlibrary{decorations.pathreplacing}
\ifdefined\STOC
\begin{tikzpicture}[scale=0.8]
\else
\begin{tikzpicture}
\fi
	\tikzstyle{vertex}=[circle,draw,fill,inner sep=1pt]
	\tikzstyle{inpath}=[blue, thick]
\node [vertex] (v1) at (-4.5,2.5) {};
\node [vertex] (v2) at (4,2.5) {};
\draw  (v1) edge (v2);
\node at (-4.5,2.18) {$s$};
\node at (4,2.18) {$t$};
\node [vertex] (v5) at (-3,2.5) {};
\node [vertex] (t') at (2.5,2.5) {};
\node at (-3,2.18) {$s'$};
\node at (2.5,2.18) {$t'$};
\node [vertex] (v4) at (-0.5,2.5) {};
\node [vertex] (v3) at (-4,2.5) {};
\node at (-4,2.18) {$u_1$};
\node at (-0.5,2.18) {$v$};
\draw[bend left, inpath]  (v3) edge (v4);
\draw  plot[smooth cycle, tension=.7] coordinates {(-1.9,4.3) (-1.7,4.7) (-1.2,4.7) (-0.6,5) (0.1,4.6) (0.8,4.6) (1.4,4.7) (1.7,4.4) (1.8,3.8) (1.6,3.6) (1,3.3) (0.7,3.4) (0.4,3.5) (-0.1,3.3) (-0.8,3.6) (-1.4,3.5) (-2.3,3.7) (-1.9,4.3)};
\node [vertex] (v6) at (-2,3.8) {};
\draw[inpath]  (v5) edge (v6);
\node [vertex] (v7) at (1.6,4) {};
\node [vertex] (v8) at (3,2.5) {};
\node at (3,2.18) {};
\draw[inpath]  (v7) edge (v8);
\node at (0.1,5.2) {$H'$};
\draw [inpath] plot[smooth, tension=.7] coordinates {(v6) (-1.4,4.1) (-0.9,4) (0,4.1) (0.6,3.9) (1.3,4) (v7)};
\draw[inpath]  (v1) edge (v3);
\draw[inpath]  (v5) edge (v4);
\draw[inpath]  (v8) edge (v2);
\node [vertex] at (1.6,4) {};
\node [vertex] at (-2,3.8) {};
\draw [decorate,decoration={brace,amplitude=5pt,mirror,raise=4ex}]
  (v1) -- (v5) node[midway,yshift=-3em]{$P_1$};
  \draw [decorate,decoration={brace,amplitude=5pt,mirror,raise=4ex}]
  (v5) -- (t') node[midway,yshift=-3em]{$P'$};
  \draw [decorate,decoration={brace,amplitude=5pt,mirror,raise=4ex}]
  (t') -- (v2) node[midway,yshift=-3em]{$P_2$};
\end{tikzpicture}
		\caption{Construction of a long path in $G'$ when $P'$ is shorter than $\delta(G-B)$.}\label{fig:eg_path_p'_short}
	\end{figure}

	One of its neighbors, say $u_1$, is different from $s'$.
	Construct an $(s,t)$-path as follows. Start from $s$, move to $u_1$ along $P_1$, then from $u_1$ to $v$, then follow the path $P'$ backwards from $v$ to $s'$. By the construction of 
	$P_1$,   there is at least one component $H'$ in $G'- V(P)$ that is connected with $s'$. Thus from 
	 $s'$ we enter $H'$, and follow a path of length at least $\delta(G'-B)-2$ in $H'$ (such path always exists by either \Cref{claim:eg_path_bic_degree} or \Cref{claim:eg_path_sep_degree}) to reach some vertex in $P_2$.
	  We complete the construction of the path by 
	 following along $P_2$  to $t$ (see \Cref{fig:eg_path_p'_short}).
	The length of the constructed path is at least $$\underbrace{1}_{s\leadsto u_1\leadsto v}+\underbrace{\delta(G-B)/2-k-|B|}_{v\leadsto s'}+\underbrace{1+(\delta(G-B)-2)+1)}_\text{$s'\leadsto t$ through $H'$},$$ which equals $\frac{3}{2}\delta(G-B)-k-|B|+1>\delta(G-B)+k$.
\end{claimproof}

By the claim, if the length of $P'$ is less than $\delta(G- B)$, then we find in polynomial time the desired path and stop. 	
Otherwise,  the distance between $P_1$ and $P_2$ in $P$ is at least $\delta(G- B)$, hence $|V(P_1)|+|V(P_2)|\le k+1$ as $|V(P)|\le \delta(G-B)+k$.

	\begin{claim}\label{claim:p'_is_bic}
		The connected component $H'$ from \Cref{claim:p'_no_neighbours} is of type \ref{enum:tunnel_path_bic} in $G'-(V(P_1) \cup V(P_2))$, or a path of length at least $\delta(G-B)+k$ in $G'$ can be found in polynomial time.
	\end{claim}
	\begin{claimproof}
		We first show that $H'$ is $2$-connected after $B$-refinements are applied to it.
		Denote the component $H'$ with applied $B$-refinements by $H''$ and assume that $G'=\gbref{B}{H'}$.
		If $H''$ is not $2$-connected, then it contains at least two leaf-blocks, as $|V(P'-\{s',t'\})|\ge \delta(G-B)-1>2$.
		Since $\delta(H'-B)\ge \delta(G-(B\cup(V(P_1)\cup V(P_2))))\ge \delta(G-B)-k-1$, each leaf-block of $H''$ should contain at least $\delta(G-B)-k$ vertices outside $B$.
		Hence, $H'-B$ consists of at least $2(\delta(G-B)-k)-1\ge \delta(G-B)+k\ge |V(P)|$ vertices.
		This is not possible since $V(H'-B)\subseteq V(P'-\{s',t'\})$ and $|V(P'-\{s',t'\})|<|V(P)|$.
		
		It is left to show that the matching conditions of type \ref{enum:tunnel_path_bic} are also satisfied.
		Assume that these conditions do not hold.
		Without loss of generality, assume that the maximum matching size between $V(H'')$ and $V(P_1)$ is at least two in $\gbref{B}{H'}$.
		Then there are two edges $v_1 h_1, v_2 h_2 \in E(G')$ with $v_1,v_2 \in V(P_1)$ and $h_1, h_2 \in V(H'')$.
		Without loss of generality, we assume that $v_1$ is closer to $s$ on $P$ than $h_2$.
		In particular, $v_1 \neq s'$.
		As $H''$ is $2$-connected, then by \Cref{thm:relaxed_st_path},  it contains a path of length at least $\delta(H''-B)=\delta(H'-B)\ge \delta(G-B)-k-1$ between $h_1$ and $h_2$.
		As discussed above in the proof of \Cref{claim:p'_long} (see \Cref{fig:eg_path_p'_short}), there is a path connecting $s'$ with some vertex in $P_2$ going through a component $H$ in $G-V(P)$. 
		Hence, there is an $(s',t)$-path of length at least $\delta(G-B)$ that does not have common vertices with $H''$. Then we concatenate the following paths. 
		Take the $(s,v_1)$-subpath of $P_1$, proceed further with the edge $v_1h_1$ and the $(h_1,h_2)$-path inside $H''$, then with the edge $h_2v_2$ and the $(v_2,s')$-subpath of $P_1$.
		Finish with the $(s',t)$-path.
		The obtained path is an $(s,t)$-path of length at least $$\underbrace{1}_{s\leadsto h_1}+\underbrace{\delta(G-B)-k-1}_{h_1\leadsto h_2}+\underbrace{1}_{h_2\leadsto s'}+\underbrace{\delta(G-B)}_{s'\leadsto t},$$
		which equals to $2\delta(G-B)-k+1>\delta(G-B)+k$.
		Thus, if $H''$ is not of type \ref{enum:tunnel_path_bic}, then we can find a long path in $G'$ in polynomial time.
	\end{claimproof}
	
	Note that every connected component in $G'-V(P_1 \cup  P_2)$ corresponds either to Case 2,  or to Case 3,  or to \Cref{claim:p'_is_bic}, or is fully contained in $B$.
	A connected component from Case 2 or \Cref{claim:p'_is_bic} corresponds to \ref{enum:tunnel_path_bic}-type connected components of {\bananadec}s.
	The connected components from Case 3 correspond to \ref{enum:tunnel_path_cut_left}-type and \ref{enum:tunnel_path_cut_right}-type connected components depending on whether the vertex $v$ is from $V(P_1)$ or from $V(P_2)$.
	Thus, $P_1$ and $P_2$ induce an \bananadec  for $P$ and $B$ in $G'$, and hence in $G$.
%
\end{proof}

The following proposition about long paths inside 
\banana{s} is clear from the proof of \Cref{lemma:st_path_or_tunnel}.

\begin{proposition}
	For any \banana of any \bananadec in $G$ for $B\subseteq V(G)$, there is a path of length at least $\delta(G-B)-2$ between any pair of vertices of this \banana.
\end{proposition}

We start to establish the properties of \banana{s} that will be exploited by the algorithm. To state the first property, we need the following definition. 
\begin{definition} We say that a path $P$ \emph{enters} a subgraph $H$, if at least one edge of $H$ is also an edge of $P$. 
\end{definition}
Informally, the  property is the following. Consider an \banana $M$ for some  \bananadec and consider also an $(s,t)$-path $P'$.  Path $P'$ can hit some vertices of $M$. However, if $P'$ enters $M$, then  all vertices of $H$ hit by $P$, that is, all  common vertices of $P$ and $M$,  appear  consecutively in $P'$.  

\begin{lemma}\label{lemma:st_path_banana_consecutive}
	Let $G$ be a $2$-connected graph, $B\subseteq V(G)$, $P$ be an $(s,t)$-path in $G$. Let paths  $P_1, P_2$ induce an \bananadec  for $P$ and $B$ in $G$. Let also $G'$ be the graph 
obtained 	after $B$-refinements   of connected components of 
	$G- V(P_1  \cup P_2)$, 
	 and let  $M$ be an \banana. Then for every $(s,t)$-path $P'$ in $G'$, if $P' $ enters $M$, then all vertices of $M\cap V(P')$   appear consecutively in $P'$.
\end{lemma}
\begin{proof}
Targeting towards a contradiction, assume that the statement of the lemma does not hold. 
Then there is an $(s,t)$-path $P'$ that contains at least one edge of $M$, but vertices of $M$ does not appear consecutively in $P'$.
	That is, there are  vertices $v_1, v_2 \in V(M)$, $v_1\neq v_2$,  such that $P'$ is of the form 
	$s, \dots, v_1, \dots, x,\dots, v_2, \dots, t$, where  $x\not \in V(M)$. No internal vertex of 
	the $(s,v_1)$-subpath   and the $(v_2,t)$-subpath of $P'$  belongs to $V(M)$.
	Moreover, the $(v_1, v_2)$-subpath of $P'$ contains at least one edge of $M$ and at least one edge outside of $M$.

	Let $G'$ be the graph obtained from  $G$ after applying all possible   $B$-refinements. 
	 According to the definition of \banana, $M$ can be one of the following three types. Either it is a connected component of 
	 $G'- V(P_1  \cup P_2)$ (this corresponds to type \ref{enum:tunnel_path_bic}), or it is a leaf-block of a  connected component of 
	 $G'- V(P_1  \cup P_2)$ (this corresponds to types \ref{enum:tunnel_path_cut_left} and \ref{enum:tunnel_path_cut_right}).
	Therefore, we consider three cases. 
	
	\medskip\noindent
	\textbf{Case 1.}
	Suppose that $M$ is an \banana of type \ref{enum:tunnel_path_bic}. That is, $M$ is a connected component of $G'- V(P_1  \cup P_2)$  and also $M$ is 2-connected. 
	Consider the $(s,v_1)$-subpath of $P'$ in $G$. Since $s\in V(P_1)$, there exists vertex 
	  $w_1$ that is  the last vertex on this subpath that is from $V(P_1\cup P_2)$.
	Then the subpath is of form $s\leadsto w_1 \leadsto v_1$, where all inner vertices of the subpath $w_1\leadsto v_1$ are from  $V(H)\setminus  V(M)$, where $H$ is the connected component $M$ before the $B$-refinements.  But after the $B$-refinement of $H$ in $G$, all inner edges of this path are contracted. Then in $G'$ this $(w_1,v_1)$-subpath consists of just single edge $w_1 v_1$.

	Analogously, consider the $(v_2, t)$-subpath of $P'$ and let $w_2$ be the first vertex from $V(P_1)\cup V(P_2)$ on this subpath.
	The $(v_2, w_2)$-subpath goes only through vertices in $V(H)\setminus  V(M)$ in $G$ and turns into  the edge between $v_2$ and $w_2$ in $G'$.
	
	The last subpath to consider is the $(v_1,v_2)$-subpath of $P'$.
	It goes between vertices in $M$ and contains at least one edge outside $M$; hence it should contain at least one vertex in $V(P_1)\cup V(P_2)$.
	Let $u$ be the first vertex on this subpath that is from $V(P_1)\cup V(P_2)$.
	Then either the $(v_1,u)$-subpath or the $(u,v_2)$-subpath contains an edge of $M$.

	First, suppose that the $(v_1,u)$-subpath contains an edge of $M$. 
	Denote by $v_3$ the last vertex from $V(M)$ on this subpath.
	Then $v_3\neq v_1$ and the $(v_3,u)$-subpath contains only vertices in $V(H)\setminus  V(M)$ as internal vertices.
	Hence, in this case there is an edge between $v_3$ and $w_3=u$ in $G'$.

	Now  for the case when the $(u,v_2)$-subpath contains an edge of $M$.
	Denote by $v_3$ the first vertex on this subpath that is from $M$.
	Then $v_3\neq v_2$ and the $(u,v_3)$-subpath does not contain vertices of $M$ as internal vertices.
	Denote by $w_3$ the last vertex in $V(P_1)\cup V(P_2)$ on this subpath.
	We obtain a path between $w_3$ and $v_3$ that goes only through $V(H)\setminus  V(M)$ in $G$, so there is an edge between $w_3$ and $v_3$ in $G'$.
	
	\medskip\noindent\textbf{Conclusion of Case 1.}
	If $M$ is an \banana corresponding to a connected component of type \ref{enum:tunnel_path_bic}, then there is a matching $v_1w_1, v_2w_2, v_3w_3$ of size three between $V(M)$ and $V(P_1)\cup V(P_2)$ in $G'$.
	Hence, there is a matching between $V(M)$ and $V(P_i)$ of size two for some $i\in\{1,2\}$.
	This contradicts to the corresponding condition  \ref{enum:tunnel_path_bic} of {\bananadec}s; hence Case~1 cannot occur. 
	
\medskip\noindent\textbf{Case 2.}
	Now suppose that there is a type \ref{enum:tunnel_path_cut_left} connected component $H$ in $G-(V(P_1)\cup V(P_2))$ such that $M$ is a leaf-block of the component obtained after some edges of $H$ were contracted in the process of $B$-refinement $\gbref{B}{H}$. 
	Denote 
	the cut-vertex of this leaf-block $M$ by $c$. We will refer to all remaining vertices of $M$ as to \emph{inner} vertices. 
	By the definition of \ref{enum:tunnel_path_cut_left}-type components, $N_{G'}(V(M))\cap V(P_1)=\{w\}$ for some $w \in V(P_1)$.
	Again consider the $(s,v_1)$-subpath, the $(v_1, v_2)$-subpath, and the $(v_2,t)$-subpath of $P'$.
	The $(v_1, v_2)$-subpath contains a vertex $u \in V(P_1)\cup V(P_2)$ as internal vertex, so we can also break it into $(v_1,u)$-subpath and $(u,v_2)$-subpath.
	
	Note that at least two of these four subpaths do not contain $c$.
	Each of these subpaths is an $(x,y)$-path for $x\in V(P_1)\cup V(P_2)$ and $y\in V(H)$.
	We claim that if such $(x,y)$-path does not contain $c$, then it contains $w$.
	Suppose that an $(x,y)$-path does not contain $c$, so $y$ is an inner vertex of $M$.
	This path does not contain $c$, and to reach $y$ it should reach some inner vertex of $M$ from the outside, since the path starts in $V(P_1)\cup V(P_2)$.
	Hence, this path should contain $w$.  Otherwise there is an edge between $V(P_2)$ and some inner vertex of $M$ in $G'$, which contradicts the property \ref{enum:tunnel_path_cut_left}.
	
	Thus at least two of the four subpaths contain $w$.
	The only possible option for this is when $u=w$ and both $(v_1,u)$-subpath and $(u,v_2)$-subpath do not contain $c$.
	Then both $(s,v_1)$-subpath and $(v_2,t)$-subpath do not  contain $w$, since $w$ can appear only once in $P'$.
	Both of them reach an inner vertex of $M$ from the outside of $M$.
	If a path reaches an inner vertex of $M$ and avoids $w$, then it should contain the cut-vertex $c$.
	Therefore,  the $(s,v_1)$-subpath and the $(v_2,t)$-subpath both contain $c$.
	This is contradiction, since these two paths are vertex-disjoint.

	\medskip\noindent\textbf{Conclusion of Case 2.}
	If $M$ is an \banana corresponding to a connected component of type \ref{enum:tunnel_path_cut_left}, then $P'$ necessarily contains an edge between $V(P_2)$ and an inner vertex of $M$.
	This contradicts the definition of \ref{enum:cycle_tunnel_path_cut_left}-type components.
	
	\medskip\noindent\textbf{Case 3.} The case when $M$ is an \banana of  type \ref{enum:tunnel_path_cut_right}  is symmetrical.
	
	In each of the three cases we obtained a contradiction with one of the properties of an \bananadec. This completes the proof.
\end{proof}

In order to proceed further with the structural properties of \bananadec{s}, we need the following definition and lemma. 

\begin{definition}[\textbf{$B$-leaf-block separator}]\label{definition:b_leaf_block_sep}
	Let $H$ be a connected graph that is not $2$-connected and $B$ be a subset of its vertices.
	Let $I$ be the set of inner vertices of all leaf-blocks of $H$.
	We say that $S \subseteq V(H)\setminus I$ is a \emph{$B$-leaf-block separator} of $H$, if $S$ separates at least one vertex in $V(H)\setminus (I\cup B)$ from $I$ in $H$.
\end{definition}
\begin{lemma}\label{lemma:separator_in_non_2c}
	Let $H$ be a connected graph with at least one cut-vertex and let $B$ be a subset of its vertices.
	Let $S$ be a $B$-leaf-block separator of $H$.
	Then for any vertex $v$ that is not an inner vertex of a leaf-block of $H$, there is a cut-vertex $c$ of a leaf-block of $H$ and a $(c,v)$-path of length at least $\frac{1}{2}\left(\delta(H-B)-|S|\right)$ in $H$.
\end{lemma}
\begin{proof}
	We assume that $\delta(H-B)> |S|$, since the other case is trivial.
	
	Consider   graph $H-(B\cup S)$.
	We know that there is at least one connected component in this graph that does not contain any vertex from $I$ and contains at least one vertex not in $B$.
	Denote this connected component by $T$.
	We know that $\delta(T)\ge \delta(H-(B\cup S))>1$.
	By \Cref{prop:cycle_delta}, $T$ contains a cycle $C$ of length at least $\delta(T)+1$.

	We know that $C$ is fully contained in some non-leaf-block of $H$.
	Denote this block by $K$.
	Now let $v$ be a vertex in $V(H)\setminus I$ given from the lemma statement.
	It is easy to see that we can always choose the vertex $c$ in a way that any $(c,v)$-path contains at least one edge of $K$.
	Take such vertex and an arbitrary $(c,v)$-path.
	Edges of $K$ induce a subpath of non-zero length in this path.
	Let $x,y$ be the endpoints of this subpath.
	We know that $x\neq y$.

We need the following claim. 	
	\begin{claim}\label{lemma:biconnected_cycle_to_any_path}
	If a $2$-connected graph contains a cycle on $k$ vertices, then it contains a path of length at least $\lceil \frac{k}{2} \rceil$ between any pair of vertices.
\end{claim}
\medskip
\noindent \emph{Proof of Claim~\ref{lemma:biconnected_cycle_to_any_path}.}  
	Take two distinct vertices $s$, $t$.
	To show that there is a path between $s$ and $t$ of length at least $\lceil \frac{k}{2} \rceil$, we  apply Menger's theorem to $\{s,t\}$ and the vertex set of the cycle of length $k$.
	This gives two vertex-disjoint paths going from $s$ and $t$ to two vertices $s'$ and $t'$ on the cycle.
	Take the longer arc between $s'$ and $t'$ on the cycle and combine it with the two paths. The resulting path is of length at least $\lceil \frac{k}{2} \rceil$. 
\claimqed \medskip
	
	By Claim~\ref{lemma:biconnected_cycle_to_any_path}, there is a path of length at least $\delta(T)/2>\frac{1}{2}(\delta(H-B)-|S|)$ between $x$ and $y$ in $K$.
	Replace the subpath of the initial $(c,v)$-path with this subpath.
	This yields a $(c,v)$-path of desired length.
\end{proof}

In Lemma~\ref{lemma:st_path_banana_consecutive}, we proved that if a path enters an \banana, then after leaving it, it cannot come back. The following lemma guarantees, that if we have a yes-instance, then there is a solution path that enters at least one \banana.
\begin{lemma}\label{lemma:st_path_edge_of_banana}
	Let $G$ be a graph, $B\subseteq V(G)$ be a subset of its vertices and $P_1, P_2$ induce an \bananadec  for an $(s,t)$-path $P$ in $G$ of length less than $\delta(G-B)+k$.
	Let $k$ be an integer such that $5k+4|B|+6 < \delta(G-B)$.
	If there exists an $(s,t)$-path of length at least $\delta(G-B)+k$ in $G$, then there exists $(s,t)$-path of length at least $\delta(G-B)+k$ in $G$ that  
enters an \banana.
\end{lemma}
\begin{proof}
	Since the length of $P$ is less than $\delta(G-B)+k$, we may assume that $|V(P_1)\cup V(P_2)|<k+2$.
	
	Assume that there is an $(s,t)$-path $P'$ of length at least $\delta(G-B)+k$ in $G$ that contains no edge of an \banana.
	We show that there exists an $(s,t)$-path of length at least $\delta(G-B)+k$ that enters some \banana.
	
	The path $P'$ path can contain only edges with endpoints in $V(P_1)\cup V(P_2) \cup B$ or edges of non-leaf-blocks of \ref{enum:tunnel_path_cut_left}- or \ref{enum:tunnel_path_cut_right}-type components.
	All other edges are edges of  \banana{s}.
	There are at most $2|V(P_1)\cup V(P_2)\cup B|\le 2(k+2+|B|)$ edges in $P'$ that have endpoints in the corresponding set.
	Hence, $P'$ contains at least $\delta(G-B)+k-2(k+2+|B|)=\delta(G-B)-k-4-2|B|$ edges that lie inside non-leaf-blocks of separable components of the \bananadec .

	Let $u$ be the vertex on $P'$ such that the $(s,u)$-subpath of $P'$ is of length exactly $k+1$.
	Denote this subpath by $P'_1$.
	Analogously, let $v$ be the vertex on $P'$ such that the $(v,t)$-subpath of $P'$ is of length exactly $k$, and denote this subpath by $P'_2$.
	Note that $P'_1$ and $P'_2$ are on a distance at least $\delta(G-B)-k>0$ from each other on $P'$.
	The $(u,v)$-subpath of $P'$ consists of at least $(\delta(G-B)-k-4-2|B|)-2k> 2|B|$ edges of the non-leaf-blocks.
	Hence, at least one non-leaf-block edge in $P'$ is not incident to any vertex in $B$.
	
	Let $H$ be the connected component in $G-(V(P_1)\cup V(P_2))$ that contains this edge and let $H'$ be its $B$-refinement.
	The graph $H'$ contains at least one edge of the $(u,v)$-subpath of $P'$, and none of these edges are incident to an inner vertex of its leaf-blocks.
	We note that the whole path $P'$ cannot go through any inner leaf-block vertex of $H'$.
	Suppose that this is not true and it contains such vertex.
	Since it does not contain any leaf-block edge, this path should enter and leave this inner vertex from the outside of $H$.
	And the only way to enter a \ref{enum:tunnel_path_cut_left}-type or a \ref{enum:tunnel_path_cut_right}-type component of the \bananadec  is to go from the only vertex of $V(P_1)$ or of $V(P_2)$ correspondingly.
	Thus, this vertex of either $V(P_1)$ or $V(P_2)$ is contained twice on the path, and that is not possible.
	Consider now the graph $H'-(V(P'_1)\cup V(P'_2))$.
	
	\textbf{Case 1.}
	Suppose that there is a connected component in this graph that contains an inner vertex of a leaf-block of $H'$ and some vertex of the $(u,v)$-subpath of $P'$ simultaneously.
	Denote this leaf-block by $L$ and the vertex of the $(u,v)$-subpath by $w$.
	Note that all paths between $w$ and vertices of $L$ go through the cut-vertex of $L$.
	If there are multiple choices of $w$ for $L$, choose the one which is the closest to the cut-vertex of $L$.
	As $P'$ does not contain any inner leaf-block vertex of $H'$, the connected component of $w$ in $H'-(V(P'_1)\cup V(P'_2))$ contains the whole leaf-block $L$.
	Hence, there is a path connecting $w$ with any inner vertex of $L$.
	Choose any inner vertex of $L$ that is connected to $V(P_1)$ (if $H'$ is \ref{enum:tunnel_path_cut_left}-type) or to $V(P_2)$ (if $H'$ is \ref{enum:tunnel_path_cut_right}-type).
	Denote this vertex by $z$.
	Since $L$ is a \banana, there is a path in $L$ of length at least $\delta(G-B)-2$ connecting $z$ with the cut-vertex of $L$, so there is a $(w,z)$-path of length at least $\delta(G-B)-2$ in $H'$.
	Note that the only common vertex of this path and $P'$ is the vertex $w$.
	We also know that the $(s,w)$-subpath and the $(w,t)$-subpath of $P'$ are of length at least $k+1$, since $w$ is not in $V(P'_1)\cup V(P'_2)$.
	
	Now prolong the $(w,z)$-path in $G$ by going  outside $H'$ from $z$ to the vertex from $V(P_1)$ or $V(P_2)$ depending on the type of $H$, and finally go from this vertex to $t$ following the initial path $P$.
	We obtain a $(w,t)$-path $Q$ that has at least two common vertices with $P'$.
	
	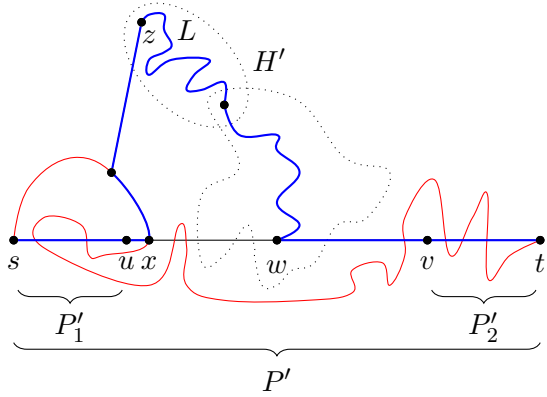
\begin{figure}
		\centering
		\usetikzlibrary{decorations.pathreplacing}
\begin{tikzpicture}
	\tikzstyle{vertex}=[circle,draw,fill,inner sep=1pt]
	\tikzstyle{inpath}=[blue, thick]
\begin{scope}
\node [vertex] (s) at (-3,0.5) {};
\node [vertex] (t) at (4,0.5) {};
\node  at (-3,0.2) {$s$};
\node  at (4,0.2) {$t$};
\draw [decorate,decoration={brace,amplitude=5pt,mirror,raise=4ex}]
  ($(s)+(0,-0.7)$) -- ($(t)+(0,-0.7)$) node[midway,yshift=-3em]{$P'$};
\node [vertex] (v1) at (-1.5,0.5) {};
\node [vertex] (v2) at (2.5,0.5) {};
\draw  (s) edge (v1);
\draw  (v2) edge (t);
\draw [decorate,decoration={brace,amplitude=5pt,mirror,raise=4ex}]
  (s) -- (v1) node[midway,yshift=-3em]{$P'_1$};
  \draw [decorate,decoration={brace,amplitude=5pt,mirror,raise=4ex}]
  (v2) -- (t) node[midway,yshift=-3em]{$P'_2$};
\draw  (v1) edge (v2);
\node [vertex] (v3) at (0.5,0.5) {};
\node at (0.5,0.15) {$w$};
\draw[dotted]  plot[smooth cycle, tension=.7] coordinates {(0.2,0.3) (0.3,-0.1) (0.5,-0.1) (0.7,0.1) (1,0.2) (1.3,0.7) (1.7,0.8) (2,1.4) (1.4,1.6) (0.7,2.4) (-0.4,2.4) (0,1.5) (-0.5,1.1) (-0.5,0.2) (-0.2,0.3) (0.1,0.7) (0.2,0.4)};

\draw[dotted, rotate=-45]  (-2.5,1.5) ellipse (1 and 0.6);
\node at (-0.7,3.3) {$L$};
\node [vertex] (v4) at (-0.2,2.3) {};
\node at (-1.5,0.2) {$u$};
\node at (2.5,0.2) {$v$};
\draw[inpath]  plot[smooth, tension=.7] coordinates {(v3) (0.8,0.7) (0.5,1) (0.8,1.4) (0.5,1.7) (0.5,1.9) (0,1.9) (v4)};
\node [vertex] (v5) at (-1.3,3.4) {};
\draw [inpath] plot[smooth, tension=.7] coordinates {(-0.2,2.3) (-0.2,2.6) (-0.5,2.5) (-0.7,2.5) (-0.6,2.7) (-0.4,2.9) (-0.9,2.9) (-1.2,2.7) (-1.2,2.9) (-1,3) (-0.9,3.1) (-1,3.3) (-1,3.5) (-1.2,3.5) (v5)};

\draw[red]  plot[smooth, tension=.7] coordinates {(s) (-2.8,1.2) (-2.2,1.6) (-1.6,1.3) (-1.2,0.5) (-1.8,0.3) (-2,0.3) (-2.1,0.7) (-2.6,0.8) (-2.6,0.4) (-1.2,-0.1) (-0.8,0.7) (-0.7,0.4) (-0.7,-0.1) (-0.2,-0.3) (0.9,-0.3) (1.5,-0.2) (1.5,0.1) (1.9,0.2) (2,-0.1) (2.2,0.3) (2.2,0.8) (2.6,1.3) (2.8,0.3) (3.1,0.8) (3.3,1.1) (3.3,0) (3.6,0.3) (t)};

\node [vertex] (v6) at (-1.7,1.4) {};

\node at (-1.2,3.2) {$z$};
\node at (0.4,2.9) {$H'$};
\draw  (v5) edge (v6);
\node [vertex] (v7) at (-1.2,0.5) {};
\node at (-1.2,0.2) {$x$};
\draw[inpath]  (s) edge (v7);
\begin{scope}[]
    \clip ($(v7)+(0.1,0)$) rectangle (v6);
    \draw[inpath]  plot[smooth, tension=0.7] coordinates {(s) (-2.8,1.2) (-2.2,1.6) (-1.6,1.3) (-1.2,0.5) (-1.8,0.3) (-2,0.3) (-2.1,0.7) (-2.6,0.8) (-2.6,0.4) (-1.2,-0.1) (-0.8,0.7) (-0.7,0.4) (-0.7,-0.1) (-0.2,-0.3) (0.9,-0.3) (1.5,-0.2) (1.5,0.1) (1.9,0.2) (2,-0.1) (2.2,0.3) (2.2,0.8) (2.6,1.3) (2.8,0.3) (3.1,0.8) (3.3,1.1) (3.3,0) (3.6,0.3) (t)};
  \end{scope}

\draw[inpath]  (v6) edge (v5);
\draw[inpath]  (v3) edge (t);

\node [vertex] (v6) at (-1.7,1.4) {};
\node [vertex] (v7) at (-1.2,0.5) {};
\node [vertex] (s) at (-3,0.5) {};
\node [vertex] (t) at (4,0.5) {};
\node [vertex] (v4) at (-0.2,2.3) {};
\node [vertex] (v3) at (0.5,0.5) {};
\node [vertex] (v5) at (-1.3,3.4) {};
\node [vertex] (v1) at (-1.5,0.5) {};
\node [vertex] (v2) at (2.5,0.5) {};
\end{scope}

\end{tikzpicture}
		\caption{Constructing a path entering an \banana in Case 1. The path $P$ is highlighted red. The constructed path is thick blue.}\label{fig:lemma_banana_edge}
	\end{figure}
	
	Denote by $x$ the second vertex on this $(w,t)$-path that is common with $P'$.
	Denote the $(w,x)$-subpath of $Q$ by $Q'$.
	The path $Q'$ is of length at least $\delta(G-B)-1$ and contains at least one \banana edge, since it contains the $(w,z)$-path as a proper subpath.
	Suppose that $x$ is a part of the $(s,w)$-subpath of $P'$.
	Then consider constructing the following path (see \Cref{fig:lemma_banana_edge}).
	Take the $(s,x)$-subpath of $P'$, then go following the path $Q'$ from $x$ to $w$, and finish with the $(w,t)$-subpath of $P'$.
	The constructed path is of length at least $0+(\delta(G-B)-1)+(k+1)\ge \delta(G-B)+k$, so we are done.
	If $x$ is not a part of the $(s,w)$-subpath of $P'$, then it is a part of the $(w,t)$-subpath of $P'$.
	Then the required path is combined of the $(s,w)$-subpath of $P'$, then of $Q'$ and of the $(x,t)$-subpath of $P'$.
	Its total length is at least $(k+1)+(\delta(G-B)-1)+0\ge \delta(G-B)+k$.
	
	\textbf{Case 2.}
	It is left to consider the case when $V(P'_1)\cup V(P'_2)$ separates all inner leaf-block vertices from all vertices of the $(u,v)$-subpath of $P'$ that are from $V(H')$.
	Note that at least one vertex of the $(u,v)$-subpath is from $V(H')\setminus B$.
	Denote the set of all inner leaf-block vertices in $H'$ by $I$.
	We know that $V(P'_1)\cup V(P'_2)$ separates at least one vertex in $V(H')\setminus B$ from $I$.
	Apply Lemma~\ref{lemma:separator_in_non_2c} to $H'$ and $S=V(P'_1)\cup V(P'_2)$.
	
	Suppose that $H'$ is of type \ref{enum:tunnel_path_cut_left}.
	Then take any vertex in $V(H')$ that is connected with $V(P_2)$ by an edge (after the edge contractions).
	Denote this vertex by $y$.
	We know that $y$ is not in $I$, so there is a $(c,y)$-path of length at least $\frac{1}{2}(\delta(H'-B)-|S|)$ in $H'$ for cut-vertex $c$ of some leaf-block $L$ in $H'$.
	This leaf-block has at least one inner vertex that is connected to $V(P_1)$ by an edge.
	Denote such vertex by $x$.
	There is a path of length at least $\delta(G-B)-3$ between $x$ and $c$ inside $L$.
	Combine this $(x,c)$-path with the $(c,v)$-path and obtain a path of length at least $1+(\delta(G-B)-3)+(\frac{1}{2}(\delta(H'-B)-|S|))+1$ between $V(P_1)$ and $V(P_2)$.
	This path does not intersect internally with $P_1$ or $P_2$.
	Hence, there is an $(s,t)$-path in $G$ of length at least $\delta(G-B)+\frac{1}{2}\delta(H'-B)-\frac{1}{2}|S|-1$.
	We know that $\delta(H'-B)\ge \delta(G-(B\cup V(P_1)\cup V(P_2)))\ge \delta(G-B)-k-1$.
	Thus, our $(s,t)$-path is of length at least $\delta(G-B)+\frac{1}{2}((\delta(G-B)-k-1)-(2(k+1))-2)=\delta(G-B)+\frac{1}{2}(\delta(G-B)-3k-5)\ge \delta(G-B)+k$.
	This path contains an edge of the leaf-block $L$, which is an \banana edge, so we are done.
	
	When $H'$ is of type \ref{enum:tunnel_path_cut_right}, the proof is symmetrical.
	The proof of the lemma is complete.
\end{proof}

We are now ready to formulate a very crucial lemma of this section.
It serves as a basic tool for applying recursion in \bananadec in the algorithm for \probstP.
Basically, it provides a way to search for a long part of the $(s,t)$-path inside an \banana wrapped up in a $2$-connected subgraph of $G$.

\begin{lemma}\label{lemma:st_path_banana_to_2_connected}
	Let  paths $P_1,P_2$ induce an \bananadec  for an $(s,t)$-path $P$ and $B\subseteq V(G)$ in graph $G$.
	Let $M$ be an \banana in $G$. Then there is a polynomial time algorithm that outputs a  $2$-connected subgraph $K$ of $G$ and two vertices $s', t' \in V(K)$, such that every 
	 $(s,t)$-path $P'$ in $G$ that enters $M$, the following hold
	\begin{enumerate}
		\item $V(K)\setminus B=(V(M)\cup\{s',t'\})\setminus B$;
		\item $P'[V(K)]$ is an $(s',t')$-subpath of $P'$ and an $(s',t')$-path in $K$;
		\item $\delta(K-(B\cup\{s',t'\})) \ge \delta(G-(B\cup\{s',t'\}))$;
	\end{enumerate}
\end{lemma}

\begin{proof}
	We consider several cases depending on the type of the connected component $H$ of $G-(V(P_1)\cup V(P_2))$ that contains $M$.
	We start with the simpler case, when $H'=\gbref{H}{B}$ is separable.
	By $G'$ we as usual denote the graph $G$ where all edges corresponding to $B$-refinements of the \bananadec  induced by $P_1,P_2$ are applied.
	
\medskip\noindent\textbf{Case 1}.
	The component $H$ is of type \ref{enum:tunnel_path_cut_left} (type \ref{enum:tunnel_path_cut_right} is symmetrical as ususal).
	Then $M$ is some leaf-block of $H'$.
	If an $(s,t)$-path $P'$ enters $M$, then $P'[V(M)\cap V(P')]$ is a path in $L$ by Lemma~\ref{lemma:st_path_banana_consecutive}.
	Moreover, we know that this path starts in an inner vertex of $M$ and ends in the cut-vertex of $M$.
	Denote these two vertices by $v$ and $c$ respectively.
	
	Also, $N_G(H)=\{w\}$ for some $w\in V(P_1)$, and by definition of \ref{enum:tunnel_path_cut_left}-type connected components, $P'$ contains a $(w,v)$-subpath going (in either direction) internally only through vertices in $V(H)\cap B$.
	There are two cases of how $K$ should be constructed.
	
	If there is a single and the only inner vertex $v \in V(M)$ such that there is an edge between $w$ and $v$ in $G'$, then any path $P'$ that enters $M$ contains a path between $v$ and $c$ as a subpath.
	Thus, put $K:=M$ and $s':=v$, $t':=c$.
	Clearly, $K,s',t'$ satisfy all three conditions in the lemma statement.
	
	The other case is when there are at least two inner vertices in $M$ that are neighbors to $w$ in $G'$. 
	We cannot put $s'$ equal to any vertex of $M$, because we cannot be sure that $P'$ passes through a concrete inner vertex.
	But we are sure that $P'$ passes through $w$.
	Construct $K$ in the following way.
	Denote by $B'$ the set of vertices in $B$ that are reachable from $V(M)\setminus\{c\}$ in $H-\{c\}$.
	Then put $K:={G[V(M)\cup B' \cup \{w\}]}$, $s'=w$, $t'=c$.
	Note that $K$ is an induced subgraph of $G$ and is $2$-connected as $G'[V(M)\cup \{w\}]$ is $2$-connected.
	The first and the last two conditions in the lemma statement are satisfied, and we claim that the second one is satisfied as well.
	
	We already know that $P'$ contains a $(w,c)$-subpath.
	This subpath goes from $w$ to an inner vertex of $M$ through the vertices in $B'$, and then follows a path inside $M$.
	Hence, this subpath is contained in $K$.
	It is left to show that no vertex from $V(M)\cup B'\cup \{w\}$ can appear in $P'$ outside of the $(w,c)$-subpath.
	We do it by contradiction.
	Assume that there is such vertex $v \in V(M)\cup B' \cup \{w\}$.
	If $v \in V(M)$, then $v \neq c$, hence $v$ is an inner vertex of $M$.
	Then $P'$ should contain a $(s,v)$-subpath or a $(v,t)$-subpath that does not go through $w$ nor $c$, but $\{w,c\}$ separates $V(M)$ from $V(P_1)\cup V(P_2)$.
	Thus, $v \in B'$.
	Then there exists either $(s,v)$-subpath or $(v,t)$-subpath in $P'$ that does not contain $w$ and any vertex from $V(M)$.
	Hence, this subpath connects $v$ with some vertex $u \in V(P_2)$ and goes only through $B'$.
	We know that after the edge contractions for $G'$ the vertex $v$ becomes identified with an inner vertex of $M$, so there is an edge between $u$ and this inner vertex.
	This is not possible by the definition of type \ref{enum:tunnel_path_cut_left} connected components.
	We obtain a contradiction.
	
	\medskip\noindent\textbf{Case 2.}
	$H$ is of type \ref{enum:tunnel_path_bic}, so $H'$ is $2$-connected and $M=H'$.
	We know that the maximum matching size between $V(P_i)$ and $V(M)$ in $G'$ is exactly one for each $i\in\{1,2\}$.
	For each $i$, it splits into two possible options: either $|N_{G'}(V(P_i))\cap V(M)|=1$ or $|N_{G'}(V(M))\cap V(P_i)|=\{w_i\}$, where $w_i$ has at least two neighbors in $V(M)$ in $G'$.
	We now consider several subcases of Case 2 depending on the combinations of these options.
	
	If for each $i\in \{1,2\}$, $|N_{G'}(V(P_i))\cap V(M)|=1=\{v_i\}$ for some $v_i\in V(M)$, then, an $(s,t)$-path $P'$ can enter or leave $M$ only through the vertices $v_1$ and $v_2$.
	Note that $v_1\neq v_2$, since $\{v_1,v_2\}$ separates $V(M)$ from the rest of the graph in $G$.
	Thus, if $P'$ enters $M$, then it necessarily contains a $(v_1, v_2)$-subpath inside $M$.
	By Lemma~\ref{lemma:st_path_banana_consecutive}, we have that $P'[V(P')\cap V(M)]$ is exactly the $(v_1, v_2)$-subpath inside $M$.
	Thus, it is enough to put $K:=M$ and $s':=v_1$, $t':=v_2$.
	The first two and the last conditions of the lemma are satisfied for this choice of $K, s'$ and $t'$.
	Also, no vertex in $V(M)\setminus \{v_1,v_2\}$ has neighbors outside $V(H)$ in $G$, so $\delta(K-(B\cup \{s',t'\}))=\delta(M- \{v_1,v_2\}])\ge \delta(G-(B\cup\{s',t'\}))$, and the third condition is also satisfied.
	
	The other case is when for each $i \in \{1,2\}$, $|N_{G'}(V(M))\cap V(P_i)|=\{w_i\}$, where $w_i$ has at least two neighbors in $V(M)$ in $G'$.
	It is easy to see that to enter or leave any vertex of $H$ in $G$, an $(s,t)$-path $P'$ should go through $w_1$ and $w_2$.
	Since $G$ is $2$-connected, $w_1\neq w_2$ and $P'$ contains a $(w_1, w_2)$-subpath going internally only through vertices in $V(H)$.
	Put $K:=G[V(H)\cup \{w_1,w_2\}]$, $s':=w_1$, $t':=w_2$.
	Clearly, $K$ is $2$-connected because $G$ is $2$-connected, $\{w_1,w_2\}$ separates $V(H)$ from the rest of $G$, and degrees of $w_1$ and $w_2$ in $K$ are at least two.
	
	We need to show that the second condition is satisfied as well.
	If it is not satisfied, then $P'[V(K)]$ consists of at least two disjoint paths.
	We know that one of these paths is the $(w_1,w_2)$-subpath.
	Hence, the other one contains at least one vertex from $V(H)$ but does not contain $w_1$ or $w_2$.
	This is not possible since $\{w_1,w_2\}$ separates $V(H)$ from the rest of the graph.

	Thus, the first two and the last condition are satisfied.
	It is easy to see that the third condition is satisfied as well, because vertices in $V(H)$ have no outside neighbors apart from $s'$ and $t'$ in $G$.
	
	It is left to consider the case when $N_{G'}(V(M))\cap V(P_1)=\{w_1\}$, where $w_1$ has at least two neighbors in $V(M)$ in $G'$, and $N_{G'}(V(P_2))\cap V(M)=\{v_2\}$ (the case when $1$ and $2$ are interchanged is symmetrical).
	This is the most non-clear case.
	We know that if $P'$ enters $M$, then it should pass through both $w_1$ and $v_2$.
	Moreover, the $(w_1,v_2)$-subpath of $P'$ goes internally only through $V(H)$.
	Let $B'$ be the set of vertices reachable from $v_2$ by the edges in $E(H)\setminus E(H')$.

	The difficulty beyond choosing $K$ in this case is to satisfy the second condition.
	We split on two cases.
	
	Assume that there is no edge between $w_1$ and $v_2$ in $G'$.
	Then put $K:=G[(V(H)\cup \{w_1\})\setminus B']$, $s':=w_1$ and $t':=v_2$.
	Clearly, $K$ is equal to $G[V(H)\cup \{w_1\}]$ with applied $B$-refinements so it is $2$-connected.
	The first, the third and the fourth condition of the lemma are satisfied by the arguments similar to the cases considered above.
	It is left to show that the second condition is satisfied.
	Suppose that $P'[V(K)]$ contains a path different from the $(w_1,v_2)$-subpath.
	Then there is at least one vertex $u \in V(K)\setminus V(M)$ that is not on this subpath.
	Then $P'$ should contain either an $(s,u)$-subpath or an $(t,u)$-subpath that does not go through $w_1$ or $v_2$.
	Moreover, this subpath does not contain any vertex of $V(M)$ by Lemma~\ref{lemma:st_path_banana_consecutive}.
	Denote by $x$ the last vertex on this supbath that is not from $V(K)$.
	Then $P'$ contains an $(x,u)$-subpath, where $x \in V(P_1)\cup V(P_2)$
	After the $B$-refinement of $H$, this path yields an edge in $G'$ between $x$ and $y$ for some $y \in V(M)$.
	This is only possible when either $x=w_1$ or $y=v_2$.
	
	The case $x=w_1$ is not possible because the $(x,u)$-subpath does not contain $w_1$.
	Hence, it should be the case that $y=v_2$.
	Then $u$ is a vertex reachable from $v_2$ in $H$ outside $V(M)$.
	That is, $u \in B'$.
	Hence, $u \notin V(K)$.
	We obtain a contradiction, so all four conditions are satisfied.

	It is left to consider the case when there is an edge between $w_1$ and $v_2$ in $G'$.
	It is clear that in this case the graph $G[V(H)\cup \{w_1\}]$ is $2$-connected.
	Unfortunately, we cannot put $K$ equal to this graph because this might break the second condition of the lemma.

	We already know, however, that the graph $K:=G[(V(H)\cup \{w_1\})\setminus B']$ with $s':=w_1$ and $t':=v_2$ would satisfy all conditions of the lemma except, possibly the first.
	Thus, there are two cases.
	
	When the graph $G[(V(H)\cup \{w_1\})\setminus B']$ is $2$-connected, then consider $K,s'$ and $t'$ similarly to the case when there is no edge between $w_1$ and $v_2$.
	
	Otherwise, the graph $G[(V(H) \cup \{w_1\})\setminus B']$ is not $2$-connected.
	Then $w_1$ is connected to exactly two vertices from $V(M)$ in $G'$.
	One of these two vertices is $v_2$.
	The other one we denote by $v_1$.
	Then $P'$ necessarily contains a $(v_1,v_2)$-subpath inside $M$.
	Then it is sufficient to put $K:=M$, $s':=v_1$, $t':=v_2$, as $M$ is $2$-connected.
	
	The proof is complete.
\end{proof}

%

 \subsection{Algorithm for \probstP}\label{subsec:algorithmstpath}
 We are almost set to proceed with the proof of Theorem~\ref{thmEG}. The   algorithm is based on Lemmata
 ~\ref{lemma:st_path_or_tunnel}, \ref{lemma:st_path_banana_consecutive},  
\ref{lemma:separator_in_non_2c}, and \ref {lemma:st_path_banana_to_2_connected} on  properties of \bananadec{s}. For the proof of the correctness of the algorithm, we will need one more lemma.

\begin{lemma}\label{lemma:almost_ham_path}
	Let $G$ be a $2$-connected graph with $B\subseteq V(G)$ such that $\frac{6}{5}\delta(G-B)\ge |V(G)|$ and $\delta(G-B)\ge 4|B|$.
	Then for any pair of distinct vertices $s,t\in V(G)$, the longest $(s,t)$-path in $G$ contains all vertices from $V(G-B)$.
\end{lemma}
\begin{proof}
	The proof is by contradiction.
	Suppose that there is an $(s,t)$-path $P$ in $G$ such that the length of $P$ is maximum possible, but there is $v \in V(G)\setminus B$ with $v \notin V(P)$.
	By Corollary~\ref{thm:relaxed_st_path}, the length of $P$ is at least $\delta(G-B)$.
	Hence, $|V(P-B)|> \delta(G-B)-|B|$.
	Then $v$ has at most $|V(G-B)|-|V(P-B)|<\frac{1}{5}\delta(G-B)+|B|$ neighbors outside $V(P)$.
	Hence, $v$ has more than $\frac{4}{5}\delta(G-B)-|B|$ neighbors from $V(P)$.
	
	Note that $v$ should not have any two consecutive vertices in $P$ as neighbors, otherwise $P$ can be made longer.
	Hence, $2(\frac{4}{5}\delta(G-B)-|B|)<|V(P)|$.
	Equivalently, $|V(P)|>|V(G)|+\frac{3}{5}\delta(G-B)-2|B|\ge |V(G)|$.
	This is a contradiction.
\end{proof}

For reader's convenience, we restate Theorem~\ref{thmEG}  here.  

\medskip\noindent\textbf{Theorem~\ref{thmEG}.} \emph{\probstP is solvable in $2^{\mathcal{O}(k+|B|)}\cdot n^{\mathcal{O}(1)}$ running time on $2$-connected graphs.}

\begin{proof} 
	The recursive algorithm is presented in Algorithm~\ref{alg:long_eg_st_path}.
	Note that this algorithm requires that $s,t\in B$ in the given input instance.
	Any instance can be reduced to instance with this restriction by adding $s,t$ into $B$ and increasing $k$ by at most two.
	This changes the parameters by a constant value and does not significantly affect the running time of the algorithm.
	Also, this algorithm does not just determine whether the given instance is a yes-instance.
	If the given instance is a no-instance, the algorithm also outputs the maximum length of an $(s,t)$-path in $G$ in the form $\delta(G-B)+x$, where $x\ge 0$ and $x < k$.
	Note that algorithm actually also finds a path of such length, and it possible to change it so the path is in the output of the algorithm.
	We now go through the lines of the algorithm to explain its correctness.
	
	\SetNlSty{}{}{}

	\let\oldnl\nl
	\newcommand\nonl{%
	\renewcommand{\nl}{\let\nl\oldnl}}
	\begin{algorithm}[!h]
		\SetKwFunction{LongestPath}{longest\_path}
		\SetKwFunction{LongSTPath}{long\_st\_path}
		\SetKwFunction{LongErdosSTPath}{long\_eg\_st\_path}
		\SetKwFunction{LongSTCycle}{long\_st\_cycle}
		\SetKwFunction{HamPath}{hamiltonian\_path}
		\Indm\nonl\LongErdosSTPath{$G,B,s,t,k$}
		
		\Indp
		
		\KwIn{an instance $(G,B,s,t,k)$ of \probstP, where $G$ is $2$-connected and $s,t\in B$}
		\KwResult{$k$, if $(G,B,s,t,k)$ is a yes-instance, or an integer $x$, such that the maximum length of an $(s,t)$-path in $G$ is $\delta(G-B)+x$.}
		
		\If{$k=0$}{
			\Return{$k$}\;
		}
		\If{$5k+5|B|+6\ge\delta(G-B)$}{
			$x\longleftarrow k$\;
			\While{\LongSTPath$(G,s,t,\delta(G-B)+x)$ is \textsc{No}}{
				$x \longleftarrow x-1$\;
			}
			\Return $x$\;
		}
		
		\uIf{the algorithm of Lemma~\ref{lemma:st_path_or_tunnel} applied to $G,B,s,t,k$ returns $P$ with $P_1, P_2$}{
			$x\longleftarrow 0$\;
			\ForEach{\banana $M$ of the \bananadec  of $(G,B)$ induced by $P_1, P_2$\label{alg:line:st_path_bananas}}{
				$K,s',t' \longleftarrow $ result of Lemma~\ref{lemma:st_path_banana_to_2_connected} applied to $G,B,P_1,P_2$ and $M$\;
				$x'\longleftarrow$ \LongErdosSTPath $(K,B\cup\{s',t'\},s',t',k)$\;

				$H \longleftarrow (V(G)\cup \{a,b\}, (E(G)\setminus E(G[T])) \cup \{as,at,bs', bt'\})$\;
				$r \longleftarrow \max\{(\delta(G-B)+k)-(\delta(K-(B\cup \{s',t'\}))+x'),0\}$\;
				\While{\LongSTCycle$(H,a,b,r+4)$ is \textsc{No}}{$r \longleftarrow r-1$\;}
				$x \longleftarrow \max\{x, (\delta(K-(B\cup \{s',t'\}))+x'+r)-\delta(G-B)\}$\;				
			}
			\Return $x$\;
		}\uElseIf{it returns $P$ with $V(P)\cup B=V(G)$\label{alg:line:ham_path}}{
			\ForEach{$B'\subseteq B\setminus\{s,t\}$}{
				$H \longleftarrow (V(G-B')\cup \{s',t'\}, E(G-B')\cup \{s's, tt'\})$\;
				\If{\HamPath$(H)$}{
					$x \longleftarrow \max\{x, |V(H)|-1-\delta(G-B)\}$\;
				}
			}
			\Return $\min\{x,k\}$\;
		}\Else{
			\Return {$k$}\;
		}
		\caption{Recursive algorithm solving \probstP on $2$-connected graphs.}
		\label{alg:long_eg_st_path}
	\end{algorithm}

	The first two conditional operators handle the most trivial cases of the problem.
	The first conditional operator is for  the case $k=0$, which corresponds to trivial yes-instances by Corollary~\ref{thm:relaxed_st_path}.
	The second operator ensures that parameters $k$ and $|B|$ are small enough compared to $\delta(G-B)$ to apply results discussed earlier in this section.
	If they are not, the algorithm just employs the algorithm from \Cref{prop:longest_cycle} 
	for \textsc{Long $(s,t)$-Path}, which works in $2^{\Oh(\delta(G-B)+k)}\cdot\polyn=2^{\Oh(k+|B|)}\cdot\polyn$.

	When the third conditional operator is reached, Lemma~\ref{lemma:st_path_or_tunnel} can indeed be applied to the input instance.
	Thus, in polynomial time either an $(s,t)$-path of length at least $\delta(G-B)+k$ is found, or an $(s,t)$-path $P$ with $V(P)\cup B=V(G)$ is found, or an $(s,t)$-path $P$ and two paths $P_1, P_2$ are found.
	The paths $P_1$ and $P_2$ induce an \bananadec  for $P$ in $(G,B)$.
	If the path of length at least $\delta(G-B)+k$ is found, our algorithm correctly decides that the given instance is a yes-instance and stops.
	Otherwise, it enters the third conditional operator body.

	The conditional operator in line \ref{alg:line:ham_path} checks that we should deal with the case covered by Lemma~\ref{lemma:almost_ham_path}.
	We shall now explain this in detail.
	Suppose that we enter the conditional operator body, i.e., \ $V(P)\cup B=V(G)$.
	Since the length of $P$ is at most $\delta(G-B)+k-1$, we get that $\delta(G-B)+k+|B|\ge |V(G)|$.
	Since this operator can be reached only if $5(k+|B|)\le |V(G)|$, we can apply Lemma~\ref{lemma:almost_ham_path} to our instance.
	We can now look for an $(s,t)$-path that contains all vertices in $V(G-B)$.
	Clearly, any such path is a Hamiltonian path in the graph $G-B'$ for some $B'\subseteq B$ with $s,t\notin B'$.
	To achieve that any Hamiltonian path in $G-B'$ corresponds to an $(s,t)$-path, we add two additional vertices $s', t'$ of degree one to obtain the graph $H$.
	
	Moreover, in the graph $H$ all vertices have degree at least $\delta(G-B)$, except, probably, at most $|B|+2$ vertices.
	We know that $2\delta(G-B)>|V(H)|$, so we can apply one of the two FPT-algorithms from \Cref{theorem:JansenKN} 
	 for solving \textsc{Hamiltonian Path} in $H$.
	This algorithm runs in $2^{\Oh(|B|+2)}\cdot\polyn$ time.
	Thus, the longest path in $G-B$ is found by the algorithm for the correct choice of $B'$.
	
	We now move to the most crucial part of the algorithm.
	This part deals with \banana{s} of the \bananadec  induced by $P_1$ and $P_2$.
	We note that when line~\ref{alg:line:st_path_bananas} of the algorithm is reached, there are at least two distinct \banana{s} of the \bananadec  of $(G,B)$ by  definition.
	By Lemma~\ref{lemma:st_path_edge_of_banana} and Lemma~\ref{lemma:st_path_banana_to_2_connected}, if the given instance is a yes-instance, there is an \banana that contains a long subpath of the desired path.
	Let $M$ be an \banana fixed by the foreach cycle.
	Lemma~\ref{lemma:st_path_banana_to_2_connected} applied to this \banana yields a triple $K,s',t'$.
	The following lines of the algorithm focus on finding maximum $x'$ such that there is an $(s',t')$-path in $K$ of length at least $\delta(K-(B\cup\{s',t'\}))+x'$.
	We know that such path in $K$ exists for $x'=0$ by Corollary~\ref{thm:relaxed_st_path}.
	We shall analyze the running time of this recursion later in this proof.
	
	Note that any $(s',t')$-path in $K$ can be expanded to an $(s,t)$-path in $G$ using at least $p:=|\{s',t'\}\setminus\{s,t\}|$ edges.
	Also, $s,t\in B$, so $|B\cup \{s',t'\}|\le |B|+p$.
	Hence, if there is an $(s',t')$-path in $K$ of length at least $\delta(K-(B\cup\{s',t'\}))+x'\ge \delta(G-(B\cup \{s',t'\}))+x'\ge \delta(G-B)-p+x'$, there is a path of length at least $\delta(G-B)+x'$ in $G$.
	It follows that if $x' \ge k$ then the algorithm can safely decide that the given instance is a yes-instance.
	
	Otherwise, the maximum possible $x'<k$ is found and it is left for the algorithm to expand the $(s',t')$-path in $K$ to an $(s,t)$-path in $G$.
	That is, it needs to find two vertex-disjoint paths of sufficient total length going from ${s',t'}$ to ${s,t}$ in $G$.
	An additional restriction for these paths is that they should not contain any edge of $M$.
	Since the sufficient total length is bounded by $k+2$, we can safely employ the algorithm for \probTLDP, Theorem~\ref{thmTLDP}  from Section~\ref{sec:tldp},  running in $2^{\Oh(k)}\cdot\polyn$ time.
	
	The correctness of the algorithm is now clear and we move to analyze the recursion running time.
	We know that without the recursive call, the algorithm runs in $2^{\Oh(k+|B|)}\cdot\polyn$ time.
	For convenience, we write this running time bound in the form $2^{\Oh(k+|B|)}\cdot(n-2)^{\Oh(1)}$.
	Note that this is possible since $n>2$ for any $2$-connected graph $G$.
	Thus, we can already assume that if the algorithm runs without making recursive calls, it runs in $2^{c_1(k+|B|)}\cdot(n-2)^{c_2}$ time, where $c_1,c_2\ge 1$ are constant integers given by the non-recursive subroutine.
	
	Since the recursive call is made when the graph contains at least two \banana{s}, it is always made to an instance with the smaller number of vertices.
	We will now prove that our algorithm runs in $2^{c_1(k+|B|)}\cdot (n-2)^{c_2+1}$ time by induction on $n$.
	
	The base of our induction are instances for which no recursive calls are made.
	Consider an instance for which at least two recursive calls are made.
	We want to prove that the algorithm running time $2^{\Oh(k+|B|)}\cdot\polyn$.
	First note that the parameter $k+|B|$ does not increase in a recursive call, because $|(B\cup \{s',t'\})\cap V(K)|\le |B\cap V(G)|$.

	Let $q\ge 2$ be the number of \banana{s} in $G$.
	For $i\in [q]$, denote by $K_i, s'_i, t'_i$ the triple given by Lemma~\ref{lemma:st_path_banana_to_2_connected} for the $i$-{th} \banana of $G$.
	Denote also $n_i:=|V(K_i)|$.
	The running time of the algorithm for the instance given by $K_i$ is at most $2^{c_1(k+|B|)}\cdot (n_i-2)^{c_2+1}$ by induction.
	Note that all $q$ sets $V(K_i)\setminus\{s'_i, t'_i\}$ are pairwise disjoint.
	Also, none of these sets contains $s$ or $t$.
	Hence, $\sum_{i=1}^{q} (n_i-2)\le n-2$.
	We now want to upper-bound the sum $\sum_{i=1}^q (n_i-2)^{c_2+1}$.
	
	\begin{proposition}
		Let $a_1, a_2, \ldots, a_q$ be a sequence of $q\ge 2$ positive integers with $\sum_{i=1}^q a_i=n$.
		Let $x>1$ be an integer.
		Then $\sum_{i=1}^q a_i^x \le (n-1)^x+1 \le n^x - n^{x-1}$.
	\end{proposition}
	\begin{proof}
		First, we show that the maximum of the sum $\sum_{i=1}^q a_i^x$ is achieved with $q=2$, $a_1=n-1$, $a_2=1$, if the sum $\sum_{i=1}^q a_i=n$ is fixed.
		To show that the maximum cannot be achieved with $q>2$, it is enough to see that replacing $a_{q-1}$ and $a_q$ with $a_{q-1}+a_q$ yields a greater total sum, as $(a_{q-1}+a_q)^x>a_{q-1}^x+a_q^x$.
		
		We know that the maximum is achieved with $a_1^x+a_2^x$ for some positive integers $a_1, a_2$ with $a_1+a_2=n$.
		Without loss of generality, we can assume that $a_1\ge a_2$.
		Suppose that $a_2>1$.
		Consider replacing $a_1$ with $a_1+1$ and $a_2$ with $a_2-1$.
		We need to show that the total sum does not decrease, i.e., \ $(a_1+1)^x+(a_2-1)^x\ge a_1^x+a_2^x$, or $(a_1+1)^x-a_1^x\ge a_2^x-(a_2-1)^x$.
		Rewrite the left and the right part to obtain
		\[\sum_{i=0}^x \binom{x}{i}a_1^i-a_1^x\ge a_2^x-\sum_{i=0}^x\binom{x}{i}a_2^i(-1)^{x-i}.\]
Then		\[\sum_{i=0}^{x-1}\binom{x}{i}a_1^i\ge -\sum_{i=0}^{x-1}\binom{x}{i}a_2^i(-1)^{x-i},\]
and 
		\[\sum_{i=0}^{x-1}\binom{x}{i}(a_1^i+a_2^i(-1)^{x-i})\ge 0.\]
		Each summand of the sum in the last inequality is non-negative since $a_1^i\ge a_2^i$ for any $i\ge 0$.
		Thus, the initial inequality holds and we can replace $(a_1,a_2)$ with $(a_1+1,a_2-1)$ if $a_2>1$ so the total sum does not decrease.
		Hence, the maximum is achieved with $a_1=n-1$ and $a_2=1$.
		
		It is left to show that $(n-1)^x+1\le n^x-n^{x-1}$.
		We rewrite it as
		\[1\le (n-1)\cdot (n^{x-1}-(n-1)^{x-1}),\]
		which  holds as $n>1$ and $x>1$.
		The proof is complete.
	\end{proof}

	With this proposition,  we have  that  the running time of the algorithm is upper-bounded by
	\begin{multline*}	
	2^{c_1(k+|B|)}\cdot (n-2)^{c_2}+2^{c_1(k+|B|)}\cdot \sum_{i=1}^q n_i^{c_2+1}\le 2^{c_1(k+|B|)}\cdot ((n-2)^{c_2}+((n-2)^{c_2+1}-(n-2)^{c_2})),
	\end{multline*}
	so the induction hypothesis holds.
	This concludes the proof.

\end{proof}

\section{Algorithm for small vertex covers}\label{sec:vcalgo}
In this section we prove Theorem~\ref{thmVCad} stating that \textsc{\probDC\ / Vertex Cover Above Degree} is solvable in $2^{\Oh(p+|B|)}\cdot n^{\Oh(1)}$ running time. Recall that the task of this problem is, given a graph $G$, a subset of vertices $B$, a vertex cover $S$ of $G$ of size $\delta(G-B)+p$ and a nonnegative integer $k$, decide whether $G$ has a cycle of length at least $2\delta(G-B)+k$.
We start by assembling combinatorial results about paths and vertex covers, which we later use in the algorithm.

The following lemma provides conditions when a part of  long cycle $C$ can be rerouted through any sufficiently large independent set. 

\begin{lemma}\label{lemma:cycle_contains_x}
	Let $G$ be a graph with a given subset of vertices $B$ and a vertex cover $S$ such that $S\supseteq B$ and $|S|=\delta(G-B)+p$ for some $p\geq 0$.  
	Let $k$ 
	be a non-negative integer and let $X\subseteq I=V(G)\setminus S$ be such that $|X|=\delta(G-B)-3p$.
	If $G$ has a cycle $C$ of length $2\delta(G-B)+k$, then it also has  a cycle $C'$   such that
	\begin{itemize}
	\item
	The length of $C'$ is   $2\delta(G-B)+k$,
	\item $C'$  contains all vertices of $X$, and  
	\item  $V(C)\cap S=V(C')\cap S$.
	\end{itemize}
\end{lemma}
\begin{proof}
Because $S$ is a vertex cover, the length of any cycle in $G$ does not exceed $2|S|$. Hence if $G$ contains a cycle of length 
$2\delta(G-B)+k$, we have that $k\le 2p$.
	
	Suppose that $G$ contains a cycle $C$ of length $2\delta(G-B)+k$. 
	Among all cycles of length $2\delta(G-B)+k$, we select a cycle $C'$ such that   $V(C)\cap S=V(C')\cap S$
	and, subject to that,  
	 with the maximum number of vertices from $X$. We claim that all vertices of $X$ are in $C'$. 
	Targeting towards a contradiction, assume that there is a vertex $x\in X$ that is not in $C'$. Let $S'=S\cap V(C')$.
	Note that $|S'|\ge \frac{|C'|}{2}= \delta(G-B)+\frac{k}{2}$.
 Because $S$ is a vertex cover, all neighbors of $x$ are in $S$. Then $x$ has at least $\delta(G-B)$ neighbors in $S$ and, therefore, all but $p$ vertices of $S'$ are adjacent to $x$.
	 
	
	Let $v_1, v_2, \ldots, v_{|S'|}$ be the vertices of $S'$ in the order they appear on the cycle $C'$.
	Note that for each  $i\in\{1,\dots, |S'|\}$,  vertices $v_{i}$ and $v_{i+1}$ (and $v_{|S'|}$,  $v_1$) are either  adjacent vertices in $C'$, or there exists exactly one vertex from $I$ that is between them in $C'$.
We want to show that there exists at least one pair $\{v_{i},v_{i+1}\}$ such that both $v_{i}$ and $v_{i+1}$ are adjacent to $x$ and a vertex $u\in I\setminus X$ is between  $v_{i}$ and $v_{i+1}$ in $C'$. If such a pair exists, then by swapping $u$ and $x$ in $C'$, we would obtain a cycle that has a larger number of vertices from $X$ leading to a contradiction.

	There are at least $\delta(G-B)+k-p$ vertices from $I$ in $C'$, so there are at least $\delta(G-B)+k-p$ pairs $\{v_i, v_{i+1}\}$ that have a vertex from  $I$ between them.
	We know  that at most $p$ vertices in $S'$ are not adjacent to $x$.
	Since each vertex in $S'$ is a member of at most two pairs, vertex $x$ is adjacent to all  but $2p$ such pairs  $\{v_i, v_{i+1}\}$.


	Suppose that $C'$ already contains $t\geq 0$ vertices from $X$. Note that by our assumption,  $t<|X|=\delta(G-B)-3p$, thus   $\delta(G-B)+k-p-2p-t>0$. Therefore, at least one pair of vertices  $\{v_i, v_{i+1}\}$ is adjacent to $x$ and  $v_i  u  v_{i+1}$ is a subpath of $C'$ for some  $u\in I\setminus X$. Therefore, by rerouting $C'$ through $x$, instead of $u$, we construct a cycle $C''$ of length 
 $2\delta(G-B)+k$,  such that   $V(C)\cap S=V(C'')\cap S$, and $C''$ containing $t+1$ vertices of $X$. 
But by our assumption, cycle $C'$ contains the maximum number of vertices $t$ from $X$.	We achieved the contradiction that concludes the proof of the lemma. \end{proof}

We will need the following two simple facts about the number of vertices in a vertex cover that have a small amount of neighbors outside the vertex cover. 

\begin{lemma}\label{thm:many_vertices_with_large_degree}
	Let $G$ be a graph, $S\subseteq V(G)$ be a vertex cover of $G$, and let $I=V(G)\setminus S\neq\emptyset$. 
	Let $d=\min_{v \in I}\deg_G(v)$,  $b=|S|-d \ge 0$, and 
	  $\beta=\frac{d}{|I|}$. Then for any 
	  $\alpha\in (0, \frac{1}{\beta})$, 
	  the number of vertices in $S$ having less that $\alpha d$ neighbors in $I$ is strictly less than $\frac{b}{1-\alpha\beta}$.
\end{lemma}
\begin{proof}
%

Let $s$ be the number of vertices in $S$ with less than $\alpha d$ neighbors in $I$.
	On one hand, the number of edges between $I$ and $S$ is at least $d|I|$.
	On the other hand, it is less than $\alpha d s + (|S|-s)|I|=\alpha ds + (d+b-s)|I|$.
	Hence, \[d|I|<\alpha ds+(d+b-s)|I|.\]
This is equivalent to 
	\[d<\frac{\alpha ds}{|I|}+d+b-s.\]
	Thus 
	\[b>s\cdot \left(1-\alpha \cdot \frac{d}{|I|}\right)=s\cdot\left(1-\alpha\beta\right),\]
and we conclude that 	
  $s < \frac{b}{1-\alpha\beta}$.

\end{proof}


\begin{lemma}\label{lemma:many_vertices_with_neighbors_in_x}
	Let $G$ be a graph with $B\subseteq V(G)$ and a vertex cover $S\supseteq B$ with $|S|=\delta(G-B)+p$, where $0<p< \delta(G-B)/8$. 
	Then for any $X\subseteq V(G)\setminus S$ with $|X|\geq\delta(G-B)-3p$, at most $2p$ vertices in $S$ have less than $2p$ neighbors in $X$.
\end{lemma}
\begin{proof}
	Consider the graph $G[S\cup X]$.
	Apply Lemma~\ref{thm:many_vertices_with_large_degree} to this graph with $I=X$.
	Clearly, $d\ge \delta(G-B)>8p$, because $B\subseteq S$ and $p< \delta(G-B)/8$. Therefore, $b\leq p$. 
	Because $p< \delta(G-B)/8$, we also have that 	
	 $\beta= d/|X|\le |S|/(\delta(G-B)-3p)=(\delta(G-B)+p)/(\delta(G-B)-3p)\le 1+4p/(\delta(G-B)-3p)< \frac{9}{5}$. 	Pick $\alpha=\frac{1}{4}$, so $\alpha\beta<\frac{9}{20}$.
	By Lemma~\ref{thm:many_vertices_with_large_degree}, at most $p/(1-\alpha \beta)<\frac{20}{11}p<2p$ vertices in $S$ have less than $\alpha d > \frac{1}{4} \cdot 8p = 2p$ neighbors in $I=X$. 
\end{proof}

The following structural  lemma provides necessary and sufficient conditions for the existence of a long cycle in graph crossing specified subsets of the vertex cover and the independent set. These conditions can be checked in FPT time (Lemma~\ref{lemma:path_cover_dp}), and both lemmata are the crucial components in the proof of Theorem~\ref{thmVCad}.

\begin{lemma}\label{thm:path_cover}
	Let $G$ be a graph, $B\subseteq V(G)$, and let $p> 0$ be an integer such that $p<\delta(G-B)/8$. Assume that $G$ has a vertex cover $S$ such that    
	$|S|=\delta(G-B)+p$ and $B\subseteq S$.
	%
	Let $k\ge 0$ be an integer and let $X \subseteq I=V(G)\setminus S$ such that $|X|\geq \delta(G-B)-3p$.
	Let $A \subseteq S$  be the set of  vertices of $S$ with at least $p+1$ neighbors in $X$,  and let $Z=S\setminus A$.
	Then there is a cycle of length $2\delta(G-B)+k$ in $G$ containing all vertices in $X \cup Z$ if and only if there is a set $Y\subseteq I$ and a path cover $\mathcal{P}$ of $G[S \cup Y]$, such that:
	\begin{itemize}
		\item[(i)] $\mathcal{P}$  consists of $|S|+k-2q-|Y|$ paths, where $\frac{k}{2} \le q \le p$,  
		\item[(ii)] $\mathcal{P}$ contains no path with an endpoint in $Z$ or $Y$,  
		\item[(iii)]  At least 
		$p-q$ of  paths in $\mathcal{P}$  are paths of length $0$, that is,  covering a single vertex of $A$,  
		\item[(iv)] $|Y|\le 2|Z|$,  
		\item[(v)] $|X\cup Y| \le \delta(G-B)+k-q$, and
		\item[(vi)] $|I| \ge \delta(G-B)+k-q$.
	\end{itemize}
\end{lemma}
\begin{proof}
	Let $C$ be a cycle of length $2\delta(G-B)+k$ in $G$ containing all vertices from $X\cup Z$.
	Define $Y\subseteq I$ to be the set of the vertices of $C$ in $I$ having neighbors in $Z$ in the cycle.  
	Clearly, $|Y|\leq 2|Z|$ satisfying (iv).
	Let $S'=S\cap V(C)$ and define $q:=|S'|-\delta(G-B)$.
	Note that  $q \ge \frac{k}{2}$ because $2|S'|\ge |C|$,  and that $q \le p$ because  $|S'|\le |S|$. Hence, the conditions for $q$ in (i) are satisfied.
Since $X\cup Y \subseteq  V(C)$ and $| V(C)\setminus (X\cup Y)|\ge|S'|=\delta(G-B)+q$, we have that $|X\cup Y|\le \delta(G-B)+k-q$ and (v) holds. Notice that 
$|I|\geq |V(C)\cap I|=|V(C)|-|S'|=(2\delta(G-B)+k)-(q+\delta(G-B))=\delta(G-B)+k-q$ and (vi) is fulfilled.

	Because $p<\delta(G-B)/8$ and  $|X|\geq \delta(G-B)-3p$, $|Z|<2p$ by Lemma~\ref{lemma:many_vertices_with_neighbors_in_x}. Then $|Y|\leq 2|Z|< 4p$ and, therefore, $|S'|+|Y|< \delta(G-B)+q+4p\leq \delta(G-B)+5p$. Since $C$ has $2\delta(G-B)+k\geq 2\delta(G-B)$ vertices and $\delta(G-B)\geq 8p$, we obtain that $|S'|+|Y|< |C|$. This means that $C[S'\cup Y]$ is a proper subgraph of $C$, that is, the union of disjoint paths. 
Consider the path cover $\mathcal{P}'$ of $S'\cup Y$ which is produced by $C$, that is, $\mathcal{P}'$ is the set of paths that are connected components of $C[S'\cup Y]$ (see Figure~\ref{fig:pathcover}).
	It consists of $|C|-|S'|-|Y|$ paths, as each vertex from $V(C)\setminus (S'\cup Y)$ on $C$ is a neighbor to exactly two endpoints in the path cover produced by $C$. Note also that the endpoints of each path of $\mathcal{P}'$ are in $A$. 
	This  path cover still does not cover vertices in $S\setminus S'$, so we  add $|S|-|S'|=p-q$ paths of zero length covering each vertex from $S\setminus S' \subseteq A$; this satisfies (iii). 
	The obtained path cover  $\mathcal{P}$  is a path cover of $S\cup Y$ consisting of exactly $|C|+|S|-2|S'|-|Y|=|S|+k-2q-|Y|$ paths implying (i). Since all the path in $\mathcal{P}'$ have their endpoints in $A$ and each or $p-q$ trivial paths is a vertex of $A$, we obtain that (ii) is fulfilled. We conclude that $\mathcal{P}$ satisfies conditions (i)--(vi) in the  statement of the lemma.

	
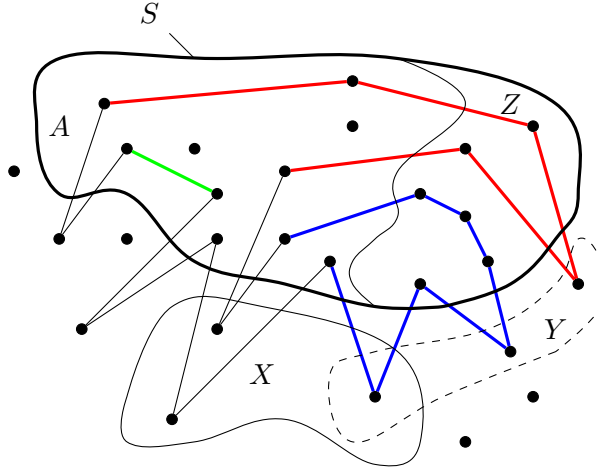
\begin{figure}[ht]
	\centering
	\ifdefined\STOC
\begin{tikzpicture}[scale=0.55]
\else
\begin{tikzpicture}[scale=0.6]
\fi
	\tikzstyle{vertex}=[draw, fill, circle, black, minimum size=4,inner sep=0pt]
	\tikzstyle{path1}=[red, very thick]
	\tikzstyle{path2}=[blue, very thick]
	\tikzstyle{path3}=[green, very thick]
	\tikzstyle{regular}=[]

\node [vertex] (v9) at (5,2.5) {};
\node [vertex] (v4) at (4.5,5) {};
\node [vertex] (v2) at (6,5.5) {};
\node [vertex] (v3) at (7,2) {};
\node [vertex] (v1) at (-3.5,6) {};
\node [vertex] (v5) at (0.5,4.5) {};
\node [vertex] (v22) at (2,5.5) {};
\node [vertex] (v23) at (2,6.5) {};
\draw [path1](v1) -- (v23) -- (v2) -- (v3) -- (v4) -- (v5);
\node [vertex] (v10) at (5.5,0.5) {};
\node [vertex] (v11) at (3.5,2) {};
\node [vertex] (v12) at (2.5,-0.5) {};
\node [vertex] (v8) at (4.5,3.5) {};
\node [vertex] (v7) at (3.5,4) {};
\node [vertex] (v6) at (0.5,3) {};

\node [vertex] (v13) at (1.5,2.5) {};
\draw [path2](v6);
\draw [path2](v6) -- (v7) -- (v8) -- (v9) -- (v10) -- (v11) -- (v12) -- (v13);
\node [vertex] (v14) at (-2,-1) {};
\node [vertex] (v15) at (-1,3) {};
\node [vertex] (v16) at (-4,1) {};
\node [vertex] (v17) at (-1,4) {};
\node [vertex] (v18) at (-3,5) {};
\draw [regular](v13) -- (v14) -- (v15) -- (v16) -- (v17);
\draw [path3](v17) -- (v18);
\node [vertex] (v19) at (-1,1) {};
\draw [regular](v5) -- (v19) -- (v6);
\node [vertex] (v20) at (-4.5,3) {};
\draw [regular](v1) -- (v20) -- (v18);

\draw [very thick]  plot[smooth, tension=.7] coordinates {(4,1.5) (5.5,2) (6.5,3) (7,4) (6.5,6) (3,7) (-1,7) (-4.5,7) (-5,5.5) (-4.5,4) (-3,4) (-1.5,2.5) (0.5,2) (2.5,1.5) (4,1.5)};

\draw  plot[smooth, tension=.7] coordinates {(3,7) (4.5,6) (3,4.5) (3,3.5) (2.5,3) (2,2.5) (2,2) (2.5,1.5)};
\draw [dashed] plot[smooth, tension=.7] coordinates {(6.5,0.5) (4,-0.5) (2,-1.5) (1.5,0) (2.5,0.5) (5,1) (6.5,2) (7,3) (7.5,2.5) (7.5,1.5) (6.5,0.5)};
\draw  plot[smooth, tension=.7] coordinates {(-3,-0.5) (-2,1.5) (0,1.5) (3,0.5) (3.5,-1.5) (2.5,-2) (0.5,-1) (-1.5,-1.5) (-3,-1.5) (-3,-0.5)};
\node [vertex] at (4.5,-1.5) {};
\node [vertex] at (-1.5,5) {};

\node [vertex] at (6,-0.5) {};
\node [vertex] at (-3,3) {};
\node [vertex] at (-5.5,4.5) {};
\node at (-4.5,5.5) {$A$};
\node at (5.5,6) {$Z$};
\node at (0,0) {$X$};
\node at (6.5,1) {$Y$};
\node (v21) at (-2.5,8) {$S$};
\draw [regular](v21) -- (-1.5,7);
\end{tikzpicture}
	\caption{Illustration of how a cycle forms a path cover  $\mathcal{P}$ from Lemma~\ref{thm:path_cover}.
		Edges belonging to different paths are colored with different colors.
		Vertices in $A$ that have no incident colored edge are covered by   zero-length paths in $\mathcal{P}$.}\label{fig:pathcover}
\end{figure}
	
	We now prove the opposite  direction.
	Let   $Y\subseteq I$ and  let $\mathcal{P}$ be  a path cover of $G[S\cup Y]$ satisfying conditions (i)--(vi) of the lemma. In particular, $|\mathcal{P}|=|S|+k-2q-|Y|$, where $\frac{k}{2}\leq q\leq p$,  by (i). 
	We show that there exists a cycle of length $2\delta(G-B)+k$ containing all vertices of  $X\cup Z$.
	We remove from  $\mathcal{P}$  arbitrary $p-q$ zero length paths covering  single vertices of $A$ using (iii).
	The obtained set of paths $\mathcal{P}'$ consists of $|S|-p-q+k-|Y|=\delta(G-B)+k-q-|Y|$ paths and covers $G[S'\cup Y]$, where $S'\subseteq S$ and  $|S'|=|S|-(p-q)=\delta(G-B)+q$.
	Now choose an arbitrary subset $I'$   of $I$ of size $\delta(G-B)+k-q$ containing all vertices from  $X\cup Y$ that exists due to (v) and (vi).
	We consider $H=G[S'\cup I']$.  Notice that $|V(H)|=|S'|+|I'|=2\delta(G-B)+k$.
	We claim that graph $H$ contains a  Hamiltonian cycle. Clearly, this suffices for the proof, because the length of such a cycle is $2\delta(G-B)+k$ and it contains all the vertices of $X\cup Z$ as required. 
	 
	Let $H'$ be the graph obtained from $H$ be the deletion of edges $e\in H[S']$ that are not included in the paths of $\mathcal{P}'$. It is straightforward to see that it is sufficient to show that $H'$ has a Hamiltonian cycle. By construction, $H'[S']$ is the union of paths that are subpaths of the elements of $\mathcal{P}'$. By (ii), no path of $\mathcal{P}'$ has an endpoint in $Y$. 
	Since $\mathcal{P}'$ consists of paths with endpoints in $S'$ covering $S'\cup Y$ and there is no edges between vertices in $Y$, removal of each vertex from $Y$ breaks one path into two.
	Hence, the number of disjoint paths forming  $H'[S']$ is exactly $|\mathcal{P}'|-|Y|=\delta(G)+k-q-|Y|+|Y|=|I'|$.
	This implies that the graph $H''$ obtained from $H'$ by making every pair of distinct vertices of $I'$ adjacent has a Hamiltonian cycle if and only if the same holds for $H'$, because no Hamiltonian cycle of $H''$ cannot contain an edge $uv$ with $u,v\in Y$. Otherwise,  such a cycle would  
	cover 
	$S'$ by less than $|I'|$ paths.
	Now the degree of each vertex from $I'$ in $H''$ is at least $\delta(G-B)-(p-q)+|I'|-1=2\delta(G-B)+k-p-1$.
	
	Take a vertex   $v\in S'\setminus Z$. Then $v\in A$ and, by the definition of $A$, $v$  has at least $p+1$ neighbors in $X\subseteq I'$. Hence it has at least $p+1$ neighbors in $H''$.
	Therefore, the sum of vertex degrees  of a  vertex from $S'\setminus Z$ and a vertex from $I'$ in $H''$,  is at least $2\delta(G-B)+k=|V(G')|$.  We construct $H'''$ from $H''$ by making adjacent every pair of vertices $u$ and $v$ with $u\in I'$ and $v\in S'\setminus Z$.  	Theorem~\ref{thm:bh} implies that $H'''$ has a Hamiltonian cycle if and only if $H''$ has a Hamiltonian cycle. 
	
	
	
	Finally, we construct  a Hamiltonian cycle in $H'''$ using the paths of $\mathcal{P}'$. For this, recall that each path of $\mathcal{P}'$ has its endpoints in $A$ by (ii). 
	Notice that there are exactly $|I'\setminus Y|=\delta(G-B)+k-q-|Y|$ vertices of $I'$ that are not covered by the paths. Since the number of paths in $\mathcal{P}'$ is $\delta(G-B)+k-q-|Y|$ and every endpoint of a path is adjacent to every vertex of $I'\setminus Y$, it is straightforward to see that we can construct a Hamiltonian cycle joining the paths of $\mathcal{P}'$ via the vertices of $I'\setminus Y$. 

Thus we conclude that $H'''$ has a Hamiltonian cycle. This implies that $H$ has a Hamiltonian cycle and competes the proof. 
\end{proof}

By Lemma~\ref{thm:path_cover}, to find a cycle of length 
$2\delta(G-B)+k$ in $G$ containing all vertices in $C \cup Z$, it suffices to identify a path cover $\mathcal{P}$. Such a path cover can be computed by making use of color-coding. More precisely. 

\begin{lemma}\label{lemma:path_cover_dp}
	Given $G, B, S, k,$ and $X,A,Z$ defined in the same way as in Lemma~\ref{thm:path_cover}, the existence of $Y$ and a path cover $\mathcal{P}$ of $G[S\cup Y]$ satisfying (i)--(vi)
	can be determined in $2^{\Oh(p)}\cdot\polyn$ running time.
\end{lemma}

\begin{proof}
         Because $p<\delta(G-B)/8$ and  $|X|\geq \delta(G-B)-3p$, $|Z|<2p$ by Lemma~\ref{lemma:many_vertices_with_neighbors_in_x}. Then we are looking for $Y\subseteq I$ with $|Y|\leq 2|Z|< 4p$ by (iv).
Also by (i), $\frac{k}{2}\leq q\leq p$. 
	We assume without loss of generality that $q$ and the cardinality $r$ of $Y$ are  fixed, as an algorithm can iterate over all $\Oh(p^2)$ possible pairs of these values in an outer loop. We also assume that (vi) holds for the given value of $q$.
	The algorithm is now to find a set of disjoint paths $\mathcal{P}$ covering all vertices in $S$ and a set $Y\subseteq I$ of  size $r$.
	Since Lemma~\ref{thm:path_cover} requires an upper bound (v) on $|X\cup Y|$, we will aim to 
	maximize
	$|X\cap Y|$, i.e.\  the number of vertices from $X$ used by the paths of $\mathcal{P}$.
	
	As the paths of $\mathcal{P}$ cover  exactly $|S|+|Y|$ vertices and their number is exactly $|S|+k-2q-|Y|$ by (i), the total length of these paths is exactly $2|Y|+2q-k\leq 10p$.
	This allows us to deal with a bounded number of paths of positive length.
	By (ii), there is no path in $\mathcal{P}$ with an endpoint in $Z\cup Y$. 
	In particular, this means that all paths of zero length are vertices in $A$ and the endpoint of nontrivial paths are in $A$. 
	Each nontrivial path has exactly two endpoints in $A$. Then, because the total number of path $\mathcal{P}$ is $|S|+k-2q-|Y|$, the number of nontrivial paths $t$ is at most $|A|-|\mathcal{P}|=|A|- (|S|+k-2q-|Y|)=|Y|-|Z|+2q-k\leq 6p$. Note also that because  $|\mathcal{P}|=|S|+k-2q-|Y|$, the nontrivial paths should cover exactly $s=|S|+k-2q-|Y|-t$ vertices of $A$  and they should leave uncovered at least $p-q$ vertices of $A$ to satisfy (iii). Clearly, $s\leq 20p$, because the total length of the nontrivial paths is at most $10p$.
	Thus, our task is reduced to deciding whether there is a set $Y\subseteq I$ of size $r\leq 4p$
	 and 
	 a family of $t\leq |Y|-|Z|+2q-k\leq 6p$ disjoint nontrivial paths  $\mathcal{P}'$ 
	 such that 
	 \begin{itemize}
	 \item[(a)] the endpoints of the paths of $\mathcal{P}'$ are in $A$, 
	 \item[(b)] the paths  cover the vertices of $Y$ and exactly
	 $s=|S|+k-2q-|Y|-t\leq 20p$ 
	 vertices of $A$, and they leave uncovered at least $p-q$ vertices of $A$,
	 \item[(c)] subject to (a)--(b), $|Y\cap X|$ is maximum. 
	 \end{itemize}
	 The color-coding technique of Alon, Yuster, and Zwick~\cite{AlonYZ95} is a standard tool for solving problems of this type. Since the approach is standard (see, e.g, the book~\cite[Chapter~5]{cygan2015parameterized}), we only briefly sketch the algorithm. In the same way as in the proof of Theorem~\ref{thmTLDP},  we give a sketch of a randomized Monte Carlo algorithm and then explain how it can be derandomized.	 
	 
For each positive integer $t\leq |Y|-|Z|+2q-k$, we verify whether there are $Y$ and $\mathcal{P}'$ satisfying (a) and (b) and find the maximum size of $|X\cap Y|$. After iterating over all possible values of $t$, the algorithm returns a solution that gives the maximum value of $X\cap Y$. For a given $t$, we compute $s=|S|+k-2q-|Y|-t$ and verify whether $|A|-s\geq p-q$. We discard the current choice of $t$ if $|A|-s< p-q$. 
From now we assume that the value of $t$ is fixed and $|A|-s\geq p-q$. 
	
	We use the following randomized procedure. 
	We color the vertices of $I$ by $r=|Y|$ distinct colors uniformly at random and  then the vertices of  $A$ are colored uniformly at random with another set of  $s$ distinct colors.	
	We also assume that the vertices of $Z$ are colored as well by pairwise distinct colors that are different from the colors used for $I$ and $A$. We denote by $C_I$, $C_A$, and $C_Z$ the sets of colors used to color $I$, $A$, and $Z$, respectively. Let also $C=C_I\cup C_A\cup C_Z$. Clearly, $|C|=\Oh(p)$.
	We say that $Y\subseteq I$ and a set of disjoint nontrivial paths  $\mathcal{P}'$ satisfying (a) and (b) is a 	\emph{coloful} solution if the vertices of the paths are colored by distinct colors.  
	
	The  main steps of our algorithm either finds the maximum $|X\cap Y|$ for a colorful solution or reports that a colorful solution does not exist.

	For a set of colors $R\subseteq C$, denote by $\alpha(R)$ the maximum number of vertices of $X$ that can be covered by a nontrivial path $P$ with $|R|$ vertices such that their  the endpoint are in $A$ and the vertices of $P$ are colored by distinct colors from $R$; we assume that $\alpha(R)=-\infty$ if such a path does not exist.  We observe that for every $R\subseteq C$, the value of $\alpha(R)$ can be computed in $2^{\Oh(p)}\cdot n^{\Oh(1)}$ time by a straightforward modification of the standard dynamic programming algorithm for finding a colorful $|R|$-path (see~\cite{AlonYZ95} and~\cite[Chapter~5]{cygan2015parameterized}). It is easy to incorporate the condition that the endpoits are in $A$. To maximize the number of vertices of $X$ used by a path, we can assume that the vertices of $X$ are of weight one and the vertices of $V(G)\setminus X$ are given zero weights. Then we use the variant of the algorithm that finds a colorful path of maximum weight.  From now, we assume that we are given the table of values of $\alpha(R)$ for all $R\subseteq C$. Note that this table of size $2^{\Oh(p)}$ can be constructed in $2^{\Oh(p)}\cdot n^{\Oh(1)}$ time. 
	
	Let $R\subseteq C$, and $\ell\leq t$ be a positive integer. Denote by $\beta(R,\ell)$ the maximum number of vertices of $X$ that can be  covered by exactly $\ell$ nontrivial path with $|R|$ vertices in total such that their endpoint are in $A$ and the vertices of the paths are colored by distinct colors from $R$; we assume that $\beta(R,\ell)=-\infty$ if such paths do not exist; in particular $\beta(R,\ell)=-\infty$ if $|R|\leq 1$.  It is straightforward to see that $\beta(R,1)=\alpha(R)$ for every $R\subseteq C$. To compute $\beta(R,\ell)$ for $\ell>1$, we use the following straightforward 	recurrence for $|R|\geq 2$.
\begin{equation}\label{eq:rec-beta}	
\beta(R,\ell)=\max\{\alpha(R')+\beta(R\setminus R',\ell-1)\mid \emptyset\neq R'\subset R\}.
\end{equation} 
	We use (\ref{eq:rec-beta}) to compute the table of values of $\beta(R,t)$ for all nonempty $R\subseteq C$. Because $|C|=\Oh(p)$, computing the table can be done in $2^{\Oh(p)}\cdot n^{\Oh(1)}$ time. 
	
	By the choice of $C_I$, $C_A$ and $C_Z$, we have that $\beta(C,t)$ is the maximum number of vertices of $X$ that can be covered by a colorful solution, and $\beta(C,t)=-\infty$ if there is no colorful solution. 
	
	To obtain an optimum (non-colorful) solution, we define $N=\lceil e^{s+t}\rceil\geq  \frac{r^r\cdot s^s}{r!\cdot s!}$ and iterate the randomized procedure $N$ times. Then the algorithm returns  a solution that gives the maximum value $|X\cap Y|$ over all coloful solution or reports that there is no solution if the algorithm fails to find a colorful solution in every iteration.	
	
	Suppose that $Y\subseteq I$ of size $r$ and $\mathcal{P}'$ of size $t$ satisfy (a) and (b) and provide the maximum value of $|X\cap Y|$. Then with probability at least $\frac{r!}{r^r}$, the vertices of $Y$ are colored by distinct colors from $C_I$ by a random coring. Similarly, with probability at least $\frac{s!}{s^s}$, the $s$ vertices of $A$ covered by the paths of $\mathcal{P}'$ are colored by distinct colors of $C_A$. Then with probability at least $\frac{r!\cdot s!}{r^r\cdot s^s}$, the vertices of the paths of $\mathcal{P}'$ are colored by distinct colors. Respectively, the probability that this does not holds, that is, there are at least two vertices of the same color, is at most $(1-\frac{r!\cdot s!}{r^r\cdot s^s})$. By the choice of $N$, we obtain that the probability that for every iteration, 	at least two vertices of paths of $\mathcal{P}$ have the same color, is at most 
$(1-\frac{r!\cdot s!}{r^r\cdot s^s})^N\leq e^{-1}$. Thus, the probability that the randomized algorithm fails to return an optimum solution is at most $e^{-1}<1$.

To evaluate the running time, recall that the tables of values of $\alpha(\cdot)$	 and $\beta(\cdot,t)$ can be computed in $2^{\Oh(p)}\cdot n^{\Oh(1)}$ time. Since $r<4p$ and $s\leq 20p$, $N=2^{\Oh(p)}$ and, therefore, the total running time is $2^{\Oh(p)}\cdot n^{\Oh(1)}$.

To derandomize the algorithm, we use the standard technique (see~\cite{AlonYZ95} and~\cite[Chapter~5]{cygan2015parameterized}). For given $r$ and $s$, we construct the $(|I|,r)$ and $(|A|,s)$-perfect hash families 
of the functions $\mathcal{F}_I$ and $\mathcal{F}_A$, respectively, of sizes   $e^{r}r^{\Oh(\log k)}\cdot \log |I|$ and $e^{s}s^{\Oh(\log s)}\cdot \log |A|$, respectively, using the results of Naor, Schulman, and Srinivasan~\cite{NaorSS95}. These families can be constructed in time $2^{\Oh(p)}\cdot n\log n$. Then we replace the random colorings of $I$ and $A$ by the functions from $\mathcal{F}_I$ and $\mathcal{F}_A$, respectively, and iterate the main step over all these functions. This gives deterministic $2^{\Oh(p)}\cdot n^{\Oh(1)}$ running time.

To conclude the proof, note that algorithms finds the maximum possible size of $|X\cap Y|$ for $Y\subseteq I$ of size $r$ such that $S\cap Y$ can be covered by a set of paths $\mathcal{P}$ satisfying conditions (i)--(iv) and (vi) of Lemma~\ref{lemma:many_vertices_with_neighbors_in_x}. To verify (v), it is sufficient to check additionally whether $|X\cup Y|\leq \delta(G-B)+k-q$, by the maximality of $|X\cap Y|$. This concludes the proof.
\end{proof}
%

Everything is settled for the proof of Theorem~\ref{thmVCad}. For convenience, we restate the theorem here. 

\medskip\noindent\textbf{Theorem~\ref{thmVCad}}. \emph{\textsc{\probDC\ / Vertex Cover Above Degree} is solvable in $2^{\Oh(p+|B|)}\cdot n^{\Oh(1)}$ running time.}
\medskip 
%
%
\begin{proof}
	Let $(G, B, S, k)$ be a given instance of the problem. 	We assume without loss of generality that $B \subseteq S$; otherwise we can set $S:=S\cup B$ and $p:=p+|B\setminus S|$, 
	which increases $p$  by at most $|B|$. Let $I=V(G)\setminus S$. 
	Note that $G$ has no cycle longer than   $2|S|\le 2\delta(G-B)+2p$.
	In particular, if $k>2p$, then the given instance is a no-instance. Therefore, we can assume that $k\leq 2p$. 
	If $\delta(G-B)\leq 8p$, then $2\delta(G-B)+k\leq 18p$ and one can verify whether $G$ has a cycle of length $2\delta(G-B)+k$ in  $2^{\Oh(p)}\cdot \polyn$ time using, e.g., the algorithm given by Zehavi~\cite{Zehavi16}.  	
	From now on, we assume that $\delta(G-B)>8p$. It is also convenient to assume that our aim is to verify the existence of a cycle of length \emph{exactly} $2\delta(G-B)+k$; for this we iterate over all possible values of the parameter from the initial given value of $k$ and $2p$.


	Also, if $p=0$, then $k=0$ and each vertex in $I$ is adjacent to all vertices in $S=\delta(G-B)$.
	Then $G$ contains all edges between $S$ and $I$, so a cycle of length at least  $2\delta(G-B)=2|S|$ exists in $G$ if and only if $|S|\geq |I|$ and $|S|\geq 2$.
	Thus, we can now assume that $p>0$.

        If $|I| < \delta(G-B)+k-p$, then $(G, B, S, k)$ is a no-instance. Hence, we can assume that this is not the case. 
	Our algorithm chooses an arbitrary $X \subseteq I$ of size $\delta(G-B)-3p$.
	By Lemma~\ref{lemma:cycle_contains_x}, the algorithm can now look for a cycle of length $2\delta(G-B)+k$ in $G$ containing all vertices from $X$.
	
	Then we partition $S$ into two subsets $A$ and $Z$.
	The subset $A$ consists of all vertices in $S$ that have at least $p+1$ neighbors in $X$.
	The subset $Z$ consists of all other vertices in $S$. The running time of the procedure computing $Z$ is clearly polynomial. 
	By Lemma~\ref{lemma:many_vertices_with_neighbors_in_x}, the cardinality of $Z$ is at most $2p$.
	
	Before we can apply Lemmata~\ref{thm:path_cover} and~\ref{lemma:path_cover_dp}, we need to ensure that the cycle we are looking for contains \emph{all} vertices from $Z$.
	To achieve that, we allow our algorithm to brute-force over all $2^{|Z|}=2^{\Oh(p)}$ options of how the cycle intersects $Z$.
	When an option is fixed, consider deleting  from $G$ all vertices of $Z$ outside the fixed intersection.
	This can change the value of $p$, as $p=|S|-\delta(G-B)$, and both $|S|$ and $\delta(G-B)$ may change after the deletion.
	As a consequence, the equality $|X|=\delta(G-B)-3p$ could no longer hold, so we need to change $X$ correspondingly.
	Rewrite $\delta(G-B)-3p=4\delta(G-B)-3|S|$.
	Note that the removal  of a single vertex of $Z$ from $G$ always decreases $|S|$ by one and can decrease $\delta(G-B)$ by at most one.
	Hence, the value $\delta(G-B)-3p$ can only increase.
	Thus, after the deletion, to ensure $|X|=\delta(G-B)-3p$,  we   add some vertices from $I$ to $X$.
	By Lemma~\ref{lemma:cycle_contains_x}, the choice of these vertices can be arbitrary and we can be sure that there is a cycle containing $X$ while its intersection with $S$ remains the same.
	Each vertex in $A$ still has at least $p+1$ neighbors in $X$.
	Since $X$ now can containin some new vertices from $I$, a vertex in $Z$ may   have at least $p+1$ neighbors in $X$.
	If such a vertex exists, we simply move it from $Z$ to $A$. 
	Observe that the value of the parameter $p$ may be only decreased and the deletion does not violate the property  $\delta(G-B)>8p$.
	Note that the deletion operation discussed above also can imply an increment in $k$ as $\delta(G-B)$ can decrease.
	This is safe as Lemma~\ref{lemma:cycle_contains_x} does not depend on the value of $k$ other than for estimating the length of the cycle.
	
	After the intersection of the cycle with $Z$ is fixed and all vertices from $Z$ outside it are deleted from $G$, the algorithm finally employs the routine from Lemma~\ref{lemma:path_cover_dp} to find the path cover from Lemma~\ref{thm:path_cover}, hence to find the cycle. 
	The total running time of the algorithm (under the assumption that $B\subseteq S$) is proportional to the number of sets $Z$, which is $2^{\Oh(p)}$,  times the time required to compute the path cover for each of the sets, which is $2^{\Oh(p)}\cdot \polyn$ by 
Lemma~\ref{lemma:path_cover_dp}. Hence the total running time is 	
	$2^{\Oh(p)}\cdot \polyn$. Taking into account that to ensure the assumption  that $B\subseteq S$ we may increase the initial value of $p$ by at most $|B|$, we conclude that the algorithm runs in  $2^{\Oh(p+|B|)}\cdot \polyn$ time.
	 \end{proof}

\section{Finding almost Hamiltonian cycles} \label{sec:HamCycles}
This section is dedicated to the proof of Theorem~\ref{theorem:hamiltonian}. To recall, the theorem states that given a graph $G$ with a set $B \subset V(G)$ and a parameter $k$ such that $|B| \le k$ and 
$\delta(G - B) \ge \frac{n}{2} - k$, in time $2^{\Oh(k)}n^{\Oh(1)}$ we can find the longest cycle in $G$. Before we move on to prove the theorem itself, we show how to deal with the special case where there is a small separator in the graph, as it is an important subroutine in the main algorithm. Another key ingredient to the proof of Theorem~\ref{theorem:hamiltonian} is our \textsc{\probDC\ / Vertex Cover Above Degree} result, presented in Section~\ref{sec:vcalgo}.

    \subsection{Small separator lemma}

    We show an algorithm for \textsc{Almost Hamiltonian Dirac Cycle} when there is a small (i.e. of size $\Oh(k)$) separator $B$ in $G$. Intuitively, the presense of a small separator makes the problem easier in the following sense. Each component of $G - B$ still has high minimal degree, slightly less than $\frac{n}{2}$. Thus, essentially, we must have exactly two components of size roughly $\frac{n}{2}$ in $G - B$, which means they are very dense.
As was proven in~\cite{fomin_et_al:LIPIcs:2020:11724}, in this situation, we can always partition a component into paths that start and end at the given vertices, and span the whole component. We restate their result formally in the next lemma.

\begin{lemma}[Lemma 1 in \cite{fomin_et_al:LIPIcs:2020:11724}]
    \label{lemma:many_paths}
    Let $G$ be an $n$-vertex graph and $p$ be a positive integer such that $\delta(G) \ge \max\{5p - 3, n - p\}$. Let $\{s_1, t_1\}$, \ldots, $\{s_r, t_r\}$, $r \le p$, be a collection of pairs of vertices of $G$ such that (i) $s_i \notin \{s_j, t_j\}$ for all $i \ne j$, $i, j \in \{1, \ldots, r\}$, and (ii) there is at least one index $i \in \{1, \ldots, r\}$ such that $s_i \ne t_i$. 
    Then there is a family of pairwise vertex-disjoint paths $\mathcal{P} = \{P_1, \ldots, P_r\}$
    in $G$ such that each $P_i$ is an ($s_i$, $t_i$)-path and $\cup_{i = 1}^r V(P_i) = V(G)$, that is, the paths cover all vertices of $G$.
\end{lemma}
We note that the proof of Lemma~\ref{lemma:many_paths}
given in \cite{fomin_et_al:LIPIcs:2020:11724} is actually constructive.
That is, there is a polynomial time algorithm that given $G$, $p$, and the respective set of pairs of vertices,
returns the family of paths $\mathcal{P}$ from the statement of Lemma~\ref{lemma:many_paths}.

For simplicity, suppose there is a Hamiltonian cycle $C$ in $G$. The cycle induces a certain partition of $B$ into paths. On the other hand, if we are able to find any such path cover $\mathcal{P}$, we can construct the whole Hamiltonian cycle. Namely, on each component $H$ of $G - B$, we invoke Lemma~\ref{lemma:many_paths} with a collection of pairs being a certain matching on ends of $\mathcal{P}$ belonging to $H$. In this way we connect the paths together while also visiting every vertex of $H$. If the pairs are selected in a certain way in both components, the union of all these parts will actually form a Hamiltonian cycle. We find the path cover itself with the help of dynamic programming and the color coding technique of Alon, Yuster, and Zwick~\cite{AlonYZ95}. In what follows, we prove the above in more detail.

\begin{lemma}
	\label{lemma:hamiltonian_separator}
	Let $G$ be a given $2$-connected graph on $n$ vertices and let $k\ge 0$ be a given integer.
	Let $B\subseteq V(G)$ be such that $|B|\le k$, $\delta(G - B) \ge \frac{n}{2} - k$, and the graph $G- B$ is not connected.
	There is a $2^{\Oh(k)}\cdot n^{\Oh(1)}$ running time algorithm that finds the longest simple cycle in $G$.
\end{lemma}
\begin{proof}
    Assume $n \ge 12k$, otherwise we invoke the general $2^{\Oh(n)}$ algorithm for the \textsc{Longest Cycle} problem from \Cref{prop:longest_cycle} polynomial number of times to find the longest cycle in $G$.
    
    First, observe that there are exactly two connected components in $G - B$. There must be at least two of them since $G - B$ is not connected. Suppose there are at least three components.
    Each of them contains a vertex of degree at least $\frac{n}{2} - k$ in $G - B$, therefore the size of each component is at least $\frac{n}{2} - k + 1$. The total number of vertices is then at least $3 \frac{n}{2} - 3k + 3 = n + (\frac{n}{2} - 3k) + 3 > n$. This is a contradiction.   
    From now on, let $H_1$ and $H_2$ be the two connected components of $G - B$.
    
    Consider the longest cycle $C$ in $G$. Recall that by Theorem~\ref{thm:relaxed_long_cycle} the length of $C$ is at least $\min\{n - 2k,n-|B|\}\ge n-2k$, thus it necessarily contains vertices from all of $H_1$, $H_2$ and $B$. We say that $C$ induces a path cover $\mathcal{P}$ of $B$, where $\mathcal{P}$ is the set of paths that $C$ forms when restricted to the edges incident to $B$. In other words, remove from $C$ all the edges that are not incident to $B$, and all the vertices that became isolated after that. The resulting collection of vertex-disjoint paths is the path cover $\mathcal{P}$. Note that $\mathcal{P}$ satisfies the following properties.
    \begin{enumerate}
        \item Every path $P \in \mathcal{P}$ starts and ends in $V(G) \setminus B$.
        \item Each path $P \in \mathcal{P}$ has at least one vertex in $B$ and no two consecutive vertices in $V(G) \setminus B$.
        \item The paths of $\mathcal{P}$ contain at most $3|B|$ vertices in total.
        \item The number of paths in $\mathcal{P}$ that start and end in different components of $G - B$ is even and at least two.
    \end{enumerate}
    Since for every vertex of $B$, its degree in $C$ is exactly two even when restricted to the edges incident to $B$, the property (1) follows. Each path goes through $B$, and two vertices in $V(G) \setminus B$ cannot be adjacent via an edge incident to $B$, thus (2) follows. Property (3) follows directly from property (2). Finally, (4) holds since $C$ must leave both $H_1$ and $H_2$ an even number of times. Moreover, if there are no paths in $\mathcal{P}$ that start and end in different components, $H_1$ and $H_2$ cannot be connected via $C$, thus $C$ is not a cycle of length at least $n - 2k$.

    We call a set of vertex-disjoint paths in $G$ satisfying (1)--(4) \emph{a good path cover}. Now we claim that any good path cover can be used to construct a long cycle in $G$, i.e. we can collect all the vertices of $V(G) \setminus B$ in a cycle by going along the paths in the cover. The proof is essentially by pairing endpoints of the paths carefully and then applying Lemma~\ref{lemma:many_paths} to both $H_1$ and $H_2$. The illustration is shown in Figure~\ref{fig:small_sep} and the proof follows next.

    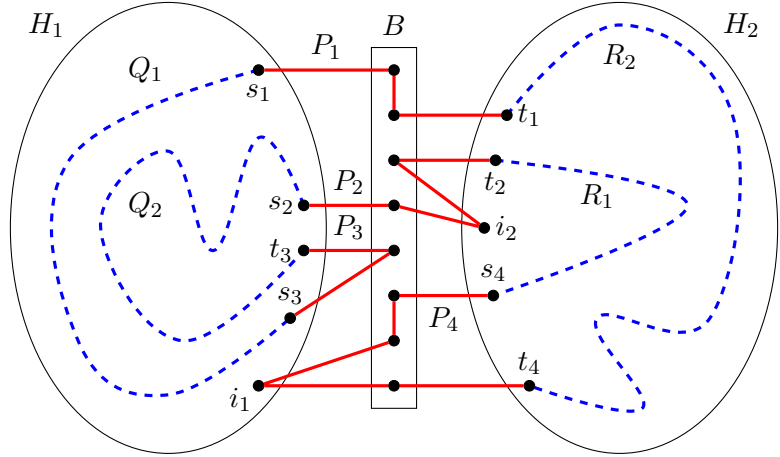
\begin{figure}[ht]
        \centering
        \ifdefined\STOC
        \begin{tikzpicture}[scale=0.46]
        \else
        \begin{tikzpicture}[scale=0.6]
        \fi
	        \tikzstyle{vertex}=[draw, fill, circle, black, minimum size=4,inner sep=0pt]
	        \tikzstyle{path1}=[red, very thick]
	        \tikzstyle{path2}=[blue, very thick, dashed]
	        \tikzstyle{regular}=[]

            \draw[path2] plot[smooth, tension=.7] coordinates {(2, 7) (-2.2, 5) (-2, 1) (0, -0.2) (2.7, 1.5)};
            \node at (-0.5, 7) {$Q_1$};
            \draw[path2] plot[smooth, tension=.7] coordinates {(3, 4) (2, 5.5) (1, 3) (0, 5.2) (-1.5, 3.5) (0.3, 1) (3, 3)};
            \node at (-0.5, 4) {$Q_2$};

            \draw[path2] plot[smooth, tension=.7] coordinates {(7.5, 6) (10, 8) (12.5, 6) (12, 1.5) (9.5, 1.5) (10.5, -0.5) (8, 0)};
            \node at (10, 7.3) {$R_2$};
            \draw[path2] plot[smooth, tension=.7] coordinates {(7.25, 5) (11.5, 4) (7.2, 2)};
            \node at (9.5, 4.2) {$R_1$};

            \node [vertex] (b1) at (5,7) {};
            \node [vertex] (b2) at (5,6) {};
            \node [vertex] (b3) at (5,5) {};
            \node [vertex] (b4) at (5,4) {};
            \node [vertex] (b5) at (5,3) {};
            \node [vertex] (b6) at (5,2) {};
            \node [vertex] (b7) at (5,1) {};
            \node [vertex] (b8) at (5,0) {};
            \draw[regular] (4.5, 7.5) rectangle  (5.5, -0.5);
            \node at (5, 8) {$B$};
            \draw[regular] (0, 3.5) circle [x radius=3.5cm, y radius = 5cm];
            \node at (-2.7,8) {$H_1$};
            \draw[regular] (10, 3.5) circle [x radius=3.5cm, y radius = 5cm];
            \node at (12.7,8) {$H_2$};

            \node[vertex] (s1) at (2, 7) {};
            \node at (2, 6.5) {$s_1$};
            \node[vertex] (t1) at (7.5, 6) {};
            \node at (8, 6) {$t_1$};
            \draw[path1] (s1) -- node[above, black]{$P_1$} (b1) -- (b2) -- (t1);
            \node[vertex] (s2) at (3, 4) {};
            \node at (2.5, 4) {$s_2$};
            \node[vertex] (t2) at (7.25, 5) {};
            \node at (7.25, 4.5) {$t_2$};
            \node[vertex] (i2) at (7, 3.5) {};
            \node at (7.5, 3.5) {$i_2$};
            \draw[path1] (s2) -- node[above, black]{$P_2$} (b4) --(i2) -- (b3) -- (t2);

            \node[vertex] (t3) at (3, 3) {};
            \node at (2.5, 3) {$t_3$};
            \node[vertex] (s3) at (2.7, 1.5) {};
            \node at (2.7, 2) {$s_3$};
            \draw[path1] (t3) -- node[above, black]{$P_3$} (b5)  -- (s3);

            \node[vertex] (i1) at (2, 0) {};
            \node at (1.6, -0.3) {$i_1$};
            \node[vertex] (t4) at (8, 0) {};
            \node at (8, 0.5) {$t_4$};
            \node[vertex] (s4) at (7.2, 2) {};
            \node at (7.2, 2.5) {$s_4$};
            \draw[path1] (s4) -- node[below, black]{$P_4$} (b6)  -- (b7) -- (i1) -- (b8) -- (t4);
        \end{tikzpicture}
        \caption{Reconstructing the cycle from the good path cover $\mathcal{P} = \{P_1, P_2, P_3, P_4\}$. The paths $Q_1$ and $Q_2$ are obtained by applying Lemma~\ref{lemma:many_paths} to $H_1' = H_1 - \{i_1\}$, the same for the paths $R_1$, $R_2$ and the graph $H_2' = H_2 - \{i_2\}$. The resulting concatenation of paths is a Hamiltonian cycle in $G$.}
        \label{fig:small_sep}
    \end{figure}

    \begin{claim}
        \label{claim:good_cover}
        There is a polynomial time algorithm that given a good path cover $\mathcal{P}$ finds a cycle of length $n - t$ in $G$, where $t$ is the number of vertices in $B$ not covered by the paths in $\mathcal{P}$.
    \end{claim}
    \begin{proof}
        Denote $\mathcal{P} = \{P_1, \ldots, P_r\}$, and for each $i \in \{1, \ldots, r\}$, denote the two ends of the path $P_i$ by $s_i$ and $t_i$. We may assume that the paths are ordered in a way that paths $P_1$, \ldots, $P_a$ lead from $H_1$ to $H_2$, paths $P_{a + 1}$, \ldots, $P_b$ start and end in $H_1$, and paths $P_{b + 1}$, \ldots, $P_r$ start and end in $H_2$, for certain integers $a$ and $b$, such that $1 < a \le b \le r$, and $a$ is even by property (5). Additionally, for $i \in \{1, \ldots, a\}$ assume that $s_i \in V(H_1)$, $t_i \in V(H_2)$.

        Let $I$ be the set of internal vertices of paths in $\mathcal{P}$,
        let $H_1' = H_1 - I$, $H_2' = H_2 - I$. The graphs $H_1'$ and $H_2'$ are targets for applying Lemma~\ref{lemma:many_paths}. By property (2), the size of $I \setminus B$ is at most $k$, thus $\delta(H_1') \ge \delta(G - B) - k = \frac{n}{2} - 2k$, and by the same argument $\delta(H_2') \ge \frac{n}{2} - 2k$.

        Consider the following set $T_1$ of $b - \frac{a}{2}$ pairs of vertices in $H_1'$. If $b = a$, the pairs are $\{s_1, s_2\}$, $\{s_3,s_4\}$, \ldots, $\{s_{a - 1}, s_a\}$. If $b > a$, the pairs are $\{s_{2i - 1}, s_{2i}\}$ for $1 \le i < \frac{a}{2}$, $\{s_{a - 1}, s_{a + 1}\}$, $\{t_j, s_{j + 1}\}$ for $a + 1 \le j < b$, and $\{t_b, s_a\}$.

        Now, we apply Lemma~\ref{lemma:many_paths} to the graph $H_1'$, the set of pairs $T_1$, and we set the parameter $p$ to be $2k$. Since pairs in $T_1$ are disjoint, and $\max\{5p - 3, n - p\} = \max\{10k - 3, n - 2k\} \le \delta(H_1')$, all conditions of the lemma are satisfied. Thus, there exist vertex-disjoint paths $Q_1$, \ldots, $Q_{b - \frac{a}{2}}$ that have the respective endpoints from $T_1$ and cover all vertices of $H_1'$.

        We deal with $H_2'$ similarly. We only need to connect $t_1$, \ldots, $t_a$ in a shifted way compared to $s_1$, \ldots, $s_a$, so that we obtain a cycle at the end.
        Consider the following set $T_2$ of $\frac{a}{2} + r - b$ pairs of vertices in $H_2'$. If $b = r$, the pairs are $\{t_2, t_3\}$, $\{t_4, t_5\}$, \ldots, $\{t_{a - 2}, t_{a - 1}\}$, and $\{t_1, t_a\}$.
        If $b < r$, the pairs are $\{t_{2i}, t_{2i + 1}\}$ for $1 \le i < \frac{a}{2}$, $\{t_a, s_{b + 1}\}$, $\{t_j, s_{j + 1}\}$ for $b + 1 \le j < r$, and $\{t_1, t_r\}$.
        Again, we apply Lemma~\ref{lemma:many_paths} to the graph $H_2'$, the set of pairs $T_2$, and $p = 2k$.
        We obtain vertex-disjoint paths $R_1$, \ldots, $R_{r - b + \frac{a}{2}}$ that have the respective endpoints from $T_2$ cover all vertices of $H_2'$. 

        The resulting cycle $C$ with $V(C)=(V(G)\setminus B )\cup I$ is a cyclic concatenation of paths $P_1$, \ldots, $P_r$, $Q_1$, \ldots, $Q_{b - \frac{a}{2}}$, $R_1$, \ldots, $R_{r - b + \frac{a}{2}}$ in a certain order. Namely,
        \[C = P_1Q_1P_2R_1\cdots P_{a - 1}Q_{\frac{a}{2}}P_{a + 1}Q_{\frac{a}{2} + 1}P_{a + 2}\cdots P_b Q_{b - \frac{a}{2}} P_a R_{\frac{a}{2}} P_{b + 1}\cdots P_r R_{r - b + \frac{a}{2}},\]
        where we understand the notation $PQ$ for paths $P$ and $Q$ with a common endpoint as their natural concatenation. Clearly, $C$ is a cycle, and it spans all the previously defined paths. By construction, these paths cover all vertices in $I$, $V(H_1')$, and $V(H_2')$, thus they cover all vertices in $V(G)$ except those vertices in $B$ that are not covered by $\mathcal{P}$.
    \end{proof}

    Now it only remains to find a good path cover that covers the maximum number of vertices in $B$. By Claim~\ref{claim:good_cover}, a good path cover immediately gives us a cycle of the corresponding length, and we have also showed that a long cycle in $G$ induces a good path cover.

    To find the desired good path cover, first we observe that the number of vertices covered by the paths in the cover is at most $3|B|$ by property (3) of a good path cover. We proceed with a color-coding scheme using $r = 3|B|$ colors: color each vertex in $B$ in its own color, and each vertex in $V(G) \setminus B$ randomly and independently in one of the remaining $r - |B|$ colors, with equal probability for each color. Denote this coloring by $c : V(G) \to \{1, \ldots, r\}$. Now we look for a colored good path cover, that is, a good path cover that covers at most one vertex of each color.

    We find a colored good path cover with the help of dynamic programming. Define a \emph{state} as a tuple $(C, v, i, \ell, p)$ where $C$ is a subset of $\{1, \ldots, r\}$, $v$ is a vertex in $V(G)$, $i \in \{1, 2\}$, $\ell \in \{1, \ldots, r\}$, and $p \in \{0, 1, \ldots, |B|\}$.
    We call a state $(C, v, i, \ell, p)$ \emph{feasible} if there exists a set of vertex-disjoint paths $\mathcal{P} = \{P_1, \ldots, P_t\}$ in $G$ such that the following holds.
    \begin{enumerate}
        \item Every path $P_1$, \ldots, $P_{t - 1}$ starts and ends in $V(G) \setminus B$ and has the length of at least three, $P_t$ starts in $V(H_i)$,  ends in $v$, and its length is $\ell$.
        \item No path $P \in \mathcal{P}$ has two consecutive vertices in $V(G) \setminus B$.
        \item The paths in $\mathcal{P}$ cover exactly one vertex of each color in $C$, and no vertices of other colors.
        \item The number of paths in $\{P_1, \ldots, P_{t - 1}\}$ that start and end in different components of $G - B$ is exactly $p$.
    \end{enumerate}
    Note that $\mathcal{P}$ in the definition of a feasible state is essentially an ``unfinished'' good path cover that agrees with the state $(C, v, i, \ell, p)$. Our goal now is to compute the set of all feasible states $S$. We start by setting
    \[S_1 = \big\{\big(\{c(v)\}, v, 1, 1, 0\big) : v \in V(H_1)\big\} \cup \big\{\big(\{c(v)\}, v, 2, 1, 0\big) : v \in V(H_2)\big\}.\]
    These are our initial states, corresponding to sets containing one path of length one. Trivially, each such state is feasible, and these are all feasible states that use exactly one color. Next, for each $j$ in $\{1, \ldots, r - 1\}$, we show how to compute the set of feasible states $S_{j + 1}$ of size $j + 1$ from $S_j$, the set of feasible states of size $j$. Here by the size of the state $(C, v, i, \ell, p)$ we mean $|C|$, the number of colors used, which is the same as the total number of vertices covered by any set of paths corresponding to the state.

    To compute $S_{j + 1}$ from $S_j$, we iterate over all states in $S_j$ and try to extend each of them by an additional vertex. Intuitively, we either extend the unique unfinished path corresponding to the state, or declare it finished and start a new path. Fix a state $(C, v, i, \ell, p) \in S_j$, there is a set of paths $\mathcal{P} = \{P_1, \ldots, P_t\}$ satisfying the feasibility definition for $(C, v, i, \ell, p)$. Consider each $u \in N_G(v)$ such that $c(u) \notin C$. If both $v$ and $u$ are not in $B$, we do nothing. Otherwise, add to $S_{j + 1}$ the state $(C \cup c(u), u, i, \ell + 1, p)$. Clearly, the size of this state is $j + 1$, and it is easy to verify that the set of paths $\mathcal{P}' = \{P_1, \ldots, P_tu\}$ satisfies the feasibility definition for $(C \cup c(u), u, i, \ell + 1, p)$.

    For the ``new path'' kind of extending $(C, v, i, \ell, p)$ with $\ell > 2$, consider each vertex $u \in V(G) \setminus B$ such that $c(u) \notin C$. If $v \in B$ do nothing, otherwise add to $S_{j + 1}$ the state $(C \cup c(u), u, i', 1, p')$, where $i'$ is such that $u \in H_{i'}$ and $p' = p$ if $v \in V(H_i)$, or $p' = p + 1$ if $v \notin V(H_i)$. To see that this state is feasible, consider the set of paths $\mathcal{P}' = \{P_1, \ldots, P_t, u\}$.
    Indeed, every path among $P_1$, \ldots, $P_{t - 1}$ starts and ends in $V(G) \setminus B$, and $P_t$ as well, since $v \in V(G) \setminus B$. The length of $P_t$ is $\ell$ so at least three, and for $P_1$, \ldots, $P_{t - 1}$ this holds by feasibility of the original state. The last path $u$ starts in $V(H_{i'})$ by definition of $i'$, ends in $u$, and has the length of one. Properties (2) and (3) are preserved automatically. The value $p'$ reflects exactly how $p$ is changed with respect to the newly finished path $P_t$.


    Now we show that $S_{j + 1}$ contains all feasible states of size $j + 1$, provided that $S_j$ contains all feasible states of size $j$. Consider a state $(C', u, i', \ell', p') \in S_{j + 1}$ and a corresponding set of paths $\mathcal{P}' = \{P_1, \ldots, P_t\}$. Recall that $|P_t| = \ell'$, if $\ell' > 1$, consider a state $(C, v, i', \ell' - 1, p')$ where $v$ is the previous vertex to $u$ in $P_t$, $C = C' \setminus \{c(u)\}$. Observe that $(C, v, i', \ell' - 1, p')$ is feasible as witnessed by the set of paths $\mathcal{P} = \{P_1, \ldots, P_t'\}$ where $P_t'$ is $P_t$ without its last vertex $u$. Since $u \in N_G(v)$, $c(u) \notin C$, and $v$ and $u$ are not simultaneously in $V(G) \setminus B$ by property (2) for $\mathcal{P}'$, the state $(C', u, i', \ell', p')$ is added to $S_{j + 1}$ when the algorithm considers extending the state $(C, v, i', \ell'-1, p') \in S_j$ by the vertex $u$. If $\ell' = 1$, consider a state $(C, v, i, \ell, p)$ where $v$ is one of the endpoints of $P_{t - 1}$, $i$ is the index of the component of the other endpoint of $P_{t - 1}$, $\ell = |P_{t - 1}|$, $C = C' \setminus \{c(u)\}$, and $p$ is either $p'$ or $p' - 1$, depending on whether $v$ belongs to $H_i$ or not. The set of paths $\{P_1, \ldots, P_{t - 1}\}$ witnesses the feasibility of $(C, v, i, \ell, p)$, and thus $(C', u, i', \ell', p')$ is added to $S_{j + 1}$ on the corresponding ``new path'' step.

    Therefore, we have shown that for each $j$ in $\{1, \ldots, r - 1\}$, we correctly compute the set $S_{j + 1}$ from $S_j$, so in the end we have the sets $S_j$ of feasible states of size $j$, for each $j \in \{1, \ldots, r\}$. Finally, we consider a subset $\mathcal{C}$ of the feasible states $(C, v, i, \ell, p) \in \bigcup_{j  = 1}^r S_j$, such that $v \notin B$, $\ell > 2$, and $p'$ is at least 2 and even, where $p' = p$ if $v \in H_i$ and $p'= p + 1$ if $v \notin H_i$. Note that $\mathcal{C}$ is not empty since a long cycle in $G$ guaranteed by Theorem~\ref{thm:relaxed_long_cycle} induces a good path cover, and thus a feasible state of the form above. From $\mathcal{C}$, we pick a state maximizing $|C \cap \{1, \ldots, |B|\}|$. The set of paths $\{P_1, \ldots, P_t\}$ corresponding to this state is a good path cover in $G$ that covers the maximum number of vertices in $B$. Note that the actual good path cover may be found by the usual means of backtracking in dynamic programming.
    Together with Claim~\ref{claim:good_cover} this concludes the algorithm, and the proof of its correctness.


    \textbf{Running time analysis.} In the dynamic programming part, the number of states is at most $2^r \cdot n \cdot 2 \cdot r \cdot (k + 1)$. While considering a state, we update $\Oh(n)$ other states, thus the total running time of the dynamic programming subroutine is $2^{\Oh(k)} n^{\Oh(1)}$. For a fixed long cycle $C$ in $G$, the probability that we guess the coloring that assigns different colors to all vertices of the induced by $C$ good path cover, is at least $e^{-r}$, since there are at most $r$ vertices in the good path cover. By performing $\lceil e^r \rceil$ iterations of the color coding subroutine, we amplify the success probability to at least $1 - (1 - e^r)^{e^r} \ge 1 - e^{-1}$. Therefore, we obtain a Monte Carlo algorithm with constant success probability and running time $\Oh(k^2 \cdot e^{3k} \cdot 2^{3k} \cdot n^2) = \Oh(2^{\Oh(k)} n^{\Oh(1)})$. Finally, the algorithm could be derandomized in the standard fashion by using perfect hash families~\cite{NaorSS95}. 
\end{proof}

\subsection{Main theorem}
Now we move on to Theorem~\ref{theorem:hamiltonian}, the main result of this section. We restate the theorem here for convenience of the reader.

%

    \theoremhamiltonian*
	\begin{proof}
        First, we may assume that $n > 40k$, otherwise the problem can be solved by the classical $2^{\Oh(n)}$ algorithm for \textsc{Longest Cycle}.
	    Instead of proving the theorem directly, we show that there exists an algorithm that in time $2^{\mathcal{O}(k)}\cdot n^{\mathcal{O}(1)}$ either 
	    \begin{enumerate}
		    \item finds the longest cycle in $G$, or
		    \item finds a vertex cover of $G$ of size at most $\frac{n}{2}+9k$, or
		    \item finds a set $B'\supseteq B$ of size at most $35k$ such that $G-B'$ is not connected.
	    \end{enumerate}
	    We say that (1)--(3) are the \emph{terminal states} of the algorithm. 

	    If state (3) is reached, we simply invoke the algorithm from Lemma~\ref{lemma:hamiltonian_separator} with the respective separating set $B'$ of size at most $35k$. This gives us immediately the longest cycle in $G$.
	    Similarly, reaching terminal state (2) also suffices to solve the problem, as shown in the next claim.
	    
	    \begin{claim}
	    	If terminal state (2) is reached, the longest cycle in $G$ can be found in $2^{\Oh(k)}\cdot\polyn$ time.
	    \end{claim}
    	\begin{claimproof}
    		Denote the obtained vertex cover of $G$ of size at most $\frac{n}{2}+9k$ by $S$.
    		We would like to invoke the algorithm given by \Cref{thmVCad}, but we are not guaranteed that the longest cycle in $G$ has length of the form $2\delta(G-B)+k'$ for $k'\ge 0$.
			
			By \Cref{thm:relaxed_long_cycle}, we have that there is a cycle of length at least $\min\{2\delta(G-B), n-|B|\}$ in $G$, as $G$ is $2$-connected.
			We aim to achieve $2\delta(G-B)\le n-|B|$.
			Each vertex in $G-B$ has at most $|S|$ neighbours.
			Take the vertex in $G-B$ with smallest degree.
			It has at least $\frac{n}{2}-k> 19k$ neighbours in $G-B$.
			Obtain $B'$ by adding $19k$ neighbours of this vertex in $G-B$ to $B$.
			We have that $\delta(G-B')=\delta(G-B)-19k$, and $G-B'$ still contains at least one vertex.
			
			Note that $\delta(G-B')\le |S|-19k\le \frac{n}{2}-10k\le \frac{n}{2}-\frac{|B'|}{2}$, as $|B'| \le 20k$, so $2\delta(G-B')\le n-|B'|$. 
			Thus, by \Cref{thm:relaxed_long_path}, the length of the longest cycle in $G$ is of form $2\delta(G-B')+k'$ for $k'\ge 0$.
			The size of the vertex cover $S$ is at most $\frac{n}{2}+9k\le \delta(G-B)+10k= \delta(G-B')+29k$.

	    	Recall that Theorem~\ref{thmVCad} provides a $2^{\mathcal{O}(p + k')}\cdot n^{\mathcal{O}(1)}$-time algorithm that finds a cycle of length at least $2 \delta(G - B') + k'$ given a vertex cover of $G$ of size $\delta(G - B') + p$, if there is any.
	    	By trying  all possible $k'$ from $n - 2\delta(G - B') \le 40k$ to zero, we find the longest cycle in $G$ in time $2^{\mathcal{O}(k)}\cdot n^{\mathcal{O}(1)}$ as $p\le 29k$. By Theorem~\ref{thm:relaxed_long_cycle} there is a cycle of length at least $2\delta(G-B')$ in $G$, thus invoking Theorem~\ref{thmVCad} with $k' = 0$ necessarily provides us with a cycle.
    	\end{claimproof}
    
     Therefore, in what follows we assume that reaching any of the terminal states solves the problem immediately.

    Now consider a cycle $C$ of maximum length in $G$. Identically to the proof of Lemma~\ref{lemma:many_paths}, $C$ induces a path cover of a subset of $B$. Namely, in this proof, we call a set of vertex-disjoint paths $\mathcal{P}$ in $G$ \emph{a good path cover} if $\mathcal{P}$ satisfies the following properties.
    \begin{enumerate}
        \item Every path $P \in \mathcal{P}$ starts and ends in $V(G) \setminus B$.
        \item Each path $P \in \mathcal{P}$ has at least one vertex in $B$ and no two consecutive vertices in $V(G) \setminus B$.
        \item The paths of $\mathcal{P}$ contain at most $3|B|$ vertices in total.
    \end{enumerate}
    Note that this definition is the same as in Lemma~\ref{lemma:many_paths}, except for the property (4) there. Intuitively, we do not need it in this lemma since we may now assume that $G - B$ is connected. Since the current definition is strictly less restrictive, it follows immediately from the proof of Lemma~\ref{lemma:many_paths} that
    \begin{itemize}
        \item for each $0 \le t \le |B|$, if there is a cycle of length $n - t$ in $G$, there is also a good path cover in $G$ that covers all but $t$ vertices of $B$,
        \item in time $2^{\Oh(k)} n^{\Oh(1)}$ we can find a good path cover $\mathcal{P}$ that covers the maximum number of vertices in $B$, by the combination of color coding and dynamic programming. 
    \end{itemize}
    Note that the empty set is a good path cover, thus a good path cover always exists.

    So for the rest of the proof we deal with the case where we have computed a good path cover $\mathcal{P}$ of $G$, possibly an empty one. Denote by $r$ the number of paths in $\mathcal{P}$, and by $B'$ the set of vertices covered by paths in $\mathcal{P}$ together with the rest of vertices of $B$. By definition, $B \subset B'$, and by property (2) of a good path cover $|B'| \le 3k$. If $G - B'$ is not connected, we have a small separator: the algorithm outputs $B'$ and stops, reaching terminal state (3). If $G - B'$ is not $2$-connected, we add to $B'$ an arbitrary cut vertex of $G - B'$ and return $B'$. Thus, from now on we may assume that $G - B'$ is $2$-connected. The minimum degree of $G - B'$ is at least
    \[\delta(G - B') \ge \frac{n}{2} - k - |B' \setminus B| \ge \frac{n - |B' \setminus B|}{2} - 2k > \frac{n - |B'| + 2}{3},\]
    since $|B' \setminus B| \le 2k$ and $n > 16k$. By Theorem~\ref{proposition:cycle_or_is}, in time $\Oh(n^3)$ we find either a Hamiltonian cycle $C_0$ in $G - B'$, or an independent set of size $\delta(G - B') + 1$. If an independent set is found, its complement in $G - B'$ together with $B'$ is a vertex cover of $G$ of size at most $\frac{n}{2} + k + 2|B'| \le \frac{n}{2} + 7k$. In this case we output the vertex cover and stop, reaching terminal state (2).

    Otherwise, we have a Hamiltonian cycle $C_0$ in $G - B'$. Now, we iteratively insert the paths of $\mathcal{P} = \{P_1, \ldots, P_r\}$ into the cycle. Namely, for each $i \in \{1, \ldots, r\}$ we prove that given a cycle $C$ that contains exactly the vertices of the cycle $C_0$ and the paths $P_1$, \ldots, $P_{i - 1}$,
    we can either modify the cycle $C$ such that it satisfies the same property for $i + 1$, i.e. contains the vertices of the path $P_i$ as well, or reach one of the terminal states. Clearly, applying the above for each $i \in \{1, \ldots, r\}$, starting from the cycle $C_0$, proves the theorem. Thus from now on we focus on this statement.

    Consider the path $P_i$ and the obtained cycle $C$ that contains all vertices of $C_0$ and $P_1$, \ldots, $P_{i - 1}$. Denote the endpoints of $P_i$ by $s$ and $t$, observe that both $s$ and $t$ have at least $\frac{n}{2} - 3k$ neighbors on $C$. That holds since $s \notin B$, so $\deg_{G - B} (s) \ge \frac{n}{2} - k$, and at most $2k$ vertices of $G$ belong to $B' \setminus B$ and are neither on $C$ nor in $B$, analogously for $t$. 

    Denote by $C_s$ the set of neighbors of $s$ on $C$, and by $C_t$ the set of neighbors of $t$ on $C$.
    Consider a vertex $c_s \in C_t$ and a vertex $c_t \in C_t$. If $c_s$ and $c_t$ are next to each other on $C$ then we can immediately insert $P_i$ in $C$. If these vertices are not adjacent, but are at distance two on $C$ with a vertex $c' \notin B$ between them, we do the following. Insert $P_i$ in $C$ by going from $c_s$ to $c_t$ through $P_i$ and not through $c'$, denote the resulting cycle by $C'$. The vertex $c'$ is the only vertex that is in $V(C) \cup V(P_i)$, but not on $C'$, thus we are done as long as we insert $c'$ back in $C'$. By the same argument as for $s$ and $t$, $c'$ has at least $\frac{n}{2} - 3k$ neighbors on $C'$. If there are two neighbors of $c'$ on $C'$ that are consecutive on $C'$ again we can immediately insert $c'$ in $C'$, thus we assume this is not the case. Now on $C'$ between every two consecutive neighbors of $c'$ there is a group of at least one and possibly several non-neighbors of $c'$. Since there are at least $\frac{n}{2} - 3k$ neighbors of $c'$ on $C'$, there are also at least $\frac{n}{2} - 3k$ such groups of consecutive non-neighbors. Since there are at most $\frac{n}{2} + 3k$ non-neighbors of $c'$ on $C'$, at most $6k$ of the groups may contain more than one vertex. Thus at least $\frac{n}{2} - 9k$ groups consist of a single vertex, denote the set of all such vertices by $I$. Each vertex of $I$ is not adjacent to $c'$, but both of its neighbors on $C'$ are adjacent to $c'$. We claim that if two vertices in $I$ are adjacent in $G$, there is a cycle that goes through $c'$ and all vertices of $C'$. Denote these vertices by $u$ and $v$, go from $c'$ to a neighbor of $u$, then to $v$ along the arc of $C'$ that does not contain $u$, then take the edge $uv$, and finally collect the rest of $C'$ going from $u$ to a neighbor of $v$ and returning to $c'$. 
    If no two vertices in $I$ are adjacent in $G$, then $I$ is an independent set of size at least $\frac{n}{2} - 9k$ in $G$. Thus the complement of $I$ is a vertex cover of $G$ of size at most $\frac{n}{2} + 9k$, and we are in the terminal state (2).
    
    \begin{figure}[ht]
	\centering
	\begin{tikzpicture}[scale=0.6]
		\tikzstyle{vertex}=[draw, fill, circle, black, minimum size=4,inner sep=0pt]
		\tikzstyle{path1}=[very thick]
		\tikzstyle{path2}=[blue, very thick, dashed]
		\tikzstyle{regular}=[]
		
		\draw[path1, red] plot[smooth, tension=.7] coordinates {(-1, 10) (2, 10.5) (4, 10) (7, 10.5) (9, 10)};
		\node at (2, 10) {$P$};
		\node[vertex] (s) at (-1, 10) {};
		\node at (-1.5, 10) {$s$};
		\node[vertex] (t) at (9, 10) {};
		\node at (9.5, 10) {$t$};
		
		\draw[path1, blue] plot[smooth, tension=.7] coordinates {(-2, 5) (-1, 2.5) (4, 1.5) (9, 2.5) (10, 5)};
		\node at (-2.4, 3.7) {$C$};
		\draw[path1, blue] plot[smooth, tension=.7] coordinates {(0, 7) (3, 7.5) (5, 7.5) (8, 7)};
		\node[vertex] (s1) at (-2, 5) {};
		\node[vertex] (s2) at (-1.1, 6.2) {};
		\node at (-0.9, 5.7) {$c_s$};
		\node[vertex] (s3) at (0, 7) {};
		\node[vertex] (t3) at (10, 5) {};
		\node[vertex] (t2) at (9.1, 6.2) {};
		\node at (9, 5.7) {$c_t$};
		\node[vertex] (t1) at (8, 7) {};
		\draw[dashed] (s) -- (s1);
		\draw[dashed] (s) -- (s2);
		\draw[path1] (s) -- (s3);
		\draw[path1, blue] (s1) -- (s2);
		\draw[dashed, blue] (s2) -- (s3);
		\draw[dashed] (t) -- (t1);
		\draw[dashed] (t) -- (t2);
		\draw[path1] (t) -- (t3);
		\draw[path1, blue] (t1) -- (t2);
		\draw[dashed, blue] (t2) -- (t3);
		\draw[path1] (s2) -- (t2);
		
	\end{tikzpicture}
	\caption{Inserting the path $P$ (in red) into the cycle $C$ (in blue) in the presence of an edge between an internal $s$-vertex $c_s$ and an internal $t$-vertex $c_t$. The resulting cycle is in solid.}
	\label{fig:insertion}
\end{figure}
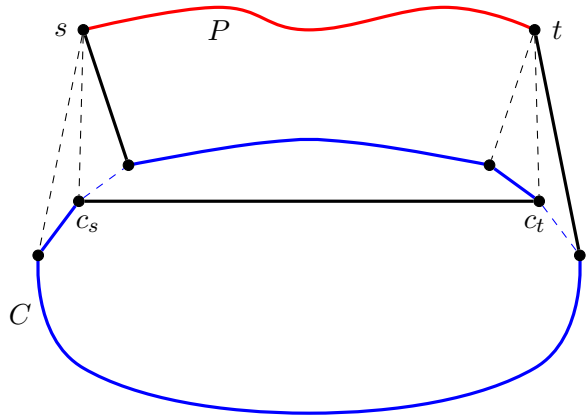

    Now we deal with the case where for every $c_s \in C_s$ and every $c_t \in C_t$, there is either a vertex of $B$ or at least two other vertices between them on $C$. First, we bound the number of common neighbors of $s$ and $t$ on $C$, denote $C_s \cap C_t$ by $C_{st}$. Fix an ordering on $C$, and consider a vertex $u \in C_{st}$ and the next vertex $v$ along the cycle that belongs to either $C_s$ or $C_t$. Between $u$ and $v$, there must be at least two vertices that belong to neither $C_s$ nor $C_t$, or a vertex of $B$. Thus with each vertex of $C_{st}$ we can uniquely associate either two vertices of $V(C) \setminus C_s \setminus C_t$, or a vertex of $B$. We get that apart from the vertices of $C_{st}$, $C_s \setminus C_{st}$ and $C_t \setminus C_{st}$, there are at least $2(|C_{st}| - |B|)$ other vertices in $C$. Summing the sizes of these four disjoint sets together, we get
    \begin{align*}
        &|C_{st}| + (|C_s| - |C_{st}|) + (|C_t| - |C_{st}|) + 2(|C_{st}| - |B|) \le n,\\
        &|C_{st}| \le n - |C_s| - |C_t|  + 2 |B| \le 8k,
    \end{align*}
    since both $C_s$ and $C_t$ contain at least $\frac{n}{2} - 3k$ vertices, and the size of $B$ is at most $k$. From this bound, we also immediately get that the number of vertices on $C$ that are not adjacent to both $s$ and $t$ is at most $n - |C_s| - |C_t| + |C_{st}| \le 14k$. Thus nearly all vertices of $C$, except for $\Oh(k)$, are adjacent either to $s$ but not to $t$, or to $t$ but not to $s$. As vertices from $C_s$ cannot be next to vertices from $C_t$ on $C$, they must come in large consecutive chunks along the cycle. To formalize this intuition, let us call a vertex in $C_s$ an \emph{internal $s$-vertex} if both of its neighbors along the cycle are also from $C_s$, and \emph{internal $t$-vertices} are defined analogously. We claim that except for $\Oh(k)$ vertices, all the vertices of $C$ are either internal $s$-vertices or internal $t$-vertices. Vertices from $C_s$ that are not internal $s$-vertices must have at least one neighbor along $C$ that is not from $C_s$ nor $C_t$, and the same holds for $C_t$. However, there are at most $14k$ vertices in $V(C) \setminus C_s \setminus C_t$, and each of them can ``spoil'' at most two vertices of $C_s$ or $C_t$. Also note that a vertex of $C_{st}$ must have vertices of $V(C) \setminus C_s \setminus C_t$ on both sides, as a vertex from $C_s$ cannot lie next to a vertex of $C_t$ on $C$. Thus the total number of internal $s$-vertices and internal $t$-vertices is at least 
    \begin{multline*}(|C_s| - |C_{st}|) + (|C_t| - |C_{st}|) - 2 (|V(C) \setminus C_s \setminus C_t| - |C_{st}|) \ge 2(\frac{n}{2} - 3k) - 28k = n - 34k.\end{multline*}

    Now assume there is an edge between an internal $s$-vertex and an internal $t$-vertex. If this holds, the path $P_i$ can be inserted in $C$ in the same way as in the case of a single high-degree vertex above, see Figure~\ref{fig:insertion} for an illustration. 
    On the other hand, if there are no edges between internal $s$-vertices and internal $t$-vertices, then the graph induced on the sets of internal $s$-vertices and $t$-vertices is not connected, as these sets are both non-empty.
    Then removing at most $34k$ vertices from $G$ leaves these sets disconnected. Thus we arrive to the terminal state (3) where we have a small separator.
    In order to apply \Cref{lemma:hamiltonian_separator}, it should contain $B$ as a subset, so after taking the union with $B$ its size is at most $35k$.
	\end{proof}

\section{\cyclebananadec}\label{sec:bananas}

In this section, we define \cyclebananadec{s} and show that, given a \cyclebananadec for a cycle in $G$, we can either find a longer cycle or solve the instance $(G,B,k)$ of \probDC in time single-exponential in $k+|B|$.

\begin{figure}[ht]
	\begin{center}
		\ifdefined\STOC
\begin{tikzpicture}[scale=0.4]
\else
\begin{tikzpicture}[scale=0.8]
\fi
\tikzstyle{vertex}=[draw, fill, circle, black, minimum size=3,inner sep=0pt]

\node[vertex] (v1) at (0, 7){};
\node[vertex] (v2) at (0, 6){};
\node[vertex] (v3) at (0, 5){};
\node[vertex] (v4) at (0, 4){};
\node[vertex] (v5) at (10, 7){};
\node[vertex] (v6) at (10, 6){};
\node[vertex] (v7) at (10, 5){};
\node[vertex] (v8) at (10, 4){};

\draw[very thick, red] (v1) -- (v2) -- (v3) -- (v4);
\draw[very thick, red] (v5) -- (v6) -- (v7) -- (v8);

\node[vertex] (v9) at (1, 8.5) {};
\node[vertex] (v10) at (9, 8.5) {};

\draw (v1) -- (v9);
\draw (v5) -- (v10);

\draw  plot[smooth, tension=.7] coordinates {(1, 8.5) (2, 9.3) (3, 9.2) (4, 9.7) (5, 9.4) (6, 9.7) (7, 9.2) (8, 9.3) (9, 8.5) };
\draw[blue, very thick]  plot[smooth cycle, tension=.7] coordinates { (1, 9) (3, 9.9) (7, 9.9) (9, 9 ) (9, 8.1) (7, 9) (3, 9) (1, 8.1)};

\node[vertex] (v11) at (1, 2.5) {};
\node[vertex] (v12) at (9, 2.5) {};

\draw (v4) -- (v11);
\draw (v8) -- (v12);

\draw  plot[smooth, tension=.7] coordinates {(1, 2.5) (2, 1.7) (3, 1.8) (4, 1.3) (5, 1.6) (6, 1.3) (7, 1.8) (8, 1.7) (9, 2.5) };
\draw[blue, very thick]  plot[smooth cycle, tension=.7] coordinates { (1, 2) (3, 1.1) (7, 1.1) (9, 2 ) (9, 2.9) (7, 2) (3, 2) (1, 2.9)};

\draw[fill=lightgray]  plot[smooth cycle, tension=.7] coordinates {(3.8, 3.1) (5, 2.4) (5.7, 3.3)};
\node[vertex] (v13) at (4.2, 3) { };
\node[vertex] (v14) at (5, 2.6) { };
\node[vertex] (v15) at (5.5, 3.1) { };
\draw (v13) -- (v14);
\draw (v14) -- (v15);
\draw (v13) -- (v15);
\draw (v4) -- (v13);
\draw (v15) -- (v7);

\node[vertex] (v16) at (2.5, 7.5) { };
\node[vertex] (v17) at (5.5, 7.3) { };
\draw[blue, very thick, rotate around={176:(4,7.4)}]  (4, 7.4) ellipse (53pt and 15pt) ;

\draw (v2) -- (v16);

\draw[fill=lightgray]  plot[smooth cycle, tension=.7] coordinates {(6.6, 7) (7.5, 7.3) (7.4, 6.9) (6.8, 6.8)};
\node[vertex] (v19) at (5.85, 6.15) { };
\node[vertex] (v18) at (3.5, 5) { };
\draw (v2) -- (v18);
\draw (v7) -- (v19);
\draw[blue, very thick, rotate around={25:(4.5,5.5)}]  (4.5, 5.5) ellipse (53pt and 15pt) ;
\draw  plot[smooth cycle, tension=.7] coordinates {(5.8, 7.8) (7, 7) (6.5, 5.8) (5.2, 6.5)};
\node[vertex] (v20) at (6.8, 6.9) { };
\node[vertex] (v21) at (7.4, 7.1) { };
\draw (v20) -- (v21);
\draw (v6) -- (v21);

\node at (-0.5, 5.5) {$P_1$};
\node at (10.5, 5.5) {$P_2$};
\node at (0.5, 9) {D1};
\node at (9.5, 1.9) {D1};
\node at (6.4, 5.4) {D2};
\node at (6.2, 3) {D0};
\node at (9.9, 7.8) {$C$};
\end{tikzpicture}
	\end{center}
	\caption{A schematic example of a \cyclebananadec, vertices belonging to $B$ are in light gray. Removing the paths $P_1$ and $P_2$ leaves two \ref{enum:cycle_tunnel_path_bic}-type components that correspond to the long arcs $P'$ and $P''$ of the starting cycle $C$, one \ref{enum:cycle_tunnel_path_cut_left}-type component, and a component consisting only of vertices from $B$, denoted by D0. The four Dirac components are in thick blue.}
	\label{fig:bananasoncycle}
\end{figure}
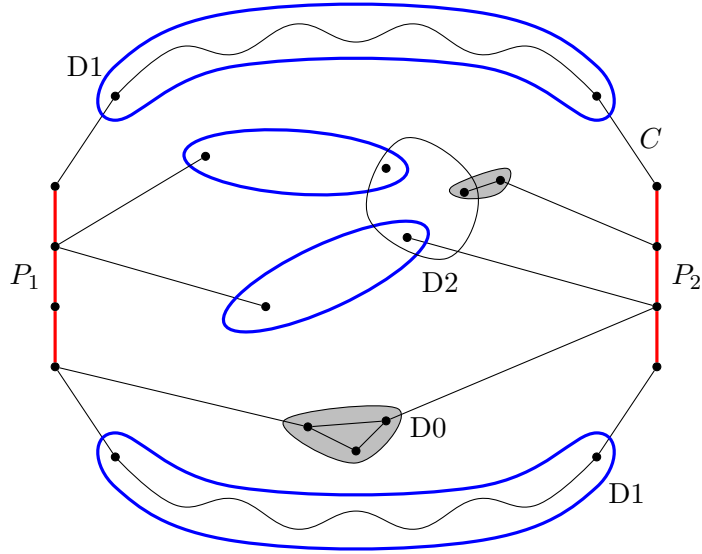
\begin{definition}[\textbf{\Cyclebananadec and \cyclebanana}]
	Let $G$ be a 2-connected graph, let $B$ be a subset of $V(G)$, and let $C$ be a cycle in $G$ of length at least $2\delta(G-B)$.
	We say that  two disjoint paths $P_1$ and $P_2$ in $G$ induce \emph{a \cyclebananadec   for   $C$ and $B$}  in $G$ if
	\begin{itemize}
		\item 
		The cycle $C$ is of the form $C=P_1 {P'}P_2{P''}$, where  each of the paths ${P'}$ and ${P''}$ has at least $\delta(G- B)-2$ edges.
		%
		\item 
		    Let $G'$ be the graph obtained from $G$ by applying $B$-refinement to every  connected component $H$ of $G- V(P_1  \cup P_2)$,  except those components $H$ with  $V(H)\subseteq B$. Note that no edges of the paths $P_1$ and $P_2$ are contracted. 
		Then for  every connected component $H'$ of $G'-V(P_1  \cup P_2)$, except those with $V(H')\subseteq B$, holds $|V(H')|\ge 3$ and one of the following. 
		\begin{enumerate}[label=(D\arabic*)]
			\item\label{enum:cycle_tunnel_path_bic} $H'$ is $2$-connected and the maximum size of a matching in  $G'$ between $V(H')$ and $V(P_1)$  is one,  and between $V(H')$ and $V(P_2)$ is also  one;
			\item\label{enum:cycle_tunnel_path_cut_left} $H'$ is not 2-connected,   
			    exactly one vertex of $P_1$ has neighbors in $H'$, that is, 			
			$|N_{G'}(V(H'))\cap V(P_1)|=1$, and no inner vertex from a  leaf-block of $H'$ has a neighbor in $P_2$;
			\item\label{enum:cycle_tunnel_path_cut_right} The same as  \ref{enum:cycle_tunnel_path_cut_left}, but with $P_1$ and $P_2$ interchanged. That is, 
			$H'$ is not 2-connected,  			
			$|N_{G'}(V(H'))\cap V(P_2)|=1$, and no inner vertex from  a leaf-block of $H'$ has a neighbor in $P_1$.			
		\end{enumerate}
	
		\item There is exactly one connected component $H$ in $G-V(P_1\cup P_2)$ with $V(H)\setminus B=V(P')\setminus (B\cup\{s',t'\})$, where $s'$ and $t'$ are the endpoints of $P'$.
		Analogously, there is exactly one connected component $H$ in $G-V(P_1\cup P_2)$ with $V(H)\setminus B=V(P'')\setminus (B\cup\{s'',t''\})$.
		
	\end{itemize}
	The set of \emph{\cyclebanana}s for a \cyclebananadec 
	is defined as follows.
	First,  for each component $H'$ of type \ref{enum:cycle_tunnel_path_bic}, $H'$ is a \cyclebanana of the \cyclebananadec.
	Second, for each leaf-block of each $H'$ of type \ref{enum:cycle_tunnel_path_cut_left}, or of type \ref{enum:cycle_tunnel_path_cut_right}, this leaf-block  is also  a \cyclebanana of the \cyclebananadec.	
	For an example of a Dirac decomposition, see Figure~\ref{fig:bananasoncycle}.
\end{definition}


Note that Lemma~\ref{lemma:st_path_banana_consecutive} holds for an arbitrary cycle $C$ if we replace \banana{s} and \bananadec{s} with \cyclebanana{s} and \cyclebananadec{s}.
We give the analogue of this lemma below without proof, since it is identical to the proof of Lemma~\ref{lemma:st_path_banana_consecutive}.

\begin{lemma}\label{lemma:dirac_cycle_banana_consecutive}
	Let $G$ be a $2$-connected graph, $B\subseteq V(G)$, $C$ be a cycle in $G$. Let paths  $P_1, P_2$ induce a \cyclebananadec  for $C$ and $B$ in $G$.
	Let $M$ be a \cyclebanana of the \cyclebananadec and $P$ be a path in $G$ such that $P$ contains at least one vertex in $V(P_1)\cup V(P_2)$.
	If $P$ enters $M$, then all vertices of $M$ hit by $P$ appear consecutively on $P$.
\end{lemma}


We now want to prove an analogue of \Cref{lemma:st_path_edge_of_banana} showing that if a long cycle in $G$ exists, then it suffices to look for a long cycle entering a \cyclebanana.
For that we first require the following weaker result.

\begin{lemma}\label{lemma:dirac_b_leaf_block_separator}
	Let $G$ be a $2$-connected graph, $B\subseteq V(G)$, $C$ be a cycle in $G$ of length less than $2\delta(G-B)+k$. Let paths  $P_1, P_2$ induce a \cyclebananadec  for $C$ and $B$ in $G$.
	If $H'$ is a \ref{enum:cycle_tunnel_path_cut_left}-type or a \ref{enum:cycle_tunnel_path_cut_right}-type component of the \cyclebananadec and $S$ is a $B$-leaf-block separator of $H'$, then there is a cycle of length at least $\frac{1}{2}(5\delta(G-B)-|S|-(k+5))$ that enters a \cyclebanana in $G$.
\end{lemma}
\begin{proof}
	Without loss of generality, let $H'$ be a \ref{enum:cycle_tunnel_path_cut_left}-type component of the \cyclebananadec.
	Take a vertex $v \in V(H')$ that is not an inner vertex of a leaf-block of $H'$ and has a neighbour in $P_2$.
	Such vertex always exists by definition of a \cyclebananadec.
	
	Let $S$ be a $B$-leaf-block separator of $H'$.
	By \Cref{lemma:separator_in_non_2c}, there is a $(c,v)$-path of length at least $\frac{1}{2}\delta(H'-B)-\frac{1}{2}|S|$ in $H'$ for some cut-vertex $c$ of a leaf-block of $H'$.
	Also $\delta(H'-B)\ge \delta(G-B)-|V(P_1)\cup V(P_2)|\ge \delta(G-B)-(k+5)$, since the total length of $P_1$ and $P_2$ is at most $k+3$.
	Denote this leaf-block of $H'$ by $L$.

	Note that $\delta(L-(B\cup\{c\}))\ge \delta(G-B)-2$ by properties of \ref{enum:cycle_tunnel_path_cut_left}-type components.
	By \Cref{thm:relaxed_st_path}, there is a path of length at least $\delta(G-B)-2$ between $c$ and any other vertex in $L$.
	Let $u$ be an inner vertex in $L$ that has a neighbour in $P_1$.
	
	Combine the $(u,c)$-path inside $L$ with the $(c,v)$-path going outside $L$ in $H'$.
	The obtained path is a $(u,v)$-path of length at least $(\delta(G-B)-2)+(\frac{1}{2}\delta(H'-B)-\frac{1}{2}|S|)$.
	
	Since $u$ and $v$ have neighbours in $V(P_1)$ and $V(P_2)$ respectively, we obtain a chord of $C$ of length at least $\delta(G-B)+\frac{1}{2}\delta(H'-B)-\frac{1}{2}|S|$.
	The chord splits $C$ into two arcs, one of which is of length at least $\delta(G-B)$.
	Combine this arc with the chord and obtain a cycle of length at least $2\delta(G-B)+\frac{1}{2}\delta(H'-B)-\frac{1}{2}|S|\ge \frac{5}{2}\delta(G-B)-\frac{1}{2}|S|-\frac{1}{2}(k+5)$.
\end{proof}

The following lemma is an analogue of \Cref{lemma:st_path_edge_of_banana} for \cyclebanana{s}.
In contrast to \Cref{lemma:dirac_cycle_banana_consecutive}, the proof is significantly different from the proof of \Cref{lemma:st_path_edge_of_banana}.

\begin{lemma}\label{lemma:dirac_cycle_edge_of_banana}
	Let $G$ be a graph, $B\subseteq V(G)$ be   a subset of its vertices and $P_1, P_2$ induce a \cyclebananadec  for a cycle $C$ of length less than $2\delta(G-B)+k$ in $G$.
	Let $k$ be an integer such that $6k+4|B|+6 < \delta(G-B)$.
	If there exists a cycle of length at least $2\delta(G-B)+k$ in $G$ that contains at least one vertex in $V(P_1)\cup V(P_2)$, then there exists a cycle of length at least $2\delta(G-B)+k$ in $G$ that  
	enters a \cyclebanana.
\end{lemma}
\begin{proof}
	Suppose that there exists a cycle $C'$ of length at least $2\delta(G-B)+k$ in $G$ that contains at least one vertex in $V(P_1)\cup V(P_2)$.
	If $C'$ already contains an edge of a \cyclebanana, we are done.
	We now assume that $C'$ does not contain any edge of any \cyclebanana.
	We show how to use $C'$ and construct a cycle of length at least $2\delta(G-B)+k$ in $G$ that contains an edge of a \cyclebanana of the given \cyclebananadec.
	
	Let $W$ be the set of all vertices of $G$ that are vertices of non-leaf-blocks of \ref{enum:cycle_tunnel_path_cut_left}-type or \ref{enum:cycle_tunnel_path_cut_right}-type components in the \cyclebananadec. 
	We start with the following claim.
	
	\begin{claim}
		$|W\cap V(C')|> 5k$.
	\end{claim}
	\begin{claimproof}
		This is a counting argument.
		Note that $C'$ cannot contain an edge with both endpoints inside a \cyclebanana of $G$.
		Since \cyclebanana{s} of $G$ are \ref{enum:cycle_tunnel_path_bic}-type components of the \cyclebananadec and leaf-blocks of \ref{enum:cycle_tunnel_path_cut_left}-type or \ref{enum:cycle_tunnel_path_cut_right}-type connected components, each edge of $C'$ has an endpoint either in $V(P_1)\cup V(P_2)\cup B$, or inside a non-leaf-block of a \ref{enum:cycle_tunnel_path_cut_left}-type or a \ref{enum:cycle_tunnel_path_cut_right}-type connected component.
		The union of the vertex sets of the non-leaf-blocks form the set $W$.
		Hence, $(W\cap V(C'))\cup V(P_1)\cup V(P_2)\cup B$ is a vertex cover of $C'$.
		
		Note that a vertex cover of any cycle consists of at least half of its vertices.
		Then $$2|(W\cap V(C'))\cup V(P_1)\cup V(P_2)\cup B|\ge |V(C')|\ge 2\delta(G-B)+k.$$
		Immediately we get that \begin{multline*}2|W\cap V(C')|\ge 2\delta(G-B)+k-2|V(P_1)\cup V(P_2)|-2|B|\ge 2\delta(G-B)+k-2(k-2)-2|B|>10k.\end{multline*}
	\end{claimproof}

	The following claim is useful for constructing long chords of $C'$ going through the \cyclebanana{s} that are leaf-blocks.
	
	\begin{claim}\label{claim:hard_lemma_no_inner_vertices}
		Let $H'$ be a \ref{enum:cycle_tunnel_path_cut_left}-type or a \ref{enum:cycle_tunnel_path_cut_right}-type component in the \cyclebananadec.
		$C'$ does not contain any inner vertex of the leaf-blocks of $H'$.
	\end{claim}
	\begin{claimproof}
		Suppose that $C'$ contains some vertex $u \in V(H')$ that is an inner vertex of some leaf-block $L$ of $H'$.
		As $L$ is a \cyclebanana of $G$, $C'$ cannot contain any edge of $L$, so $C'$ should enter $L$ from $V(P_1) \cup V(P_2)$ through $u$ and leave it immediately.
		By definition of \cyclebananadec{s}, the only option to enter or leave $L$ is to go through the only vertex in $V(P_1)$ (if $H'$ is of type \ref{enum:cycle_tunnel_path_cut_left}) or in $V(P_2)$ (if $H'$ is of type \ref{enum:cycle_tunnel_path_cut_right}).
		As $C'$ cannot contain any vertex twice, this is not possible.
	\end{claimproof}

	We now use the above claims to construct either a family of long chords of $C'$ going through \cyclebanana{s}, or a $B$-leaf-block separator of some of the \ref{enum:cycle_tunnel_path_cut_left}-type or \ref{enum:cycle_tunnel_path_cut_right}-type components in the \cyclebananadec.
	
	To construct the first chord of $C'$, take a vertex $w_1 \in W$.
	Since $w_1$ is a vertex of a  separable component $H'$, there is a cut vertex $c_1$ of a leaf-block $L_1$ of $H'$ reachable from $w_1$ inside $H'$.
	The leaf-block $L_1$ contains also at least one vertex $v_1\neq c_1$ that is connected to $V(P_1)$ (if $H'$ is of type \ref{enum:cycle_tunnel_path_cut_left}) or to $V(P_2)$ (if $H'$ is of type \ref{enum:cycle_tunnel_path_cut_right}) outside $H'$.
	We know that $\delta(L_1-(B\cup\{c_1\}))\ge \delta(G-(B\cup \{c_1\}))-1\ge \delta(G-B)-2$, since the only outside neighbour of vertices in $L_1-(B\cup{c_1})$, apart from vertices in $B$, is a single vertex in $V(P_1)$ or $V(P_2)$.
	By Corollary~\ref{thm:relaxed_st_path}, there exists an $(c_1,v_1)$-path inside $L_1$ of length at least $\delta(G-B)-2$.
	Combine this with $(w_1,c_1)$-path inside $H'$ and obtain a $(w_1,v_1)$-path inside $H'$.
	
	Note that the constructed $(w_1,v_1)$-path can contain vertices from $W$ apart from $w_1$.
	Let $w'_1 \in W$ be the vertex on the $(w_1,v_1)$-path farthest from $w_1$.
	Note that the $(w'_1,v_1)$-subpath does not contain any vertex from $W$ except $w'_1$, and it still contains the $(c_1,v_1)$-path as a subpath by Claim~\ref{claim:hard_lemma_no_inner_vertices}.
	Hence, we obtain a $(w'_1,v_1)$-path of length at least $\delta(G-B)-2$ inside $H'$ that does not contain any vertex in $W\setminus \{w'_1\}$.
	To obtain a long chord of $C'$, it is left to reach the vertex in $V(P_1)\cup V(P_2)$ from $v_1$ outside $H'$, and then follow the cycle $C$ until a vertex $v'_1$ of $C'$ is reached. This is always possible since $V(C)\cap V(C')\supseteq (V(P_1)\cup V(P_2))\cap V(C')\neq \emptyset$.
	We obtain a chord of length at least $\delta(G-B)-1$ connecting $w'_1$ and $v'_1$.
	
	To construct the second chord, we follow the same process for a vertex $w_2 \in W\setminus \{w'_1\}$.
	When constructing the path going from $w_2$ to a cut vertex of a leaf-block, we prohibit this path from going through $w'_1$.
	If $w'_1$ separates $w_2$ from all leaf-block cut vertices, then we obtain a small $B$-leaf-block separator of $H'$.
	Otherwise, we obtain a $(w'_2,v'_2)$-chord of $C'$ of length at least $\delta(G-B)-1$ that does not contain any vertex in $W\setminus\{w'_1,w'_2\}$.
	It is important that during the construction of different chords we always follow $C$ in the same direction.
	
	Repeat this process $3k$ times and obtain either a family of $3k$ $(w'_i,v'_i)$-chords of $C'$ of length at least $\delta(G-B)-1$, or a $B$-leaf-block separator of size at most $3k$.
	If it is the latter case, then, by Lemma~\ref{lemma:dirac_b_leaf_block_separator}, there is a cycle of length at least $2\delta(G-B)+\frac{1}{2}(\delta(G-B)-3k-(k+5))> 2\delta(G-B)+k$ that enters a \cyclebanana.
	We now assume that a family of chords is obtained.
	The following claim is useful.
	
	\begin{claim}
		If for some $i,j\in [3k]$ we have $v'_i\neq v'_j$, then the chords between $w'_i$ and $v'_i$ and between $w'_j$ and $v'_j$ do not have any common vertex.
	\end{claim}
	\begin{claimproof}
		Note that $v'_i$ or $v'_j$ depend only on the vertex in $V(P_1)$ or in $V(P_2)$ which we start following $C$ from.
		If $v'_i\neq v'_j$, then they were found when starting from different vertices.
		Then the $i^\text{th}$ and the $j^\text{th}$ chords were constructed from different components of the Dirac decomposition, as for each separable component there is only one vertex in $V(P_1) \cup V(P_2)$ that is adjacent to inner vertices of leaf-blocks of this component.
		Therefore, the $(w'_i,v_i)$- and $(w'_j,v_j)$-subpaths of the chords do not have common vertices.
		The $(v_i,v'_i)$- and $(v_j,v'_j)$-subpaths also cannot have any common vertex as $v'_i$ and $v'_j$ were found as the first vertices from $V(C')$ on $C$ when following $C$ in the same direction.
	\end{claimproof}

	\begin{definition}
		Consider two chords of $C$ that do not have common endpoints.
		Denote the endpoints of one chord by $s$ and $t$ and of the other by $p$ and $q$.
		We say that these two chords of $C$ \emph{intersect graphically}, if the vertices $s$, $t$, $p$ and $q$ are located in the order $s, p, t, q$ on $C$ when following $C$ in one or the other direction.
		In other words, two chords intersect graphically if each chord cuts $C$ into two arcs each containing exactly one endpoint of the other chord.
	\end{definition}

	We can now show that if two chords in the constructed family intersect graphically, then we can find a long cycle that contains these two chords.
	Assume that there are $i, j \in [3k]$ such that $v'_i\neq v'_j$ and the vertices $w'_i,v'_i,w'_j,v'_j$ are located in the order $w'_i,w'_j,v'_i,v'_j$ on $C'$ when following $C'$ in one of the two directions.
	These four vertices split $C'$ into four arcs.
	A pair of opposite arcs together with the two chords constitute a cycle in $G$.
	Take the pair of arcs with the longest total length.
	This total length is at least $\frac{1}{2}|V(C')|\ge \delta(G-B)+\frac{k}{2}$.
	Combining these two arcs with the two chords, we obtain a cycle of length at least $3\delta(G-B)+\frac{k}{2}-2>2\delta(G-B)+k$.
	This cycle enters two \cyclebanana{s}, as each chord enters a \cyclebanana.
	Hence, in this case the desired cycle exists and the lemma is proved.
	
	\begin{figure}
		\center
		\begin{tikzpicture}
\tikzstyle{vertex}=[circle,draw,fill,inner sep=1pt]
\draw  (-4,2.5) ellipse (3.5 and 2.5);
\node at (-0.5,4) {$C'$};
\node [vertex] (v1) at (-6.05,4.5) {};
\node [vertex] (v5) at (-5.35,4.8) {};
\node [vertex] (v6) at (-4.55,4.95) {};
\node [vertex] (v7) at (-3.4,4.9) {};
\node [vertex] (v11) at (-2.25,4.65) {};
\node [vertex] (v2) at (-5.9,0.4) {};
\node [vertex] (v3) at (-5.1,0.15) {};
\node [vertex] (v4) at (-4.2,0) {};
\node [vertex] (v9) at (-2.8,0.15) {};
\node [vertex] (v10) at (-2.1,0.4) {};
\node [vertex] (v12) at (-1.55,0.7) {};
\node [vertex] (v8) at (-3.35,0.05) {};
\node [vertex] (v13) at (-6.8,3.95) {};
\node [vertex] (v14) at (-6.8,1) {};
\draw[blue]  plot[smooth, tension=.7] coordinates {(v1) (-6.15,4.15) (-6.05,4) (-5.55,3.7) (-5.8,3.45) (-6.05,3.2) (-5.85,2.8) (-5.7,2.5) (-5.4,2.3) (-5.6,2.1) (-6.05,1.85) (-5.9,1.7) (-5.75,1.55) (-5.7,1.3) (-5.9,1.1) (-6,0.9) (-5.9,0.65) (v2)};
\draw[blue]  plot[smooth, tension=.7] coordinates {(-6.05,4.5) (-5.6,4.1) (-5.25,3.7) (-5.1,3.25) (-5.2,2.95) (-5.15,2.7) (-5.1,2.2) (-4.95,1.85) (-5.1,1.5) (-5.3,1.25) (-4.9,0.9) (-5.3,0.6) (v3) };
\draw[blue]  plot[smooth, tension=.7] coordinates {(v4) (-4.25,0.25) (-4.4,0.4) (-4.55,0.75) (-4.4,0.9) (-4.35,1.35) (-4.55,1.65) (-4.55,1.95) (-4.35,2.25) (-4.35,2.6) (-4.55,2.85) (-4.65,3.15) (-4.45,3.4) (-4.5,3.65) (-4.7,3.85) (-5,4.15) (-4.85,4.35) (-5.05,4.55) (v5)};
\draw[blue]  plot[smooth, tension=.7] coordinates {(v4) (-3.95,0.35) (-3.9,0.75) (-4,1.2) (-3.95,1.7) (-3.85,2) (-3.95,2.4) (-3.85,2.9) (-4.2,3.2) (-4.2,3.55) (-4.2,3.95) (-4.4,4.3) (-4.4,4.7) (v6)};
\draw[blue]  plot[smooth, tension=.7] coordinates {(v7) (-3.7,4.45) (-3.6,4.3) (-3.45,4) (-3.75,3.8) (-3.6,3.55) (-3.5,3.35) (-3.4,3) (-3.5,2.7) (-3.45,2.45) (-3.3,1.95) (-3.25,1.6) (-3.45,1.3) (-3.5,0.95) (-3.5,0.6) (-3.55,0.25) (v8)};
\draw[blue]  plot[smooth, tension=.7] coordinates {(v7) (-3.2,4.35) (-3.05,4.05) (-3.25,3.85) (-3.25,3.55) (-3.2,3.2) (-3.05,2.9) (-3,2.65) (-3,2.3) (-3,1.85) (-2.85,1.45) (-3,1.15) (-3.05,0.8) (-2.8,0.45) (v9)};
\draw[blue]  plot[smooth, tension=.7] coordinates {(v7) (-3.2,4.65) (-3.05,4.55) (-2.85,4.4) (-2.8,4.1) (-2.95,3.85) (-2.95,3.45) (-2.9,3.05) (-2.7,2.75) (-2.6,2.2) (-2.6,1.8) (-2.65,1.65) (-2.65,1.35) (-2.35,1.1) (-2.45,0.8) (-2.4,0.65) (-2.3,0.5) (v10)};
\draw[blue]  plot[smooth, tension=.7] coordinates {(v11) (-2.3,4.25) (-2.15,3.9) (-2.15,3.55) (-2.15,3.2) (-2.25,2.75) (-2.05,2.45) (-2,2) (-2,1.5) (-1.85,1.1) (-1.8,0.8) (v12)};
\draw[blue]  plot[smooth, tension=.7] coordinates {(v13) (-6.6,3.7) (-6.4,3.6) (-6.4,3.3) (-6.3,3) (-6.3,2.55) (-6.55,2.45) (-6.45,2.15) (-6.3,2) (-6.25,1.6) (-6.5,1.45) (-6.35,1.25) (-6.6,1.1) (v14)};
\node at (-6.95,4.15) {$a_1$};
\node at (-6.4,4.7) {$a_2,a_3$};

\node [vertex] (v1) at (-6.05,4.5) {};
\node [vertex] (v5) at (-5.35,4.8) {};
\node [vertex] (v6) at (-4.55,4.95) {};
\node [vertex] (v7) at (-3.4,4.9) {};
\node [vertex] (v11) at (-2.25,4.65) {};
\node [vertex] (v2) at (-5.9,0.4) {};
\node [vertex] (v3) at (-5.1,0.15) {};
\node [vertex] (v4) at (-4.2,0) {};
\node [vertex] (v9) at (-2.8,0.15) {};
\node [vertex] (v10) at (-2.1,0.4) {};
\node [vertex] (v12) at (-1.55,0.7) {};
\node [vertex] (v8) at (-3.35,0.05) {};
\node [vertex] (v13) at (-6.8,3.95) {};
\node [vertex] (v14) at (-6.8,1) {};
\node at (-5.34,4.995) {$a_4$};
\node at (-4.635,5.135) {$a_5$};
\node at (-3.35,5.15) {$a_6,a_7,a_8$};
\node at (-2.1,4.85) {$a_9$};
\node at (-6.85,0.75) {$b_1$};
\node at (-5.96,0.235) {$b_2$};
\node at (-5.175,-0.025) {$b_3$};
\node at (-4.17,-0.21) {$b_4,b_5$};
\node at (-3.3,-0.2) {$b_6$};
\node at (-2.7,-0.1) {$b_7$};
\node at (-2,0.15) {$b_8$};
\node at (-1.4,0.45) {$b_9$};
\end{tikzpicture}
		\caption{The family of long chords of the cycle $C'$ that do not intersect graphically pairwise.}\label{fig:long_chords}
	\end{figure}
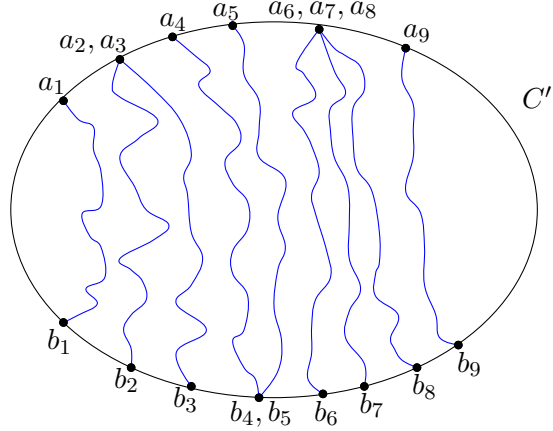

	We now assume that no two chords of the family intersect graphically.
	In this case we can arrange them in an order from left to right.
	That is, we can choose a permutation $\pi \in S_{3k}$ and for each $i\in[3k]$ either a pair $a_i=w'_{\pi_i}$ and $b_i=v'_{\pi_i}$ or a pair $a_i=v'_{\pi_i}$ and $b_i=w'_{\pi_i}$ in such a way that following the cycle $C'$ starting from $a_1$ one will read the $6k$ vertices in the order $a_1, a_2, \ldots, a_{3k}, b_{3k}, b_{3k-1}, \ldots, b_1$ (see Figure~\ref{fig:long_chords}).
	It is important to note that $a_i={a_{i+1}}$ or $b_i=b_{i+1}$ might hold true for any $i\in [3k-1]$,
    but at least one of $a_i\neq a_{i+1}$ and $b_i\neq b_{i+1}$ always holds.
	
	Take an arbitrary $i\in[3k]$.
	The chord between $a_i$ and $b_i$ splits $C'$ into two arcs.
	We call them \emph{left arc of $(a_i,b_i)$}, that is, the arc that contains vertices $a_1,b_1,a_2,b_2,\ldots,a_i,b_i$ and the \emph{right arc of $(a_i,b_i)$}, that is, the arc that contains vertices $a_i,b_i,a_{i+1},b_{i+1},\ldots,a_{3k},b_{3k}$.
	If at least one of these arcs has length at least $\delta(G-B)+k+1$, we call the chord between $a_i$ and $b_i$ a \emph{good} chord.
	Then the chord together with the longer arc constitute a cycle of length at least $2\delta(G-B)+k$.
	This cycle enters a \cyclebanana, so the proof is complete if there is a good chord in the constructed family of chords.
	We now show that there exists at least one good chord in the family.

	If the chord between $a_1$ and $b_1$ is good, then we are done.
	Otherwise, both arcs of $(a_1,b_1)$ are of length at most $\delta(G-B)+k$.
	Since the length of $C'$ is at least $2\delta(G-B)+k$, it follows that the length of both these arcs is at least $\delta(G-B)$.
	Consider the chord between $a_2$ and $b_2$.
	Note that the left arc of $(a_2,b_2)$ is longer than the left arc of $(a_1,b_1)$ because $(a_1,b_1)\neq (a_2,b_2)$.
	Hence, the length of the left arc of $(a_2,b_2)$ is at least $\delta(G-B)+1$.
	Analogously, we can show that for each $i\in[3k]$ the length of the left arc of $(a_i,b_i)$ is at least $\delta(G-B)+i-1$.
	Hence, for any $j \in [2k]$ the chord between $a_{k+j}$ and $b_{k+j}$ is a good chord.
	The proof of the lemma is complete.
\end{proof}

Finally, we state and prove the main theorem of this section.

\begin{theorem}\label{thm:cyclebanana}
	Let $(G,B,k)$ be a given instance of \probDC.
	There is an algorithm that, given a cycle $C$ in $G$ and two paths $P_1, P_2$ that induce an \cyclebananadec for $C$ and $B$ in $G$, in time $2^{\Oh{(k+|B|)}}\cdot\polyn$ either
	\begin{itemize}
		\item Solves $(G,B,k)$, or
		\item Finds a cycle longer than $C$ in $G$.
	\end{itemize}
\end{theorem}
\begin{proof}
	The algorithm considers several cases.
	If the given cycle $C$ is of length at least $2\delta(G-B)+k$, then the algorithm correctly determines that $(G,B,k)$ is a yes-instance.
	Hence, we assume that $|V(C)|< 2\delta(G)+k$.
	From now on, we also assume that $6k+4|B|+6 < \delta(G-B)$, otherwise the algorithm solves $(G,B,k)$ using the algorithm from Proposition~\ref{prop:longest_cycle}.
	
	Suppose now that $G$ contains a cycle $C'$ of length at least $2\delta(G-B)+k$.
	We show how the algorithm finds \emph{some} cycle of length at least $2\delta(G-B)+k$ in $G$ or enlarges $C$ provided that $C'$ exists.
	We are now interested in the set $X=V(C')\cap (V(P_1)\cup V(P_2))$.
	Depending on its cardinality, there are several cases.
	In most of the cases, it is possible for the algorithm to replace one arc of $C$ with a longer arc that is found using the algorithm for \probstP given by Theorem~\ref{thmEG}.
	The algorithm is usually applied to a component in $G-V(P_1\cup P_2)$ with the goal of finding a path of length at least $\delta(G-B)+k/2$.
	Since $\delta(G-(V(P_1\cup P_2)\cup B))> \delta(G-B)-k-4$, the running time of the algorithm is still bounded by $2^{\Oh(k+|B|)}\cdot\polyn$.
	
	\medskip\noindent\textbf{Case 1:} $|X|=0$.
	Then $C'$ is completely contained in some connected component $H$ of $G-V(P_1\cup P_2)$.
	This component cannot be contained in $B$, since $|B| < \delta(G-B)$.
	So after the $B$-refinements, this component is a component $H'$ of type \ref{enum:cycle_tunnel_path_bic}, \ref{enum:cycle_tunnel_path_cut_left} or \ref{enum:cycle_tunnel_path_cut_right}.
	Note that only the leaf-blocks that have all inner vertices in $B$ are contracted in $H$ to obtain $H'$.
	The cycle $C'$ does not pass through cut vertices.
	Moreover, its length is greater than $B$.
	Hence, no edge of $C'$ is contracted during the $B$-refinements of $H$, so $C'$ is fully contained in $H'$.
	
	By the last property of \cyclebananadec{s}, there are exactly two connected components in $G-(V(P_1)\cup V(P_2))$ containing vertices of $C$.
	Both of them contain at most $\delta(G-B)+k+|B|$ vertices, as the length of $C$ is less than $2\delta(G-B)+k$.
	Hence, $H'$ does not share any vertices with the initial cycle $C$, as $|V(H')|\ge |V(C')|\ge 2\delta(G-B)+k$.
	
	If $H'$ is of type \ref{enum:cycle_tunnel_path_bic}, then it contains a path of length at least $|V(C')|/2$ between any pair of vertices by \Cref{lemma:biconnected_cycle_to_any_path}.
	Take any pair $(s,t)$ of neighbours of $H'$ in $V(P_1)$ and $V(P_2)$ respectively.
	There is an $(s,t)$-path in $G$ of length at least $|V(C')|/2+2>\delta(G-B)+k/2$ that contains only vertices in $V(H')\cup B$ as internal vertices.
	One of the arcs of $C$ between $s$ and $t$ have length less than $\delta(G-B)+k/2$, so it can be replaced with the obtained $(s,t)$-path, making $C$ longer.
	
	If $H'$ is of type \ref{enum:cycle_tunnel_path_cut_left}, then it is not $2$-connected.
	Denote the only neighbour of $H'$ in $V(P_1)$ by $s$.
	Note that the graph $G'[V(H')\cup \{s\}]$ is $2$-connected and still contains the cycle $C'$.
	Hence, it contains a path of length at least $|V(C')|/2$ between any pair of vertices.
	Now take any neighbour of $H'$ in $V(P_2)$, say $t$.
	It is easy to obtain an $(s,t)$-path in $G$ of length at least $|V(C')|/2+1$ going only through vertices in $V(H')\cup B$.
	Again, this path is a replacement for one of the arcs between $s$ and $t$ in $G$.
	The case of type \ref{enum:cycle_tunnel_path_cut_right} is symmetrical.
	
	\textbf{Conclusion of Case 1.}
	To handle this case, the algorithm unconditionally iterates over all components in $G-V(P_1\cup P_2)$ and tries to find a suitable path  of length at least $\delta(G-B)+k/2$ in a $2$-connected subgraph of $G$ using the algorithm of Theorem~\ref{thmEG} for \probstP.
	Note that a subgraph picked by the algorithm is always a graph $H'$ with $\delta(H'-B')\ge \delta(G-B'-(V(P_1)\cup V(P_2))$, where $|B'|\le |B|+1$.
	Hence, $\delta(H'-B')\ge \delta(G-B)-(k+5)$, so the algorithm for \probstP always runs in $2^{\Oh(k+|B|)}\cdot\polyn$ running time.
	If a suitable path is found, $C$ is made longer by the algorithm, and the algorithm outputs the longer cycle and terminates.
	Otherwise, there are no long cycles $C'$ with $|X|=0$ in $G$.
	
	\medskip\noindent\textbf{Case 2.}
	$|X|=1$.
	Denote the only vertex in $X$ by $v$.
	Note that $C'$ passes through only one connected component in $G-V(P_1\cup P_2)$, since $C'-v$ is a path having no common vertices with $P_1$ or $P_2$.
	Denote the component containing $C'-v$ in $G-V(P_1\cup P_2)$ by $H$.
	We know that $H$ consists of at least $2\delta(G-B)+k-1$ vertices, so it is not fully contained in $B$ and does not share any vertex with the initial cycle $C$, just as in the previous case.
	
	Without loss of generality, assume that $v \in V(P_1)$.
	Denote by $H'$ the connected component $H$ after the $B$-refinements.
	Independently of the type of $H'$ in the \cyclebananadec, there is a vertex in $H$ with a neighbour in $V(P_2)$.
	Hence, there is a path starting in a certain vertex $u \in V(P_2)$ and going to a certain vertex $z \in V(C'-v)$ through $H$ in $G$.
	Take the longer arc between $z$ and $v$ on $C'$ and combine it with the path between $u$ and $z$.
	The obtained path is of length at least $|V(C')|/2+1$ and is a replacement for the shorter arc between $u$ and $v$ on $C$.
	
	The only obstacle here is that $H$ is not necessarily $2$-connected.
	However, the graph $G[V(H)\cup \{v\}]$ still contains the whole cycle $C'$.
	If $H'$ is of type \ref{enum:cycle_tunnel_path_bic}, then $G[V(H)\cup \{v\}]$ is necessarily $2$-connected after $B$-refinements.
	If $H'$ is of type $\ref{enum:cycle_tunnel_path_cut_left}$, then $G[V(H)\cup \{v\}]$ also becomes $2$-connected after $B$-refinements, as $v$ is the only neighbour in $V(P_1)$ connecting all leaf-blocks of $H'$ together.
	In either of the two cases, if $z$ is fixed, the path between $v$ and $z$ can be found using the algorithm for \probstP.
	
	Finally, if $H'$ is of type \ref{enum:cycle_tunnel_path_cut_right}, then $u$ is the only neighbour of $H'$ in $V(P_2)$ after the $B$-refinements.
	Then the graph $G[V(H)\cup\{u,v\}]$ is necessarily $2$-connected after $B$-refinements and the desired $(u,v)$-path can be found inside it.
	
	\textbf{Conclusion of Case 2.}
	The algorithm iterates over all suitable pairs of $v$ and $H$, and iterates over all possible options of $z$ or $u$ when necessary.
	When this triple is fixed, it is left to apply the algorithm of Theorem~\ref{thmEG} to the corresponding $B$-refinement as described above.
	Again, in time $2^{\Oh(k+|B|)}\cdot\polyn$ our algorithm either makes the initial cycle $C$ longer and stops or correctly determines that no long cycle $C'$ with $|X|=1$ exists.
	
	\medskip\noindent\textbf{Case 3:}
	$|X|=2$.
	Let $X=\{s,t\}$.
	Starting from this case, we need to consider \cyclebanana{s} that $C'$ enters.
	By Lemma~\ref{lemma:dirac_cycle_edge_of_banana}, we can assume that $C'$ enters some \cyclebanana $M$.
	The cycle $C'$ has two arcs between $s$ and $t$.
	At least one of them enters $M$ and, by Lemma~\ref{lemma:dirac_cycle_banana_consecutive}, we know that all vertices of $M$ appear consecutively on this arc.
	
	Suppose that both arcs between $s$ and $t$ enter $M$.
	If both $s,t \in V(P_1)$ or $s,t\in V(P_2)$, then we obtain a matching of size two between $V(P_i)$ and a \cyclebanana, which is not possible by the definition of \cyclebananadec.
	Hence, we can assume that $s \in V(P_1)$ and $t \in V(P_2)$.
	Since both arcs enter $M$, there is a connected component $H$ in $G-V(P_1\cup P_2)$ that contains $M$ and both arcs of $C'$.
	Thus $C'$ is contained in $G[V(H)\cup\{s,t\}]$.
	After the $B$-refinements, $G[V(H)\cup\{s,t\}]$ also contains the whole cycle $C'$ similarly to the arguments above.
	Note that after the $B$-refinements this graph is $2$-connected, since if $H'$ is of type \ref{enum:cycle_tunnel_path_cut_left} or of type \ref{enum:cycle_tunnel_path_cut_right}, the vertex $s$ or the vertex $t$ correspondingly is the vertex connecting all its leaf-blocks together.
	Hence, the algorithm can look for an $(s,t)$-path of length at least $\delta(G-B)+k/2$ inside the graph $G[V(H)\cup\{s,t\}]$ with applied $B$-refinements.
	The component $H$ contains at least $2\delta(G-B)+k-2$ vertices, so it does not share vertices with $C$.
	Thus, this $(s,t)$-path is a suitable replacement for a shorter arc between $s$ and $t$ on $C$.
	
	Suppose now that only one arc of $C'$ between $s$ and $t$ enters $M$.
	Then, by Lemma~\ref{lemma:dirac_cycle_banana_consecutive}, we can be sure that \emph{all} vertices of $M$ appear consecutively on $C'$.
	That is, there are two vertices $u,v \in V(M)\cap V(C')$ such that one of the arcs of $C'$ between $u$ and $v$ is a $(u,v)$-path inside $M$, and the other arc is a $(u,v)$-path in $G$ that does not contain any vertex of $M$ as internal vertex.
	In this case, the algorithm can find these two arcs in the following way.
	
	When $s,t$ and $M$ are fixed, the algorithm iterates over all pairs of distinct vertices $u, v \in V(M)$.
	Firstly, the algorithm tries to find a path between $u$ and $v$ outside $M$.
	Since there is a path of length at least $\delta(G-B)-2$ between any pair of vertices in $M$, the outer path length $\delta(G-B)+k+2$ is sufficient to construct a cycle of length $2\delta(G-B)+k$ in $G$.
	Hence, to find the $(u,v)$-path outside $M$, the algorithm removes all vertices in $V(M)\setminus\{u,v\}$ from $G$, and adds a single edge between $u$ and $v$ in $G$.
	Note that if the outer path between $u$ and $v$ exists, then $G$ remains $2$-connected after these operations, since $u$ and $v$ still belong to the same cycle.
	If $G$ is not $2$-connected, then the choice of $u$ and $v$ was wrong.
	Otherwise, we apply the algorithm of Theorem~\ref{thmEG} to the changed graph $G$ to find a long path between $u$ and $v$.
	If a $(u,v)$-path of length at least $\delta(G-B)+k+2$ exists, then $(G,B,k)$ is a yes-instance.
	Otherwise, the algorithm finds the longest path between $u$ and $v$.
	This is done in $2^{\Oh(k+|B|)}\cdot\polyn$ time.
	Note that a $(u,v)$-path of length at least $\delta(G-(V(P_1)\cup V(P_2)\cup B\cup \{u,v\}))\ge \delta(G-B)-k-6$ always exists in the modified graph $G$ by \Cref{thm:relaxed_st_path}.
	
	If the outer $(u,v)$-path is found, it is left for us to find a long path between $u$ and $v$ inside $M$.
	If this path is of length at least $\delta(G-B)+2k+6$, then $(G,B,k)$ is a yes-instance of \probDC, since the outer $(u,v)$-path is of length at least $\delta(G-B)-k-6$.
	Thus, using the algorithm for \probstP, we either find a sufficiently long path between $u$ and $v$ inside $M$, such that the total length of this path and the outer path is at least $2\delta(G-B)+k$, or conclude that none exists and move on to the next choice of $u$ and $v$.
	
	\textbf{Conclusion of Case 3.}
	To handle this case, the algorithm iterates over all possible pairs of $s$ and $t$.
	To handle the case when $C'$ enters just one connected component of $G-V(P_1\cup P_2)$, the algorithm behaves similarly to previous cases.
	Additionally, to handle the case when $C'$ contains a consecutive path inside a \cyclebanana,  the algorithm iterates over all possible \cyclebanana{s} $M$, and pairs $u,v \in V(M)$ and tries to construct a long cycle using two calls to the algorithm for \probstP.
	
	\textbf{Case 4.}
	$|X|\ge 3$.
	Then $X$ contains three distinct vertices $v_1, v_2, v_3$.
	These vertices split $C'$ into three arcs $A_1, A_2, A_3$, where $A_1$ is the arc between $v_1$ and $v_2$ that does not contain $v_3$, $A_2$ is the arc between $v_2$ and $v_3$ that does not contain $v_1$, and $A_3$ is the arc between $v_3$ and $v_1$ that does not contain $v_2$.
	By Lemma~\ref{lemma:dirac_cycle_edge_of_banana}, we can assume that $C'$ enters a \cyclebanana $M$.
	Without loss of generality, assume that $A_1$ enters $M$.
	By Lemma~\ref{lemma:dirac_cycle_banana_consecutive}, all vertices of $M$ appear consecutively on this arc.

	\begin{claim}
		$A_2$ and $A_3$ do not contain any vertex of $M$.
	\end{claim}
	\begin{claimproof}
		Take the arc between $v_1$ and $v_3$ that contains $v_2$, i.e.\ the union of $A_1$ and $A_2$.
		We now that this arc enters $M$, so by Lemma~\ref{lemma:dirac_cycle_banana_consecutive} all vertices of $M$ appear consecutively on it.
		But $A_1$ contains at least two vertices of $M$.
		Hence, $A_2$ cannot contain any vertex of $M$, as $v_2 \notin V(M)$ separates $A_1$ and $A_2$ on the arc between $v_1$ and $v_3$.
		
		To show that $A_3$ does not contain any vertex of $M$, take the arc between $v_3$ and $v_2$ that contains $v_1$, i.e.\ the union of $A_3$ and $A_1$.
		Again, by Lemma~\ref{lemma:dirac_cycle_banana_consecutive} this arc contains vertices of $M$ consecutively, but $v_1$ divides $A_3$ and $A_1$ on the arc.
		Since $A_1$ contains at least two vertices of $M$, $A_3$ cannot contain any of them.
	\end{claimproof}

	The claim shows that vertices of $M$ induce an arc of $C'$, similarly to the second part of Case 3.
	Hence, this case can be handled by the algorithm in exactly the same way as in Case 3.
	
	\textbf{Conclusion of Case 4.}
	To cover this case, our algorithm first fixes $v_1, v_2 \in V(P_1\cup P_2)$. 
	Then it iterates over all \cyclebanana{s} of the \cyclebananadec and tries to combine a long cycle from two paths, one inside the \cyclebanana, and one outside.
	This is done in exactly the same way as in the second part of Case 3.
	
	The list of cases is exhaustive, so if $C'$ exists, our algorithm enlarges the initial cycle $C$ or finds a cycle of length at least $2\delta(G-B)+k$ in $G$, determining that $(G,B,k)$ is a yes-instance.
	If $C'$ does not exist, the algorithm does not find any long arc or long cycle in $G$, and safely decides that $(G,B,k)$ is a no-instance.
	This concludes the proof.
\end{proof}

\section{Long Dirac Cycle: Putting all together}\label{sec:longDC}

In this section we finalize the proof of \Cref{theorem:main} by combining the main results of previous sections.
This relies crucially on the following lemma.
The most important part of this lemma is the construction of a \cyclebananadec.

\begin{lemma}\label{lemma:main_cycle_lemma}
	Let $G$ be an $n$-vertex $2$-connected graph, $B\subseteq V(G)$, and $k$ be an integer such that   $0< k \le \frac{1}{24}\delta(G-B)$, and 
	\[2k+2|B|+12\leq \delta(G- B)< \frac{n}{2}.
	\]
	Then there is an algorithm that, given a cycle $C$ of length less than $2\delta(G- B)+k$ with $V(G-B)\not\subseteq V(C)$ in polynomial time finds either
	\begin{itemize}
		\item Longer cycle in $G$, or
		\item Vertex cover of $G-B$ of size at most $\delta(G-B)+2k$, or
		\item Two paths $P_1, P_2$ that induce a \cyclebananadec for $C$ and $B$ in $G$.
	\end{itemize}
\end{lemma}

Before proceeding with the proof of the lemma, we show  how to use it for the proof of  \Cref{theorem:main}.

\subsection{Proof of \Cref{theorem:main}}


We combine the main results of Sections \ref{sec:vcalgo}, \ref{sec:HamCycles}, \ref{sec:bananas}, and \Cref{lemma:main_cycle_lemma}.
Let $(G, B, k)$ be an instance of \probDC. First we consider the cases that do not fit the conditions of  \Cref{lemma:main_cycle_lemma}. 
If $\delta(G-B)<12$ or if $24 k > \delta(G-B)$, we can find a cycle of length at least $2\delta(G-B)+k> 48k+24$ in time $2^{\Oh(k)}\cdot\polyn$ by calling  the algorithm for \textsc{Longest Cycle} from  \Cref{prop:longest_cycle}.

If $2k+2|B|+12> \delta(G- B)$, we have that $5k+4|B|+48> 2\delta(G- B)+k$.
By  \Cref{prop:longest_cycle}, a cycle of length at least $5k+4|B|+48$ could be found in time $2^{\Oh(k+|B|)}\cdot\polyn$.

The remaining reason why \Cref{lemma:main_cycle_lemma} could not be applied to a cycle $C$ is that $C$ satisfies $V(G-B)\subseteq V(C)$.
Then $|V(C)|\ge n-|B|$,
implying that $\delta(G-B)\ge n/2-(|B|+k)/2$.
In this case,
we apply the algorithmic results of Section~\ref{sec:HamCycles}.
We put $k':=\max\{|B|,\frac{|B|+k}{2}\}$,  and obtain that $\delta(G-B)\ge \frac{n}{2}-k'$ for $|B|\le k'$. We apply Theorem~\ref{theorem:hamiltonian} for $G, B$ and $k'$. Then the problem is solvable in time  $2^{\Oh(k')}\cdot\polyn=2^{\Oh(k+|B|)}\cdot\polyn$.

From now we assume that $k$ and $B$ satisfy the conditions of  \Cref{lemma:main_cycle_lemma}.
Then start with arbitrary cycle $C$ in $G$.
If its length is at least $2\delta(G-B)+k$, then report that $(G,B,k)$ is a yes-instance.
Otherwise, apply \Cref{lemma:main_cycle_lemma} to $G,B,k$ and $C$.
In polynomial time we either find a longer cycle, a vertex cover,  or a \cyclebananadec of $G$.
If a longer cycle is found and the length of this cycle is still less than $2\delta(G-B)+k$, we call
  \Cref{lemma:main_cycle_lemma} with the longer cycle. 
  If a vertex cover of $G-B$ of size at most $\delta(G-B)+2k$ is found, then the vertex cover of $G$ is at most $\delta(G-B)+2k +|B|.$ We 
    apply Theorem~\ref{thmVCad} to solve the problem  in time $2^{\Oh(k+|B|)}\cdot\polyn$.
Finally, if a \cyclebananadec for $C$ and $B$ is found in $G$, we use Theorem~\ref{thm:cyclebanana} to solve $(G,B,k)$ in running time single-exponential in $k+|B|$ or find a longer cycle in $G$ and repeat the application of  \Cref{lemma:main_cycle_lemma}.

The proof of Theorem~\ref{theorem:main} (up to the proof of \Cref{lemma:main_cycle_lemma})  is complete.
\qed

\subsection{Last piece: proof of \Cref{lemma:main_cycle_lemma}}
The remaining part of the section is devoted to the postponed proof of \Cref{lemma:main_cycle_lemma}.

\begin{proof}[Proof of \Cref{lemma:main_cycle_lemma}]
The proof is algorithmic. We try 
 to replace an arc of $C$, that is, a path in $C$,  with a path in $G- V(C)$.
	This process of \emph{enlarging} $C$  is similar to the process of enlarging a path in Lemma~\ref{lemma:st_path_or_tunnel}.
	We consider connected components $H$ in $G-V(C)$ that contain at least one vertex in $V(G)\setminus B$.
	Note that at least one such component exists since $V(G)\setminus V(C)\not\subseteq B$.
	
	To simplify our job, we first apply $B$-refinements to all connected components in $G-V(C)$.
	Without loss of generality, we assume that $G$ is a graph with all possible $B$-refinements applied, i.e.,  $\gbref{B}{H}=G$ for any connected component $H$ in $G-V(C)$ with $V(H)\not\subseteq B$.
	Note that this assumption preserves all resulting points of the lemma statement: if a longer cycle, or a vertex cover, or a \cyclebananadec is found for the graph with applied $B$-refinements, they can be easily restored in the original graph.

	Similarly to the proof of Lemma~\ref{lemma:st_path_or_tunnel}, we consider several cases depending on the structure of a connected component $H$ with $V(H)\not\subseteq B$.
	The difference is that isolated vertices in $G-V(C)$ now do not lead to an immediate enlargement of $C$.
	However, we show that they contribute  to a construction of a vertex cover of $G-B$. 
	
 In what follows we  prove the following.  If  there is a component $H$ with 
$G-V(C)$ with exactly two vertices, then cycle $C$ can be always enlarged. If there is a component $H$ with at least $3$ vertices, call it a large component,  then either $C$ can be enlarged, or 
$H$ has a very special structure. The special structure of large components  is used twice. First, we    show that if there is at least one single-vertex component   and at least one large component, then $C$ can be enlarged. Thus if we cannot enlarge $C$, it means that either $G-V(C)$ is an independent set or all components are large. In the first case, we prove that the vertex cover of   $G-B$ is at most $\delta(G-B)+2k$. In the second case, the structural properties of large components are used to construct a \cyclebananadec.

We start with two claims that will be used in several places of the proof.
The first claim shows that if there is a pair of distant consecutive neighbors of a vertex $h\not\in V(C)$ in  $C$, then $C$ can be enlarged.
	
	\begin{claim}\label{claim:dirac_cycle_isolated_close_neighbors}
		Let $h \in V(G)\setminus V(C)$ be a vertex with at least $\delta(G-B)-2$ neighbors in $V(C)$ and such that there is a pair of   neighbors $u,v$ of $h$ on $C$  such that one of the $(u,v)$-arcs is of length  at least $8k$  containing no other neighbors of $h$. Then $C$ can be enlarged in polynomial time.
	\end{claim}	
	\begin{claimproof}
		Suppose that there are two neighbors of $h$, say $u, v \in V(C)$ such that one arc of $C$ between $u$ and $v$ is of length at least $8k$ and does not contain any neighbor of $h$.  Hence 
 the other arc between $u$ and $v$ contains all neighbors of $h$ on $C$. Moreover, since the length of $C$ is at most  $2\delta(G-B)+k-1$, the length of this arc is at most $2\delta(G-B)-7k-1$.

		There are at least $\delta(G-B)-2$ neighbors of $h$ on $C$. 
		Since $2(\delta(G-B)-3)>2\delta(G-B)-7k-1$, by the pigeonhole principle, there is a pair of neighbors of $h$ that are adjacent $C$.
		Then $h$ can be inserted in $C$ between these neighbors so the length of $C$ increases by one.
	\end{claimproof}

	The following claim allows to eliminate the existence of large connected components in $G-V(C)$, when there are isolated vertices in $G-V(C)$.
	This claim will be useful later in this proof. Recall that a chord of a cycle $C$ is a path connecting two vertices of $C$ and containing no other vertices of $C$.
	
	\begin{claim}\label{claim:dirac_cycle_isolated_and_chord}
		If there is a vertex $h \in V(G)\setminus V(C)$ with at least $\delta(G-B)$ neighbors in $V(C)$ and there is a chord of $C$ of length at least $16k$ that does not pass through $h$, then $C$ can be enlarged in polynomial time.
	\end{claim}
	\begin{claimproof}
		By Claim~\ref{claim:dirac_cycle_isolated_close_neighbors}, we can  assume that for every pair of   neighbors $u,v$ of $h$ on $C$, each of the  $(u,v)$-arcs is either of length less than $8k$ or  contains other neighbors of $h$.

		Let the endpoints of the chord be $c_1, c_2 \in V(C)$.
		If the distance between $c_1$ and $c_2$ in $C$ is less than the length of the chord, then $C$ can be made longer by replacing an arc between $c_1$ and $c_2$ with the chord.
		Otherwise, both arcs between $c_1$ and $c_2$ are of length at least $16k$.
		
		Each of these two arcs should contain a neighbor of $h$ as an internal vertex.	Select one of the two arcs between $c_1$ and $c_2$.
		Let $v_1\neq c_1$ be the neighbor of $h$ that is closest to $c_1$ on this arc. Since there are no other neighbors of $h$ between $c_1$ and $v_1$, the 
  distance in $C$ between $c_1$ and $v_1$ is at most $8k$.
		Analogously, take the other arc between $c_1$ and $c_2$ and let $v_2$ be the neighbor of $h$ on this arc that is closest to $c_2$, but is different from it.
		Again, the distance between $c_2$ and $v_2$ is at most $8k$.
		
		Now construct the following path between $v_1$ and $v_2$: go from $v_1$ to $c_2$ following the first arc, then go from $c_2$ to $c_1$ following the chord, then go from $c_1$ to $v_2$ following the second arc. See \Cref{fig:chord}.
		This path contains all but at most $16k$ edges of the cycle $C$, since $c_i$ and $v_i$ are close to each other on $C$ for each $i \in \{1,2\}$.
		Additionally, this path contains at least $16k$ edges of the chord.
		Hence, the length of the constructed $(v_1,v_2)$-path is at least the length of the cycle $C$.
		This path does not contain $h$, so adding two edges between $v_1$ and $h$ and between $h$ and $v_2$ to it,  yields a cycle of length at least $|V(C)|+2$.
		
\begin{figure}[ht]
 \begin{center}
 \includegraphics[scale=.25]{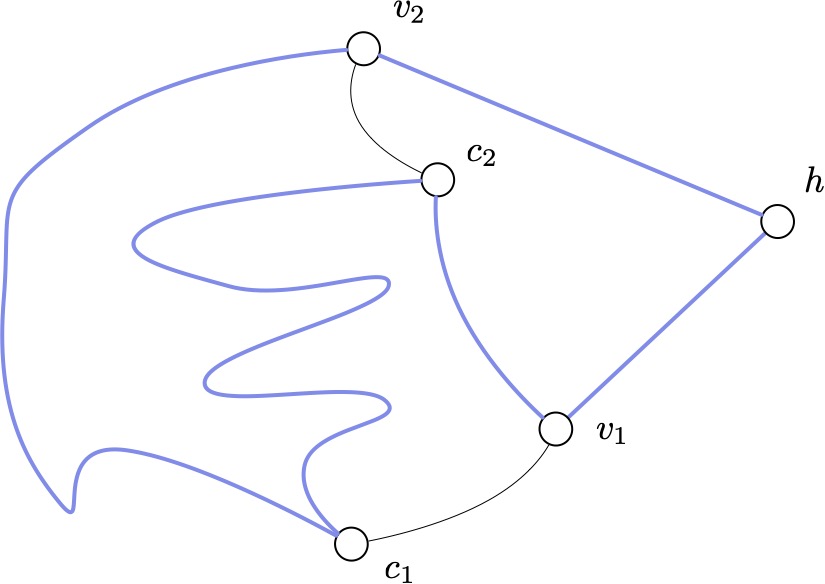}
 \caption{Rerouting through a chord.}\label{fig:chord}
 \end{center}
 \end{figure}		
\end{claimproof}

 \medskip 
 Depending on the number of vertices in a component $H$   of $G-V(C)$, we consider difference cases. We start with the simplest case. 
 
 	\medskip\noindent\textbf{Case 1:} 
	\emph{At least one component $H$ consists of two vertices.} In this case   we can always enlarge $C$ in polynomial time.

Let	 $V(H)=\{h_1,h_2\}$ for $h_1\neq h_2$. Then both $h_1$ and $h_2$ have at least $\delta(G-B)-1$ neighbors in $V(C)$ and are connected by an edge.
	In this case, $C$ can be made longer in polynomial time.
	We formulate this slightly more generally in the following claim.

	\begin{claim}\label{claim:dirac_cycle_two_connected_isolated}
		If there are two distinct vertices  $h_1, h_2 \in V(G)\setminus V(C)$,  each  having at least $\delta(G-B)-1$ neighbors in $V(C)$,  and that are connected by a path in $G-V(C)$, then the length of $C$ can be increased  in polynomial time.
	\end{claim}
	\begin{claimproof}
		Let $S$ be the set of neighbors of $h_1$ and $h_2$ in $V(C)$.
		Let $a$ be the number of the common neighbors of $h_1$ and $h_2$ in $S$.
		Then $|S|\ge 2\delta(G-B)-2-a$ and $S$ splits $C$ into at least $2\max\{a,\delta(G-B)-1\}-a$ arcs. If we have an arc of length $1$, we can always enlarge $C$ by inserting one or both of the $h_i$. Moreover, if one of the endpoints of an arc is a common neighbor of $h_1$ and $h_2$, then the length of this arc should be at least $3$. Indeed, if an arc having a common neighbor of $h_1$ and $h_2$ as its endpoint and is of length less than three, then we can insert a path between $h_1$ and $h_2$ and two boundary edges instead of this arc in $C$; thus  $C$ becomes longer.
	
	Therefore, if  $C$ cannot be enlarged, its length is  at least 
	$2(|S|-a)+3a$. By the conditions of the lemma, we have that  $\delta(G-B)\geq 2k+12$. 
	 If $a\geq  \delta(G-B) -1$, then 
	 \[2(|S|-a)+3a\geq 3a \geq 3(\delta(G-B) -1)> 2\delta(G-B)+k.
	 \]
If $a< \delta(G-B) -1$, then 
 \begin{multline*}2(|S|-a)+3a=2|S|+a \geq   2(2\delta(G-B)-2-a)+a =4\delta(G-B)-a -8 > 2\delta(G-B)+k.\end{multline*}
 
In both cases, we have that  the length of cycle $C$ is more than $ 2\delta(G-B)+k$.
This contradicts our assumption that 		 $|V(C)|< 2\delta(G-B)+k$.	
	\end{claimproof}

The next two cases consider the situation when a component $H$ of $G-V(C)$ contains at least 3 vertices. Then $H$ could be $2$-connected or it contains a cut-vertex.

	\medskip\noindent\textbf{Case 2:} \emph{$H$ is $2$-connected.} We show that either we can enlarge $H$, or $H$ has very specific properties described in 
	\Cref{claim:dirac_cycle_matching}	and \Cref{claim:dirac_cycle_matching_degree}. These properties will be used in handling isolated components and in constructing  \cyclebananadec.

\begin{claim}\label{claim:dirac_cycle_matching}	
	Either the maximum size of a  matching between $V(H)$ and $V(C)$ in $G$ is   two, or $C$ can be enlarged in polynomial time.
	\end{claim}
	\begin{claimproof}	
Since $G$ is 2-connected, the maximum matching size  	between $V(H)$ and $V(C)$ is always at least $2$. 
		
	Suppose first that at most one vertex in $V(H-B)$ has neighbors in $V(C)$.
	If such vertex exists, let $h \in V(H-B)$ be that vertex, otherwise let $h$ be an arbitrary vertex in $H-B$.
	We know that $\delta(H-(B\cup \{h\}))\ge \delta(G-B)-1$, since $H$ is a connected component in $G-V(C)$.
	We now claim that if there is a matching of size at least three between $V(H)$ and $V(C)$ in $G$, then $C$ can be made longer by replacing one of its arcs with a path in $H$.
	By Theorem~\ref{thm:relaxed_st_path}, there is a path of length at least $\delta(H-(B\cup\{h\}))\ge \delta(G-B)-1$ between an arbitrary pair of vertices in $H$.
	The endpoints of the matching in $V(C)$ split $C$ into at least three arcs.
	If at least one of these arcs is of length less than $(\delta(G-B)-1)+2$, it can be replaced with a path in $H$ connecting corresponding endpoints of the matching.
	Hence, if $C$ cannot be made longer, its length is at least $3\delta(G-B)+3$.
	Since $|V(C)|<2\delta(G-B)+k<3\delta(G-B)+3$, we obtain that either $C$ can be made longer or the maximum matching size between $V(H)$ and $V(C)$ in $G$ is two.
	
	Now we assume that at least two vertices in $V(H-B)$ have neighbors in $V(C)$.
	Take the vertices $h_1, h_2 \in V(H-B)$ that have the most and the second most number of neighbors in $V(C)$.
	Denote $n_i=|N_G(h_i)\cap V(C)|$ for each $i\in\{1,2\}$. Thus $n_1\geq n_2$. 
	By Theorem~\ref{thm:relaxed_st_path}, there is a path of length at least $\delta(H-(B\cup \{h_1\}))$ between $h_1$ and $h_2$ in $H$.
	Let  $t=\max\{\delta(H-(B\cup\{h_1\})), 1\}$.
	Note that the path between $h_1$ and $h_2$ is of length at least $t$.
	
	Assume that $\delta(H-(B\cup \{h_1\}))<\delta(G-B)-1-n_2$.
	Then at least one vertex in $V(H-(B\cup \{h_1\}))$ has at most $\delta(G-B)-2-n_2$ neighbors in $V(H-(B\cup \{h_1\}))$.
	Hence, it has at most $\delta(G-B)-1-n_2$ neighbors in $V(H-B)$.
	All other neighbors of this vertex in $V(G-B)$ are from $V(C)$, so this vertex should have at least $n_2+1$ neighbors in $V(C)$.
	This contradicts the choice of $h_2$ and $n_2$.
	Thus, $t\ge \delta(G-B)-1-n_2$, or $n_2\ge \delta(G-B)-t-1$.
	
	Denote by $S$ the set of all neighbors of $h_1$ and $h_2$ in $V(C)$, i.e.\ $S=(N_G(h_1)\cup N_G(h_2))\cap V(C)$.
	Let $a$ be the number of common neighbors of $h_1$ and $h_2$ in $S$, i.e.\ $a=|N_G(h_1)\cap N_G(h_2)\cap S|$.
	Observe that vertices in $S$ split $C$ into $|S|=n_1+n_2-a$ arcs.
	Note that each arc is of length at least two, otherwise we enlarge $C$. Moreover, every arc whose endpoint is a common neighbor of $h_1$ and $h_2$ should have length at least $t+2$, because otherwise  $C$ can be made longer.
	Hence, $|V(C)|\ge 2|S|+at$.
	Since $2\delta(G-B)+k> |V(C)|$, we have that 
\begin{multline*}
			2\delta(G-B)+k>   2(n_1+n_2-a)+at \geq
		4n_2-2a+at \geq 4(\delta(G-B)-t-1)+a(t-2).
\end{multline*}		
		Therefore, 
		\begin{multline*}
			k>2\delta(G-B)-4t-4+12-12+a(t-2)>2\delta(G-B)-4(t-2)+a(t-2)-12,
		\end{multline*}
			and hence 
			\[
			(4-a)(t-2)>2\delta(G-B)-k-12.
\]	

	
In particular, $(4-a)(t-2)>0$.
	If $t=1$, then $a>2\delta(G-B)-k-8$.
	But $|V(C)|\ge 2|S|+at\ge 2a+at=a(t+2)\ge 3a$, so $3a<2\delta(G-B)+k$.
	It follows that $3(2\delta(G-B)-k-8)<2\delta(G-B)+k$, or $4\delta(G-B)<4k+24$, which contradicts the assumptions of the lemma. 
	
	Thus $t\neq 1$.
	Since $t-2\neq 0$, we obtain that $t>2$, and, consequently, $a<4$.
	Then $3(t-2)\ge (4-a)(t-2)>2\delta(G-B)-k-12$, or $3t>2\delta(G-B)-k-6$.
	It yields that $t\geq\frac{1}{2}\delta(G-B)$.
	
	Assume now that there is a matching in $G$ between $V(H)$ and $V(C)$ of size three.
	Let $c_1, c_2, c_3$ be the endpoints of this matching in $V(C)$, and $v_1, v_2, v_3$ be the corresponding endpoints in $V(H)$.
	Without loss of generality, we assume that $v_1=h_2$, as if $h_2 \notin \{v_1,v_2,v_3\}$ we can always change the matching to include the vertex $h_2$.
	Denote by $T$ the set of all neighbors of $v_1, v_2$ and $v_3$ in  $V(C)$, i.e.,\ $T=N_G(\{v_1,v_2,v_3\})\cap V(C)$.
	Note that $|T|\ge |N_G(v_1)\cap V(C)|=n_2\ge \delta(G-B)-t-1$.
	Unless $C$ can be made longer, the vertices of $T$ split $C$ into $|T|$ arcs of length at least two.
	Additionally, at least three arcs (the arcs that are incident to $c_1,c_2,c_3\in T$) should be of length at least $t+2$, as there is a path of length at least $t$ between $v_i$ and $v_j$ in $H$ for any $i\neq j$.
	We obtain that $|V(C)|\ge 2|T|+3t\ge 2\delta(G-B)+t-2>2\delta(G-B)+k$ unless $C$ can be made longer.
		\end{claimproof}
	\begin{claim}\label{claim:dirac_cycle_matching_degree}
	Either between any pair of vertices in $H$	there is a path in $H$ of length at least $\delta(G-B)-2$, or $C$ can be made longer in polynomial time.
	\end{claim}
	\begin{claimproof}
		The proof is identical to the proof of \Cref{claim:eg_path_bic_degree}.
	\end{claimproof}
	
	\medskip\noindent\textbf{Case 3:} \emph{$|V(H)|\ge 3$ and $H$ contains a cut-vertex.}
	
	Since $H$ contains a cut-vertex, it contains at least two leaf-blocks.
	Denote the leaf-blocks of $H$ by $L_1, L_2, \ldots, L_p$ and their respective cut-vertices by $c_1, c_2, \ldots, c_p$, where $p\ge 2$.

	Since $L_i$ is $2$-connected or $|V(L_i)|=2$, we can proceed similarly to Case~2 with $L_i$ and $B\cup \{c_i\}$ instead of $H$ and $B$, and make $C$ longer or conclude that the maximum matching size between $V(L_i)$ and $V(C)$ in $G$ is at most two.
	
	We now assume that for each $i\in[p]$ the maximum matching size between $V(L_i)$ and $V(C)$ is at most two.
	Then for any $i\in[p]$, accordingly to \Cref{claim:dirac_cycle_matching_degree} applied to $L_i$ and $B\cup\{c_i\}$ instead of $H$ and $B$, we obtain that there is a path of length at least $\delta(G-(B\cup\{c_i\})-2\ge\delta(G-B)-3$ between any pair of vertices in $L_i$, if $|V(L_i)|> 2$. 
	
	\begin{claim}\label{claim:dirac_cycle_leaf_block_single_neighbor}
		$\left|\bigcup_{i=1}^p N_G(V(L_i-\{c_i\}))\right|=1$, or $C$ can be made longer in polynomial time.
	\end{claim}
	\begin{claimproof}
		We first show that if there exists $i \in [p]$ with $|V(L_i)|=2$, then $C$ can be made longer in polynomial time.
		
		Assume that there exists $L_i$ with $|V(L_i)|=2$.
		Then $V(L_i)=\{u, c_i\}$ for some vertex $u\neq c_i$.
		As $\gbref{B}{H}=G$, it is true that $u\notin B$.
		Hence, $u$ has at least $\delta(G-B)-1$ neighbors in $V(C)$.
		If $u$ has two consecutive vertices of $C$ as neighbors, then $C$ can be made longer with inserting $u$.
		
		Now take $j \in [p]\setminus\{i\}$ and consider the leaf-block $L_j$.
		If $|V(L_j)|=2$, then $V(L_j)=\{u',c_j\}$, where $u'$ has at least $\delta(G-B)-1$ neighbors in $V(C)$.
		Note that $u$ and $u'$ are connected by a path in $G-V(C)$.
		By \Cref{claim:dirac_cycle_two_connected_isolated}, $C$ can be made longer in polynomial time in this case.
		
		If $|V(L_j)|>2$, then $L_j$ is $2$-connected, so there is a path of length at least $\delta(G-B)-3$ between any pair of vertices in $L_j$.
		Hence, each inner vertex of $L_j$ is connected with $u$ by a path of length at least $\delta(G-B)-2$.
		Take a vertex $u'\in V(L_j-\{c_j\})$ that has a neighbor $v' \in V(C)$.
		By \Cref{claim:dirac_cycle_isolated_close_neighbors}, there is a vertex $v \in V(C)$ that is a neighbor of $u$ and is on a distance at least one and at most $8k$ from $v'$ on $C$.
		We obtain a $(v,v')$-chord of $C$ that is of length at least $\delta(G-B)$ but the distance between $v$ and $v'$ on $C$ is at most $8k<\delta(G-B)$.
		Hence, $C$ can be made longer in polynomial time.
		
		We now assume that $|V(L_i)|\ge 3$ for each $i \in [p]$.
		Then there is a path of length at least $\delta(G-B)-3$ for any pair of vertices in any $L_i$.
		Assume that $\bigcup_{i=1}^p N_G(V(L_i-\{c_i\})) \supseteq \{v_1, v_2\}$, where $v_1, v_2 \in V(C)$ and $v_1 \neq v_2$.
		Then either there exist $i\neq j$ such that $L_i-\{c_i\}$ contains a neighbor of $v_1$ and $L_j-\{c_j\}$ contain a neighbor of $v_2$, or there only exists $i$ such that $L_i-\{c_i\}$ contains both a neighbor of $v_1$ and a neighbor of $v_2$.
		In the latter case, we can pick $j\neq i$ and $v_3\in V(C)$ with $v_3\neq v_1$ or $v_3 \neq v_2$ such that $L_j-\{c_j\}$ contains a neighbor of $v_3$.
		Thus, without loss of generality we assume that $L_1-\{c_1\}$ contains a neighbor $u_1$ of $v_1 \in V(C)$ and $L_2-\{c_2\}$ contains a neighbor $u_2$ of $v_2 \in V(C)$ and $v_1 \neq v_2$.
		
		Observe that there exists a $(u_1,u_2)$-path in $H$ of length at least $2\delta(G-B)-6$.
		Hence, this path can be prolonged to a $(v_1, v_2)$-chord of $C$ of length at least $2\delta(G-B)-4$.
		Note that at least one of $(v_1, v_2)$-arcs of $C$ is of length at most $\delta(G-B)-\frac{k-1}{2}<2\delta(G-B)-4$, so $C$ can be made longer in polynomial time.
	\end{claimproof}

	The following claim shows that $H$ yields at least one long chord of $C$.
	
	\begin{claim}\label{claim:dirac_cycle_separable_long_path}
		Either for any $i \in [p]$ and any $u \in V(L_i - c_i)$, $v \in V(H)\setminus u$, there is a $(u,v)$-path of length at least $\delta(G-B)-2$ in $H$, or $C$ can be made longer in polynomial time.
	\end{claim}
		
	\begin{claimproof}
		Take $i \in [p]$.
		From \Cref{claim:dirac_cycle_leaf_block_single_neighbor} follows that $\delta(L_i-(B\cup\{c_i\}))\ge \delta(G-B)-2$, as each vertex in $V(L_i-c_i)$ has at most one neighbour outside $L_i$.
		By \Cref{thm:relaxed_st_path}, there is a path of length at least $\delta(G-B)-2$ between any pair of vertices inside $L_i$.
		
		Take $u \in V(L_i-c_i)$ and $v \in V(H)\setminus u$.
		If $v \in V(L_i)$, then we are done.
		If $v$ is outside $L_i$, then a path between $u$ and $v$ should go through $c_i$.
		Since $u\neq c_i$, there is a $(u,c_i)$-path of length at least $\delta(G-B)-2$ inside $L_i$.
		Combine this path with any $(c_i,v)$-path outside $L_i$ in $H$ to obtain the required $(u,v)$-path.
	\end{claimproof}
		
	\medskip\noindent\textbf{Case 4:} \emph{At least one component $H$   of $G-V(C)$ consists of one vertex.} In this case we show that either we can enlarge $C$ in polynomial time, or construct a vertex cover of $G-B$ of size at most $\delta(G-B)+2k$.

Let $V(H)=\{h\}$ for some vertex $h\in V(G-B)$.
	All neighbors of $h$ are from $V(C)$, so $h$ has at least $\delta(G-B)$ neighbors in $V(C)$.
	We first claim that if $G-V(C)$ contains both an isolated vertex and some non-isolated connected component, then we can make $C$ longer.
	
	\begin{claim}\label{claim:dirac_cycle_isolated_and_non_isolated}
		Let $H_1$ and $H_2$ be two connected components in $G-V(C)$ with $V(H_i)\not\subseteq B$.
		If $|V(H_1)|=1$ and $|V(H_{2})|\neq 1$, then $C$ can be made longer in polynomial time.
	\end{claim}
	\begin{claimproof}
		We can assume that $|V(H_2)|\ge 3$, so $V(H_2)$ is either $2$-connected or contains a cut-vertex.
		In both of the cases, by \Cref{claim:dirac_cycle_matching_degree} and \Cref{claim:dirac_cycle_separable_long_path},  we can find a chord of $C$ of length at least $\delta(G-B)-2> 16k$ that passes through $H_2$.
		By \Cref{claim:dirac_cycle_isolated_and_chord} , the single vertex of $H_1$ and the chord passing through $H_2$ help making $C$ longer in polynomial time.	
	\end{claimproof}

By \Cref{claim:dirac_cycle_isolated_and_non_isolated}, we can assume that if there is one connected component of $G-V(C)$ which is an isolated vertex, then all other components are also isolated vertices. 

	Our next step is to  to show that if an isolated vertex exists, then we can find a large independent set in $C$ that has no neighbors outside $C$.
	For an isolated vertex $h$ in $G-V(C)$, we define the set of its \emph{$101$-neighbors}.	
	A vertex $v \in V(C)$ is a $101$-neighbor of $h$, if it is not a  neighbor of $h$, i.e., $v \notin N_G(h)$, but both  neighbors of $v$ in $C$ are also the neighbors of $h$.
	In other words, the set of all $101$-neighbors of $h$ is the set of all isolated vertices in $C-N_G(h)$.
	We now claim that if $C$ cannot be enlarged, then $101$-neighbors of a vertex $h$  form an independent set in $C$ and do not have neighbors in $V(G)\setminus V(C)$.
	
	\begin{claim}\label{claim:vczeroone}
		Let $h\notin B$ be an isolated vertex in $G-V(C)$.
		If at least one $101$-neighbor of $h$ on $C$ is not in $B$ and has at least one neighbor in $V(G-V(C)-B)$ or two $101$-neighbors of $h$ are connected by an edge, then $C$ can be made longer in polynomial time.
	\end{claim}
	\begin{claimproof}
		Suppose first that two $101$-neighbors of $h$, say $v_1, v_2 \in V(C)$, are connected by an edge in $G$.
		Let the neighbors of $v_i$ on $C$ be $u_i$ and $w_i$ for $i \in \{1,2\}$.
		Without loss of generality, we assume that the six vertices appear in the order $u_1, v_1, w_1, u_2, v_2, w_2$ when following $C$, and possibly $w_1=u_2$ or $w_2=u_1$.
		Then construct a new cycle as following: $u_1 \rightarrow v_1 \rightarrow v_2 \rightarrow u_2 \rightsquigarrow w_1 \rightarrow h \rightarrow w_2 \rightsquigarrow u_1$, where $\rightarrow$ corresponds to following a single edge in $G$, while $\rightsquigarrow$ corresponds to following an arc of $C$. See \Cref{fig:vc_101}.
		Note that the vertex set of the new cycle is $V(C)\cup\{h\}$, so $C$ is enlarged in this case.
		
\begin{figure}
\begin{center}
\includegraphics[scale=.25]{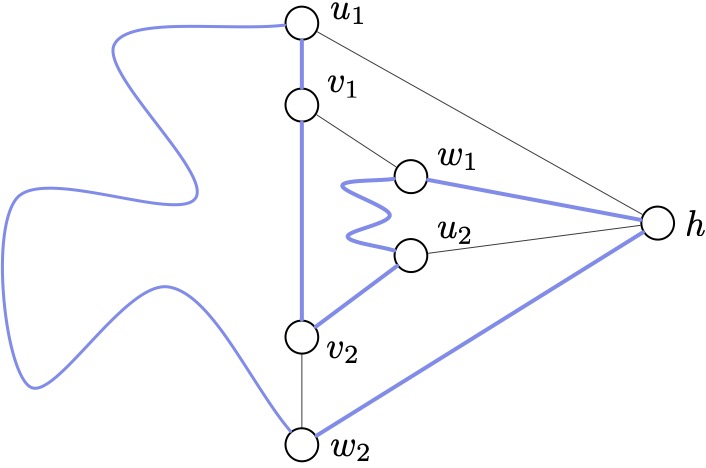}
\caption{Rerouting through adjacent  101-neighbors.}\label{fig:vc_101}
\end{center}
\end{figure}		
		
		Now suppose that a $101$-neighbor of $h$, say $v \in V(C)$,  has a neighbor outside $V(C)$, say $h' \in V(G-V(C)-B)$.
		By \Cref{claim:dirac_cycle_isolated_and_non_isolated}, we can assume that all vertices in $G-V(C)$ are isolated.
		Assume that $h'$ is the only neighbor of $v$ in $V(G-V(C))$.
		Then replace $v$ with $h$ in $C$, so $v$ becomes a vertex outside $C$.
		Then $v$ and $h'$ form a connected component of size two in $G-V(C)$.
		Since $v \notin B$, $v$ has at least $\delta(G-B)-1$ neighbors in $V(C)$.
		By \Cref{claim:dirac_cycle_two_connected_isolated}, $C$ can be made longer in polynomial time.

		If $h'$ is not the only neighbor of $v$, then after the replacement $v$ connects two vertices with at least $\delta(G-B)-1$ neighbors in $V(C)$.
		We can again apply \Cref{claim:dirac_cycle_two_connected_isolated} and make $C$ longer.
	\end{claimproof}

	\medskip\noindent\textbf{Constructing vertex cover.}
	The construction of vertex cover of $G-B$ of size at most $\delta(G-B)+2k$ is possible when there is at least one isolated vertex in $G-V(C)$.
	Take an isolated vertex $h$ in $G-V(C)-B$.
	Denote by $a$ the number of its $101$-neighbors.
	The neighbors of $h$ on $C$ split $C$ into arcs.
	Since each $101$-neighbor corresponds to an arc of length two, and all other arcs are of length at least three, we obtain that $2a+3(\delta(G-B)-a)\ge |V(C)|$, so $a\ge 3\delta(G-B)-|V(C)|>\delta(G-B)-k$.
	Now denote by $S$ the set of all $101$-neighbors of an isolated vertex in $G-V(C)$, so $|S|=a$.
By 	\Cref{claim:vczeroone}, $(V(G)\setminus V(C)) \cup (S\setminus B)$ is an independent set in $G$, so $V(C)\setminus (S \cup B)$ is a vertex cover of $G-B$.
	The size of this vertex cover is at most $(2\delta(G-B)+k-1)-(\delta(G-B)-k+1)<\delta(G-B)+2k$. Finally, the 
	  the desired vertex cover of $G-B$ can be trivially found in polynomial time by taking an isolated component $h$ and constructing the set of its $101$-neighbors.
	
	\medskip\noindent\textbf{Constructing \cyclebananadec.}
	When no isolated vertex is presented in $G-V(C)-B$, then $G-V(C)$ consists of non-empty connected components, apart from components that are completely contained in $B$.
	We show how to construct a \cyclebananadec in this case.
	Before proceeding with claims, it is convenient to define the following notion agreeing with the definition of \cyclebananadec{s}.
	
	\begin{definition}[Dirac layouts]
		We say that a vertex set $X \subseteq V(C)$ is in \emph{Dirac layout} on $C$, if the vertices of $X$ split $C$ into arcs such that two of these arcs are of length at least $\delta(G-B)$.
	\end{definition}

	In what follows, we show that neighbors of $V(G-V(C))$ on $C$ are in Dirac layout, unless $C$ can be made longer.
	We start showing this first for every connected component in $G-V(C)$.
		
	\begin{claim}\label{claim:dirac_cycle_component_dirac_layout}
		Let $H$ be a connected component in $G-V(C)$ with $|V(H)|\ge 3$ and $V(H)\not\subseteq B$.
		If $N_G(V(H))$ is not in Dirac layout on $C$, then $C$ can be made longer in polynomial time.
	\end{claim}
	\begin{claimproof}
		Denote $S=N_G(V(H))$.
		Note that $S\subseteq V(C)$.
		We know that $S$ splits $C$ into $|S|$ arcs.
		Denote $S=\{v_1,v_2,\ldots, v_t\}$, where $v_1,v_2,\ldots,v_t$ are the vertices of $S$ on $C$ in the order when following $C$ in some direction.
		We also assume that $v_{t+1}=v_1$.
		
		Assume first that $H$ is $2$-connected.
		Then assign to each vertex $v_i \in S$ a set of its neighbors in $H$.
		That is, make an assignment $\sigma: S\to 2^{V(H)}$ with $\sigma(v_i)=N_G(v_i)\cap V(H)$.
		As $G$ is $2$-connected, $|\bigcup_{i=1}^t \sigma(v_i)|\ge 2$.
		If for at least one $i\in [t]$ holds $|\sigma(v_i)|=2$, then there exist at least two $j \in [t]$ with $\max\{|\sigma(v_j)|,|\sigma(v_{j+1})|\}\ge 2$.
		For each such $j$, we can pick $h_j \in \sigma(v_j)$ and $h_{j+1} \in \sigma(v_{j+1})$ with $h_j\neq h_{j+1}$.
		Since there is a path of length at least $\delta(G-B)-2$ in $H$, the length of the arc between $s_j$ and $s_{j+1}$ should be at least $\delta(G-B)$.
		Otherwise we can make $C$ longer.
		
		Now consider that $|\sigma(v_i)|=1$ for each $i \in [t]$.
		But not all values of $\sigma(v_i)$ are equal, since their union is of size at least two.
		Then there exist at least two $j \in [t]$ with $\sigma(v_j)\neq \sigma(v_{j+1})$.
		Hence, we can again assign distinct $h_j$ and $h_{j+1}$ and obtain that the $(v_j,v_{j+1})$-arc should be of length at least $\delta(G-B)$.
		
		It is left to consider the case when $H$ is not $2$-connected.
		We again make an assignment $\sigma$, but now this assignment is slightly different and is denoted $\sigma:S\to 2^{\{0,1\}}$.
		If a vertex $v_i$ has a neighbor in $H$ that is an inner vertex of a leaf-block of $H$, then $1 \in \sigma(v_i)$.
		If $v_i$ has a neighbor in $H$ that is not an inner vertex of a leaf-block, put $0 \in \sigma(v_i)$.
		Thus, $\sigma(v_i)$ denotes the set of types of neighbors that $v_i$ has in $V(H)$.
		Note that $\bigcup_{i=1}^t\sigma(v_i)=\{0,1\}$ by \Cref{claim:dirac_cycle_leaf_block_single_neighbor} and $2$-connectivity of $G$.
		Analogously to the $2$-connected case, there are two $j \in [t]$ with $0\in\sigma(v_j)$ and $1\in\sigma(v_{j+1})$ or vice versa.
		Since there is a path of length at least $\delta(G-B)-2$ between any inner leaf-block vertex and any other vertex, we obtain that the arcs between $v_j$ and $v_{j+1}$ should be of length at least $\delta(G-B)$.
	\end{claimproof}

	\begin{claim}
		Assume that $G-V(C)$ contains no isolated vertex.
		Let $X$ be the union of vertex sets of all connected components $H$ in $G-V(C)$ with $V(H)\not\subseteq B$.
		If $N_G(X)$ is not in Dirac layout on $C$, then $C$ can be made longer in polynomial time.
	\end{claim}
	\begin{claimproof}
		Take a connected component $H$ in $G-V(C)$ with $V(H)\not\subseteq B$.
		By \Cref{claim:dirac_cycle_component_dirac_layout}, we assume that $N_G(V(H))$ is in Dirac layout on $C$.
		Hence, the vertices in $N_G(V(H))$ can be covered by two arcs of $C$ of total length at most $|V(C)|-2\delta(G-B)$ and the distance between these arcs on $C$ is at least $\delta(G-B)$.
		Let $u_1,u_2$ and $v_2,v_1$ be the endpoints of these arcs.
		Among all possible ways to choose the arcs we choose the way when the total length of the $(u_1,u_2)$-arc and $(v_2,v_1)$-arc is the minimum possible.
		Hence, $u_1, u_2, v_1, v_2 \in N_G(V(H))$ and the $(u_1,u_2)$-arc and the $(v_2,v_1)$-arc together contain all neighbors of $N_G(V(H))$ on $C$.
		Note that these arcs can be of zero length.
		For example, if $|N_G(V(H))|=2$, then $u_1=u_2$ and $v_1=v_2$, so $N_G(V(H))=\{u_1,v_1\}$.

		We also assume that the order of the vertices on $C$ is $u_1, u_2, v_2, v_1$ when following $C$ in some direction.
		Thus, the chords between $u_1$ and $v_1$ and between $u_2$ and $v_2$ do not intersect graphically but can only coincide in one or two endpoints.
		From the proof of \Cref{claim:dirac_cycle_component_dirac_layout} follows that $H$ yields a $(u_1,v_1)$-chord or a $(u_2,v_2)$-chord of $C$ of length at least $\delta(G-B)$.
		
		Let $S$ be the union of the sets $\{u_1,u_2,v_1,v_2\}$ among all connected components of $G-V(C)$.
		It is easy to see that $S$ is in Dirac layout on $C$ if and only if $S$ is on Dirac layout on $C$.
		It is left to show that $S$ is in Dirac layout on $C$ or $C$ can be made longer in polynomial time.
		
		Consider the vertices in $S$ on $C$.
		They are connected by chords of length at least $\delta(G-B)$ yielded by their connected components.
		If there is a pair of these chords that intersect graphically, then the chords in this pair correspond to distinct connected components of $G-V(C)$.
		Hence, if such a pair exists, we can enlarge $C$ as we did in the proof of \Cref{lemma:dirac_cycle_edge_of_banana}.
		We can now assume that no two chords of $S$ intersect graphically.
		But we also know that no chord splits $C$ into two arcs such that one of them is shorter than $\delta(G-B)$.
		Hence, there are two arcs of length at least $\delta(G-B)$ that do not contain any vertex in $S$ as inner vertex.
		Then $S$ is in Dirac layout on $C$ by definition.
	\end{claimproof}

	The claim shows that $N_G(X)$ can be covered by two arcs of $C$ of total length at most $k-1$ at a distance at least $\delta(G-B)$ between them.
	Let $P_1$ and $P_2$ be these two arcs chosen in the unique way that minimizes their total length.
	It is left for us to show that $P_1$ and $P_2$ induce a \cyclebanana for $C$ and $B$ in $G$.
	
	The first property from the defition of \cyclebanana is satisfied by the way $P_1$ and $P_2$ are constructed.
	It is easy to verify the second property for each connected component in $G-V(C)$: $2$-connected components form \ref{enum:cycle_tunnel_path_bic}-type components and components containing cut vertices form \ref{enum:cycle_tunnel_path_cut_left} and \ref{enum:cycle_tunnel_path_cut_right}-type components of the \cyclebananadec.
	If the matching size conditions are not satisfied for one of these components, then $C$ can be trivially made longer in polynomial time using a long chord yielded by the component.
	
	It is important to verify that the second property holds for all connected components in $G-V(P_1\cup P_2)$.
	Note that a connected component $H$ in $G-V(C)$ with $V(H)\not\subseteq B$ is a connected component in $G-V(P_1\cup P_2)$ as well.
	Connected components that appear in $G-V(C)$ but do not appear in $G-V(P_1\cup P_2)$ are connected components that contain vertices in $V(C)\setminus V(P_1\cup P_2)$.
	
	Note that there is either one or two such connected components, because the vertex set $V(C)\setminus V(P_1\cup P_2)$ is a union of vertex sets of two arcs of $C$.
	If there is just one such connected component $H$, then $V(C)\setminus V(P_1\cup P_2)\subseteq V(H)$.
	We claim that if such $H$ exists in $G-V(P_1\cup P_2)$, then $C$ can be made longer in polynomial time (except in some very specific cases).
	
	\begin{figure}
		\begin{center}
			\ifdefined\STOC
\begin{tikzpicture}[scale=0.85]
\else
\begin{tikzpicture}
\fi
	\tikzstyle{vertex}=[circle,draw,fill,inner sep=1pt]
	
\node [vertex] (v1) at (-0.7,0.6) {};
\node [vertex] (v2) at (-0.7,-0.6) {};
\draw[very thick]  plot[smooth, tension=.7] coordinates {(v1) (-1,0) (v2)};
\node [vertex] (v3) at (6.3,0.5) {};
\node [vertex] at (6.3,-0.5) {};
\draw[very thick]  plot[smooth, tension=.7] coordinates {(v3) (6.5,0) (6.3,-0.5)};
\draw  plot[smooth, tension=.7] coordinates {(-0.7,0.6) (1,2) (4.5,2) (6.3,0.5)};

\draw  plot[smooth, tension=.7] coordinates {(6.3,-0.5) (4.5,-2) (1,-2) (v2)};
\node [vertex] (v4) at (1.7,2.1) {};
\node [vertex] (v5) at (4.3,-2) {};
\draw  plot[smooth, tension=.7] coordinates {(v4) (2,0.7) (2.8,0.6) (3.1,0.1) (3.4,-0.7) (4,-0.9) (v5)};
\node at (-1.4,0.1) {$P_1$};
\node at (6.9,0) {$P_2$};
\node at (2.6,2.4) {$P'$};
\node at (3.1,-2.5) {$P''$};
\node at (1.7,2.3) {$u$};
\node at (4.3,-2.2) {$v$};
\node at (0.4,1.6) {$a'$};
\node at (4.5,1.7) {$b'$};
\node at (1,-1.7) {$a''$};
\node at (5.4,-1.3) {$b''$};

\draw[blue]  plot[smooth, tension=.7] coordinates {(v2) (-0.1,-0.3) (0.7,-0.5) (1.2,-0.2) (2,-0.4) (2.5,-0.2) (3.4,-0.2) (3.9,-0.3) (4.6,-0.1) (5.1,-0.3) (5.7,-0.5) (6.3,-0.5)};
\node at (-0.8,0.8) {$s'$};
\node at (-0.9,-0.8) {$s''$};
\node at (6.4,0.8) {$t'$};
\node at (6.5,-0.7) {$t''$};
\end{tikzpicture}
		\end{center}
		\caption{A schematic picture of an existence of a chord between $P'$ and $P''$ passing through $B$.
		A blue chord represents a chord of $C$ passing through a component of $G-V(C)$.}\label{fig:dirac_cycle_p'_p''}
	\end{figure}
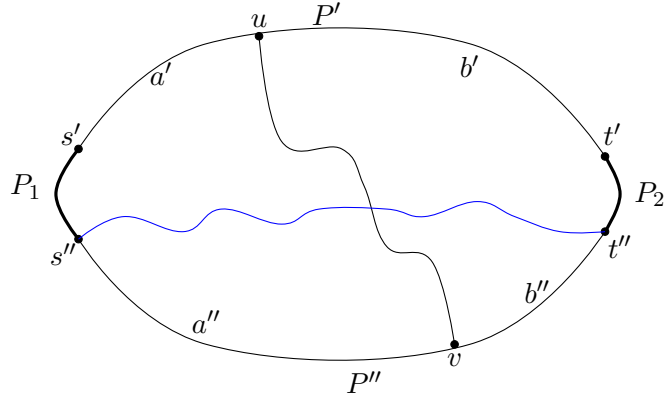
	
	Assume that such $H$ exists.
	Then two arcs of $C$ of length at least $\delta(G-B)$ (denoted by $P'$ and $P''$ in the definition of \cyclebananadec and here) are connected by a chord that can pass internally only through vertices in $B$.
	Note that the length $\delta(G-B)$ does not match the lower bound in the definition of \cyclebananadec{s}.
	This is intentional.
	In one of the cases below, we have to expand the paths $P_1$ and $P_2$ and reduce the length of $P'$ and $P''$ by one or two.
	
	Denote by $s'$ and $t'$ and by $s''$ and $t''$ the endpoints of the arcs $P'$ and $P''$ respectively.
	Note that $V(P'-\{s',t'\})\cup V(P''-\{s'',t''\})\subseteq V(H)$, but $s',t',s'',t'' \notin V(H)$.
	Since $P_1$ and $P_2$ are an $(s',s'')$-arc and an $(t',t'')$-arc of $C$ respectively.
	Hence, the chord connecting $P'$ and $P''$ has endpoints in inner vertices of $P'$ and $P''$.
	For clarity of presentation, we formulate the following intermediate claim.
	
	\begin{claim}
		If there exists a connected component $H$ in $G-V(P_1\cup P_2)$ with $V(C)\setminus V(P_1\cup P_2)\subseteq V(H)$, then either the only chords connecting $P'-\{s',t'\}$ and $P''-\{s'',t''\}$ are between their respective endpoints or $C$ can be made longer in polynomial time. 
	\end{claim}
	\begin{claimproof}
	Let $u \in V(P'-\{s',t'\})$ and $v \in V(P''-\{s'',t''\})$ be the endpoints of this chord.
	Denote by $a'$ and $b'$ the length of the paths that $u$ splits $P'$ into.
	Analogously, by $a''$ and $b''$ denote the length of the paths that $v$ splits $P''$ into, as shown in \Cref{fig:dirac_cycle_p'_p''}.
	If $\max\{a'+b'',a''+b'\}+\delta(G-B)\ge |V(C)|$, then we can find a cycle longer than $C$ in polynomial time using a chord passing though some connected component in $G-V(C)$ and the $(u,v)$-chord of $C$.
	
	Note that if $u$ is the neighbor of $s'$ in $P'$ and $v$ is the neighbor of $t''$ in $P_2$, then $b'\ge \delta(G-B)-1$ and $a''\ge \delta(G-B)-1$ so $a''+b'+\delta(G-B)>|V(C)|$ and $C$ can be made longer.
	The situation when $u$ is the neighbor of $t'$ and $v$ is the neighbor of $s''$ is symmetrical.
	
	We now assume that $a'+b''< \delta(G-B)+k$ and $a''+b'<\delta(G-B)+k$ (and, consequently, $a'+b''\ge \delta(G-B)$ and $a''+b'\ge \delta(G-B)$) for each choice of $u$ and $v$.
	That is, each such $(u,v)$-chord should split $C$ in a way that the difference between $a'+b''$ and $a''+b'$ is at most $k$.	
	
	Consider a fixed $u \in V(P')$.
	Without loss of generality, we assume that $u$ is not the neighbor of $s'$ in $P'$.
	Note that distinct choices of $v \in V(P'')$ provides distinct values of $a''$ and $b''$ with fixed sum.
	Hence, if there are at least $2k+1$ choices of a pair $(u,v)$ for a fixed $u$, there are $2k+1$ different values of $a'+b''$.
	Since the sum of $a', b', a'', b''$ is also fixed, in at least one of these choices the difference between $a'+b''$ and $a''+b'$ is at least $k+1$.
	It follows that if $u$ has at least $2k+1$ neighbors in $V(P'')$, then $C$ can be made longer in polynomial time.
	Note that the same arguments apply to a fixed choice of $v \in V(P'')$.
	
	We have that for each $u \in V(P')$, $|N_G(u)\cap V(P'')|\le 2k$.
	As soon as vertices in $P'-\{s',t'\}$ can have neighbors outside only in $V(P'')$, $V(P_1\cup P_2)$ and $B$, we have that, $\delta(G[V(P'-\{s',t'\})]-B)\ge \delta(G-B)-|V(P_1\cup P_2)|-2k$.
	Since the total length of $P_1$ and $P_2$ is at most $k-1$ and $|V(P_1\cup P_2)|\le k+2$, we have that $\delta(G[V(P'-\{s',t'\})]-B)\ge \delta(G-B)-3k-1$.
	Denote by $H'$ the graph $G[V(P'-\{s',t'\})]-B$.
	As the length of $P'$ is less than $\delta(G-B)+k$, we have that $|V(H')|\le \delta(G-B)+k-2$.
	Hence, $\delta(H')> |V(H')|-4k$.
	On the other hand, $\delta(H')\ge \delta(G-B)-4k \ge 20k$.
	
	We can now apply \Cref{lemma:many_paths} to $H'$ with $p=4k, r=1$ and $\{s_1,t_1\}=\{u',u\}$, where $u'$ is the neighbor of $s'$ in $P'$.
	Note that $u' \neq u$ by our assumption.
	By \Cref{lemma:many_paths}, there is a Hamiltonian $(u',u)$-path in $H'$.
	This path is of length at least $\delta(G-B)-2-|B|$ that is found in polynomial time.
	Hence, we obtain a $(s',u)$-path of length at least $\delta(G-B)-1-|B|$ that contains only vertices of $P'-t'$.
	
	The arguments of constructing a $(s',u)$-path for $P'$ are applicable for constructing a $(t'',v)$-path for $P''$, if $v$ is not the neighbor of $t''$ in $P''$.
	Then we are able to construct a $(t'',v)$-path of length at least $\delta(G-B)-1-|B|$.
	Combine the $(s',u)$-path with $P_1$ and the $(t'',v)$-path and two chords: the $(u,v)$-chord and a $(s'',t'')$-chord of length at least $\delta(G-B)$ (it is depicted in \Cref{fig:dirac_cycle_p'_p''}) to obtain a cycle of length at least $3\delta(G-B)-2|B|-1\ge 2\delta(G-B)+k$.
	The last chord always exists by the construction of $P_1$ and $P_2$.
	
	Note that we only required in the above construction that if $v$ is not the neighbor of $t''$ in $P''$.
	If $v$ is the neighbor of $t''$, then we can consider constructing a $(t'',u)$-path instead of $(s',u)$-path, but only if $u$ is not the neighbor of $t'$ in $P'$.
	The long path between $s''$ and $v$ required for construction is then given by $P''$, and is of length $a''\ge\delta(G-B)-1$.
		\end{claimproof}
	
	We are left with the cases when $u$ and $v$ are simultaneously the neighbors of $t'$ and $t''$ in $P'$ and $P''$ respectively, or the neighbors of $s'$ and $s''$ in $P'$ and $P''$ respectively.
	That is, the cases when $b'=b''=1$ or $a'=a''=1$.
	In these cases, we cannot construct a pair of long paths and combine them with two chords, because we cannot apply \Cref{lemma:many_paths} to both $u$ (e.g.\ to $\{u',u\}$, as $u'=u$) and $v$.
	In other cases, we can make $C$ longer in polynomial time.
	
	We now assume that the only two $(u,v)$-chords between $P'$ and $P''$ can be only a chord between the neighbor of $s'$ in $P'$ and the neighbor of $s''$ in $P''$ and a chord between the neighbor of $t'$ in $P'$ and the neighbor of $t''$ in $P''$.
	In this case, we need to expand $P_1$ or $P_2$ to contain two more vertices.
	If there is a chord between the neighbors of $s'$ and $s''$, expand $P_1$ with two edges so it starts containing these neighbors.
	Analogously, expand $P_2$ if there is a chord between the neighbors of $t'$ and $t''$.
	
	Observe that such expansion of $P_1$ or $P_2$ with two edges does not influence the properties for connected components in $G-V(C)$.
	We now have that $V(P'-(V(P_1)\cup V(P_2)))$ and $V(P''-(V(P_1)\cup V(P_2)))$ belong to distinct connected components in $G-(V(P_1)\cup V(P_2))$.
	The length of $P'$ and $P''$ is now at least $\delta(G-B)-2$ and the total length of $P_1$ and $P_2$ is at most $k+4$.
	The first and the last properties of a \cyclebananadec are satisfied by $P_1$ and $P_2$.
	
	Denote the two connected components of $G-(V(P_1)\cup V(P_2))$ that contain inner vertices of $P'$ and $P''$ by $H'$ and $H''$ respectively.
	We know that$|V(H')|,|V(H'')|\le \delta(G-B)+|B|$ while $\delta(H'-B)\ge \delta(G-B)-|V(P_1)\cup V(P_2)|\ge \delta(G-B)-k-6\ge \frac{1}{2}\delta(G-B)+|B|$.
	Consider the $B$-refinements of $H'$ and $H''$.
	If one of them is not $2$-connected, then it should contain two leaf-blocks each consisting of at least $\frac{1}{2}\delta(G-B)+|B|+1$ vertices.
	Then, the total number of vertices in this component would be $2(\frac{1}{2}\delta(G-B)+|B|+1)-1>\delta(G-B)+|B|$, which is not possible.
	Hence, the $B$-refinements of $H'$ and $H''$ are $2$-connected.
	It is left to prove that they satisfy the properties of \ref{enum:cycle_tunnel_path_bic}-type components of \cyclebananadec{s}.
	
	That is, we have to prove that the maximum matching size between $V(H')$ or $V(H'')$ and $V(P_1)$ or $V(P_2)$ is exactly one after the $B$-refinements.
	Consider that the matching size between $V(H')$ and $V(P_1)$ equals two.
	If the path $P_1$ was not expanded, then there is a long chord of $C$ passing though a component in $G-V(C)$ and connecting $s'$ with a vertex in $P_2$.
	Hence, we can take a cycle of length at least $2\delta(G-B)-2$ combined of this chord, $P''$, $P_1$ and a part of $P_2$.
	Then \Cref{thm:relaxed_st_path} and the maximum matching between $V(H')$ and $V(P_1)$ yields a chord of this cycle with endpoints in $V(P_1)$ of length at least $\delta(H'-B)+2\ge \delta(G-B)-k-4$.
	Since the length of $P_1$ is at most $k$, we can enlarge this cycle and obtain a cycle of length at least $(2\delta(G-B)-2)+(\delta(G-B)-k-4)-k\ge 3\delta(G-B)-2k-6> 2\delta(G-B)+k$.
	
	If $P_1$ was expanded, then there is no long chord of $C$ connecting the common endpoint of $P_1$ and $P'$ with a vertex in $P_2$.
	However, then there exists a short chord of $C$ connecting the endpoints of $P_1$ and passing only through $B$ without visiting $H'$ or $H''$ or any component in $G-V(C)$.
	Also, there is still a chord connecting $s'$ with some vertex in $V(P_2)$, though $s'$ now is not an endpoint of $P_1$ but the neighbor of the common endpoint of $P_1$ and $P'$ in $P_1$.
	If the endpoints of the matching in $V(P_1)$ do not include either $s'$ or the endpoint of $P_1$, we can proceed in the same way as when $P_1$ was not expanded.
	As $V(H')$, the matching and the edge between $s'$ and the endpoint of $P_1$ produce the required long chord of the new cycle.
	
	The case that requires explanation is when the endpoints of the matching are $s'$ and the endpoint of $P_1$.
	Then $V(H')$ only yields an $(s',s')$-chord, which is not appropriate.
	In this case, we have to use the chord between the endpoints of $P_1$ instead of the edge between $s'$ and the endpoint of $P_1$.
	It is easy to see that $V(H')$ together with the matching and this chord provide a long chord between $s'$ and the other endpoint of $P_1$ (the one closer to $t'$).
	Note that this endpoint is different from $s'$, as the length of $P_1$ is at least two.
	
	We have shown that if the matching size between $V(H')$ and $V(P_1)$ is at least two, then we can find a longer cycle in polynomial time.
	The other cases are symmetrical.
	Hence, $H'$ and $H''$ satisfy the properties of \ref{enum:cycle_tunnel_path_bic}-type components.
	This concludes the proof of the lemma.
\end{proof}

 \section{Conclusion}\label{sec:conclusion}
 In this paper, we developed an algorithmic extension of the classical theorem of Dirac. Our main result,  Theorem~\ref{theorem:main}, is   
that \probDC is solvable in $2^{\Oh (k+|B|)} \cdot n^{\Oh (1)}$ time on 2-connected graphs.
An important step in the proof of Theorem~\ref{theorem:main} is  Theorem~\ref{thmEG}:    \probstP is solvable in $2^{\mathcal{O}(k+|B|)}\cdot n^{\mathcal{O}(1)}$  time on $2$-connected graphs. 
In this section we provide lower bounds complementing  Theorems~\ref{theorem:main}
and~\ref{thmEG}, and then conclude with open questions for further research.

\subsection{Tightness of results}
We have already observed that the dependency on $k$ in the running times of Theorems~\ref{theorem:main} and \ref{thmEG} is tight  up to ETH.
Here we show that the dependency on $|B|$ is similarly tight. Additionally, we show that for any $\varepsilon > 0$, it is $\classNP$-hard to find a cycle of length at least $(1 + \varepsilon)2\delta(G)$, meaning that our starting bound of $2\delta(G)$ is tight. We start with the first hardness result.

\begin{theorem}
    Unless ETH fails, there is no algorithm solving \probDC or \probDP in time $2^{o(|B|)} \cdot |V(G)|^{\Oh(1)}$, even when $k = 1$.
    \label{thm:hard_on_B}
\end{theorem}
\begin{proof}
    \begin{figure}[h]
        \centering
\begin{tikzpicture}[scale=0.8]
\tikzstyle{vertex}=[draw, fill, circle, black, minimum size=2,inner sep=0pt]
\tikzstyle{subgraph}=[draw, ellipse, black, minimum height=1cm, minimum width=3cm,inner sep=10pt]
\tikzstyle{connect}=[draw]

\node[vertex] (s) at (-0.5, 0) {};
\node[vertex] (t) at (7.5, 0) {};
\node[subgraph, name path=Hborder] (H) at (3.5, 3) {$H$};
\node[subgraph, name path=K1border] (K1) at (3.5, 0) {$K_{n - 1}$};
\node[subgraph, name path=K2border] (K2) at (3.5, -3) {$K_{n - 1}$};
\path[name path=sHedge1] (s) -- ($ (H.west)!0.35!(H) $);
\draw[connect,name intersections={of=Hborder and sHedge1}]  (s) edge (intersection-1);
\path[name path=sHedge2] (s) -- ($ (H.west)!0.8!(H) $);
\draw[connect,name intersections={of=Hborder and sHedge2}]  (s) edge (intersection-1);
\draw[connect] (s) -- (H.west);
\path[name path=sK2edge1] (s) -- ($ (K2.west)!0.35!(K2) $);
\draw[connect,name intersections={of=K2border and sK2edge1}]  (s) edge (intersection-1);
\path[name path=sK2edge2] (s) -- ($ (K2.west)!0.8!(K2) $);
\draw[connect,name intersections={of=K2border and sK2edge2}]  (s) edge (intersection-1);
\draw[connect] (s) -- (K2.west);
\path[name path=sK1edge1] (s) -- ($ (K1.west)!0.7!(K1.south) $);
\draw[connect,name intersections={of=K1border and sK1edge1}]  (s) edge (intersection-1);
\path[name path=sK1edge2] (s) -- ($ (K1.west)!0.7!(K1.north) $);
\draw[connect,name intersections={of=K1border and sK1edge2}]  (s) edge (intersection-1);
\draw[connect] (s) -- (K1.west);

\path[name path=tHedge1] (t) -- ($ (H.east)!0.35!(H) $);
\draw[connect,name intersections={of=Hborder and tHedge1}]  (t) edge (intersection-1);
\path[name path=tHedge2] (t) -- ($ (H.east)!0.8!(H) $);
\draw[connect,name intersections={of=Hborder and tHedge2}]  (t) edge (intersection-1);
\draw[connect] (t) -- (H.east);
\path[name path=tK2edge1] (t) -- ($ (K2.east)!0.35!(K2) $);
\draw[connect,name intersections={of=K2border and tK2edge1}]  (t) edge (intersection-1);
\path[name path=tK2edge2] (t) -- ($ (K2.east)!0.8!(K2) $);
\draw[connect,name intersections={of=K2border and tK2edge2}]  (t) edge (intersection-1);
\draw[connect] (t) -- (K2.east);
\path[name path=tK1edge1] (t) -- ($ (K1.east)!0.7!(K1.south) $);
\draw[connect,name intersections={of=K1border and tK1edge1}]  (t) edge (intersection-1);
\path[name path=tK1edge2] (t) -- ($ (K1.east)!0.7!(K1.north) $);
\draw[connect,name intersections={of=K1border and tK1edge2}]  (t) edge (intersection-1);
\draw[connect] (t) -- (K1.east);

\node at (-1, 0) {$s$};
\node at (8, 0) {$t$};
\end{tikzpicture}%
        \caption{An illustration to the hardness reduction in Theorem~\ref{thm:hard_on_B}, from \textsc{Hamiltonian Path} to \probDC. The graph $H$ is the starting \textsc{Hamiltonian Path} instance. The reduction to \probDP looks similarly, only without the vertex $t$.}
        \label{fig:hardnessB}
    \end{figure}
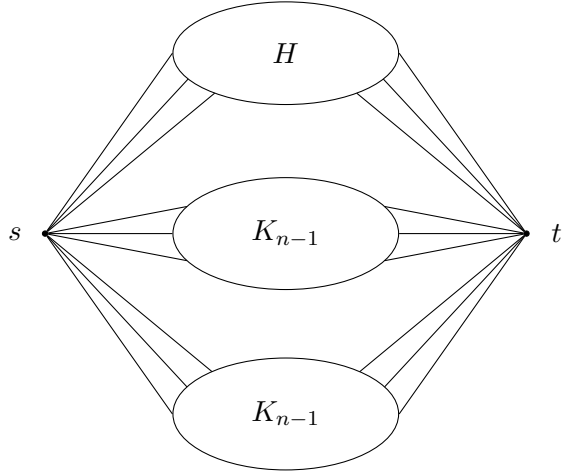
    First, we show a reduction from \textsc{Hamiltonian Path} to \probDC. Consider an instance $H$ of \textsc{Hamiltonian Path}, let $n = |V(H)|$. Take a disjoint union of $H$ and two disjoint copies of $K_{n - 1}$, the clique on $(n - 1)$ vertices. Add two additional vertices $s$ and $t$ that are adjacent to all previously listed vertices (but not to each other). This finishes the description of the graph $G$ that our reduction constructs from $H$, see Figure~\ref{fig:hardnessB} for the illustration. Finally, set $B$ to $V(H) \subset V(G)$, and $k$ to one. Observe that $2\delta(G - B)$ and $|V(G - B)|$ are both equal to $2n$. Our aim is now to show that $H$ has a Hamiltonian path if and only if $G$ has a cycle of length at least $2n + 1 = \min\{2\delta(G-B), |V(G)|-|B|\}+k$.

    In the forward direction, if there is a Hamiltonian path $P$ in $H$, consider a Hamiltonian path $P'$ in one of the $K_{n - 1}$ components. Connect $P$ and $P'$ in a cycle by going through the vertices $s$ and $t$. This results in a cycle of length $|V(G)| + |V(K_{n - 1})| + 2 = 2n + 1$.

    In the other direction, let $C$ be a cycle of length at least $2n + 1$ in $G$. Since $|V(G - B)| = 2n$, $C$ necessarily intersects $B$, and since $|B| = n$, $C$ also intersects $V(G - B)$. Since in $G - \{s, t\}$ the set $B$ is disconnected from the rest of the graph, the cycle $C$ necessarily enters $B$ from $s$ and exits via $t$. The two $K_{n - 1}$ copies are also disconnected in $G - \{s, t\}$, thus $C$ intersects exactly one of the cliques. Thus, $C$ has at most $n + 1$ vertices outside of $B$. Since $|C| = 2n + 1$ and $|B| = n$, $C$ must traverse all vertices of $B$. Since $C$ induces a path on $B$, this path is also a Hamiltonian path in $H$. This finishes the proof of correctness of the reduction.

    Finally, following the reduction above, a $2^{o(|B|)} \cdot |V(G)|^{\Oh(1)}$-time algorithm for \probDC would immeidately imply a 
    $2^{o(n)}$-time algorithm for \textsc{Hamiltonian Path} since $|B| = n$ and $|V(G)| = \Oh(n)$, and the existence of the latter would contradict ETH.

    For \probDP, the reduction follows a similar idea. From an instance $H$ of \textsc{Hamiltonian Path}, construct a graph $G$ as follows.
Take a disjoint union of $H$ and two disjoint copies of $K_{n - 1}$, and add an additional apex vertex $s$. Set $B$ to be $V(H) \subset V(G)$, and  set $k$ to one. Clearly, $2 \delta(G - B)  = |V(G)| - |B| - 1 = 2n - 2$. If there is a Hamiltonian path in $H$, it extends to a path of length $2n - 1$ in $G$ by continuing through $s$ into one of the cliques, as it is always possible to traverse through all vertices of the clique. On the other hand, if there is a path $P$ of length at least $2n - 1$ in $G$, it necessarily goes from $B$ to $V(G) \setminus B$, since $|B| = n$ and $|V(G) \setminus B| = 2n - 1$. Such a path can only go through $s$ to one of the cliques while completely avoiding the other, since $s$ is an articulation point. Thus, outside of $B$ the path $P$ visits at most $n$ vertices, and since a path of length at least $2n - 1$ has to visit at least $2n$ distinct vertices, $P$ necessarily traverses through all vertices of $V(H)$, yielding a Hamiltonian path in $H$. This finishes the proof for \probDP.
\end{proof}

Next, we show that the bound $2\delta(G)$ cannot be improved unless $\classP=\classNP$ by proving the following theorem.

 \begin{theorem}\label{thm:tightness}
 For every positive $\varepsilon<1$, it is \classNP-complete to decide whether 
 \begin{itemize}
 \item[(a)] a $2$-connected graph $G$ with two given vertices $s$ and $t$ has an $(s,t)$-path of length at least $(1+\varepsilon)\delta(G)$;
 \item[(b)]  a $2$-connected graph $G$ has a cycle of length at least $(2+\varepsilon)\delta(G)$. 
 \end{itemize}
 \end{theorem}
 
 \begin{proof}
 Both claims are shown by reduction from the classical \textsc{Hamiltonian Path} problem that is well-known to be \classNP-complete~\cite{GareyJ79}. Both reductions exploit the same idea. We first show the claim for an $(s,t)$-path and then explain how to modify the reduction for the second claim. 
 
 \begin{figure}[ht]
\centering
\scalebox{0.6}{
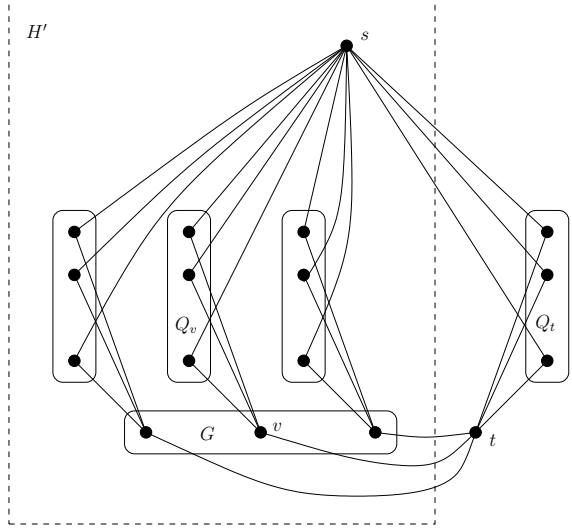}
\caption{Construction of $H$ and $H'$.}\label{fig:hard}
\end{figure} 
 
Let $0<\varepsilon<1$ and let $G$ be an $n$-vertex  graph with $n\geq 2$. 
We select a positive integer $p$ such that $\lceil\varepsilon(p+1)\rceil=n$. Clearly, such an integer exists, because $\varepsilon<1$ and $n\geq 2$. Then we construct the following graph $H$ (see Figure~\ref{fig:hard}).
\begin{itemize}
\item Construct a copy of $G$.
\item Construct a vertex $t$ and make it adjacent to every vertex of $G$.
\item Construct a vertex $s$.
\item For every vertex $v\in V(G)\cup \{t\}$, construct a clique $Q_v$ with $p$ vertices and make the vertices of $Q_v$ adjacent to $v$ and $s$.
\end{itemize}
Notice that $H$ is 2-connected and $\delta(H)=p+1$. We claim that $G$ is has a Hamiltonian path if and only if $H$ has an $(s,t)$-path of length at least $(1+\varepsilon)\delta(H)$.

In one direction, let $P$ be a Hamiltonian path in $G$ and denote by $x$ and $y$ its endpoints.  Because $Q_x$ is a clique, $H$ has an $(s,x)$-path $R$ with $V(R)=Q_x\cup\{s,x\}$. That is, $R$ is a Hamiltonian path in $H[Q_x\cup\{s,x\}]$. Consider path $P'$ obtained by   concatenating  $R$, $P$, and $yt$. Then  $P'$ is an $(s,t)$-path in $H$. Observe that the length of $P'$ is 
\begin{multline*}
(p+1)+(n-1)+1=p+n+1=p+1+\lceil\varepsilon(p+1)\rceil\geq (1+\varepsilon)(p+1)=(1+\varepsilon)\delta(H)
\end{multline*}
as required. 

For the opposite direction, assume that $P'$ is an $(s,t)$-path in $H$ of length at least $(1+\varepsilon)\delta(H)$. Then the length of $P'$ is at least 
\begin{equation*}
\lceil(1+\varepsilon)\delta(H)\rceil=\delta(H)+\lceil\varepsilon \delta(G)\rceil=(p+1)+\lceil\varepsilon(p+1)\rceil=p+1+n.
\end{equation*}
By the construction of $H$, $P'$ is the  concatenation of paths $R$ and $S$ such that $R$ is an $(s,v)$-path for some $v\in V(G)\cup\{t\}$ where $V(R)\subseteq V(Q_v)\cup\{s,v\}$ and $S$ is a $(v,t)$-path with $V(S)\subseteq V(G)\cup\{t\}$. The length of $R$ is at most $p+1=\delta(H)$. Therefore,
the length of $S$ is at least $n$. Consider the path $P$ obtained from $S$ by  deleting $t$. We have that $V(P)\subseteq V(G)$ and the length of $P$ is at least $n-1$. We obtain that $P$ is a Hamiltonian path in $G$. This concludes the proof of the first claim.

The proof of (b)  is similar.  Let $0<\varepsilon<1$ and let $G$ be an $n$-vertex connected graph with $n\geq 3$. Now we select a positive integer $p$ such that $\lceil\varepsilon(p+1)\rceil=n-1$. 
 We construct  graph $H'$ that is, in fact, the graph obtained from $H$ constructed above by deleting $t$ and the vertices of $Q_t$ (see Figure~\ref{fig:hard}). Formally, $H'$ is constructed as follows.
 \begin{itemize}
\item Construct a copy of $G$.
\item Construct a vertex $s$.
\item For every vertex $v\in V(G)$, construct a clique $Q_v$ with $p$ vertices and make the vertices of $Q_v$ adjacent to $v$ and $s$.
\end{itemize}
Because $G$ is a connected graph with at least three vertices, $\delta(H')=p+1$. Because $G$ is connected, $H'$ is 2-connected. 
 We claim that $G$ has a Hamiltonian path if and only if $H'$ has a cycle of length at least  $(2+\varepsilon)\delta(H)$. 
 
 Suppose that $P$ is a Hamiltonian path in $G$ and let $x$ and $y$ be its endpoints. Note that since $G$ has at least two vertices, $x\neq y$. Because $Q_x$ and $Q_y$ are cliques, $H$ has an $(s,x)$-path $R_x$ with $V(R_x)=Q_x\cup\{s,x\}$ and a $(y,s)$-path $R_y$ with 
 $V(R_y)=Q_y\cup\{s,y\}$. Observe that the concatenation of $R_x$, $P$, and $R_y$ is a cycle. Denote this cycle by $C$. The length of $C$ is
 \begin{multline*}
(p+1)+(n-1)+(p+1)=2(p+1)+n-1=2(p+1)+\lceil\varepsilon(p+1)\rceil\geq (2+\varepsilon)(p+1)=(1+\varepsilon)\delta(H).
\end{multline*} 
 
 Finally, let $C$ be a cycle of $G$ of length at least  $(2+\varepsilon)\delta(H)$. Then the length of $C$ is at least 
\begin{equation*}
\lceil(2+\varepsilon)\delta(H)\rceil=2\delta(H)+\lceil\varepsilon \delta(G)\rceil=2(p+1)+\lceil\varepsilon(p+1)\rceil=2(p+1)+n-1.
\end{equation*}
 Suppose that $s\notin V(C)$. Then, by the construction of $H'$, either $C$ is a cycle in $H'[Q_v\cup\{v\}]$ for some $v\in V(G)$ or $C$ is a cycle of $G$. In the first case the length of $C$ is at most $|Q_v|+1=p+1$, and in the second case the length of $C$ is at most $n$. In both cases, we have that the length of $C$ is strictly less that $2(p+1)+n-1$. This implies that $s\in V(C)$. If $|V(C)\cap V(G)|\leq 1$, then $V(C)\subseteq Q_v\cup\{s,v\}$ for some $v\in V(G)$. However, $|V(C)|\leq p+2<2(p+1)+n-1$ in this case. Hence,   $|V(C)\cap V(G)|\geq 2$. Then the construction of $H'$ implies that $|V(C)\cap V(G)|=2$. Let $\{x,y\}=V(C)\cap V(G)$. It is easy to verify that $C$ can be seen as the concatenation of three paths $R_x$, $P$, and $R_y$, where $R_x$ is an $(s,x)$-path with $V(R_x)\subseteq Q_x\cup\{s,x\}$, $P$ is an $(x,y)$-path in $G$, and $R_y$ is a $(y,s)$-path with $V(R_y)\subseteq Q_y\cup \{s,y\}$. The length of  $R_x$ and the length of $R_y$ is at most $p+1$. This means, that the length of $P$ is at least $n-1$. Therefore, $P$ is a Hamiltonian path in $G$. This concludes the proof.
 \end{proof}
 
For simplicity,  we proved Theorem~\ref{thm:tightness} for the case when $\varepsilon<1$  but let us remark that the claim also holds for $\varepsilon\geq 1$. Moreover, it can be assumed that $\varepsilon$ not a constant but an appropriate function of $\delta(G)$ like $\varepsilon(\delta)=\delta^c$ for some constant $c>-1$.  

\subsection{Open questions}
Dirac's theorem is the first fundamental result in Extremal Hamiltonian Graph Theory. The area contains many deep and interesting theorems but it  remains largely unexplored from the algorithmic perspective. 
Here we present several open questions hoping that these questions would trigger further research in this fascinating area. 

%

%
%
%
%
Our first open question concerns the problem of finding a cycle containing a specified set of vertices.
 The study of this problem can be traced back to another fundamental theorem of Dirac from 1960s about the existence of a cycle in $h$-connected graph passing through a given set of $h$ vertices \cite{Dirac1960}.  According to Kawarabayashi \cite{Kawarabayashi08}  \emph{``...cycles through a vertex set or an edge set are one of central topics in all of graph theory."}   Such type of problems  have been a popular and important topic in algorithms as well. See, e.g.,   Bj{\"{o}}rklund, Husfeldt and Taslaman~\cite{BjorklundHT12}  and Wahlstr{\"{o}}m~\cite{Wahlstrom13}, and Kawarabayashi   \cite{Kawarabayashi08}.

%

In Extremal Hamiltonian Graph Theory, the following theorem of 
  Egawa, Glas, and Locke~\cite{Egawa1991} is well-known.
  
\begin{theorem}[\cite{Egawa1991}]
Let $G$ be an $h$-connected graph, $h\geq 2$, with minimum degree $d$, and at least $2d-1$ vertices. Let $X$ be a set of $h$ vertices of $G$. Then $G$ has a cycle $C$ of length at least $2d$ such that every vertex of $X$ is on $C$.
\end{theorem}

This brings us to the following algorithmic problem. 

\begin{problem}[\textbf{Cycle above Egawa, Glas, and Locke condition}]\label{prob:through}
Given an $h$-connected graph $G$, a set of vertices $X\subseteq V(G)$ of size $h$, and a nonnegative integer $k$, how difficult is to decide whether $G$ has a cycle  of length at least $2\delta(G)+k$ containing every vertex of $X$?
 \end{problem}
 
This question is open even for   $k=1$.


\medskip
The further questions are form the area of directed graphs; we refer to the book of Bang-Jensen and Gutin~\cite{BangJensenG09}
  and the survey of Bermond and Thomassen~\cite{BermondT81} for extremal theorems for directed graphs. In particular, the classical result of Ghouila-Houri~\cite{GhouilaHouri1960} from 1960, generalizes  Theorem~\ref{thm:diracs}. Recall that a digraph $D$ is \emph{strong} if for every two vertices $u$ and $v$, $D$ has directed $(u,v)$ and  $(v,u)$-path, and the degree $\deg_D(v)$ of a vertex $v$ is the sum of its \emph{in-degree} $\deg_D^-(v)$ and \emph{out-degree} $\deg_D^+(v)$. 
  
\begin{theorem}[\cite{GhouilaHouri1960}]\label{thm:GH}
If for every vertex $v$ of a 
 strong digraph $D$ with $n$ vertices   $\deg_D(v)\geq n$, then $D$  has a Hamiltonian cycle.
\end{theorem}

The following question is the variant of the question  discussed by Jansen,  Kozma and Nederlof in~\cite{DBLP:conf/wg/Jansen0N19} for undirected graphs.

\begin{problem}[\textbf{Cycle above Ghouila-Houri condition}]\label{prob:GH}
Given an $n$-vertex strong digraph $D$ and a nonnegative integer $k$ such that at least $n-k$ vertices have degree at least $n$, how difficult is to decide whether $D$ is Hamiltonian? 
\end{problem}
Again, the simplest variant --- whether there is a polynomial time algorithm for  for $k=1$ --- is open. We also do not know the complexity of the problem when 
 every vertex has degree at least $n-k$.

\bibliographystyle{alpha}
\bibliography{ref}

\end{document}